\documentclass{article}[12pt]

\usepackage{amssymb}
\usepackage{amsmath}
\usepackage{enumerate}
\usepackage{latexsym}
\usepackage{mathrsfs}
\usepackage{amsthm}
\usepackage{verbatim}
\usepackage{graphicx}
\usepackage{epstopdf}
\usepackage{epsfig}
\usepackage{yhmath}

\usepackage{color}
\usepackage[dvipsnames]{xcolor}
\usepackage{calligra}

\usepackage{subfig}
\usepackage{float}
\usepackage{titlesec}
\usepackage{bm}
\usepackage[colorlinks,linkcolor=blue,anchorcolor=green,citecolor=red]{hyperref}
\usepackage[T1]{fontenc}

\usepackage{charter}
\usepackage{pythonhighlight}
\usepackage{fancyhdr}
\usepackage{amscd}
\usepackage{cleveref}

\usepackage{caption}
\usepackage{algorithm} 
\usepackage{algorithmic} 

\usepackage{physics}

\newtheorem{theorem}{Theorem}[section]
\newtheorem{corollary}[theorem]{Corollary}
\newtheorem{lemma}[theorem]{Lemma}
\newtheorem{proposition}[theorem]{Proposition}

\newtheorem{remark}[theorem]{Remark}
\newtheorem{definition}[theorem]{Definition}

\usepackage{apptools}
\usepackage{amsfonts}
\newcommand{\ff}{\ensuremath{\mathfrak{f}}}

\makeatletter
\@addtoreset{equation}{section}
\makeatother

\makeatletter 
\@addtoreset{equation}{section}
\makeatother  
\renewcommand\theequation{\oldstylenums{\thesection}%
	.\oldstylenums{\arabic{equation}}}

\usepackage{geometry}
\usepackage{listings}
\usepackage{setspace}

\allowdisplaybreaks 

\newcommand{\edit}[1]{{#1}} 

\newcommand{\cconst}{\frac{4\pi}{\sqrt{3}}}
\newcommand{\Icell}{\left[-\frac{2\pi}{\sqrt{3}}, \frac{2\pi}{\sqrt{3}}\right]}

\newcommand{\LtwoBrill}{L^2_k\!\left(\Icell\right)}

\newcommand{\Ucell}{\mathscr{C}_{m,n}}
\newcommand{\Cell}{\mathscr{C}}

\begin{document}
	
\title{Pseudo-magnetism in a strained discrete honeycomb lattice}
\author{Xuenan Li\thanks{Department of Applied Physics and Applied Mathematics, Columbia University, xl3383@columbia.edu} \text{ and }Michael I. Weinstein\thanks{Department of Applied Physics and Applied Mathematics,  and Department of Mathematics, Columbia University, miw2103@columbia.edu}}

\maketitle

\begin{abstract}
Slowly varying nonuniform strains of non-magnetic wave propagating media with honeycomb symmetry induce an effective- or {\it pseudo-magnetic} field, a phenomenon observed first in graphene \cite{castro2009electronic}, and later in photonic crystals, e.g. \cite{barczyk2024observation,barsukova2024direct,guglielmon2021landau}  and other physical settings. Starting with a discrete nearest-neighbor tight-binding model of a non-uniformly strained honeycomb  medium, we derive the continuum effective magnetic Dirac Hamiltonian governing the envelope dynamics of wave packets, which are spectrally localized near a Dirac point (conical band degeneracy) of the unperturbed honeycomb.  For  unidirectional deformations of bounded gradient, which preserve translation invariance along the ``armchair'' direction, we prove the existence of time-harmonic states which are plane-wave like (pseudo-periodic) along the armchair direction and exponentially localized transverse to it. We also obtain the leading order multi-scale structure of such modes for small deformation gradients. Their transverse localization are determined by the eigenstates of a one dimensional effective Dirac Hamiltonian. Our rigorous results apply to deformations which induce an approximate  perpendicular constant pseudo-magnetic field (Landau gauge), and yield states
with nearly flat band (Landau level) spectrum and hence very high density of states. In contrast, the analogous deformation which preserves translations in the zigzag direction induces no such localization. 
Corroborating numerical simulations for the different deformation types are presented. 
\end{abstract}


\tableofcontents

\section{Introduction}

Graphene and its synthetic variants are spatially periodic two-dimensional  media, whose material properties have the symmetries of a honeycomb tiling of $\mathbb R^2$. Wave propagation in these media is governed by wave equations, continuum (Schr\"odinger, Maxwell,\dots) or discrete, which are invariant under these symmetries.  Central to the remarkable properties of honeycomb media is the presence of {\it Dirac points}.  These are energy/quasi-momentum pairs at which two consecutive (Floquet-Bloch) dispersion surfaces touch conically. Wave packet initial conditions  which are spectrally localized to a size $\delta$ neighborhood of a Dirac point, give rise to a time-evolving wave packet whose envelope, on a time scale $\delta^{-2+\epsilon}$, is governed by a 2D effective Dirac equation. Thus, such quasiparticle states  evolve as though they are massless relativistic particles.
 A mathematically rigorous analysis in the context of the underlying continuum Schr\"{o}dinger equation with a honeycomb lattice potential was given in   \cite{fefferman2014wave}; see also \cite{lee2019elliptic}. 

Shortly after the graphene's discovery, it was observed that slowly varying non-uniform deformations of graphene give rise to a {\it pseudo-magnetic effect}
\cite{castro2009electronic}. Based on a tight-binding (discrete) models, it was deduced that the wave packet envelope satisfies an effective magnetic Dirac equation. Thus, in the absence of any external magnetic field, quasi-particles modeled by wave packets move as though they are charged particles in a  magnetic field. Further, when the deformation is chosen to induce an effective magnetic field  which is constant and perpendicular to the two-dimensional plane of the material (schematic of Figure \ref{fig:honeycomb-intro}(c)), the spectrum of waves consists of highly degenerate states, which arise as  ``Landau levels'' in effective magnetic Dirac Hamiltonian (schematic of Figure \ref{fig:honeycomb-intro}(d)).  The significance of Landau level and other ``flat band'' spectra is that it can be leveraged to enhance strong interactions and nonlinear effects. The rich phenomena associated with Landau Levels in condensed matter physics, have inspired investigations, in diverse physical settings, to realize the pseudo-magnetic effect.

A particular direction concerns photonics and, in particular, photonic crystals. While the tight-binding models explain 
pseudo-magnetism in condensed matter settings, this approximation does not valid in photonic settings. In \cite{guglielmon2021landau} the effective dynamics was studied in the setting of general continuum wave equations, and where both magnetic and effective electric potentials arise in the effective Dirac Hamiltonian. Experimental observation of photonic Landau levels in 
 3D photonic crystal was observed in \cite{barczyk2024observation,barsukova2024direct}. Further, an extension of  the 2D theory to 3D gives an excellent agreement with experiments \cite{barsukova2024direct}.

Our goal in this article, is to present a rigorous derivation of wave localization due to a non-uniform deformation induced pseudo-magnetic field.  We study this question in the context of class of  tight-binding with general slowly varying hopping coefficients. Our results, stated more precisely below, apply to unidirectional deformations with bounded gradients. In this setting  we prove the validity of the magnetic Dirac operator, explore (Landau gauge type) deformations which give rise to approximate Landau level spectra, as well as deformations that produce other behaviors. In particular, we prove that the localized states of an effective asymptotic Dirac operator, seed localized states of the underlying tight-binding Hamiltonian. Finally, we provide numerical corroboration of our results. 

The family of tight-binding models studied in this article captures the essential physical behavior in the condensed matter setting; see e.g. \cite{guinea2010energy,settnes2016pseudomagnetic}.  As discussed in \cite{guglielmon2021landau}, tight-binding 
  is not the appropriate approximation in, for example, photonic crystals. 
   We remark however that our discrete model  captures the main properties of waves which are spectrally localized near the Dirac point for both Schr\"odinger and Helmholtz equations, and is more amenable to full rigorous analysis since one need not control arbitrarily high energies; indeed the spectrum of the tight-binding Hamiltonian is bounded. Our plan is to extend the present work to  continuum models in a future work.


Before summarizing our main results, we next present a more precise formulation of the tight-binding model.

\begin{figure}[!htb]
    \centering
    \subfloat[]{
		\includegraphics[height=1.2in]{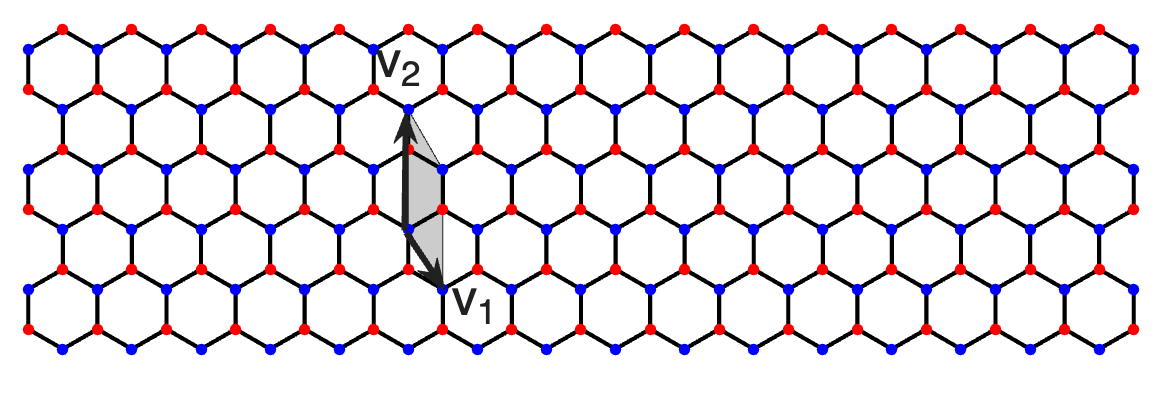}\label{fig:honeycomb-ac-undeformed}
	}\hfil
    \subfloat[]{
		\includegraphics[height=1.3in]{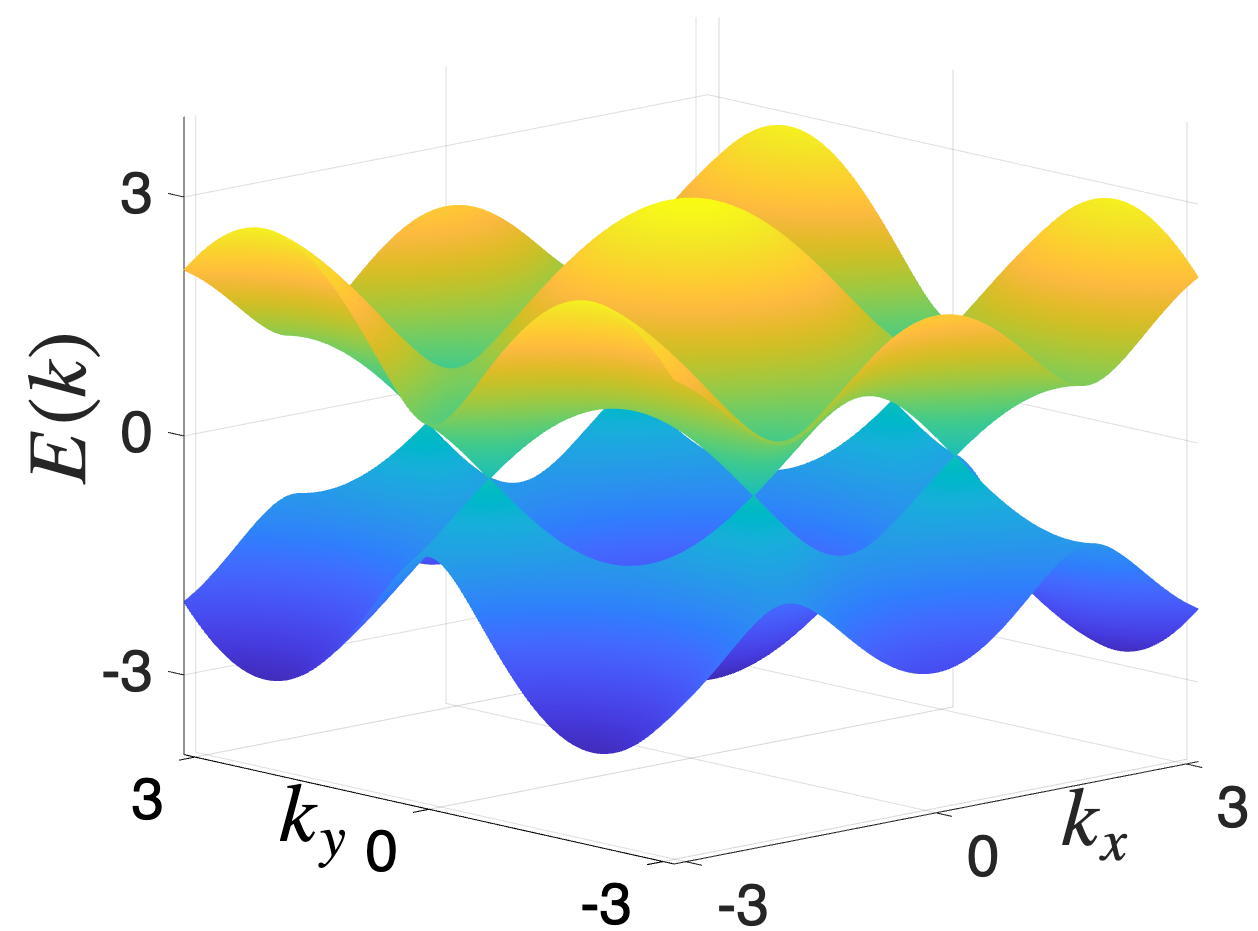}\label{fig:honeycomb-2d-band}
	}\\
    \subfloat[]{
		\includegraphics[height=1.3in]{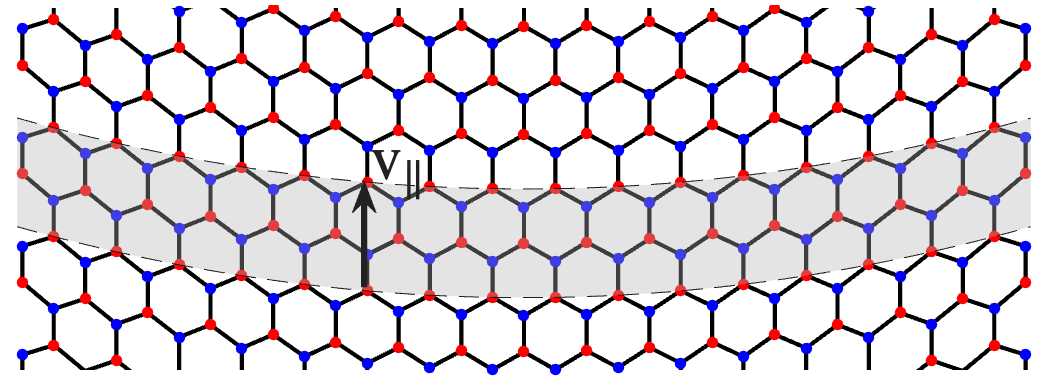}
	}\hfil
    \subfloat[]{
		\includegraphics[height=1.3in]{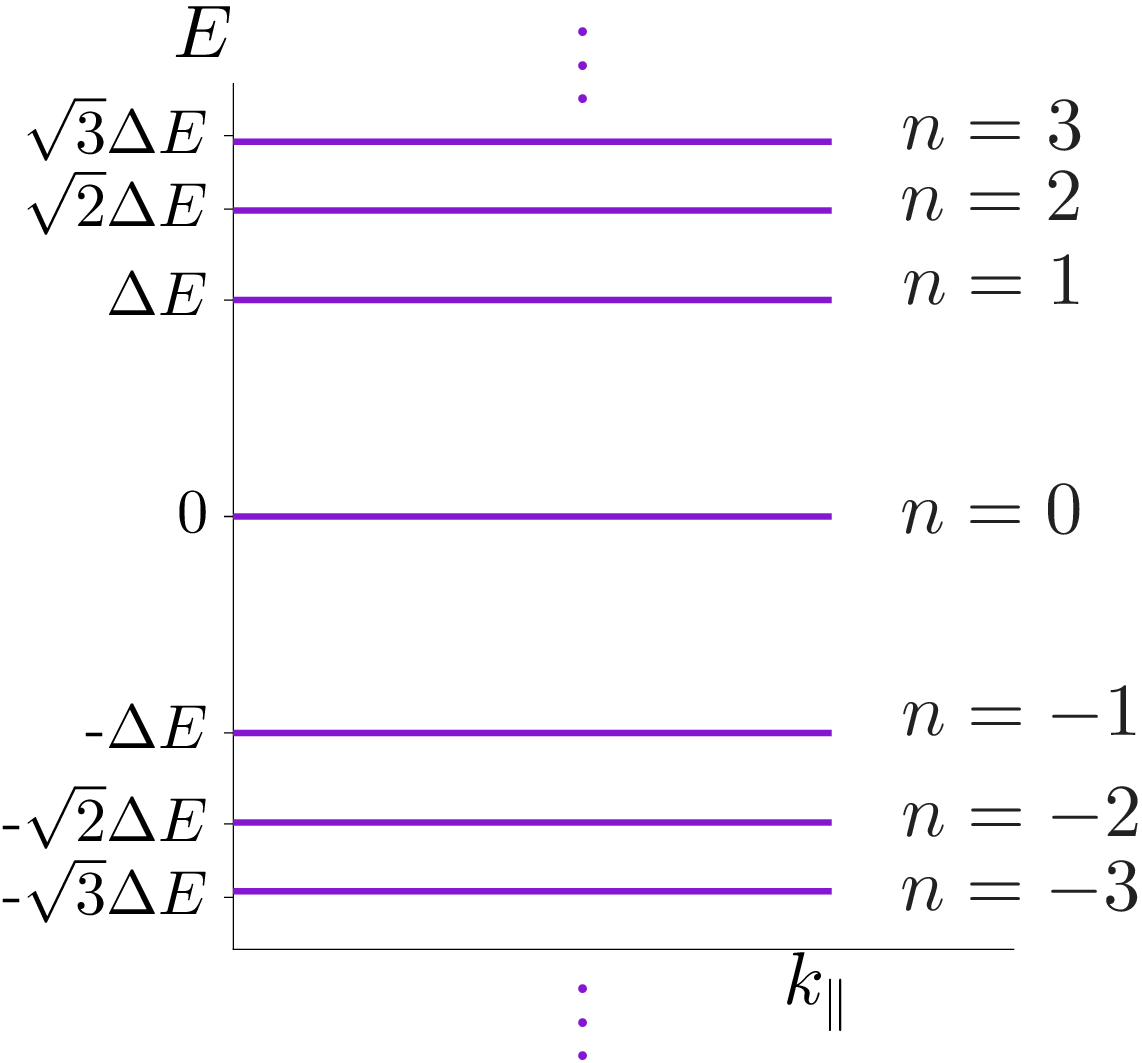}
	}
	\caption{The undeformed and deformed honeycomb and their band structures: (a) the undeformed honeycomb with the unit cell highlighted by shading; (b) the band structure of the tight-binding Hamiltonian $H^0$ for the undeformed honeycomb with Dirac points occurring at the intersection of the two dispersion surfaces, at the six vertices of the Brillouin zone $\mathcal{B}$; (c) the deformed honeycomb with a slowly-varying quadratic deformation $\bm u(\delta \bm X) = (0,\delta^2X_1^2)^T$; (d) Landau levels generated by the effect of strain in (c) on the spectrum.}
	\label{fig:honeycomb-intro}
\end{figure}

\subsection{Tight-binding model for a slowly varying deformation of the honeycomb}
Let $A_{m,n}, B_{m,n}$ denote the vertices of the honeycomb lattice; see  Figure \ref{fig:honeycomb-intro}a and the more detailed discussion in Section  \ref{sec:dirac-pt-honeycomb}. Each $A-$ site has $3$ nearest neighbor $B-$sites and each $B-$site has $3$ nearest neighbor $A-$sites. The positions of the three nearest $B-$sites to $A_{m,n}$ are $B_{m+m_\nu,n+n_\nu}$, $\nu=1,2,3$ and the three nearest $A-$sites to $B_{m,n}$ are $A_{m-m_\nu,n-n_\nu}$ with $\nu=1,2,3$. We take $|A_{m,n}-B_{m,n}|=1$, for all $(m,n)\in\mathbb Z^2$, to be the lattice spacing between nearest neighbors. 

Introduce a smooth displacement field  $\bm x\mapsto \bm u(\bm x)$.
 For $\delta>0$ and small we define a slowly-strained honeycomb lattice whose vertices are at the positions
\begin{align*}
    \widetilde A_{m,n}=A_{m,n}+\bm u(\delta A_{m,n}),\qquad \widetilde B_{m,n}=B_{m,n}+\bm u(\delta B_{m,n})\ .
\end{align*}
 The undistorted honeycomb lattice, $\delta=0$, with a choice of period lattice vectors and period cell,  is displayed in Figure
 \ref{fig:honeycomb-intro}a and an example of a weakly and non-uniformly (quadratically) strained 
lattice is displayed in Figure \ref{fig:honeycomb-intro}b.

A wave function, $\bm \psi$, is an assignment of a complex amplitude to each site. We group the amplitudes by cell and write:
\[ \bm \psi = (\bm \psi_{m,n})_{(m,n)\in\mathbb Z^2},\quad\textrm{where}\quad \bm \psi_{m,n}=\begin{pmatrix}
    \psi_{m,n}^A\\ \psi_{m,n}^B
\end{pmatrix} \in \mathbb{C}^2.\]

The  strength of interaction between the amplitude at a given site and the amplitudes at the three nearest neighbor sites, depends on the distance between nearest neighbor sites. For two nearest neighbor sites $\widetilde{A}$ and $\widetilde{B}$, we define:
\begin{equation} t(\widetilde{A}, \widetilde{B}) = h(|\widetilde{A}- \widetilde{B}|) = t(\widetilde{B}, \widetilde{A}), \label{eq:tAB-def}
\end{equation}
 where $h$ is a specified smooth scalar function.  We set $t_0=h(1)$ and $t_1=h'(1)$. 

The nearest neighbor tight-binding model, $H^\delta$ (undeformed, $\delta=0$ and deformed, $\delta>0$)  acts on $\bm \psi\in l^2(\mathbb Z^2;\mathbb C)$
 is defined by: 
\begin{align}
    H^\delta \bm \psi = \begin{pmatrix}
        (H^{\delta} \bm \psi)_{m,n}^A\\
        (H^{\delta}\bm \psi)_{m,n}^B
    \end{pmatrix} =  \begin{pmatrix}
       \sum_{\nu=1}^3\  t(\widetilde{A}_{m,n}, \widetilde{B}_{m+m_\nu,n+n_\nu}) \psi_{m+m_\nu,n+n_\nu}^B\\
       \sum_{\nu=1}^3\ t(\widetilde{B}_{m,n}, \widetilde{A}_{m-m_\nu,n-n_\nu}) \psi_{m-m_\nu,n-n_\nu}^A
    \end{pmatrix}\ .\label{eq:Hdelta}
\end{align}

For $\delta=0$, $\widetilde{A}_{m,n}={A}_{m,n}$ and $\widetilde{B}_{m,n}={B}_{m,n}$. Hence, for the undeformed lattice, $t({A}_{m,n},{B}_{m,n})=h(1)$; all hopping coefficients of the undeformed Hamiltonian $H^0$ are equal. This yields $H^0$, the tight-binding model for the undeformed honeycomb with uniform hopping coefficient $t_0$. The Hamiltonian $H_0$ has a band structure, deducible via the two-dimensional discrete Fourier transform, consisting of two band dispersion surfaces which meet conically at high symmetry quasimomenta, the vertices of the hexagonal Brillouin zone; see Figure \ref{fig:honeycomb-intro}(b).

For the deformed tight-binding Hamiltonian ($0<\delta\ll1$), as hopping coefficients in \eqref{eq:Hdelta} we replace $t(\widetilde{A}_{m,n}, \widetilde{B}_{m+m_\nu,n+n_\nu})$ and $t(\widetilde{B}_{m,n}, \widetilde{A}_{m-m_\nu,n-n_\nu})$ in \eqref{eq:Hdelta} by their
 $O(\delta)$-corrections. For example, the in-cell hopping coefficients are obtained by expansion of $\delta\mapsto h(|\widetilde{A}_{m,n}- \widetilde{B}_{m,n}|) $ through order $\delta$ and dropping terms of order $ O(\delta^2)$:
\begin{align*}
    t(\widetilde{A}_{m,n}, \widetilde{B}_{m,n}) &\equiv t_0 + \delta t_1 \bm{e}_1^T \nabla_{\bm X} \bm{u}(\delta A_{m,n}) \bm{e}_1.\
\end{align*}
Analogous expressions for the hoppings are displayed below in   
\eqref{eqn:hopping-def-honeycomb}. The resulting Hamiltonian, $H^\delta$, varies on  two spatial scales: the $O(1)$ length scale of the lattice, and the long $O(\delta^{-1})$ length scale of the lattice deformations.


\subsection{Summary of results}

\begin{itemize}
    \item Using the multiple scale structure of our Hamiltonian, we first construct formal asymptotic expansion of eigenpairs $(\bm \psi^\delta,E^\delta)$ eigenvalue problems $H^\delta \bm \psi=E\bm \psi$ subject to appropriate homogeneous boundary conditions at $\infty$. Eigenmodes have a two-scale wave packet structure:
     \begin{equation}\label{eq:formal-mode}
   \bm \psi_{m,n}^\delta = e^{i\bm K\cdot\bm x} \bm \Phi(\bm X;\delta)\Big|_{\bm X = \delta ( m\bm v_1 + n\bm v_2)},\quad \textrm{where}\quad   (m,n)\in\mathbb Z^2, 
          \end{equation}
 and are constructed as a formal series $\bm \Phi(\bm X;\delta)=\sum_{j\geq 0} \delta^j \bm \Phi_j(\bm X) $.  If, for example, $\bm X\mapsto \bm \Phi(\bm X;\delta)$ decays sufficiently rapidly at $|\bm X|\mapsto \infty$, then the discrete Fourier transform of \eqref{eq:formal-mode} is spectrally concentrated in a $\delta-$ width neighborhood to the high symmetry quasi-momentum $\bm K$.
      
  At leading order in $\delta$, we find that 
  \[ \bm \psi_{m,n}^\delta \approx e^{i\bm K\cdot\bm x} \bm \Phi_0(\bm X)\Big|_{\bm X = \delta ( m\bm v_1 + n\bm v_2)}\quad \textrm{and}\quad E^\delta \approx \delta E_1
  \]
  where -- up to a unitary change of variables --  $(\bm \Phi_0(\bm X), E_1)$ is  an eigenpair, of an effective magnetic Dirac operator:
  \begin{align}\label{eq:D_A}
  \mathcal{D}_{\bm A} &= \big(-i\partial_{X_1} -A_1(X_1,X_2)\big)\sigma_1 + \big(-i\partial_{X_2} -A_2(X_1,X_2)\big)\sigma_2\ ,\end{align}
with effective magnetic potential $\bm A_{\text{eff}} = (A_1,A_2)$, given explicitly in terms of the deformation $u(\bm X)$. 


    \item Motivated by the case of Landau gauge for a constant perpendicular magnetic field, we focus on   unidirectional deformations $\bm u=(0,d(X_1))^T$, for which the deformed structure are translation invariant in the $\bm v_2$ (vertical) direction. Hence, the spectral theory of $H^\delta$ acting in the space $l^2(\mathbb Z^2;\mathbb{C}^2)$ is reducible to a family of $q_\parallel-$ pseudo-periodic spectral problems, parameterized by $q_\parallel\in[-\pi,\pi)$.  We are particularly interested in the discrete eigenvalues. These are   nontrivial solutions of the form $\bm \psi_{m,n} =e^{in q_\parallel}\  \bm \Psi_m(q_\parallel)$, where
   \begin{align}
   \big( H^\delta(q_\parallel) \bm \Psi \big)_m  &= E(q_\parallel) \ \bm \Psi_m,\quad \bm \Psi(q_\parallel)=(\bm \Psi_m)_{m\in\mathbb Z}\in l^2(\mathbb Z;\mathbb C^2) .
   \label{eq:kpar-EVP}
 \end{align} 
Given a continuous family of solutions $q_\parallel\in I\mapsto \bm \Psi(q_\parallel)$, for $q_\parallel$ varying over a subinterval $ I\subset [-\pi,\pi)$, we may form continuous superposions to produce states which have full spatial localization over the lattice; equivalently, such states are in $l^2(\mathbb Z^2;\mathbb{C}^2)$. 


For the $q_\parallel-$ pseudo-periodic eigenvalue problem \eqref{eq:kpar-EVP} with $q_\parallel = \delta k_\parallel$, our multiple scale expansion yields approximate eigenpairs $(\bm \Psi^\delta (k_\parallel),E^\delta(k_\parallel))$ 
where $\bm \Psi^\delta(k_\parallel)$ is $k_\parallel-$ pseudo-periodic in the $\bm v_2$ direction and spatially
localized in transverse directions. Up to coordinate change, the mode profile has the structure 
\begin{subequations}\label{eqn:ansatz-1d-intro}
   \begin{align}
    &\bm \psi^{\rm approx}_{m,n}(k_\parallel) = \delta^{\frac{1}{2}}e^{i \bm K \cdot \bm x} e^{i \frac{k_\parallel}{2 \pi} \bm X \cdot \bm a_2} \bm \Psi^{\#}_0\left(\frac{\bm X \cdot \bm a_1}{|\bm a_1|}; k_\parallel\right) \bigg|_{\substack{\bm x = \Cell_{m,n}\\\bm X = \delta \Cell_{m,n}}}\\
    &E^{\rm approx}(k_\parallel) = \delta E_1(k_\parallel), \qquad |k_\parallel| \lesssim 1,
\end{align}
\end{subequations}
where $\big(\bm \Psi_0^{\#}(X_1;k_\parallel), E_1(k_\parallel)\big)$ is an $L^2_{X_1}(\mathbb R)-$ eigenpair of a one-dimensional 
Dirac operator arising from \eqref{eq:D_A}: 
    \begin{align}\label{eqn:1d-dirac-intro}
        \mathcal{D}(k_\parallel) := \frac{3}{2} \Big[-i\partial_{X_1} \sigma_1 + \kappa(X_1;k_\parallel) \sigma_2\Big], \qquad \kappa(X_1;k_\parallel) = \frac{k_\parallel}{3} - \frac{t_1}{2} d'(X_1)\ ;
    \end{align}
    see also \eqref{eqn:dirac-1d}. For $k_\parallel$ in an interval about $0$, the $L^2(\mathbb R)$ spectrum of $\mathcal{D}(k_\parallel)$ consists of a continuous part, with a spectral gap containing an odd number of eigenvalues. The eigenvalue curve $k_\parallel\mapsto E_0(k_\parallel)$, for which $E_0(0)=0$, is ``topologically protected''; it cannot be removed by a spatially localized perturbation of $\kappa$.
    
    \item Introduce $l^2_{k_\parallel}$, the space of all $\bm \psi=(\bm \psi_{m,n})_{m,n\in\mathbb Z}$ which satisfy:
    \begin{subequations}\label{eqn:intro-lk}
        \begin{align}
   &\bm \psi^\delta_{m,n+1} =e^{i(\delta k_\parallel)} \bm \psi^\delta_{m,n},\quad \textrm{(pseudo-periodicity)} \label{eqn:intro-lk-1}\\
        &m\in\mathbb Z \longmapsto  
        \bm \psi_{m,n}\in l^2_m(\mathbb Z)\ \textrm{(transverse localization)} .\label{eqn:intro-lk-2}
    \end{align}
    \end{subequations}
   Our main theorem is Theorem \ref{thm:main}, which states that for $0<\delta<\delta_0$ sufficiently small and $|k_\parallel|\lesssim 1$ the following holds: any eigenpair of $\mathcal{D}(k_\parallel)$  gives rise to an eigenpair of the {\it $l^2_{k_\parallel}$-eigenvalue problem}
        \begin{align}\label{eqn:intro-eigen-1d}
        &(H^\delta \bm \psi)_{m,n} = E^\delta \bm \psi^\delta_{m,n},\quad \bm \psi\in l^2_{k_\parallel}
        \end{align}
    Moreover, we have the following bound for the correctors:
    \begin{subequations}\label{eqn:intro-corrector-est}
        \begin{align}
        &\sup_{n \in \mathbb{Z}}\Big\| \bm \psi^\delta_{m,n}(k_\parallel)-\ \bm \psi^{\rm approx}_{m,n}(k_\parallel) \Big\|_{l^2_m(\mathbb{Z};\mathbb{C}^2)} \lesssim \delta, \label{eqn:intro-corrector-est-1}\\
        &|E^\delta(k_\parallel) - \delta E_1^{\rm approx}(k_\parallel)| \lesssim \delta^2.\label{eqn:intro-corrector-est-2}
    \end{align}
    \end{subequations}
  Here, the leading approximations to $(\psi^\delta,E^\delta)$ are displayed in  \eqref{eqn:ansatz-1d-intro}. 

    \item We perform numerical simulations for two classes of deformations of the underlying honeycomb structure: (1) armchair oriented deformations (see e.g. Figure \ref{fig:ac-edge-def}) and zigzag oriented   deformations (see e.g. Figure \ref{fig:zz-edge-def}). For these configurations, we investigate the spectral characteristics of the resulting Hamiltonian in a neighborhood of the Dirac point. 

    For the quadratic deformation along the AC orientation, which is our main focus, the numerical eigenvalues of the Hamiltonian $H^\delta$ in \eqref{eqn:intro-eigen-1d} are nearly independent of the quasi-momentum $q_\parallel$ near the Dirac point $q_\parallel=0$ (see Figure \ref{fig:intro-numerics}(b)), indicating the presence of Landau levels. The eigenmodes are localized in the bulk, as shown in Figures \ref{fig:intro-numerics}(b-c). For the quadratic deformation along the ZZ orientation, the Dirac point is not opened into flat bands; the crossing persists from one side (see Figures \ref{fig:zz-def}(b) and (d)). This agrees with the effective Hamiltonian being purely continuous spectrum near the Dirac point, so no discrete localized modes emerge.
    
    \begin{figure}[!htb]
	\centering
    \subfloat[]{
	\includegraphics[height=1.75in]{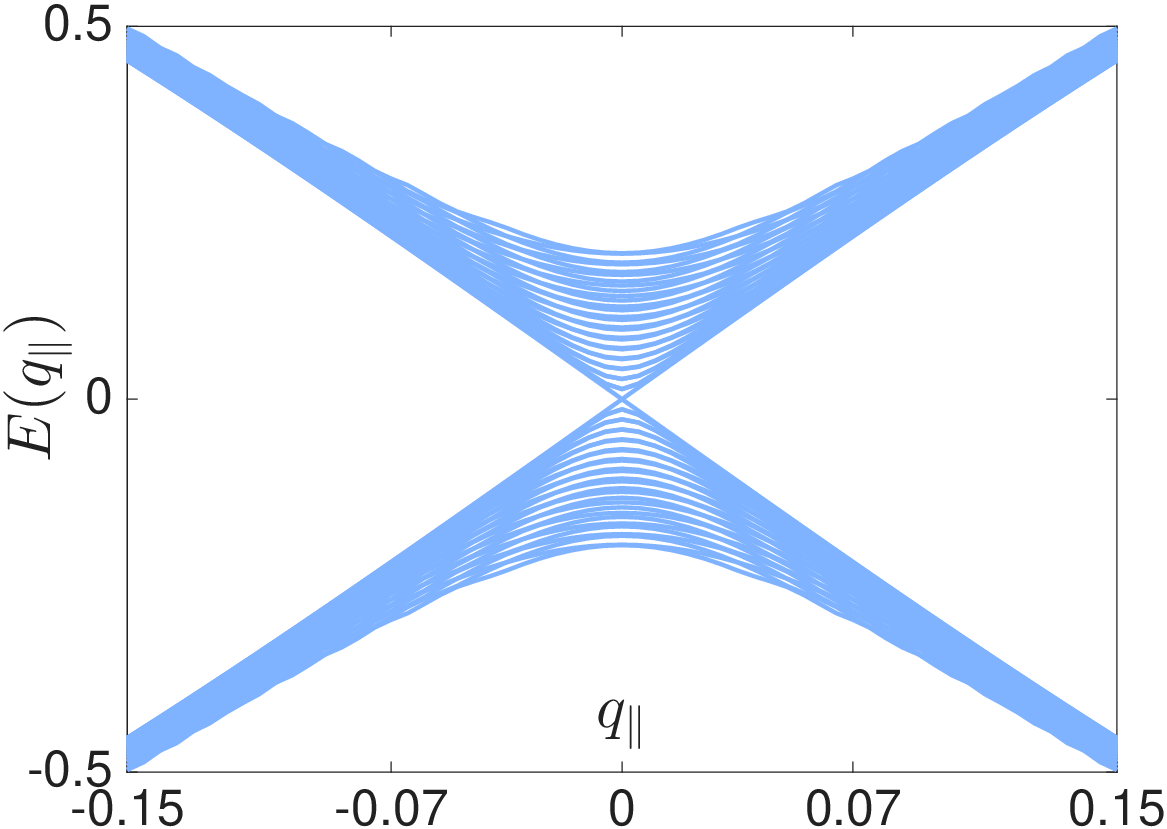}
	}\hfil
	\subfloat[]{
		\includegraphics[height=1.75in]{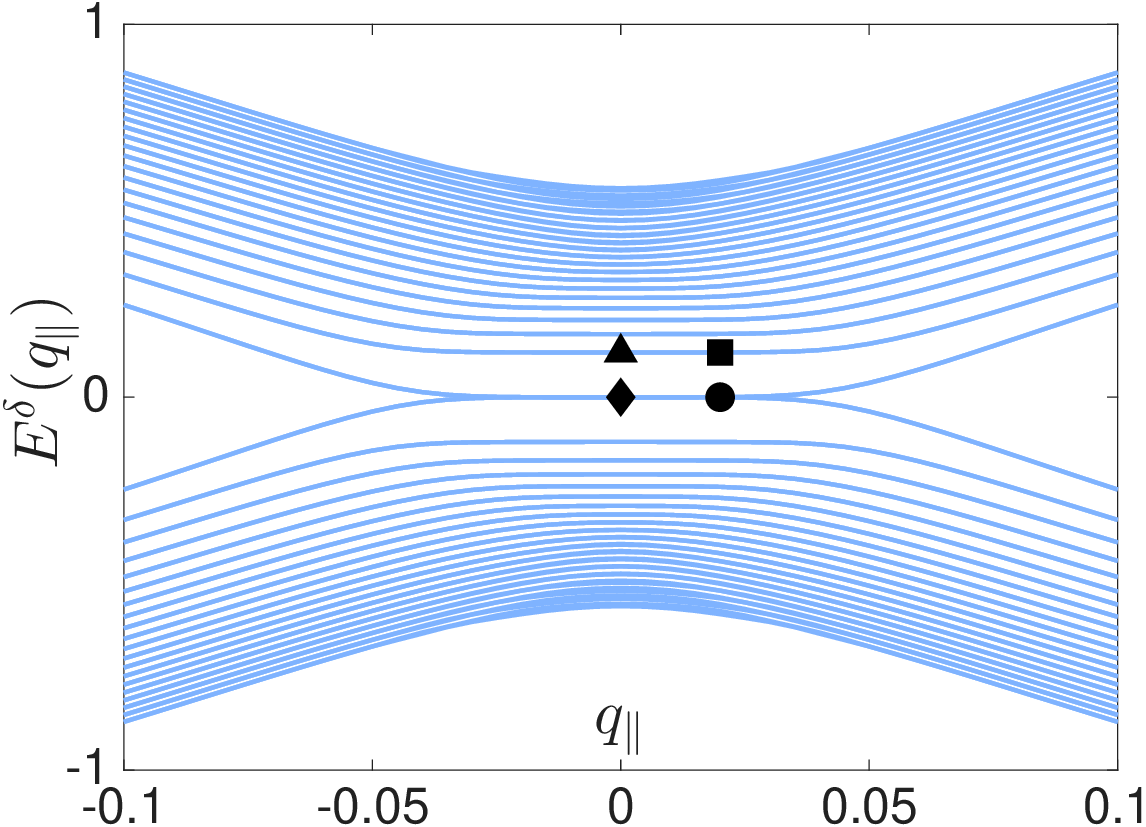}
	}\\
    \subfloat[]{
		\includegraphics[height=1.75in]{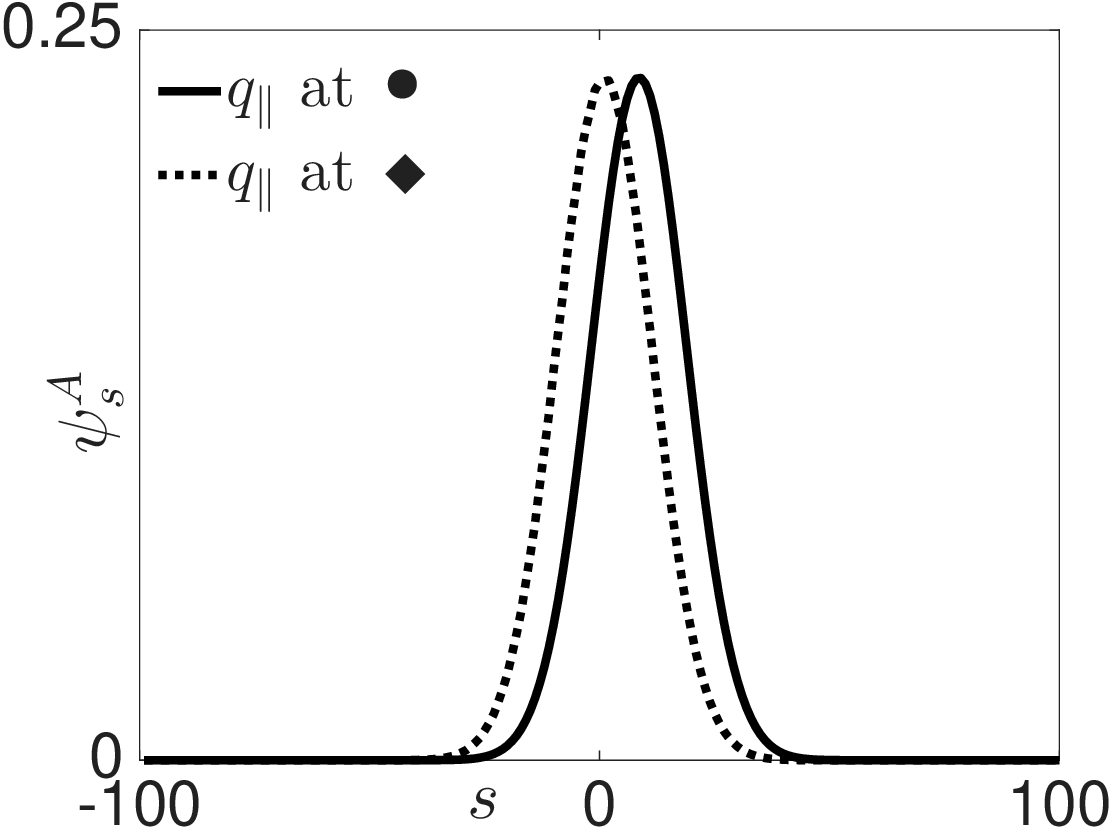}
	}\hfil
	\subfloat[]{
		\includegraphics[height=1.75in]{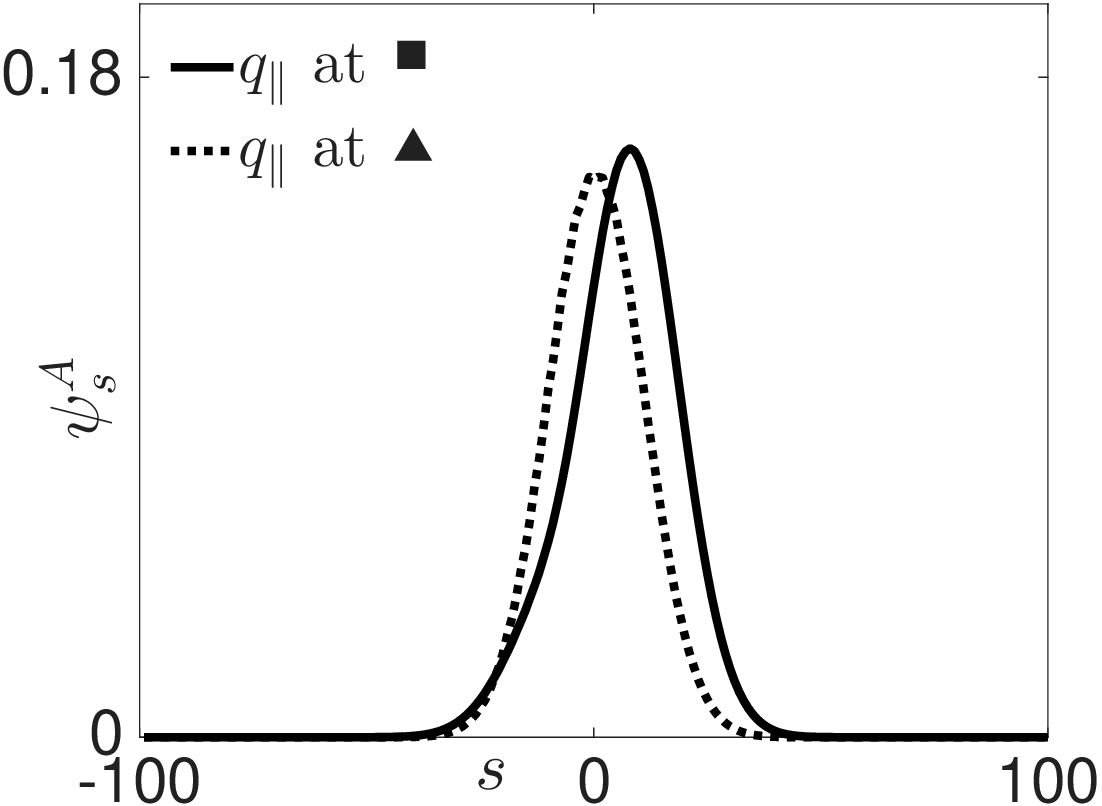}
	}
    \caption{(a) Conical (Dirac) point occurs at $q_\parallel = 0$; (b) Nonuniform deformation $\bm u = (0,X_1^2)^T$ (approximation to Landau gauge vector potential) creates a bulk gap with quite flat (Landau level like) localized mode eigenvalue curves $E^\delta(q_\parallel)$ in the gap; (c) Localized eigenstates corresponding to the diamond and circle in (b); (d) Localized eigenstates corresponding to the triangle and square in (b).}
	\label{fig:intro-numerics}
\end{figure}
\end{itemize}

\subsection{Related work}
 In \cite{zhang2025pseudo} it is shown that the dynamics of a class of wave-packets of deformed continuum honeycomb media is governed, on large but finite time scales, by the effective magnetic Dirac operator, derived in \cite{guglielmon2021landau}; see also \cite{barsukova2024direct}. 
   The approach taken is similar to that in \cite{fefferman2014wave} for the dynamics of wavepackets in undistorted honeycomb structures.
The article \cite{bal2026macroscopic} develops, via pseudo-differential calculus, a framework for studying the large (finite) time dynamics in discrete systems with slowly varying hopping coefficients.

Spatial localization  in deformed honeycomb structures, as we show arises from the eigenstates of an asymptotic (effective) Dirac Hamiltonian with a domain wall potential. Analogous localization phenomena arise in the study of {\it edge states} of honeycomb or other media with band degeneracies, when perturbed along a line defect or ``edge''; see e.g. \cite{chaban2025edge,drouot2020edge,drouot2021bulk,drouot2020defect,fefferman2017topologically,lee2019elliptic}.


 \subsection{Future directions} 
 \begin{enumerate}
     \item Our current results are based on a tight-binding model with slowly varying coefficients. A future direction is to generalize this setting to a continuum theory that captures not only the low-energy regime but also higher-energy bands.
     
     \item We show point spectrum of an effective Dirac operator give rise to 
     \edit{$l^2_{k_\parallel}$-eigenstates} of the underlying deformed tight-binding Hamiltonian, $H^\delta$. Do all such localized states of the discrete operator arise in this manner? We believe this to be the case, and that a quasi-mode approach, like that used in \cite{drouot2020edge}, would be applicable.

     \item In the present tight-binding model, only nearest-neighbor couplings are included, and the resulting effective Dirac Hamiltonian contains only a pseudo-magnetic potential. The addition of next-nearest-neighbor couplings (see e.g. the Haldane model in \cite{haldane1988model}), would give rise to effective Dirac Hamiltonians with a pseudo-electric potential as well (see e.g. \cite{guglielmon2021landau}).

     \item \textit{Spatially compact deformations and scattering resonances:} In practice, strain fields are applied over a finite region (due to sample size and fabrication limits). Therefore,  it is natural to consider the case of strain which is confined to a large compact set. In this case, we do not expect the effective Dirac Hamiltonian to support localized states. We do however expect {\it scattering resonances} (see e.g. \cite{dyatlov2019mathematical,tang2000resonance}) which capture the effect of slow energy leakage out of the region where the structure is deformed. An analytical and numerical investigation would be of great interest. 

    \item \textit{Triaxial deformation:} 
    A second deformation, unlike the unidirectional quadratic deformation $\bm u (X_1, X_2) = (0,X_1^2)^T$, which also generates a constant pseudo-magnetic field is the {\it triaxial deformation}  \cite{guinea2010energy,settnes2016pseudomagnetic}, for defined by $\bm u(X_1,X_2) = (2X_1 X_2, X_1^2 - X_2^2)^T$. This deformation is rotational symmetric and gives rise to the effective magnetic potential $\bm A_\text{eff} = 2t_1(-X_2, X_1)$ in \eqref{eqn:eff-mag}. This is the well-known  \textit{symmetric gauge} magnetic potential for a constand perpendicular magnetic field. The spectral properties of the effective Hamiltonian $\mathcal{D}_{\bm A}$ in \eqref{eqn:dirac-mag} for the symmetric gauge have been extensively studied; the spectrum of $\mathcal{D}_{\bm A}$  consists is discrete spectrum and consists of  eigenvalues of infinite multiplicity (Landau levels). We believe our strategy for unidirectional deformations can be extended to bounded gradient regularizations of the symmetric gauge. 
     



 \end{enumerate}

\subsection{Outline of the article}
The remainder of this article is organized as follows. Section \ref{sec:dirac-pt-honeycomb} reviews the discrete honeycomb lattice and the undeformed tight-binding Hamiltonian. Section \ref{sec:derivation} introduces the slowly strained tight-binding model and derives the effective magnetic Dirac description via a formal discrete multiscale expansion. In Section \ref{sec:uni-displacement}, we introduce unidirectional deformations, which preserve vertical periodicity, and reduce the eigenvalue problem to a one-dimensional one. In Section \ref{sec:main-thm}, we state our main Theorem, Theorem \ref{thm:main}, and outline the sketch of the proof. In Section \ref{sec:num-comparison}, we present numerical results for quadratic deformations and their linear regularizations. We prove Theorem \ref{thm:main} in detail in Section \ref{sec:proof}. We conclude our paper with a brief summary and some future directions in Section \ref{sec:conclusion}.

\subsection{Notation and conventions}
\begin{enumerate}[(1)]
	\item $\bar{z}$ denotes the complex conjugate of $z \in \mathbb{C}$.
	
	\item $x \lesssim y$ means that there exists a constant $C > 0$ such that $x \leq Cy$.
	
	\item $l^2(\mathbb{Z}; \mathbb{C}^2) = \{(x_n)_{n \in \mathbb{Z}} \: | \: x_n \in \mathbb{C}^2 \text{ and } \sum_{n \in \mathbb{Z}} |x_n|^2 < \infty\}$ and $l^2(\mathbb{Z}^2; \mathbb{C}^2) = \{(x_{m,n})_{m,n\in \mathbb{Z}} \: | \: x_{m,n} \in \mathbb{C}^2 \text{ and } \sum_{m,n \in \mathbb{Z}} |x_{m,n}|^2 < \infty\}$.
	
	\item The Schwartz space $\mathcal{S} = \mathcal{S}(\mathbb{R},\mathbb{C})$: $f(x) \in \mathcal{S}(\mathbb{R})$ if $f\in C^\infty(\mathbb{R})$ and for every pair of multi-indices $\alpha,\beta\in \mathbb{N}_0^n$, we have $\sup_{x\in\mathbb{R}} \, |x^\alpha \,\partial^\beta f(x)| < \infty$, where $x^\alpha := x_1^{\alpha_1}\cdots x_n^{\alpha_n}$ and $\partial^\beta := \partial_{x_1}^{\beta_1}\cdots \partial_{x_n}^{\beta_n}$.
	
	\item Hilbert space $H^s(\mathbb{R}) = H^s(\mathbb{R}, \mathbb{C}) = \{f : \mathbb{R} \mapsto \mathbb{C} \:|\: \int_\mathbb{R} (1+|k|^2)^s|\widehat{f}(k)|^2 \: dk < \infty\}$.
	
	\item $L^{2,1}(\mathbb{R}) = \Big\{ \widehat{f}(\xi) \:  \big| \: \int_\mathbb{R} (1+|\xi|^2) \: |\widehat{f}(\xi)|^2 \: d\xi < \infty\Big\}$ and the norm in $L^{2,1}(\mathbb{R})$ space is defined as
	\begin{align}
		\|\widehat{f}\|_{L^{2,1}(\mathbb{R})}^2 = \int_\mathbb{R} (1+|\xi|^2) |\widehat{f}(\xi)|^2 \: d\xi.\label{eqn:L21-space}
	\end{align}
    By Plancherel, the condition $\widehat f\in L^{2,1}(\mathbb R)$ is equivalent to $f\in H^1(\mathbb R)$.


    \item The space of $C_b^\infty(\mathbb{R})$ consists of functions on $\mathbb{R}$ whose derivatives of all orders are bounded, i.e.
    \begin{align}
        C_b^\infty(\mathbb{R}) = \Big\{ f \in L^\infty(\mathbb{R}) \ \Big| \ f^{(k)} \in L^\infty(\mathbb{R}), \text{ for all } k \geq 1\Big\}.\label{eqn:C_b_inf}
    \end{align}
	
	\item Pauli matrices $\sigma_1, \sigma_2, \sigma_3$ are 
	\begin{equation}\label{eqn:pauli}
		\sigma_1 = \begin{pmatrix}
			0 & 1 \\
			 1 & 0
		\end{pmatrix}, \quad \sigma_2 = \begin{pmatrix}
		0 & -i \\
		i & 0
		\end{pmatrix}, \quad \sigma_3 = \begin{pmatrix}
		1 & 0\\
		0 & -1
		\end{pmatrix}.
	\end{equation}
\end{enumerate}

\paragraph{Acknowledgement} This work was supported in part by NSF grants: DMS1908657, DMS-1937254 and DMS-2510769, and Simons Foundation Math+X Investigator Award \#376319. Part of this research was carried out during the 2023-24 academic year, when MIW was a Visiting Member in the School of Mathematics, Institute of Advanced Study, Princeton, supported by the Charles Simonyi Endowment, and a Visiting Fellow in the Department of Mathematics at Princeton University.
\section{
Tight-binding model on the honeycomb lattice and Dirac points}\label{sec:dirac-pt-honeycomb}

In this section we discuss the discrete honeycomb lattice and the associated tight-binding model of bulk (undeformed) graphene. This operator is well-known to have Dirac points; its  two dispersion surfaces (bands) which touch conically at the vertices of the Brillouin zone;  see e.g. \cite{castro2009electronic,fefferman2024discrete,wallace1947band}. 

\subsection{Triangular lattice} 


The equilateral triangular lattice is the set $ \Lambda = \mathbb{Z} \bm{v}_1 \oplus \mathbb{Z} \bm{v}_2$ of all integer linear combinations of the lattice vectors
\begin{equation}\label{eqn:v1v2-honeycomb}
	\bm{v}_1 = \Big(\frac{\sqrt{3}}{2}, -\frac{3}{2}\Big)^T, \qquad \bm{v}_2 = \Big(0,3\Big)^T.
\end{equation}
Introduce the dual lattice basis $\{\bm a_1,\bm a_2\}$, 
\begin{align}
    \bm{a}_1 = \frac{4\pi}{\sqrt{3}}\left(1,0\right)^T, \qquad \bm{a}_2 = \frac{4\pi}{3}\left(\frac{\sqrt{3}}{2},\frac{1}{2}\right)^T\ ,\quad \bm{v}_i \cdot \bm{a}_j = 2\pi\delta_{ij}, \label{eqn:dual-vec}
\end{align}
and the dual lattice
$\Lambda^*= \mathbb{Z} \bm{a}_1 + \mathbb{Z} \bm{a}_2$; see Figure \ref{fig:honeycomb-diff-cell} and Figure \ref{fig:honeycomb-dual}.

The Brillouin zone $\mathcal{B}$ is a choice of fundamental domain for $\mathbb{R}^2 / \Lambda^*$, here chosen to be Brillouin zone is the hexagon shown in Figure \ref{fig:honeycomb-dual}.  with high-symmetry corner points $\bm K$ and $\bm K'$ (marked by solid and open dots, respectively, in Figure \ref{fig:honeycomb-dual}). The vertices of the $\mathcal B$ are the points:
\begin{equation}\label{eqn:dirac-pt-honeycomb}
	\bm{K} = \frac{1}{3} \bm{a}_1, \qquad \bm{K}' = -\frac{1}{3} \bm{a}_1 + \bm{a}_2,
\end{equation}
and their rotations by $2\pi/3$ with respect to the center of the hexagon.

\subsection{Honeycomb lattice} 

The honeycomb lattice is the union of two interpenetrating triangular lattices, with base points at   $A_{0,0} := (0,0)^T$ and and $B_{0,0} := (-\sqrt{3}/2, -1/2)^T$: 
\begin{align}
\mathbb{H} = \Lambda_A \cup \Lambda_B,\quad\textrm{where  $\Lambda_A= A_{0,0} + \Lambda$ and $\Lambda_B=B_{0,0} + \Lambda$}.\label{eqn:def-H-lattice}
\end{align}
The plane, $\mathbb R^2$, can be tiled by parallelograms, generated by $\bm v_1$ and $\bm v_2$ and such that each of these parallelograms contains two sites or nodes of $\mathbb H$, an $A-$node and a $B-$node. The parallelograms are indexed by pairs $(m,n)\in\mathbb Z^2$, and the nodes of $\mathbb H$ in the $(m,n)$ cell are: 
\begin{align}
    A_{m,n} = A_{0,0} + m \bm{v}_1 + n\bm{v}_2\quad  {\rm and}\quad  B_{m,n} =  B_{0,0} + m \bm{v}_1 + n\bm{v}_2. \label{eqn:node-AB-mn}
\end{align}

%

The honeycomb lattice, $\mathbb H$, is a bipartite lattice; each $A-$node has three nearest neighbhor $B-$nodes and each $B-$node has three nearest neighbor $A-$nodes. Starting with an $A-$ or $B-$node), the nearest neighbor $A-$ and $B-$ nodes are  obtained by displacement  by $\pm\bm e_\nu,\ \nu=1,2,3$, where 
\begin{equation}\label{eqn:ref-e-vec}
	\bm{e}_1 = \Big(\frac{\sqrt{3}}{2}, \frac{1}{2}\Big)^T, \quad \bm{e}_2 = \Big(-\frac{\sqrt{3}}{2}, \frac{1}{2}\Big)^T, \quad \bm{e}_3 = \Big(0,-1\Big)^T.
\end{equation}
Hence, for $\nu=1,2,3$,
\begin{subequations}\label{eqn:def-nn-connection}
    \begin{align}
    &\textrm{the nearest $B-$nodes to $A_{m,n}$ are  $B_{m+m_\nu,n+ n_\nu}$, and } \label{eqn:def-nn-connection-1}\\
    &\textrm{the nearest $A-$nodes to $B_{m,n}$ are  $A_{m-m_\nu,n-n_\nu}$, where } \label{eqn:def-nn-connection-2}\\
    &\textrm{ $(m_1,n_1)=(0,0)$, $(m_2,n_2)=(2,1)$, and $(m_3,n_3)=(1,1)$ }.\label{eqn:def-nn-connection-3}
\end{align}
\end{subequations}

\begin{figure}[!htb]
\centering
\begin{minipage}[t]{0.4\textwidth}
  \centering
  \subfloat[]{\includegraphics[width=0.85\linewidth]{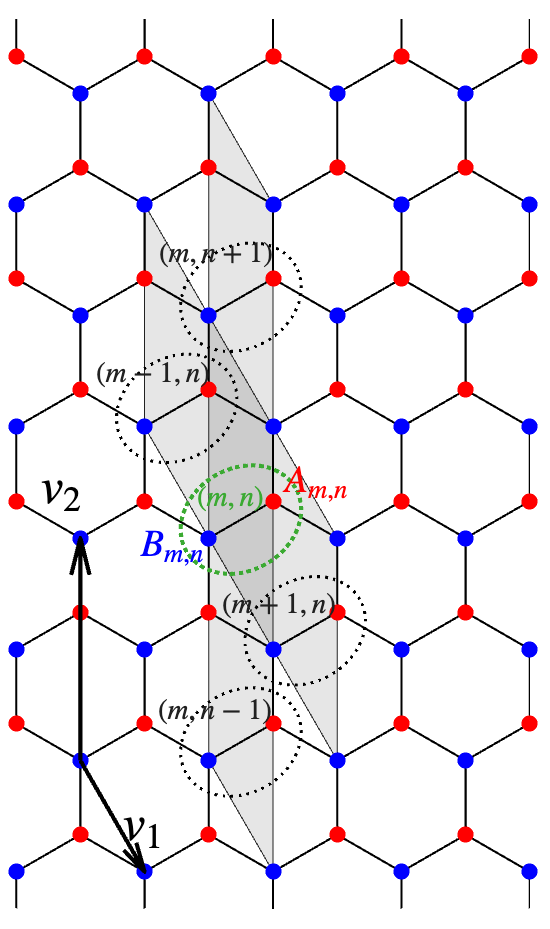}\label{fig:honeycomb-diff-cell}}
\end{minipage}
\begin{minipage}[t]{0.5\textwidth}
  \centering
  \subfloat[]{\includegraphics[width=0.6\linewidth]{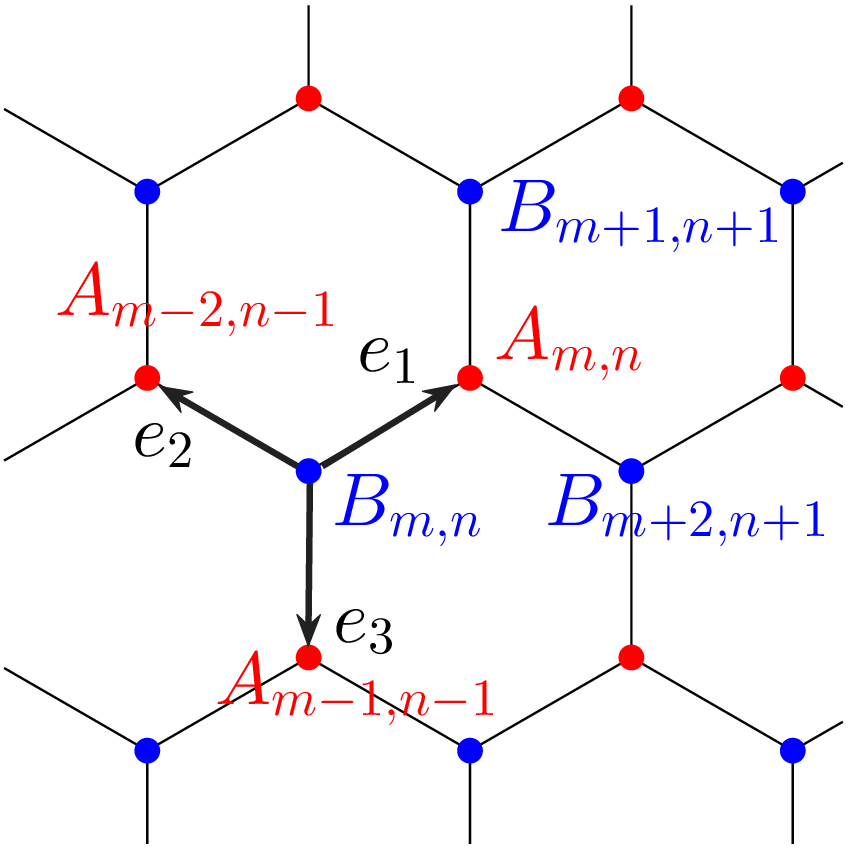}\label{fig:honeycomb-connect}}

  \subfloat[]{\includegraphics[width=0.75\linewidth]{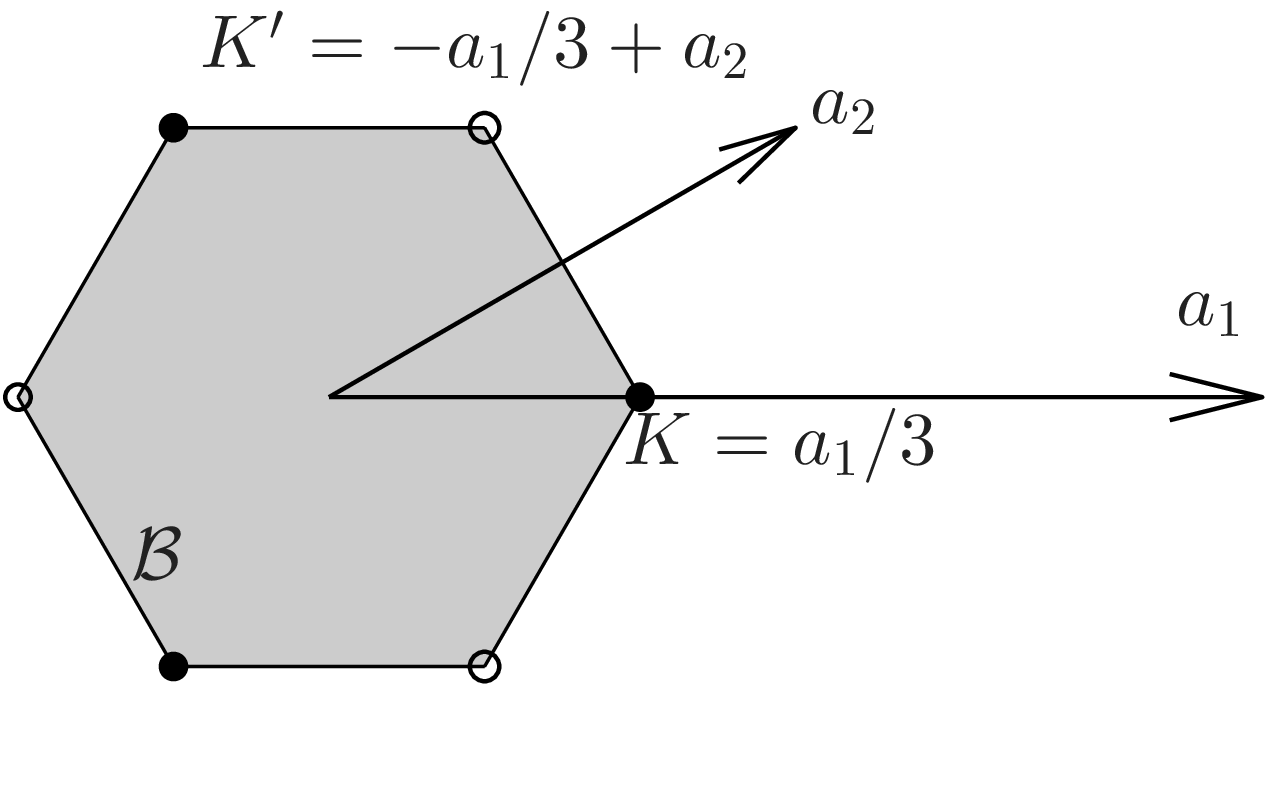}\label{fig:honeycomb-dual}}
\end{minipage}

\caption{(a) The honeycomb lattice, $\mathbb Z\bm{v}_1+\mathbb Z\bm{v}_2$, containing two nodes, $A$ in red and $B$ in blue, per unit cell. $\mathbb R^2$ is the union of cells (parallelograms), each containing two nodes (circled pair of dots), labeled with $\mathbb Z^2$ indices.   Dark-shaded parallelogram is the $(m,n)$-cell; adjacent lighter-shaded regions are translations by $\pm \bm{v}_1$ or $\pm \bm{v}_2$.  (b) nearest neighbors of $A_{m,n}$ and $B_{m,n}$, bond direction vectors $\bm{e}_i$ ($i=1,2,3$); (c) Hexagonal Brillouin zone $\mathcal{B}$ with  high symmetry quasimomenta $\bm{K},\bm{K}'$, at the vertices of $\mathcal{B}$. Dual lattice vectors $\bm{a}_1, \bm{a}_2$, such that $\bm a_l\cdot\bm v_j=2\pi\delta_{lj}$.}
\label{fig:honeycomb}
\end{figure}


\subsection{Tight-binding model on $\mathbb H$} 
    


A ``wave function'' is an assignment of a complex-valued amplitude to each node of $\mathbb H$. It is convenient to organize these
 as complex amplitude pairs, corresponding to each cell:
 \begin{align} \bm\psi = (\bm\psi_{m,n})_{m,n\in\mathbb Z},\quad 
 \bm \psi_{m,n} =
\begin{pmatrix} \psi_{m,n}^A\\ \psi_{m,n}^B
\end{pmatrix}\in \mathbb C^2.
\end{align}
We take as our Hilbert space $l^2(\mathbb{H}):=l^2(\mathbb Z^2; \mathbb C^2)$ with norm defined by
\begin{align}
 \|\bm\psi\|^2 = \sum_{(m,n)\in\mathbb Z^2} \|\bm \psi_{m,n}\|_{\mathbb C^2}^2 <\infty.\label{eqn:wave-fcn-def}
\end{align}
The nearest neighbor tight-binding Hamiltonian $H^0$ on the honeycomb lattice (the superscript 0 refers to the undeformed honeycomb), is the nearest neighbor operator:
\begin{align}\label{eqn:tb-ref-hamiltonian}
    (H^0\psi)_{m,n} :=  \begin{pmatrix}
     (H^0\psi)_{m,n}^A \\  (H^0\psi)_{m,n}^B
\end{pmatrix} = \begin{pmatrix}
      \sum_{\nu=1}^3 \psi_{m+m_\nu,n+n_\nu}^B \\  \sum_{\nu=1}^3 \psi_{m-m_\nu,n-n_\nu}^A
\end{pmatrix}\ .
\end{align}
$H^0$ is a self-adjoint and bounded linear operator on  $l^2(\mathbb Z^2;\mathbb C^2)$.




\subsection{Band structure of $H^0$} We study the spectrum of $H^0$ via the 2D discrete Fourier transform (DFT); see Appendix \ref{app:prelim-dft}. For  $\psi = (\bm \psi_{m,n})_{m,n\in \mathbb{Z}} \in l^2(\mathbb{H})$, we define 
\begin{equation}\label{eqn:psi-dft-honeycomb}
    \bm{\widetilde{\psi}}(\bm k) =
   \sum_{m,n \in \mathbb{Z}} e^{-i\bm{k}\cdot (m\bm{v}_1 + n \bm{v}_2)}\ \bm\psi_{m,n}.
   \end{equation}
The mapping $\bm{k}\mapsto \bm{\widetilde{\psi}}(\bm k)$ is in $L^2(\mathbb R^2/\Lambda^*)$ with  period cell $\mathcal{B}$; see Figure \ref{fig:honeycomb-dual}.  The band structure of $H^0$ is given by the two eigenpairs  (band dispersion curves and Floquet-Bloch modes) of 
of the $2\times2$ matrix:
\begin{equation}\label{eqn:tb-bloch-hamiltonian}
	\widetilde{H}^0(\bm{k}) = \sum_{\nu = 1}^3 \begin{pmatrix}
	0 & e^{i\bm{k}\cdot \bm w_\nu}\\
	e^{-i\bm{k}\cdot \bm w_\nu} & 0
	\end{pmatrix},\qquad \bm k\in\mathcal{B}.
\end{equation}
Here  $\bm w_\nu = m_\nu \bm{v}_1+ n_\nu \bm{v}_2 = \bm e_1 - \bm e_\nu$:
\begin{equation}\label{eqn:c-vec}
    \bm w_1 = (0,0)^T, \quad \bm w_2 = \bm e_1 - \bm e_2 = (\sqrt{3},0)^T, \quad \bm w_3 = \bm e_1 - \bm e_3 = \left(\frac{\sqrt{3}}{2}, \frac{3}{2}\right)^T.
\end{equation}
The band dispersion functions are given by:
\begin{equation}\label{eqn:tb-bloch-eigval}
    E_\pm(\bm{k}) = \pm \Bigg|\sum_{\nu = 1}^3 e^{i\bm{k}\cdot \bm w_\nu}\Bigg|.
\end{equation}
The two dispersion surfaces touch the over the vertices of $\mathcal B$; for $\bm k\in\mathcal B$, $E_+(\bm k)=0=E_-(\bm k)$ if and only if 
 $\bm k = \bm K$, $\bm K'$ or one of their rotations by $2\pi/3$; see \eqref{eqn:dirac-pt-honeycomb} and  Figure \ref{fig:honeycomb-2d-band}. The local character of the touching dispersion surfaces at these high symmetry quasi-momenta is conical and the corresponding energy / quasi-momentum pairs are called {\it Dirac points}. The spectrum 
  of $H^0$, acting in $l^2(\mathbb H)$, is continuous and equal to the interval $[-3,3]$.

\section{Slow strain induces an effective 2D Dirac magnetic Hamiltonian}\label{sec:derivation}
We now introduce a tight-binding Hamiltonian for a slowly strained honeycomb lattice and propose a two-scale wave-packet ansatz for the associated eigenvalue problem. We then present a formal discrete multiscale expansion of the resulting two-scale eigenstates, yielding an effective two-dimensional magnetic Dirac equation.

\subsection{Tight binding Hamiltonian for a  slowly-strained discrete honeycomb}\label{subsec:model-slow-strain}
We consider a slowly-strained discrete honeycomb, where the undeformed nodes,  $A_{m,n},B_{m,n}$ are mapped, under the smooth and slowly-varying displacement $\bm X\mapsto \bm u(\bm X)= \bm u(X_1,X_2)$ to: 
\begin{equation}\label{eqn:def-honeycomb}
	\widetilde{A}_{m,n} = A_{m,n} + \bm{u}(\delta A_{m,n}), \qquad \widetilde{B}_{m,n} = B_{m,n} + \bm{u}(\delta B_{m,n}),\quad 0<\delta\ll1.
\end{equation}
%

We now introduce the nearest neighbor tight-binding Hamiltonian $H^{\delta}$ for the {\it deformed} honeycomb lattice. For a wave function $\bm \psi = (\bm \psi_{m,n})_{m,n\in \mathbb{Z}}\in \ell^2(\mathbb Z^2)$, we have:
\begin{align}\label{eqn:def-hon-hamiltonian}
    (H^\delta \bm \psi)_{m,n} = \begin{pmatrix}
        (H^{\delta}\bm \psi)_{m,n}^A\\
        (H^{\delta}\bm \psi)_{m,n}^B
    \end{pmatrix} = \begin{pmatrix}
        \sum_{\nu=1}^3  t(\widetilde{A}_{m,n}, \widetilde{B}_{m+m_\nu,n+n_\nu}) \psi_{m+m_\nu,n+n_\nu}^B\\
        \sum_{\nu=1}^3  t(\widetilde{B}_{m,n}, \widetilde{A}_{m-m_\nu,n-n_\nu}) \psi_{m-m_\nu,n-n_\nu}^A
    \end{pmatrix}\ .
\end{align}
Here,  $t(\widetilde{X}, \widetilde{Y})$ is the hopping coefficient between the deformed nodes $\widetilde{X}$ and $\widetilde{Y}$, here taken to be of the form: 
\begin{align}
    t(\widetilde{A}_{m,n}, \widetilde{B}_{m,n}) := h(|\widetilde{A}_{m,n}- \widetilde{B}_{m,n}|),\label{eqn:main-h-AB}
\end{align}
where $h: \mathbb{R}_+ \mapsto \mathbb{R}$ is a smooth function such that:
\[ \textrm{$h(1) = 1$ and $h'(1) = t_1$}.\]
The parameter $t_1$ controls how the hopping coefficient depends on the strain: when $t_1 > 0$, the hopping strength increases with strain, while for $t_1 < 0$, it decreases.

For simplicity, rather than working with $H^\delta$ in  \eqref{eqn:def-hon-hamiltonian}, we seek a tight binding model which incorporates only order $\delta$ corrections to $H^0$.  
Since $|{A}_{m,n}- {B}_{m,n}|=1$, we find $
    |\widetilde{A}_{m,n}- \widetilde{B}_{m,n}| = |A_{m,n} - B_{m,n} + \bm{u}(\delta A_{m,n}) - \bm{u}(\delta B_{m,n})|= 1 + O(\delta)$.
Expanding $t(\widetilde{A}_{m,n}, \widetilde{B}_{m,n})$ in \eqref{eqn:main-h-AB} for small $\delta$, we  obtain  (details in Appendix \ref{app:hopping})
\begin{align}\label{eq:t-tA-tB}
    t(\widetilde{A}_{m,n}, \widetilde{B}_{m,n}) &= 1 + \delta t_1 \bm{e}_1^T \nabla_{\bm X} \bm{u}(\delta A_{m,n}) \bm{e}_1 + O(\delta^2),
\end{align}
where $\bm{e}_1 = A_{m,n} - B_{m,n}$ is defined in \eqref{eqn:ref-e-vec} and $\nabla_{\bm X} \bm{u} (\bm X)$ denotes the $2\times2$ Jacobian matrix. T
 Note that since $\delta B_{m,n}$ lies $O(\delta)$-away from $\delta A_{m,n}$, an  expression for the hopping coefficient $t(\widetilde{A}_{m,n}, \widetilde{B}_{m,n})$, which is asymptotically equivalent to \eqref{eq:t-tA-tB} is the expression: $ t(\widetilde{A}_{m,n}, \widetilde{B}_{m,n}) = 1 + \delta t_1 \bm{e}_1^T \nabla_{\bm X} \bm{u}(\delta B_{m,n}) \bm{e}_1 + O(\delta^2)$.
%
In expanding the hopping coefficients through order $\delta$, we'll make the consistent choice of the $A-$point in each cell, as the  reference point about which we expand. 
Thus we set 
\begin{equation}\label{eqn:ref-ab-honeycomb}
	\Ucell := A_{m,n} = m \bm{v}_1 + n\bm{v}_2 = \left(\frac{\sqrt{3}}{2}m, -\frac{3}{2}m + 3n\right)^T,
\end{equation}
This choice leads to a self-adjoint Hamiltonian is shown in Appendix \ref{app:some-bds}. 
Expansion of the hopping coefficient through order $\delta$ and dropping all terms of order $\delta^2$ and higher, we obtain the following hopping coefficients: 
\begin{subequations}\label{eqn:hopping-def-honeycomb}
\begin{align}
    t(\widetilde{A}_{m,n} , \widetilde{B}_{m\pm m_1,n\pm n_1}) &= t(\widetilde{A}_{m,n} , \widetilde{B}_{m,n}) = 1 + \delta t_1 \bm{e}_1^T\nabla_{\bm X} \bm{u}(\delta \Ucell) \bm{e}_1,\label{eqn:hopping-def-honeycomb-1}\\
    t(\widetilde{A}_{m,n} , \widetilde{B}_{m+m_2,n+n_2}) &= t(\widetilde{A}_{m,n} , \widetilde{B}_{m+2,n+1}) = 1 + \delta t_1 \bm{e}_2^T\nabla_{\bm X} \bm{u}(\delta \Ucell) \bm{e}_2, \label{eqn:hopping-def-honeycomb-2}\\
    t(\widetilde{A}_{m,n} , \widetilde{B}_{m+m_3,n+n_3}) &= t(\widetilde{A}_{m,n} , \widetilde{B}_{m+1,n+1}) = 1 + \delta t_1 \bm{e}_3^T\nabla_{\bm X} \bm{u}(\delta \Ucell) \bm{e}_3,\label{eqn:hopping-def-honeycomb-3}\\
    t(\widetilde{B}_{m,n} , \widetilde{A}_{m-m_2,n-n_2}) &= t(\widetilde{B}_{m,n} , \widetilde{A}_{m-2,n-1}) = 1 + \delta t_1 \bm{e}_2^T\nabla_{\bm X} \bm{u}(\delta \Cell_{m-m_2,n-n_2}) \bm{e}_2,\label{eqn:hopping-def-honeycomb-4}\\
    t(\widetilde{B}_{m,n} , \widetilde{A}_{m-m_3,n-n_3}) &= t(\widetilde{B}_{m,n} , \widetilde{A}_{m-1,n-1}) =  1 + \delta t_1 \bm{e}_3^T\nabla_{\bm X} \bm{u}(\delta \Cell_{m-m_3,n-n_3}) \bm{e}_3.\label{eqn:hopping-def-honeycomb-5}
\end{align}
\end{subequations}

Here, the scalar quantity $\bm{e}_i^T \big(\nabla_{\bm X} \bm{u}(\delta \Cell_{m,n}) \big)\bm{e}_i$ with $i=1,2,3$ represents the linear strain induced by $\bm{u}$ at the edge in the $\bm{e}_i$ direction in the $(m,n)$-th cell.   Our selection of hopping coefficients in \eqref{eqn:hopping-def-honeycomb} is consistent with those commonly used in the physics literature (see Equation (172) in \cite{castro2009electronic}). 




 Let 
\begin{equation}\label{eqn:f-fcn}
	f_\nu(\bm X) := \bm{e}_\nu^T \big(\nabla_{\bm X} \bm{u}(\bm X)\big) \bm{e}_\nu,\qquad \nu=1,2,3\ .
\end{equation}
Explicitly, we have:
\begin{subequations}\label{eqn:f-u-relation}
\begin{align}
    f_1(\bm{X})&= \bm{e}_1^T (\nabla_{\bm X} \bm{u}) \bm{e}_1 = \frac{3}{4} \partial_{X_1} u _1 + \frac{\sqrt{3}}{4} (\partial_{X_1} u_2 + \partial_{X_2} u_1) + \frac{1}{4} \partial_{X_2} u_2, \label{eqn:f-u-relation-1}\\
	f_2(\bm{X})&= \bm{e}_2^T (\nabla_{\bm X} \bm{u}) \bm{e}_2 = \frac{3}{4} \partial_{X_1} u_1 - \frac{\sqrt{3}}{4} (\partial_{X_1} u_2 + \partial_{X_2} u_1) + \frac{1}{4} \partial_{X_2} u_2, \label{eqn:f-u-relation-2}\\
    f_3(\bm{X}) &= \bm{e}_3^T (\nabla_{\bm X} \bm{u}) \bm{e}_3 =\partial_{X_2} u_2.\label{eqn:f-u-relation-3}
\end{align}
\end{subequations}
The deformed Hamiltonian $H^{\delta}$, an $O(\delta)$ perturbation of $H^0$, is given by:
\begin{align}\label{eqn:ham-hon-delta-approx}
   (H^\delta \bm\psi)_{m,n}\ =\  \begin{pmatrix}
        (H^{\delta} \bm \psi)_{m,n}^A\\
        (H^{\delta} \bm \psi)_{m,n}^B
    \end{pmatrix} = \begin{pmatrix}
        \sum_{\nu=1}^3 \Big(1 + \delta t_1 f_\nu(\delta \Cell_{m,n})\Big)\psi_{m+m_\nu,n+n_\nu}^B\\
        \sum_{\nu=1}^3 \Big(1 + \delta t_1 f_\nu(\delta \Cell_{m-m_\nu,n-n_\nu}) \Big)\psi_{m-m_\nu,n-n_\nu}^A
    \end{pmatrix},
\end{align}
The operator $H^\delta $ is formally self-adjoint in $l^2(\mathbb Z^2;\mathbb C^2)$; see the Supplementary Material \ref{app:some-bds}.


Our goal is to study eigenvalue problems  for $H^\delta$:
\begin{equation}\label{eqn:eig-def-hon}	
    H^\delta \bm\psi = E \bm\psi
    \end{equation}
subject to a suitable choice of self-adjoint  boundary conditions; for  example:\\
(a) $\bm \psi = (\bm \psi_{m,n})_{m,n\in\mathbb Z} \in l^2(\mathbb Z^2;\mathbb{C}^2)$, or\\
(b) in the case where the deformation is invariant under with respect to $\bm v_2$  (and thus $f_\nu(\delta \Cell_{m,n})$ depends only on $m$), we impose:
\[ \textrm{$m\mapsto \bm \psi_{m,n}\in l^2(\mathbb Z;\mathbb C^2)$\quad  and \quad $\bm \psi_{m,n+1}=e^{iq_\parallel}\ \bm \psi_{m,n}$},\]
where $q_\parallel\in[-\pi,\pi]$ is the parallel quasi-momentum related associated with this translation invariance of $(H^\delta \bm \psi)_{m,n}$ with respect to $n\in\mathbb Z$.



The eigenvalue problem depends on two spatial scales: the lattice period of order $1$ and the length scale of the slowly varying deformation, which is of order $\delta^{-1}$. Hence, we seek solutions $\bm \psi=(\bm \psi_{m,n})\in l^2(\mathbb H)$  of the eigenvalue problem \eqref{eqn:eig-def-hon}, which depend on these two-scale length scales:
\begin{equation} 
\bm\psi_{m,n} = e^{i \bm k \cdot \bm x} \bm \Phi\big(
\bm X;\delta \big)\Big|_{\substack{\bm x = \Cell_{m,n}\\\bm X = \delta \Cell_{m,n}}} = e^{i \bm{k} \cdot \Cell_{m,n}} \bm \Phi\big(
\delta \Cell_{m,n};\delta \big),
\label{eq:ansatz}
\end{equation}
where $\bm{\Phi}(\bm X;\delta)=\big(\Phi^A(\bm X;\delta),\Phi^B(\bm X;\delta)\big)^T$ is a smooth and decaying function of the slow variable $\bm X=(X_1,X_2)^T$ that needs to be determined. 

\subsection{A formal discrete multiple scale expansion}\label{subsec:multi-scale-exp}

We begin by evaluating the expressions $\bm \psi_{m\pm m_\nu,n\pm n_\nu}$, which appear in $(H^\delta\psi)_{m,n}$,  for our 
Ansatz \eqref{eq:ansatz}:
\begin{align}
    \bm \psi_{m\pm m_\nu,n\pm n_\nu} = e^{\pm i \bm{k} \cdot \Cell_{m_\nu, n_\nu}} e^{i \bm{k} \cdot \Cell_{m,n}} \bm \Phi\big(\bm X \pm \delta \Cell_{m_\nu,n_\nu};\delta \big) \Big|_{\bm X = \delta \Cell_{m,n}}.  \label{eqn:main-exp-ansatz-1}
\end{align}
Here, we have used that \[
\Cell_{m\pm m_\nu, n\pm n_\nu} - \Cell_{m,n} = \pm \Cell_{m_\nu, n_\nu}\quad {\rm and}\quad  \Cell_{m_\nu, n_\nu} = m_\nu \bm v_1 + n_\nu \bm v_2 = \bm w_\nu ,\]
$\bm w_\nu$ is defined in  \eqref{eqn:c-vec}. 
%
%
The following is a proposition which provides provides a formal relation between self-adjoint eigenvalue problems  for the 
the discrete operator $H^\delta$ and corresponding continuum eigenvalue problems:
\begin{proposition}
\begin{enumerate}
\item If the expression \eqref{eq:ansatz} is a solution of the eigenvalue problem \eqref{eqn:eig-def-hon}, $H^\delta \bm \psi=E \bm \psi$  subject to the relevant self-adjoint  boundary conditions, then
the following system holds for all $\bm X = \{ \delta \Cell_{m,n}\}_{(m,n)\in\mathbb Z^2}$:
\begin{subequations}\label{eqn:multi-calc}
    \begin{align}
    \sum_{\nu = 1}^3 e^{i \bm{k} \cdot \bm w_\nu} \Big(1 + \delta t_1 f_\nu(\bm X) \Big) \Phi^B\big(\bm X + \delta \bm w_\nu;\delta \big)  &=    E\ \Phi^A\big(\bm X;\delta \big) \label{eqn:multi-calc-1}\\
     \sum_{\nu = 1}^3 e^{-i \bm{k} \cdot \bm w_\nu} \Big(1 + \delta t_1 f_\nu(\bm X-\delta\bm w_\nu) \Big) \Phi^A\big(\bm X-\delta\bm w_\nu;\delta \big) &=  E\  \Phi^B\big(\bm X;\delta \big)\label{eqn:multi-calc-2}
    \end{align}
\end{subequations}
\item Conversely, let $\Phi(\bm X;\delta)$ denote a solution of \eqref{eqn:multi-calc} subject to the relevant self-adjoint  boundary conditions. Let  $\bm\psi = (\bm \psi_{m,n})_{m,n\in\mathbb Z}$ where we define $\bm \psi_{m,n}:= e^{i\bm k\cdot \Cell_{m,n}}\ \bm\Phi(\Cell_{m,n})$. 
Then, $\bm \psi$ solves the discrete eigenvalue problem: $H^\delta\bm \psi = E\bm \psi$, with its relevant boundary conditions.
\end{enumerate}
\end{proposition}

We next solve  \eqref{eqn:multi-calc}  via a power series in $\delta$ using the approximations:
\begin{subequations}\label{eqn:main-exp-ansatz-2}
    \begin{align}
    \Phi^A\big(\bm X - \delta \bm w_\nu;\delta \big) &= \Phi^A \big(\bm X;\delta \big) - \delta \nabla_{\bm{X}} \Phi^A\big(\bm X;\delta \big) \cdot \bm w_\nu \nonumber\\
    & \quad+ \frac{\delta^2}{2} \left \langle\bm w_\nu, \ \Big(D^2_{\bm X} \Phi^A(\bm X; \delta)\Big) \ \bm w_\nu \right \rangle  + O(\delta^3), \label{eqn:main-exp-ansatz-2a}\\
    \Phi^B\big(\bm X + \delta \bm w_\nu;\delta \big) &=  \Phi^B \big(\bm X;\delta \big) + \delta \nabla_{\bm{X}} \Phi^B\big(\bm X;\delta \big) \cdot \bm w_\nu \nonumber\\
    & \quad + \frac{\delta^2}{2} \left \langle\bm w_\nu, \ \Big(D^2_{\bm X} \Phi^B(\bm X;\delta)\Big) \bm w_\nu \right \rangle  + O(\delta^3), \label{eqn:main-exp-ansatz-2b}
\end{align}
\end{subequations}
where $D^2_{\bm X} F(\bm X)$ denotes the Hessian matrix of $F$. The term  $f_\nu(\delta \bm X - \delta \bm w_\nu)$ is similarly expanded.  

Equations  \eqref{eqn:multi-calc} and \eqref{eqn:main-exp-ansatz-2} yield and formal differential equation which can be expanded to arbitrary in $\delta$. We seek a formal solution of the form: 
\begin{subequations}
\label{eqn:EXPAND}
\begin{align}
\bm \Phi(\bm X;\delta) &= \bm \Phi_0(\bm X) + \delta \bm \Phi_1(\bm X) + \delta^2 \bm \Phi_2(\bm X) + \dots \label{eqn:Phi-eqn}\\
E^\delta &= E_0 + \delta E_1 + \delta^2 E_2 + \dots, \label{eqn:main-eig-val}
\end{align}
\end{subequations}
at first without attention to boundary conditions. When discussing particular deformations, e.g. unidirectional, we'll further constrain the formal solution constructed in this section to satisfy the appropriate boundary conditions as well.

Substituting \eqref{eqn:EXPAND} into \eqref{eqn:multi-calc}, \eqref{eqn:main-exp-ansatz-2},   and then equating terms of like order in  $\delta$, yields a hierarchy of equations for $\bm \Phi_j(\bm X)$ and $E_j$, for $j\ge0$.
%
The first equations of this hierarchy are displayed  and solved in Appendix \ref{app:dev-2d-eff-eqn}, with a leading term, $\Phi_0$,
chosen so that \eqref{eq:ansatz} has the structure of a wave packet, which is spectrally localized at a Dirac point.  We now summarize the results of these calculations.

\paragraph{At order $\delta^0$:} We collect order $\delta^0$ terms in \eqref{eqn:multi-calc} and obtain
\begin{align}\label{eqn:main-order-0}
    \widetilde{H}^0(\bm k) \begin{pmatrix}
        \Phi_0^A\big(\bm X \big)\\
        \Phi_0^B\big(\bm X \big)
    \end{pmatrix} = E_0 \begin{pmatrix}
        \Phi_0^A\big(\bm X \big)\\
        \Phi_0^B\big(\bm X \big)
    \end{pmatrix}
\end{align}
where the $2\times2$ Hermitian matrix $\widetilde{H}^0(\bm k)$ is displayed in \eqref{eqn:tb-bloch-hamiltonian}. 
 The matrix $\widetilde{H}^0(\bm k)$ vanishes at $\bm k$ for which $\sum_{\nu=1}^3e^{i\bm k\cdot\bm w_\nu}=0$, a condition which holds precisely 
at $\bm K$ or at $\bm K'$ or their dual lattice translates. Henceforth, we fix $\bm k=\bm K$; the $\bm K'$ point is treated similarly.  Since $H(\bm K)=0$,  $E_0=0$ is an eigenavlue of multiplicity two and we therefore can take $\bm \Phi_0(\bm X)\in\mathbb C^2$ to be arbitrary. The $\bm X-$ variation of $\Phi_0$ is determined at the next order.

\paragraph{At order $\delta^1$:} We obtain an eigenvalue problem
\begin{align}\label{eqn:main-order-1}
    \mathcal{H}_\text{eff}\begin{pmatrix}
        \Phi^A_0\big(\bm X \big)\\
        \Phi^B_0\big(\bm X \big)
    \end{pmatrix} = E_1\ \begin{pmatrix}
        \Phi^A_0\big(\bm X \big)\\
        \Phi^B_0\big(\bm X\big)
    \end{pmatrix}\ , 
\end{align}
along with accompanying boundary conditions on $\bm\Phi_0$.
Here,  $\mathcal{H}_\text{eff}$ is the off diagonal matrix operator, given by
\begin{align}\label{eqn:eff-ham}
\mathcal{H}_\text{eff}\bm = \begin{pmatrix}
0 & \sum_{\nu = 1}^3 e^{i \bm K \cdot \bm w_\nu} \Big(\bm w_\nu \cdot \nabla_{\bm{X}} + t_1 f_\nu(\bm X)\Big)\\
    \sum_{\nu = 1}^3 e^{-i \bm K \cdot \bm w_\nu} \Big(-\bm w_\nu \cdot \nabla_{\bm{X}} + t_1 f_\nu(\bm X)\Big) & 0 
\end{pmatrix}\ .
\end{align}
In Proposition \ref{prop:uni-equivalence}, we will show that $\mathcal{H}$ is unitarily equivalent to a two dimensional magnetic Dirac operator.


Given a  nontrivial solution $(\bm \Phi_0,E_1)$ of \eqref{eqn:main-order-1}, we proceed to order $\delta^2$. 

\paragraph{At order $\delta^2$:} In  Appendix \ref{app:dev-2d-eff-eqn} we find that $\bm \Phi_1$, $E_2$ satisfy the following non-homogeneous equation:
{\begin{equation}\label{eqn:next-order}
	 \Big(\mathcal{H}_{\text{eff}} - E_1\Big) \bm{\Phi}_1 = E_2 \bm{\Phi}_0 + \bm{R}_2[\bm{\Phi}_0], \qquad \bm{R}_2[\bm{\Phi}_0] = (R_2^A, R_2^B)^T,
\end{equation}
}where $\bm{R}_2[\bm{\Phi}_0]$ depends on $\bm{\Phi}_0$ and is given by:
\begin{subequations}\label{eqn:2d-rem-terms}
    \begin{align}
        R_2^A[\bm \Phi_0] &= - \sum_{\nu = 1}^3 e^{i \bm K \cdot \bm w_\nu} \left(\frac{1}{2}  \left \langle\bm w_\nu, \ \Big(D^2_{\bm X} \Phi^B_0 (\bm X)\Big) \bm w_\nu \right \rangle + t_1 f_\nu(\bm X) \big(\nabla_{\bm X}\Phi_0^B(\bm X) \cdot \bm w_\nu \big) \right), \label{eqn:2d-rem-terms-a}\\
        R_2^B[\bm \Phi_0] &= \sum_{\nu = 1}^3 e^{-i \bm K \cdot \bm w_\nu} \left(-\frac{1}{2} \left \langle\bm w_\nu, \ \Big(D^2_{\bm X} \Phi^A_0 (\bm X)\Big) \bm w_\nu \right \rangle + t_1 \nabla_{\bm X} \big(f_\nu(\bm X)\Phi_1^A(\bm X)\big) \cdot \bm w_\nu\right). \label{eqn:2d-rem-terms-b}
    \end{align}
\end{subequations}

The system \eqref{eqn:next-order} has a solution in an appropriate space if and only if the right hand side is orthogonal to the nullspace of $\mathcal{H}_{\text{eff}} - E_1$. This solvability condition determines  $E_2$.

\paragraph{At order $\delta^{k+1}$:} Similarly, $\bm \Phi_k(\bm X) = (\Phi_k^A(\bm X), \Phi_k^B(\bm X))^T$ solves an equation of the form
\begin{equation}
    \Big(\mathcal{H}_\text{eff} - E_1\Big) \bm \Phi_k = E_{k+1} \bm \Phi_0 + \bm{R}_{k+1}[\bm \Phi_0, \dots, \bm \Phi_{k-1}]\label{eqn:main-order-k}
\end{equation}
whose solvability is ensured by a choice of $E_{k+1}$.

\subsection{$\mathcal{H}_\text{eff}$, the effective Hamiltonian, is a magnetic Dirac operator} \label{subsec:eff-2d-mag-opt}


The effective Hamiltonian $\mathcal{H}_\text{eff}$ is unitarily equivalent to a magnetic Dirac operator $\mathcal{D}_{\bm{A}}$ with effective magnetic vector potential $\bm A_{\text{eff}}$ determined by the strain matrix $\nabla_{\bm X} \bm u$. We state this argument as the following Proposition:
\begin{proposition}\label{prop:uni-equivalence}
    The effective Hamiltonian $\mathcal{H}_\text{eff}$ is unitarily equivalent to a magnetic Dirac operator $\mathcal{D}_{\bm{A}}$
    \begin{align}
        \mathcal{H}_\text{eff} = U \mathcal{D}_{\bm{A}} U^*, \qquad U = \text{diag}(e^{-i\frac{\pi}{6}}, e^{i\frac{\pi}{6}}), \label{eqn:H-mag}
    \end{align}
    and the magnetic Dirac operator $\mathcal{D}_{\bm{A}}$ with magnetic vector potential $\bm A = (A_1, A_2)$ is given by
    \begin{equation}\label{eqn:dirac-mag}
    	\mathcal{D}_{\bm{A}} := \frac{3}{2}\Big[(-i \partial_{X_1} - A_1)\sigma_1 + (-i\partial_{X_2} -A_2) \sigma_2\Big],
    \end{equation}
    where the two components $A_1$ and $A_2$ are
    \begin{equation}\label{eqn:eff-mag}
    	A_1(\bm{X}) = -\frac{t_1}{2}(\partial_{X_1} u_1 - \partial_{X_2} u_2), \qquad A_2(\bm{X}) = \frac{t_1}{2}(\partial_{X_1} u_2 + \partial_{X_2} u_1).
    \end{equation}
\end{proposition}

The proof of Proposition \ref{prop:uni-equivalence} is given in Supplementary Material \ref{app:dirac-op}. By Proposition \ref{prop:uni-equivalence}, we know that $\mathcal{H}_\text{eff}$ is self-adjoint since $\mathcal{D}_{\bm{A}}$ is self-adjoint. 


The induced vector potential $\bm A$ gives rise to a pseudo-magnetic field $\bm B_\text{eff}$ via
\begin{equation}\label{eqn:mag-field}
    \bm B_\text{eff}=\nabla\times\bm A
    =(\partial_{X_1}A_2-\partial_{X_2}A_1)\,\widehat{\bm z},
\end{equation}
where $\widehat{\bm{z}}$ points perpendicular to the $(X_1,X_2)$-plane. 



\begin{remark}[Comparing the discrete and continuous case]
	We compare the effective operator obtained in our discrete setting and the continuous setting in \cite{guglielmon2021landau}, which is (see Equation D6 in \cite{guglielmon2021landau})
	\begin{equation*}
		\mathcal{H}_\text{eff}^\text{cont} = v_D \Big[(-i \partial_{X_1} - A_1)\sigma_1 + (-i\partial_{X_2} -A_2) \sigma_2\Big] + W_\text{eff}\:\sigma_0,
	\end{equation*}
	where $v_D$ is a constant given by the continuous material and $W_\text{eff} = \nabla \cdot \bm{u}$ is the effective electric potential. 
    
    We notice that our effective Dirac operator $\mathcal{D}_{\bm{A}}$ in \eqref{eqn:dirac-mag} does not have such an effective electric potential. A brief explanation is that our discrete setting has chiral symmetry -- $A$ nodes only commutes with $B$ nodes (vice versa). Therefore, the effective Dirac operator preserves such a chiral symmetry, forcing the diagonal entries to vanish. One may use the Haldane model (see e.g. \cite{haldane1988model}), which considers the next-nearest neighbors, to obtain the effective electric potential. Nevertheless, all deformations considered in this paper has the property $W_\text{eff} = \nabla \cdot \bm{u} = 0$.
\end{remark}




\section{Unidirectional deformations and the $l^2_{k_\parallel}$-eigenvalue problem}\label{sec:uni-displacement}




In this section, we study an important class of unidirectional deformations, defined by displacements of the form:
\begin{equation}\label{eqn:uni-disp}
	\bm{u}(X_1, X_2) = \Big(0, d(X_1)\Big)^T.
\end{equation}
Such deformations of the honeycomb are independent of $X_2$ and respect translation invariance in the $\bm a_2$ direction; see Figure \ref{fig:ac-edge-def}. We take motivation from the special  case $\bm{u} = (0, X_1^2)^T$, which induces an effective magnetic potential agreeing with the Landau gauge for a constant perpendicular magnetic field:
\begin{align}
    {\bm A}_{\rm eff} = \frac{3}{2}\Big[-i \partial_{X_1} \sigma_1 + (-i\partial_{X_2} -t_1 X_1) \sigma_2\Big],\quad B_{\rm eff}=\nabla \times {\bm A}_{\rm eff} = t_1 \widehat{\bm z}. \label{eqn:main-landau-gauge}
\end{align}
This quadratic unidirectional deformation has been explored in many physical settings (see e.g. \cite{barczyk2024observation,barsukova2024direct}). \edit{A review of the spectral properties of the magnetic Dirac operator with Landau gauge in \eqref{eqn:main-landau-gauge}, including the emergence of Landau levels, is given in the Appendix is presented in Appendix \ref{sec:ac-zz-edge}; see Lemma \ref{lemma:spectrum-landau-dirac-mag-op}.}



Given a general and smooth unidirectional displacement of the form \eqref{eqn:uni-disp}, we notice that the deformed honeycomb under the deformation $\bm x \mapsto \bm x + \bm u(\delta \bm x)$ remains periodic in the vertical direction. \edit{We focus primarily on the unidirectional displacement in \eqref{eqn:uni-disp} with bounded deformation gradients, i.e. $d'(X_1)$ (or equivalently, $\nabla_{\bm X} \bm u$) is uniformly bounded, i.e. $d'(X_1) \in C_b^\infty(\mathbb{R})$.} 


\edit{Therefore, given $\delta$ and $k_\parallel$, we seek solutions to the eigenvalue problem \eqref{eqn:eig-def-hon} in the $l_{k_\parallel}^2$-space, i.e. 
\begin{subequations}\label{eqn:lk-evp}
    \begin{align}
    (H^\delta \bm \psi)_{m,n} &= E \bm \psi_{m,n}, \qquad \bm \psi \in l^2_{k_\parallel}, \label{eqn:lk-evp-a}\\
    \text{where} \quad l^2_{k_\parallel} &:= \{\bm (\psi_{m,n})_{m,n \in \mathbb{Z}} \ | \ \bm \psi_{m,n+1} = e^{i \delta k_\parallel} \bm \psi_{m,n} \text{ and } m \in \mathbb{Z} \mapsto \bm \psi_{m,n} \in l^2_m(\mathbb{Z})\}.\label{eqn:lk-evp-b}
    \end{align}
\end{subequations}
In what follows, we solve \eqref{eqn:lk-evp} by a formal expansion; a rigorous justification of the existence of eigenstates approximated by this expansion is given in Section \ref{sec:main-thm}.
}

\paragraph{The wave-packet ansatz for unidirectional deformations} Due to the translation invariance in the vertical direction, we seek a special form of the wave-packet ansatz \eqref{eq:ansatz} that is (1) quasi-periodic vertically and (2) rapidly decaying transversely.

Recall that our ansatz \eqref{eq:ansatz} is obtained by evaluating $e^{i \bm{K}\cdot \bm{x}}\,\bm{\Phi}(\bm X;\delta)$
at $\bm{x}= \Cell_{m,n}$ and $\bm X = \delta \Cell_{m,n}$. For unidirectional deformations \eqref{eqn:uni-disp}, it is natural to impose the quasi-periodicity directly at the level of the envelope $\bm \Phi(\bm X)$ by taking
\begin{subequations}\label{eqn:uni-phi}
    \begin{align}
     &\bm \Phi(\bm X;\delta) = e^{i \frac{k_\parallel}{2\pi} \bm X \cdot \bm a_2} \bm \Psi^{\#}\left(\frac{\bm X \cdot \bm a_1}{|\bm a_1|}; k_\parallel, \delta\right), \label{eqn:uni-phi-1}\\
     &\bm \Psi^{\#}\left(\frac{\bm X \cdot \bm a_1}{|\bm a_1|}; k_\parallel,\delta\right) \rightarrow 0 \quad \text{as} \quad |\bm X \cdot \bm a_1| \rightarrow \infty,\label{eqn:uni-phi-2}
\end{align}
\end{subequations}
where $\bm a_1, \bm a_2$ are the dual lattice vectors given in \eqref{eqn:dual-vec} and $k_\parallel \in [-\pi, \pi]$ is the parallel quasi-momentum associated the translation-invariance in the vertical direction. Since $\bm a_1$ is horizontal, we have $\frac{\bm X \cdot \bm a_1}{|\bm a_1|} = X_1$, hence $\bm \Psi^{\#}(\cdot;k_\parallel,\delta)$ depends only on the scalar variable $X_1$ with
\begin{align}\label{eqn:uni-psi-simp}
     \bm \Psi^{\#}\left(\frac{\bm X \cdot \bm a_1}{|\bm a_1|}; k_\parallel, \delta\right) &= \bm \Psi^{\#}\left(X_1; k_\parallel, \delta\right),\qquad \bm \Psi^{\#}\left(X_1; k_\parallel,\delta\right) \rightarrow 0 \quad \text{as} \quad |X_1| \rightarrow \infty.
\end{align}

Our special ansatz \eqref{eqn:uni-phi} is designed to be (1) quasi-periodic in the vertical direction and (2) rapidly localized in the transverse direction. Since $\bm \Psi^{\#}(X_1;k_\parallel,\delta)$ decays rapidly, it encodes the transverse localization of $\bm \Phi(\bm X;\delta)$ in \eqref{eqn:uni-phi}. We also observe that $\bm \Phi(\bm X;\delta)$ in \eqref{eqn:uni-phi} is quasi-periodic along the $\bm v_2$ (vertical) direction with
\begin{align*}
    \bm \Phi(\bm X+\bm v_2;\delta) = e^{i \frac{k_\parallel}{2\pi}\bm v_2 \cdot \bm a_2} \bm \Phi(\bm X;\delta) = e^{i k_\parallel} \bm \Phi(\bm X;\delta).
\end{align*}
Moreover, given $\bm \Phi(\bm X;\delta)$ in \eqref{eqn:uni-phi}, our ansatz \eqref{eq:ansatz} satisfies the Bloch boundary condition with
\begin{align}
    \bm \psi_{m,n+1} &= e^{i \bm K \cdot (\bm x + \bm v_2)} \bm \Phi(\bm X +\delta \bm v_2;\delta) \Big|_{\substack{\bm x = \Cell_{m,n}\\ \bm X = \delta \Cell_{m,n}}} = e^{i\delta k_\parallel} \bm \psi_{m,n}, \label{eqn:bloch-bdry}
\end{align}
and decays in the $m$-direction. The following proposition provides a formal relation between the self-adjoint eigenvalue problem for the discrete operator $H^\delta$ in \eqref{eqn:ham-hon-delta-approx} and the corresponding continuum eigenvalue problems in the unidirectional deformation setting:

\begin{proposition} When the displacement $\bm u$ is of the form \eqref{eqn:uni-disp}: $\bm{u}(X_1, X_2) = \Big(0, d(X_1)\Big)^T$. 
\begin{enumerate}
\item If the expression 
\begin{align}\label{eqn:uni-ansatz}
    \bm\psi_{m,n} = e^{i \bm k \cdot \bm x} \bm \Phi(\bm X;\delta)\Bigg|_{\substack{\bm x = \Cell_{m,n}\\\bm X = \delta \Cell_{m,n}}}, \quad \text{with} \quad &\bm \Phi(\bm X;\delta) = e^{i \frac{k_\parallel}{2\pi} \bm X \cdot \bm a_2} \bm \Psi^{\#}\left(\frac{\bm X \cdot \bm a_1}{|\bm a_1|}; k_\parallel, \delta\right),
\end{align}
is a solution of the $l_{k_\parallel}^2$-eigenvalue problem \eqref{eqn:lk-evp}, then $\bm \Phi(\bm X;\delta)$ in \eqref{eqn:uni-ansatz} satisfies the system \eqref{eqn:multi-calc} for all $\bm X = \{ \delta \Cell_{m,n}\}_{(m,n)\in\mathbb Z^2}$ and $\bm \Psi^{\#}(X_1;k_\parallel,\delta)$ decays with
\begin{align}
    &\bm \Psi^{\#}(X_1;k_\parallel,\delta) \rightarrow 0, \quad \text{as} \quad |X_1| \rightarrow \infty, \quad \text{and} \quad \norm{\bm \Psi^{\#}(X_1;k_\parallel,\delta)}_{L^2_{X_1}(\mathbb{R})} < \infty.\label{eqn:uni-ansatz-2}
\end{align}

\item Conversely, let $\Phi(\bm X;\delta)$ in the form of \eqref{eqn:uni-ansatz} be a solution of \eqref{eqn:multi-calc} with $\bm \Psi^{\#}(X_1;k_\parallel,\delta)$ satisfying \eqref{eqn:uni-ansatz-2}. Then the expression \eqref{eqn:uni-ansatz} solves the $l^2_{k_\parallel}$-eigenvalue problem.
\end{enumerate}
\end{proposition}

To solve for the system \eqref{eqn:multi-calc} with $\bm \Phi(\bm X;\delta)$ of the form \eqref{eqn:uni-phi}, we expand $\bm \Psi^{\#}(X_1;k_\parallel,\delta)$ in \eqref{eqn:uni-phi} and the eigenvalue $E$ (now depending on $k_\parallel$) in powers of $\delta$:
\begin{equation}
    \bm \Psi^{\#}(X_1;k_\parallel,\delta) = \sum_{j\geq 0} \delta^j \bm \Psi^{\#}_j(X_1;k_\parallel), \qquad E^\delta(k_\parallel) = \sum_{j\geq 1} \delta^j E_j(k_\parallel), \label{eqn:main-psi-exp}
\end{equation}
where $\bm \Psi^{\#}_j(X_1;k_\parallel)$ are smooth and decaying functions of $X_1$. 

In Section \ref{sec:derivation}  we have derived, for general displacements $\bm u(\bm X)$,  a hierarchy of equations for terms, $\bm \Phi_j(\bm X)$, in the formal expansion of $\bm \Phi(\bm X;\delta)$. We now adapt these calculations to obtain  the expansion coefficients $\bm \Psi^{\#}_j(X_1;k_\parallel)$ of $\bm \Psi^{\#}(X_1;k_\parallel,\delta)$ in \eqref{eqn:main-psi-exp} by substituting
\begin{align}
    \bm \Phi_j(\bm X;k_\parallel) = e^{i \frac{k_\parallel}{2\pi} \bm X \cdot \bm a_2} \bm \Psi^{\#}_j\left(\frac{\bm X \cdot \bm a_1}{|\bm a_1|}; k_\parallel\right), \qquad j\geq 0 \label{eqn:phi-psi-relation-0}
\end{align}
into the equations for $\bm \Phi_j(\bm X;k_\parallel)$. 

For future convenience, we choose $\bm \Psi_j^{\#}(X_1;k_\parallel)$ in \eqref{eqn:phi-psi-relation-0} in the following form
\begin{subequations}\label{eqn:uni-psi-relation}
    \begin{align} 
     &\bm \Psi_j^{\#}\left(X_1; k_\parallel\right) = e^{-i\frac{\sqrt{3}}{3}k_\parallel X_1} U \bm \Psi_j\left(X_1; k_\parallel\right), \label{eqn:uni-psi-relation-1}\\
     & \bm \Psi_j\left(X_1; k_\parallel\right) \rightarrow 0, \quad \text{as} \quad |X_1| \rightarrow 0, \label{eqn:uni-psi-relation-2}
\end{align}
\end{subequations}
where $U=\text{diag}(e^{-i\frac{\pi}{6}},e^{i\frac{\pi}{6}})$ is defined in \eqref{eqn:H-mag}\footnote{The matrix $U$ and the extra phase $e^{-i\frac{\sqrt{3}}{3}k_\parallel X_1}$ are introduced here so that the resulting equations for $\bm \Psi_j(X_1;k_\parallel)$ take the form of one-dimensional Dirac equations; see \eqref{eqn:1st-expansion} below. }. Consequently, using \eqref{eqn:phi-psi-relation-0} and \eqref{eqn:uni-psi-relation}, $\bm \Phi_j(\bm X)$ and $\bm \Psi_j(X_1;k_\parallel)$ are related by
\begin{equation}
    \bm \Phi_j(\bm X;k_\parallel) = e^{i \frac{k_\parallel}{2\pi} \bm X \cdot \bm a_2} \Big(e^{-i\frac{\sqrt{3}}{3}k_\parallel X_1} U \bm \Psi_j\left(X_1; k_\parallel\right)\Big)= e^{i\frac{k_\parallel}{3} X_2} U \bm \Psi_j(X_1;k_\parallel),\label{eqn:phi-psi-relation}
\end{equation}
where the last equality holds by applying the coordinate of $\bm a_2$ in \eqref{eqn:dual-vec}.


We now substitute \eqref{eqn:phi-psi-relation} into equations \eqref{eqn:main-order-1} and \eqref{eqn:next-order} of $\bm \Phi_j(\bm X)$ to obtain the equations for $\bm \Psi_j(X_1;k_\parallel)$ with $j\geq 0$. Here we present the effective envelope equations for $j=0,1$; the cases $j\geq 2$ are analogous and omitted. The full derivation is given in Appendix \ref{app:dev-1d-eff-eqn}.
\begin{enumerate}[(1)]
    \item $\bm \Psi_0(X_1;k_\parallel)$ is the eigenfunction of a one-dimensional effective Dirac operator $\mathcal{D}(k_\parallel)$ with associated eigenvalue $E_1(k_\parallel)$, i.e.
\begin{equation}\label{eqn:1st-expansion}
	\mathcal{D}(k_\parallel) \bm{\Psi}_0 = E_1(k_\parallel) \bm{\Psi}_0,
\end{equation}
and $\mathcal{D}(k_\parallel)$ is given by
\begin{equation}\label{eqn:dirac-1d}
	\mathcal{D}(k_\parallel) := \frac{3}{2} \Big[-i\partial_{X_1} \sigma_1 + \kappa(X_1;k_\parallel) \sigma_2\Big], \qquad \kappa(X_1;k_\parallel) = \frac{k_\parallel}{3} - \frac{t_1}{2} d'(X_1);
\end{equation}
For $k_\parallel$ in a suitable interval, the spectrum of  $\mathcal{D}(k_\parallel)$ consists of two semi-infinite intervals tend to $\pm\infty$, separated by a spectral gap,  and a finite number of eigenvalues in this spectral gap. The eigenvalues of $\mathcal{D}(k_\parallel)$ are simple; see Lemma \ref{lemma:1d} in Section \ref{sec:uni-def-bdd}.

\item $\bm{\Psi}_1(X_1;k_\parallel)$ and $E_2(k_\parallel)$ satisfy
\begin{equation}\label{eqn:2nd-expansion}
	\begin{aligned}
		\Big(\mathcal{D}(k_\parallel)  - E_1 (k_\parallel) \Big) \bm{\Psi}_1 = E_2(k_\parallel) \bm{\Psi}_0 + \bm{\mathcal{R}}_2[\bm \Psi_0;k_\parallel],
    \end{aligned}
\end{equation}
where $\bm{\mathcal{R}}_2[\bm \Psi_0;k_\parallel] = (\mathcal{R}_2^A, \mathcal{R}_2^B)^T$ are the remaining terms related to $k_\parallel$ and $\bm \Psi_0$, given by
\begin{subequations}\label{eqn:app-rem-delta-2}
\begin{align}
    \mathcal{R}_2^A &= -\frac{1}{8} k_{\parallel}^2 \Psi_0^B - \left(\sqrt{3} e^{i\frac{5}{3}\pi} t_1 f_2  - \frac{\sqrt{3}}{4} k_\parallel i \right) \partial_{X_1} \Psi_0^B - \left( \frac{3}{2} e^{i\frac{5}{3}\pi} - \frac{3}{8} \right)\partial_{X_1}^2 \Psi_0^B, \label{eqn:app-rem-delta-2-1}\\
    \mathcal{R}_2^B &= -\frac{1}{8} k_{\parallel}^2 \Psi_0^A + \left(\frac{\sqrt{3}}{4} k_\parallel i \right) \partial_{X_1} \Psi_0^A + \sqrt{3} e^{-i\frac{5}{3}\pi} t_1 \partial_{X_1}\left(f_2 \Psi_0^A\right)- \left( \frac{3}{2} e^{-i\frac{5}{3}\pi} - \frac{3}{8} \right)\partial_{X_1}^2 \Psi_0^A. \label{eqn:app-rem-delta-2-2}
\end{align}
\end{subequations}
Since $E_1(k_\parallel)$ is a simple eigenvalue, $E_2(k_\parallel)$ is determined by a solvability condition
\begin{equation}\label{eqn:E2}
	-E_2(k_\parallel) \norm{\bm \Psi_0}^2_{L^2(\mathbb{R})} = \Big\langle \Psi_0^A, \mathcal{R}_2^A \Big\rangle_{L^2(\mathbb{R})} + \Big \langle\Psi_0^B, \mathcal{R}_2^B \Big\rangle_{L^2(\mathbb{R})},
\end{equation}
and $\bm \Psi_1(X_1;k_\parallel)$ is determined in the subspace orthogonal to $\bm \Psi_0(X_1;k_\parallel)$.

\end{enumerate}

We also construct approximate solutions to the eigenvalue problem \eqref{eqn:eig-def-hon} subject to the Bloch boundary condition \eqref{eqn:bloch-bdry}, which is stated as the following Proposition:


\begin{proposition}\label{prop:uni-approx}
Consider a unidirectional displacement $\bm u = (0,d(X_1))^T$ with bounded deformation gradient, i.e. $d'(X_1) \in C^\infty_b(\mathbb{R})$. Fix $l\geq 0$. Let $\bm \Psi_j(X_1;k_\parallel)$ satisfy the effective envelope equation for $0\leq j\leq l$. Then the pair $\bm \psi^{(l)} = (\bm \psi^{(l)}_{m,n})_{m,n \in \mathbb{Z}}$ and $E^{(l+1)}$, defined by 
\begin{subequations}\label{eqn:uni-ansatz-approx}
    \begin{align}
        &\bm \psi_{m,n}^{(l)} = e^{i \bm K \cdot \bm x} e^{i \frac{k_\parallel}{2\pi} \bm X \cdot \bm a_2} \Bigg(\sum_{j=0}^l \delta^j \bm \Psi^{\#}_j\left(\frac{\bm X \cdot \bm a_1}{|\bm a_1|}; k_\parallel\right) \Bigg) \bigg|_{\substack{\bm x = \Cell_{m,n}\\\bm X = \delta \Cell_{m,n}}} \label{eqn:uni-ansatz-approx-1}\\
        &\qquad \ =e^{i \bm K \cdot \bm x} e^{i \frac{k_\parallel}{3}X_2} U \Bigg(\sum_{j=0}^l \delta^j \bm \Psi_j\left(X_1; k_\parallel\right) \Bigg)\bigg|_{\substack{\bm x = \Cell_{m,n}\\\bm X = \delta \Cell_{m,n}}}, \label{eqn:uni-ansatz-approx-2}\\
         &E^{(l+1)}(k_\parallel) = \sum_{j=1}^{l+1} \delta^j E_j(k_\parallel), \label{eqn:uni-E-approx}
    \end{align}
\end{subequations} 
form approximate \edit{$l^2_{k_\parallel}$-eigenpairs} of $H^\delta \bm \psi_{m,n}=E\bm \psi_{m,n}$ at order $O(\delta^{l})$. Moreover, the remainders of these approximate are of order $O(\delta^{l+1})$, i.e.
\begin{align}
     \sup_{n \in \mathbb{Z}}\norm{\Big[\ \Big(H^\delta  - E^{(l+1)}\Big) \bm \psi^{(l)}\ \Big]_{m,n}}_{l^2_m(\mathbb{Z};\mathbb{C}^2)}  \lesssim O(\delta^{l+1}).\label{eqn:eig-approx-rem} 
\end{align}

\end{proposition}

The proof of Proposition \ref{prop:uni-approx} is a straightforward calculation when the effective envelope equations for $\bm \Psi_j(X_1;k_\parallel)$ are satisfied. Here we briefly outline the proof and omit the detailed calculation. When $d'(X_1)$ is bounded, the difference $H^\delta - H^0$ is a bounded operator for any fixed $\delta$. Moreover, any eigenfunction $\bm \Psi_0(X_1;k_\parallel)$ associated with an eigenvalue $E_1(k_\parallel)$ of $\mathcal{D}(k_\parallel)$ decays exponentially; see Lemma \ref{lemma:1d}(c-d). Consequently, the remainder terms involving $\bm \Psi_0(X_1;k_\parallel)$ are bounded in the given norm in \eqref{eqn:eig-approx-rem}, and the terms for $j\geq 1$ are controlled similarly.

While Proposition \ref{prop:uni-approx} provides bounds on the remainder terms, our main result, Theorem \ref{thm:main}, goes further by proving the existence of an eigenstate approximated by \eqref{eqn:uni-ansatz}, together with estimates for the corresponding correctors.

\subsection{Unidirectional deformations with bounded gradients}\label{sec:uni-def-bdd}


In this section, we present the spectral properties for $\mathcal{D}(k_\parallel)$ in \eqref{eqn:dirac-1d} when the unidirectional displacement $\bm u = (0,d(X_1))^T$ has bounded deformation gradient, i.e. $d'(X_1)$ is bounded. 


The spectrum of $\mathcal{D}(k_\parallel)$ when $d'(X_1)$ is bounded has been well-studied (see e.g. \cite{fefferman2017topologically,fefferman2014topologically,lu2020dirac}). We now state the spectral properties of $\mathcal{D}(k_\parallel)$. For simplicty, we recall $\kappa(X_1;k_\parallel) = k_\parallel/3 - t_1 d'(X_1)/2$ defined in \eqref{eqn:dirac-1d}. Notice that $\kappa(X_1;k_\parallel)$ being bounded is equivalent to $d'(X_1)$ being bounded.



\begin{lemma}[Spectral properties of $\mathcal{D}(k_\parallel)$ with bounded $\kappa(X_1;k_\parallel)$]\label{lemma:1d}
	  Fix $k_\parallel$. Assume that $\kappa(X_1;k_\parallel)$ is continuous in $X_1$ and behaves like a domain-wall type with limits of opposite sign at $\pm \infty$, i.e.
      \begin{subequations}\label{eqn:kappa-limit}
          \begin{align}
              &\kappa_\pm(k_\parallel) := \lim_{X_1 \rightarrow \pm \infty} \kappa(X_1; k_\parallel), \label{eqn:kappa-limit-1}\\
              &\kappa_+(k_\parallel) \: \kappa_-(k_\parallel) < 0.\label{eqn:kappa-limit-2}
          \end{align}
      \end{subequations}
      We further assume an integrability condition on $\kappa(X_1;k_\parallel)$, i.e.
      \begin{equation}\label{eqn:kappa-condition-2}
		\kappa(X_1;k_\parallel) - \kappa_+(k_\parallel) \in L^1([0,\infty)), \qquad \kappa(X_1;k_\parallel) - \kappa_-(k_\parallel) \in L^1((-\infty,0]).
	\end{equation}
	Then, the following statements hold for $\mathcal{D}(k_\parallel)$ in \eqref{eqn:dirac-1d}:
	\begin{enumerate}[(a)]
		\item $\mathcal{D}(k_\parallel)$ has essential spectrum equal to $(-\infty, -a(k_\parallel)] \cup [a(k_\parallel), \infty)$, where $a(k_\parallel)$ is given by
		\begin{equation}\label{eqn:kappa-pm}
			a(k_\parallel) := \frac{3}{2} \min \Big\{\big|\kappa_+(k_\parallel)\big|, \big|\kappa_-(k_\parallel)\big|\Big\},
		\end{equation}
		
		\item The eigenvalue problem \eqref{eqn:1st-expansion} has a simple zero eigenvalue, i.e. $E_1 = 0$ (independent of $k_\parallel$) and the corresponding zero eigenfunction $\bm \Psi_0(X_1;k_\parallel)$ has exponential decay, i.e. there exists $\lambda_-(k_\parallel) < 0 < \lambda_+(k_\parallel)$ and $\bm p_\pm(k_\parallel) \in \mathbb{R}^2$ such that
         \begin{align}
            &\lim_{X_1 \rightarrow +\infty} \bm \Psi_0(X_1;k_\parallel) \ e^{\lambda_+(k_\parallel) X_1} = \bm p_+(k_\parallel),  \qquad \lim_{X_1 \rightarrow -\infty} \bm \Psi_0(X_1;k_\parallel) \ e^{\lambda_-(k_\parallel) X_1} = \bm p_-(k_\parallel).\label{eqn:main-psi0-decay}
        \end{align}
		  Furthermore, when all higher-order derivatives are bounded, i.e. $\kappa(X_1;k_\parallel) \in C_b^\infty(\mathbb{R})$, all derivatives of $\bm \Psi_0(X_1;k_\parallel)$ have exponential decay and 
          \begin{equation}\label{eqn:main-psi0-Hs-bd}
			\| \bm \Psi_0(\cdot;k_\parallel)\|_{H^s(\mathbb{R})} < \infty, \qquad \forall s \in \mathbb{N}.
		\end{equation}
		
		\item $\mathcal{D}(k_\parallel)$ may have non-zero eigenvalues that live in the spectral gap, i.e. $E_1(k_\parallel) \in \big(-a(k_\parallel), a(k_\parallel)\big)$. When such an eigenvalue exists, it is simple and its corresponding eigenstates also have exponential decay and satisfy \eqref{eqn:main-psi0-decay}. Furthermore, when $\kappa(X_1;k_\parallel) \in C^\infty_b(\mathbb{R})$, all derivatives of $\bm \Psi_0(X_1;k_\parallel)$ have exponential decay and \eqref{eqn:main-psi0-Hs-bd} holds.
	\end{enumerate}
\end{lemma}

Lemma \ref{lemma:1d} is obtained by adapting classical ODE results from \cite{coddington1956theory} and spectral properties of $\mathcal{D}(k_\parallel)$ (see e.g. Theorem 4.2 in \cite{fefferman2017topologically}). We provide a detailed proof of Lemma \ref{lemma:1d} in the Supplementary Material (see the discussion of Lemma \ref{lemma:sm-lemma-ode-asy}). 




Here we briefly explain why the domain-wall type behavior in \eqref{eqn:kappa-limit} leads to a zero eigenstate with exponential decay. To see why $\mathcal{D}(k_\parallel)$ has a simple zero eigenvalue, we observe that the zero eigenvalue problem of $\mathcal{D}(k_\parallel)$ is equivalent to the following decoupled linear system
\begin{equation}\label{eqn:zero-state}
	0 = \begin{pmatrix}
		0 & -i\partial_{X_1} -i\kappa(X_1;k_\parallel)\\
		-i\partial_{X_1} + i\kappa(X_1;k_\parallel) & 0
	\end{pmatrix}\begin{pmatrix}
		\Psi_0^A\\
		\Psi_0^B
	\end{pmatrix}.
\end{equation}
When $\kappa_+(k_\parallel) > 0$ and $\kappa_-(k_\parallel) < 0$, the zero eigenfunction only lives on $B$-nodes with
\begin{equation}\label{eqn:zero-b}
	\Psi_0^A(X_1;k_\parallel) = 0, \qquad \Psi_0^B(X_1;k_\parallel) = e^{-\int_0^{X_1} \kappa(s;k_\parallel) ds}.
\end{equation}
Similarly, when $\kappa_+(k_\parallel) < 0$ and $\kappa_-(k_\parallel) > 0$, the zero eigenfunction only lives on the $A$-nodes with
\begin{equation}\label{eqn:zero-a}
	\Psi_0^B(X_1;k_\parallel) = 0, \qquad \Psi_0^A(X_1;k_\parallel) = e^{\int_0^{X_1} \kappa(s;k_\parallel) ds}.
\end{equation}
From the explicit solutions in \eqref{eqn:zero-b} and \eqref{eqn:zero-a}, we observe that when $X_1 \rightarrow \pm \infty$, the zero eigenfunction $\bm{\Psi_0}(X_1;k_\parallel)$ has exponential decay at infinity since $\kappa(X_1;k_\parallel)$ approximates two constants with opposite signs. Eigenfunctions of $\mathcal{D}(k_\parallel)$ for nonzero eigenvalues exhibit analogous behavior, though the argument is less direct.

It is worth mentioning that we have $E_2(k_\parallel) = 0$ when we choose $E_1(k_\parallel)=0$. We observe that $\mathcal{R}_2^1$ in \eqref{eqn:app-rem-delta-2} depends only on $\Psi_0^B$ and $\mathcal{R}_2^2$ depends only on $\Psi_0^A$. Moreover, as mentioned in \eqref{eqn:zero-b} and \eqref{eqn:zero-a}, any zero eigenfunction is supported entirely on either the $A$-nodes or the $B$-nodes. Therefore, if $\Psi_0^A \neq 0$, then $\Psi_0^B = \mathcal{R}_2^1 = 0$ (the same applies when $\Psi_0^B \neq 0$). Consequently, we have $E_2(k_\parallel) = 0$.

\begin{remark}[The topologically proteced zero eigenstates]\label{rmk:zero-eigenstates}
    The zero eigenvalue and its eigenfunction are topologically protected -- their existence persists under any perturbation of $\kappa(X_1;k_\parallel)$ that preserve its sign change (a similar phenomenon involving such a topologically protected zero eigenstate is investigated in \cite{fefferman2017topologically}). However, the nonzero eigenvalues may not be topologically protected (see also in \cite{lu2020dirac}).
\end{remark}

\begin{remark}[The number of nonzero eigenvalues in the spectral gap]\label{rmk:number-non-zero}
    For a general domain-wall-type $\kappa(X_1;k_\parallel)$, we have a qualitative understanding of the number of nonzero eigenvalues in the spectral gap $\big(-a(k_\parallel),a(k_\parallel)\big)$: when the difference between the two limits
    \begin{equation*}
        d(k_\parallel):= \big|\kappa_+(k_\parallel) - \kappa_-(k_\parallel)\big|
    \end{equation*}
    becomes larger, the number of nonzero eigenvalues increases. A similar phenomenon has been observed in the Schr\"odinger equation with a deep well potential (see e.g. Section 5 in \cite{hall2013quantum}).

    Although we lack a quantitative description of how many nonzero eigenvalues there are, we know they occur in pairs: if $E_1(k_\parallel)$ is an eigenvalue of $\mathcal{D}(k_\parallel)$ with eigenfunction $(\Psi_0^A,\Psi_0^B)^T$, then $-E_1(k_\parallel)$ is also an eigenvalue of $\mathcal{D}(k_\parallel)$ with eigenfunction $(\Psi_0^A,-\Psi_0^B)^T$. This can also be explained by the chiral symmetry of $\mathcal{D}(k_\parallel)$, i.e. there exists a unitary operator $\Gamma$ (here $\Gamma=\sigma_3=\text{diag}(1,-1)$) such that $\Gamma \mathcal{D}(k_\parallel) \Gamma^{-1} = -\mathcal{D}(k_\parallel)$. The spectrum of a given operator with chiral symmetry is always symmetric about zero.
\end{remark}


\section{Main theorem: Eigenpairs of the effective magnetic Dirac operator seed localized 
 states}\label{sec:main-thm}



In this section, we present the main theorem, which guarantees the existence of an eigenpair which is approximated by \eqref{eqn:uni-ansatz-approx} for the eigenvalue problem \eqref{eqn:eig-def-hon} subject to the Bloch boundary condition \eqref{eqn:bloch-bdry} when the unidirectional displacement $\bm u = (0,d(X_1))^T$ has bounded gradient.

To justify \eqref{eqn:uni-ansatz-approx} as an accurate approximation of an eigenstate of \eqref{eqn:eig-def-hon}, we estimate its corresponding remainder. For simplicity, we restrict to the leading-order case in \eqref{eqn:uni-ansatz-approx-2} with $l=0$, i.e.
\begin{align*}
    &\bm \psi_{m,n}^{(0)} = e^{i \bm K \cdot \bm x}e^{i\frac{k_\parallel}{3} X_2} U \bm \Psi_0(X_1; k_\parallel) \Bigg|_{\substack{\bm x = \Cell_{m,n}\\ \bm X = \delta \Cell_{m,n}}}, \qquad E^{(1)}(k_\parallel) = \delta E_1(k_\parallel).
\end{align*}
To represent the reminder for this leading order approximation, we observe that the slowly-varying envelope $\bm \Psi_0(X_1;k_\parallel)$ only depends on $m$, i.e.
\begin{align}
    \bm \Psi_0(X_1;k_\parallel) \Big|_{\bm X = \delta \Cell_{m,n}} = \bm \Psi_0\left(\frac{\sqrt{3}}{2}\delta m; k_\parallel\right),\label{eqn:scale-psi0}
\end{align}
where $\frac{\sqrt{3}}{2}\delta m$ is the $X_1$-coordinate of $\delta \Cell_{m,n}$ (see \eqref{eqn:ref-ab-honeycomb} for the coordinate of $\Cell_{m,n})$. Similarly, in higher-order approximations with $j\geq 1$, the slowly-varying envelopes $\bm \Psi_j(X_1;k_\parallel)$ depend only on $m$. Therefore, it is natural to consider an eigenstate decomposed into a leading-order approximation and a remainder of the form
\begin{subequations}\label{eqn:ansatz-1d}
    \begin{align}
        &\bm \psi_{m,n} = \delta^{\frac{1}{2}} e^{i \bm K \cdot \bm x} e^{i\frac{k_\parallel}{3} X_2} U \Bigg(\bm \Psi_0(X_1;k_\parallel) + \delta \bm \eta_m\Bigg)\Bigg|_{\substack{\bm x = \Cell_{m,n}\\ \bm X = \delta \Cell_{m,n}}}, \label{eqn:ansatz-1d-1}\\
        &E^\delta(k_\parallel) = \delta E_1(k_\parallel) + \delta^2 \mu(k_\parallel), \label{eqn:ansatz-1d-2}
    \end{align}
\end{subequations}
where $\bm \eta_m = (\eta_m^A, \eta_m^B)^T$ is the correctors of the eigenstate and $\mu(k_\parallel)$ is the correctors of the eigenvalue. The factor $\delta^{\frac{1}{2}}$ in \eqref{eqn:ansatz-1d-1} is added to guarantee that the leading order term in \eqref{eqn:scale-psi0} is $O(1)$ in $l^2(\mathbb{Z}; \mathbb{C}^2)$. Accordingly, we normalize $\bm \Psi_0(X_1;k_\parallel)$ with
\begin{equation}\label{eqn:norm-1}
    \norm{\bm\Psi_0(\cdot;k_\parallel)}^2_{L^2(\mathbb{R)}} = \norm{\Psi_0^A(\cdot;k_\parallel)}^2_{L^2(\mathbb{R)}} + \norm{\Psi_0^B(\cdot;k_\parallel)}^2_{L^2(\mathbb{R)}} = \sqrt{3}/2,
\end{equation}
so that the leading order term approximates norm 1 as $\delta \rightarrow 0$ (see the calculation of the leading order term in \eqref{eqn:app-l2-L2-norm}).



\subsection{Our main theorem}\label{subsec:main-thm}

We now state our main result, Theorem \ref{thm:main}, which provides an estimate for the correctors in \eqref{eqn:ansatz-1d} for unidirectional displacements with bounded deformation gradients. The conditions on the displacement $d(X_1)$ we consider in this section are the following:
\begin{enumerate}[(i)]
    \item $d(X_1) \in C^\infty(\mathbb{R})$ and $d'(X_1) \in C_b^\infty(\mathbb{R})$;
    
    \item $d'(X_1)$ satisfies a sign-changing condition, i.e.
    \begin{equation}\label{eqn:d-condition}
		\lim_{X_1 \rightarrow \pm \infty} d'(X_1) = \pm d_\infty;
	\end{equation}

    \item $d'(X_1)$ satisfies the integrability condition
	\begin{equation}\label{eqn:main-d-condition-2}
		d'(X_1) - d_\infty \in L^1([0,\infty)), \qquad d'(X_1) + d_\infty \in L^1((-\infty,0]).
	\end{equation}
 
\end{enumerate}

\begin{theorem}\label{thm:main}
	\edit{Consider the $l^2_{k_\parallel}$-eigenvalue problem: $H^\delta \bm \psi=E \bm \psi$ with $\bm \psi\in l^2_{k_\parallel}$} for the deformed honeycomb lattice with unidirectional displacement $\bm u = (0,d(X_1))^T$. Assume $d(X_1)$ satisfies the above conditions (i)-(iii) and $k_\parallel$ satisfies 
    \begin{equation}\label{eqn:k-para-cond}
        |t_1 d_\infty| > \frac{2}{3}|k_\parallel|.
    \end{equation}

    If $E_1(k_\parallel)$ is an eigenvalue of $\mathcal{D}(k_\parallel)$ in \eqref{eqn:dirac-1d} and $\bm{\Psi}_0(X_1;k_\parallel)$ is the corresponding eigenfunction, then there exists $M >0$ and a threshold $\delta_0 > 0$ (depending on $k_\parallel$) such that for any $0< \delta < \delta_0$,
    \begin{itemize}
        \item a solution to the eigenvalue problem \eqref{eqn:eig-def-hon} subject to the Bloch boundary condition \eqref{eqn:bloch-bdry} of the form \eqref{eqn:ansatz-1d} exists, i.e. the correctors $\mu(\delta)$ and $\bm \eta_m^\delta$ exist as a map $\delta\mapsto \left(\mu(\delta), \bm\eta_m^{\delta}\right) \in \{|\mu| < M\} \times l^2(\mathbb{Z}; \mathbb{C}^2)$;

        \item moreover, the correctors $\mu(\delta)$ and $\bm \eta_m^\delta$ satisfy the following estimates
	\begin{align}
	    &\|\delta^{\frac{3}{2}} \bm \eta_m^{\delta}\|_{l^2(\mathbb{Z};\mathbb{C}^2)} \lesssim \delta,\qquad \lim_{\delta \rightarrow 0} \mu(\delta) = E_2(k_\parallel),\label{eqn:err-bd}
	\end{align}
	where $E_2(k_\parallel)$ is given by \eqref{eqn:E2}.
    \end{itemize}
\end{theorem}


We briefly explain the conditions on $k_\parallel$ and $d(X_1)$, which ensures the existence of a zero eigenvalue of $\mathcal{D}(k_\parallel)$ by Lemma \ref{lemma:1d}. We observe that the conditions \eqref{eqn:d-condition} and \eqref{eqn:main-d-condition-2} on $d(X_1)$ are equivalent to the conditions \eqref{eqn:kappa-limit-1} and \eqref{eqn:kappa-condition-2}. The condition \eqref{eqn:k-para-cond} ensures that $\kappa_\pm(k_\parallel)$ defined in \eqref{eqn:kappa-limit-1} are of opposite signs when $3|t_1 d_\infty| > 2|k_\parallel|$ since
\begin{equation*}
    \kappa_\pm(k_\parallel) = \frac{k_\parallel}{3} \mp \frac{t_1 d_\infty}{2} \quad \Rightarrow \quad \kappa_+(k_\parallel) \kappa_-(k_\parallel) = \frac{k_\parallel^2}{9} - \frac{t_1^2 d_\infty^2}{4} < 0.
\end{equation*}
Therefore, condition \eqref{eqn:kappa-limit-2} is satisfied and Lemma \ref{lemma:1d}(b) ensures at least a zero eigenvalue for $\mathcal{D}(k_\parallel)$.

\paragraph{Sketch of the proof} Here we briefly sketch the proof of Theorem \ref{thm:main}, while the detailed proof is provided in Section \ref{sec:proof}. To show the existence of the correctors $\bm \eta_m$ and $\mu(\delta)$, we apply the discrete Fourier transforms (DFT) to $\bm \eta_m$, and then follow the strategy of \cite{fefferman2017topologically} by using a near- and far-momentum separation in the momentum space. The proof consists of the following steps:

\paragraph{Step 1:} We apply the discrete Fourier transform to $\bm\eta_m $ and denote the DFT of $\bm\eta_m $ as $\widetilde{\bm\eta}(k)$. We then derive the corresponding equation for $\widetilde{\bm\eta}(k)$ in Section \ref{subsec:eta-M}.

\paragraph{Step 2:} We separate $\widetilde{\bm{\eta}}(k)$ into the near-and far-momentum parts, i.e.
\begin{equation*}
    \widetilde{\bm{\eta}}(k) = \widetilde{\bm{\eta}}_\text{near}(k) + \widetilde{\bm{\eta}}_\text{far}(k), \qquad \widetilde{\bm{\eta}}_\text{near}(k) = \begin{pmatrix}
        \widetilde{\eta}^A_\text{near}(k)\\
        \widetilde{\eta}^B_\text{near}(k)
    \end{pmatrix}, \quad \widetilde{\bm{\eta}}_\text{far}(k) = \begin{pmatrix}
        \widetilde{\eta}^A_\text{far}(k)\\
        \widetilde{\eta}^B_\text{far}(k)
    \end{pmatrix},
\end{equation*}
where $\widetilde{\bm{\eta}}_\text{near}(k)$ and $\widetilde{\bm{\eta}}_\text{far}(k)$ only live in $|k|\leq \delta^\tau$ and $|k|\geq \delta^\tau$ respectively for some $\tau \in (0,1)$. We also derive the equations for $\widetilde{\bm \eta}_\text{near}(k)$ and $\widetilde{\bm \eta}_\text{far}(k)$ in Section \ref{subsec:momentum-sepration}

\paragraph{Step 3:} We solve for the far-momentum part $\widetilde{\bm \eta}_\text{far}(k)$ for fixed $\widetilde{\bm{\eta}}_\text{near}(k), \mu$, and $\delta$. More precisely, we show that the eigenvalue problem \eqref{eqn:eig-def-hon} uniquely determines $\bm{\widetilde{\eta}}_\text{far}(k) = \bm{\widetilde{\eta}}_\text{far}\Big[k;\bm{\widetilde{\eta}}_\text{near}, \mu, \delta\Big]$ (see Proposition \ref{prop:far-energy}). In fact, the eigenvalue problem \eqref{eqn:eig-def-hon} can be written as a coupled linear system for $\widetilde{\bm{\eta}}_\text{near}(k)$ and $\widetilde{\bm{\eta}}_\text{far}(k)$. Conceptually, the far-momentum part $\widetilde{\bm{\eta}}_\text{far}(k)$ is small while $\widetilde{\bm{\eta}}_\text{near}(k)$ contains the primary contribution. Therefore, we can use \textit{Schur complement} to solve for the far-momentum part (see Section \ref{subsec:far-sol}).

\paragraph{Step 4:} For fixed $\mu, \delta$, we then substitute the far-momentum solution $\bm{\widetilde{\eta}}_\text{far}(k) = \bm{\widetilde{\eta}}_\text{far}\Big[k;\bm{\widetilde{\eta}}_\text{near}, \mu, \delta\Big]$ into the eigenvalue problem \eqref{eqn:eig-def-hon} to obtain the equation for $\bm{\widetilde{\eta}}_\text{near}(k)$. By rescaling $\delta\,\widetilde{\bm \eta}_\text{near}(k)=\widehat{\bm \eta}_\text{near}(k/\delta)$ (see \eqref{eqn:near-hat}), we arrive at an equation for $\widehat{\bm \eta}_\text{near}(\xi)$, which admits a solution $\widehat{\bm \eta}_\text{near}[\xi;\mu,\delta]$ (see Proposition \ref{prop:beta}). Then we solve for $\mu$ as a function of $\delta$ by a continuity argument (see Proposition \ref{prop:solve-Jplus}).\\

The main difficulty in the proof of Theorem \ref{thm:main} is to solve for the near-momentum part $\widehat{\bm \eta}_\text{near}(\xi)$ in step 4. In fact, the equation satisfied by the inverse Fourier transform of $\widehat{\bm \eta}_\text{near}(\xi)$ is a perturbed Dirac system, obtained by perturbing the effective equation \eqref{eqn:2nd-expansion} for $\bm \Psi_1(X_1;k_\parallel)$ (see \eqref{eqn:beta-system}). 

Since the linear operator in \eqref{eqn:2nd-expansion} acting on $\bm \Psi_1(X_1;k_\parallel)$ is $\mathcal{D}(k_\parallel) - E_1(k_\parallel)$ and has a zero eigenfunction $\bm \Psi_0(X_1;k_\parallel)$, solving the associated perturbed system requires a \textit{Lyapunov Schmidt reduction}. In particular, we need the invertibility of $\mathcal{D}(k_\parallel) - E_1(k_\parallel)$ on the orthogonal space of $\bm \Psi_0(X_1;k_\parallel)$, which is provided in Proposition \ref{prop:inv-Dk} for every $E_1(k_\parallel)$ in the spectral gap.

\begin{remark}[Inverting $\mathcal{D}(k_\parallel) - E_1(k_\parallel)$]
    For the zero eigenvalue $E_1(k_\parallel) = 0$, one can square the operator $\mathcal{D}(k_\parallel)$ and obtain an elliptic operator to show the invertibility of $\mathcal{D}(k_\parallel) - E_1(k_\parallel)$ by following the same argument in Theorem 6.15 in \cite{fefferman2017topologically}. However, since we are also interested in $E_1(k_\parallel) \neq 0$, we need a different argument to show the invertibility of $\mathcal{D}(k_\parallel) - E_1(k_\parallel)$ on the orthogonal space of $\bm \Psi_0(X_1;k_\parallel)$. In Section \ref{subsec:inv-Dk}, we properly define $\left(\mathcal{D}(k_\parallel) - E_1(k_\parallel)\right)^{-1}$ and provide a bound on the inverse operator by using classical results from exponential dichotomy theory in \cite{palmer1984exponential}.
\end{remark}




\section{Corroboration of Theorem \ref{thm:main} via  numerical simulations}\label{sec:num-comparison}

The quadratic deformation\footnote{\edit{The quadratic deformation can be performed along different orientations and induce different pseudo-magnetic effect. Two important orientations are the armchair (AC) orientation and the zigzag orientation (ZZ); see Figure \ref{fig:honeycomb-edges}. We provide a short review of these orientations in Appendix \ref{sec:ac-zz-edge}.}} $\bm{u}^\text{AC} = (0,X_1^2)$ is anticipated to induce localization in $x_1$; see Lemma \ref{lemma:spectrum-landau-dirac-mag-op}. However, we cannot apply Theorem \ref{thm:main} directly, since the deformation $(X_1,X_2)\mapsto \bm{u}^\text{AC}(X_1,X_2)$, does not have a bounded gradient; see condition \eqref{eqn:d-condition}.  We observe however that Theorem \ref{thm:main} applies to a {\it linear regularization} of $\bm{u}^\text{AC}$, which keeps the quadratic behavior on a compact set about zero but transitions to linear growth. 

Specifically, let 
$\bm{u}^\text{AC}_\text{reg}(\bm{X};L) = (0, d_\text{reg}(X_1;L))^T$ where
\begin{equation}\label{eqn:lin-trun-def}
	d_\text{reg}(X_1;L) = \varphi \ast \widetilde{d}(X_1;L) ,\quad {\rm where} \qquad \widetilde{d}(X_1;L) = \begin{cases}
		2L(X_1-L) + L^2, & X_1 \geq L,\\
		X_1^2, & |X_1| \leq L,\\
		-2L(X_1+L) + L^2,& X_1 \leq -L\ .
	\end{cases}
\end{equation}
Here, $\widetilde{d}(X_1;L)$ is a continuation by a linear function for $|X_1|\ge L$. The function $\varphi*\tilde{d}$ is the convolution of $\widetilde{d}(X_1;L)$ with a smooth approximation to a delta function. Hence, $d_\text{reg}(X_1;L)$ is smooth and thus $\bm{u}^\text{AC}_\text{reg}(\bm{X};L)$ is an admissible deformation in Theorem \ref{thm:main}.

Applying Theorem \ref{thm:main} to $d_\text{reg}(X_1;L)$, we obtain that for any $k_\parallel$ satisfying $|k_\parallel| < 3L|t_1|$ (this comes from the condition \eqref{eqn:k-para-cond} by replacing $d_\infty = 2L$), the corresponding $\mathcal{D}(k_\parallel)$ has simple eigenvalues (including zero) in the spectral gap and the discrete Hamiltonian $H^\delta$ in \eqref{eqn:ham-hon-delta-approx} has eigenstates that can be approximated by 
the eigenfunction $\bm \Psi_0$ of $\mathcal{D}(k_\parallel)$ in the form of our 1-dimensional ansatz in \eqref{eqn:ansatz-1d}. 

In this section, we compare the ``numerical eigenvalue curves'' (numerical approximation of band structures, which shall be defined shortly) and eigenstates associated with the quadratic deformation $\bm{u}^\text{AC}$ and its linear regularization $\bm{u}^\text{AC}_\text{reg}$ along the armchair (AC) orientation. We also numerically compare these eigenvalues and eigenstates for the quadratic deformation $\bm{u}^\text{ZZ} = (X_2^2,0)^T$ along the zigzag (ZZ) orientation. Throughout this section, we set the hopping coefficient to be $t_1 = -2$, so that the hopping strength decreases as the deformed bond length increases. It is worth mentioning that for a negative $t_1$, the zero eigenstates live only on the $B$-node for both $\bm{u}^\text{AC}$ and $\bm{u}^\text{AC}_\text{reg}$ (see \eqref{eqn:zero-b}-\eqref{eqn:zero-a} and \eqref{eqn:zero-quantum-osc}).

\paragraph{{Spatial localization of modes for deformations with AC orientation}} We now compare the band structures and eigenstates for the deformed honeycomb by $\bm{u}^\text{AC}$ and its linear regularized version $\bm{u}^\text{AC}_\text{reg}$. We first illustrate our numerical method using the undeformed honeycomb lattice along the AC orientation; the deformed case is modeled analogously. We observe that the undeformed and deformed lattices are periodic in the vertical direction, with a periodic cell consisting of two rows of $A$- and $B$-type nodes (marked between the dotted lines in Figure \ref{fig:ac-ref-nodes}). To construct a row unit cell, we partition it into cells each containing four nodes, $A_s,B_s,C_s,D_s$, where $A_s,C_s$ are of type $A$ nodes and $B_s,D_s$ are of type $B$ nodes (see the shaded region in Figure \ref{fig:ac-ref-nodes} for the $s$th cell). We associate these nodes in the row unit cell with wave functions $\bm{\psi}_s=(\psi_s^A,\psi_s^B,\psi_s^C,\psi_s^D)^T$.

\begin{figure}[!htb]
	\centering
	\subfloat[]{
	\includegraphics[height=1.5in]{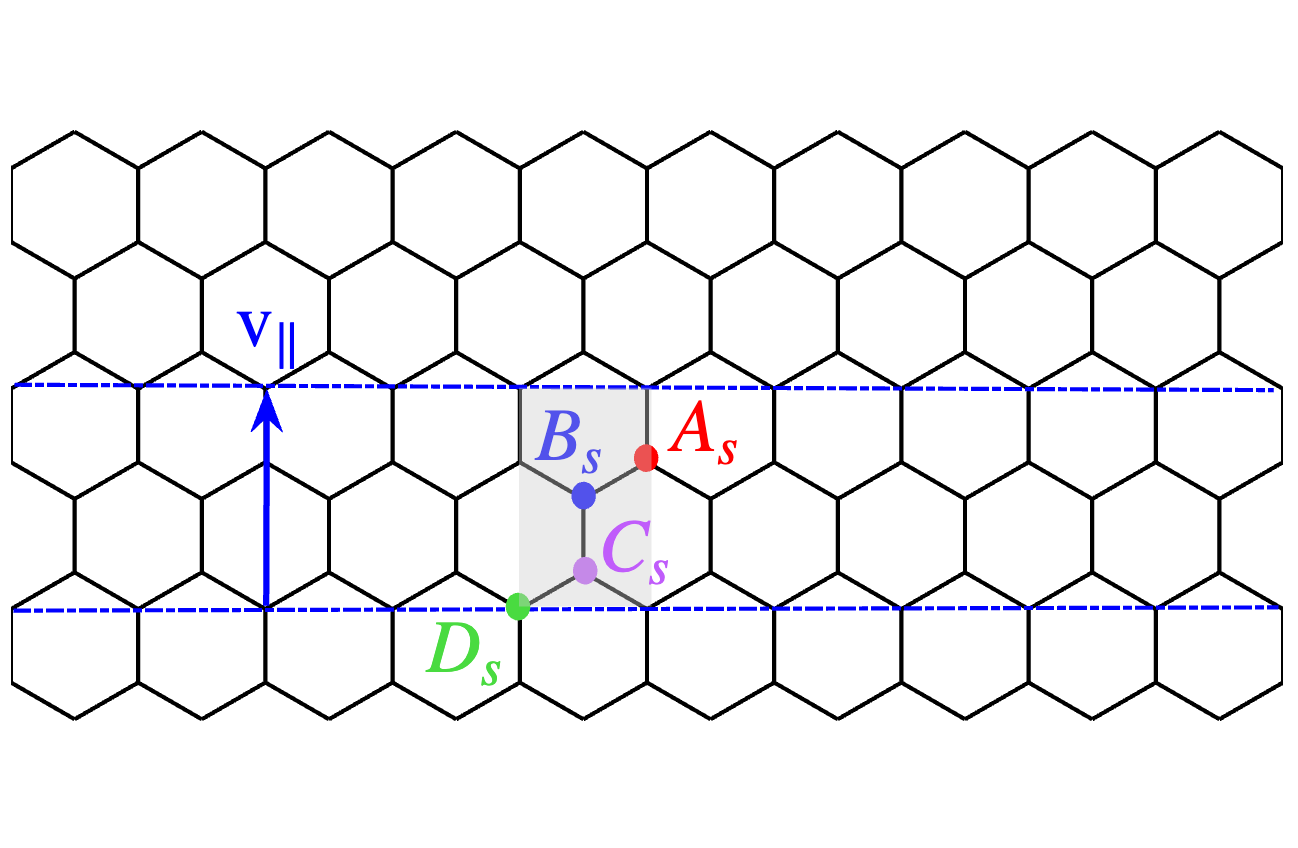}\label{fig:ac-ref-nodes}
	}\hfil
	\subfloat[]{
	\includegraphics[height=1.3in]{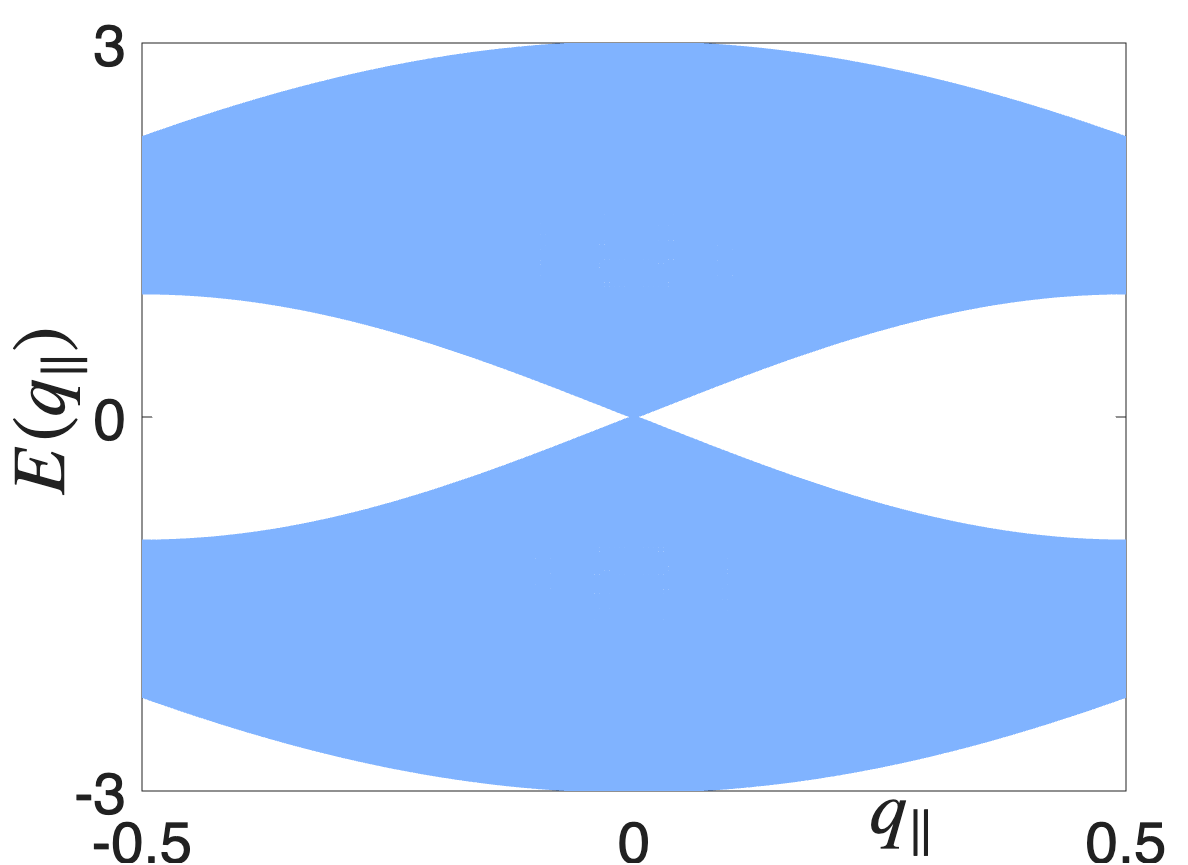}\label{fig:ac-band-ref}
	}\hfil
    \subfloat[]{
	\includegraphics[height=1.25in]{fig/ac-ref-middle-band}\label{fig:ac-middle-band-ref}
	}
	\caption{The undeformed honeycomb along the AC edge: (a) our numerical model of the honeycomb along the AC edge as a periodic structure in the vertical direction; (b) the numerical eigenvalue curves with truncation size $N_T = 200$ and the Dirac point occurs at $q_\parallel = 0$; (c) a zoomed-in view at $q_\parallel = 0$ with the lowest 40 eigenvalues.}
	\label{fig:ac-ref}
\end{figure}

Since the deformed structure is periodic along the $X_2$ direction with translation vector $\bm{v}_\parallel = (0,3)^T$, wave functions on nodes outside the row unit cell differ by the associated Bloch phase, i.e. the oscillation at node $A_s \pm \bm v_\parallel$ is related to $\psi_s^A$ by 
\begin{align}\label{eqn:num-bloch-ac}
    \psi_s^{A \pm \bm v_\parallel} = e^{\pm i3\pi q_\parallel}\psi_s^A,
\end{align}
where $q_\parallel$ is the quasi-momentum related to the translation vector $\bm{v}_\parallel = (0,3)^T$ and lives in $q_\parallel \in [-1/3,1/3)$. We notice that $q_\parallel$ plays the role of $\delta k_\parallel$ in our ansatz \eqref{eqn:ansatz-1d} with $3\pi q_\parallel = \delta k_\parallel$, since incrementing $n$ by 1 produces the Bloch phase factor $e^{i\delta k_\parallel}$ (see \eqref{eqn:bloch-bdry}). The Hamiltonian in \eqref{eqn:def-hon-hamiltonian} thereby reduces to a Bloch Hamiltonian $H^\delta(q_\parallel)$ (see Appendix \ref{app:numerical-zero} for the detailed expression of $H^\delta(q_\parallel)$). 

The band structure of $H^\delta$ corresponds to the spectrum of $H^\delta(q_\parallel)$ on wave functions $(\bm \psi_s)_{s \in \mathbb{Z}}$ as $q_\parallel \in [-1/3,1/3)$. For a given $q_\parallel$, to numerically approximate the spectrum of $H^\delta(q_\parallel)$, we consider a sharp truncation scheme with zero boundary condition by imposing $\bm \psi_s=0$ for $|s|>N_T$ with sufficiently large $N_T$. Then the spectrum of $H^\delta(q_\parallel)$ are approximated by the numerical eigenvalues of the resulting finite Hamiltonian matrix. As $q_\parallel$ varies, the numerical eigenvalues trace out curves, which we refer to as the ``numerical eigenvalue curves''. As the truncation size $N_T$ increases, these curves approximate the band structure of $H^\delta$.

For the undeformed honeycomb lattice along the AC edge, the numerical eigenvalue curves is presented in Figure \ref{fig:ac-band-ref}, which exhibits a Dirac point at $q_\parallel=0$. This value corresponds to the Dirac point $\bm{K}$, since the translation vector $\bm{v}_\parallel = \bm{v}_2$, which is associated to $\bm{a}_\parallel=\bm{a}_2$ in the dual lattice. Consequently, the value of $q_\parallel$, which is the projection of $\bm{K}$ in \eqref{eqn:dirac-pt-honeycomb} onto $\bm{a}_\parallel$, vanishes.

We now move to the deformed honeycomb under the displacements $\bm{u}^{\mathrm{AC}}$ and $\bm{u}^{\mathrm{AC}}_{\mathrm{reg}}$. To model the linearly regularized deformation, we set a length $L$ so that the linear regularization starts in the final fraction $\alpha_{\mathrm{reg}}$ of the interval. For instance, with $N_T=200$ cells in the row unit cell and $\alpha_{\mathrm{reg}}=0.5$, the regularization starts at $N_L=100$; for $|s|>N_L$, the deformation grows linearly. 

Then, we compare the numerical eigenvalue curves of the deformed honeycomb under the displacements $\bm{u}^{\mathrm{AC}}$ and $\bm{u}^{\mathrm{AC}}_{\mathrm{reg}}$ with different $\delta$ in Figure \ref{fig:num-band-ac}. For all deformations and values of $\delta$, the zero eigenvalue curve near the Dirac point $q_\parallel = 0$ in Figure \ref{fig:num-band-ac} are quite flat. At a fixed $\delta$, on each numerical eigenvalue curve, the flat region near the Dirac point $q_\parallel =0$ is wider for $\bm{u}^{\mathrm{AC}}$ than for $\bm{u}^{\mathrm{AC}}_{\mathrm{reg}}$. By comparing zoomed-in numerical eigenvalue curves near the Dirac points $q_\parallel =0$, we observe that as plotted in Figure \ref{fig:ac-lindef-middle-band}, for $\bm{u}^\text{AC}_\text{reg}$ with a small $\delta$, only the lowest few numerical eigenvalue curves are flat near the Dirac point, while the numerical eigenvalue curves near the Dirac point for $\bm{u}^\text{AC}$ are almost flat on every band in Figure \ref{fig:ac-def-middle-band}. Moreover, for a fixed deformation, the flat region expands as $\delta$ increases by comparing Figure \ref{fig:ac-bands-def-1} with $\delta = 0.04$ and Figure \ref{fig:ac-bands-def-2} with $\delta = 0.08$\footnote{The numerical eigenvalue curves for $\bm{u}^{\mathrm{AC}}_{\mathrm{reg}}$ with $\delta=0.08$ is not shown, as it is nearly identical to that in Figure \ref{fig:ac-bands-def-2}.}. 

\begin{remark}[Numerical artifacts]
    One may observe that the zero eigenvalue curve splits into two branches (see e.g. Figure \ref{fig:ac-def-middle-band}), and raise the question of the multiplicity of the zero eigenvalue for small $q_\parallel$. Numerically, the zero eigenvalue appears with multiplicity \textbf{six} for small $q_\parallel$, while nonzero eigenvalues appear with multiplicity \textbf{two}. This doubling is an artifact of our four-node cell ($A_s,B_s,C_s,D_s$ in Figure \ref{fig:ac-ref-nodes}), which is a supercell -- a larger repeated unit obtained by grouping two primitive honeycomb cells. Therefore, the spectrum is folded and each eigenvalue is duplicated (see Appendix \ref{app:numerical-zero} for a detailed explanation). In fact, the two numerical eigenstates associated with a given eigenvalue correspond to the same envelope function. 
    
    After modulating out this supercell effect, the non-zero eigenvalues are simple, which is consistent with the spectral theory in Lemma \ref{lemma:spectrum-landau-dirac-mag-op}. Furthermore, only three distinct zero modes remain: one desired eigenstate localized in the middle and two numerical artifacts localized near the left and right boundaries. We clarify these numerical artifacts at zero energy in more detail in Appendix \ref{app:numerical-zero}.
\end{remark}

\begin{figure}[!htb]
	\centering
	\subfloat[]{
		\includegraphics[height=1.4in]{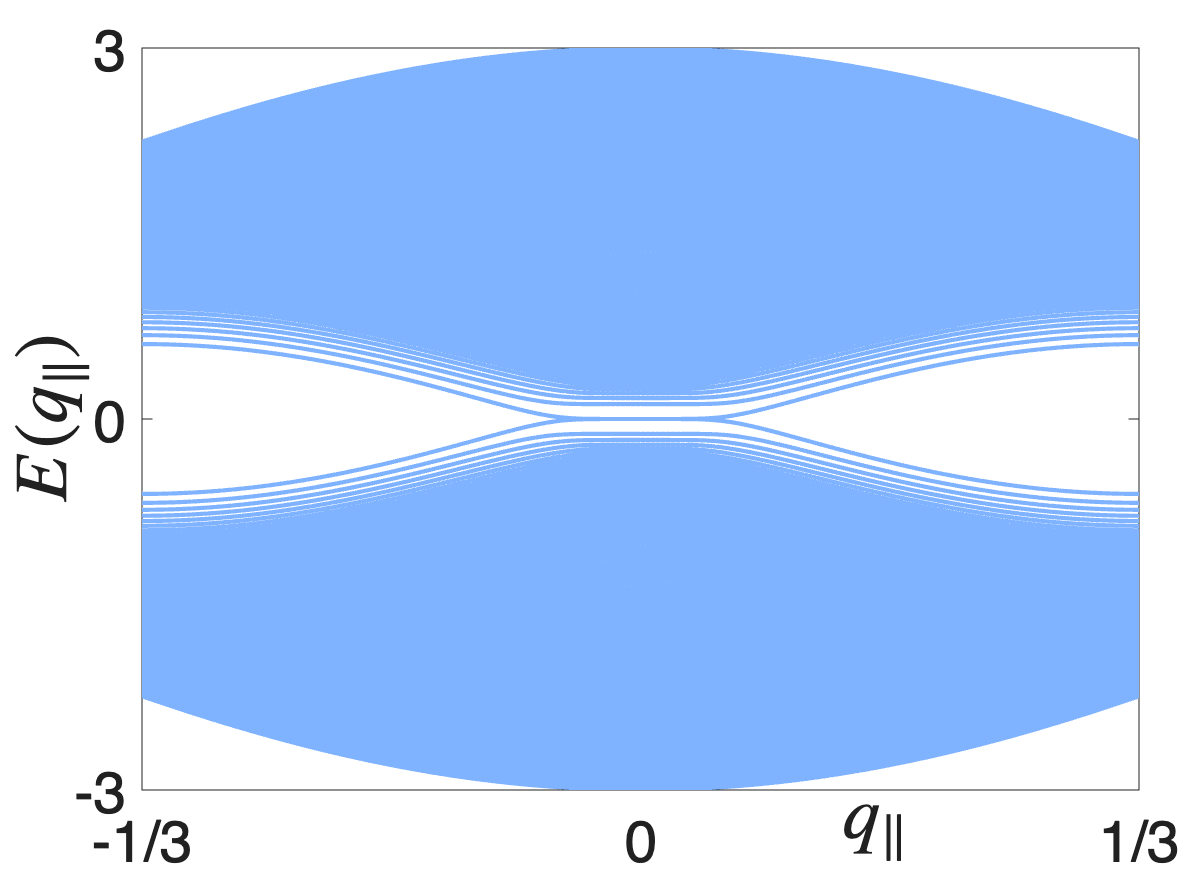}\label{fig:ac-bands-def-1}
	}\hfil
	\subfloat[]{
		\includegraphics[height=1.4in]{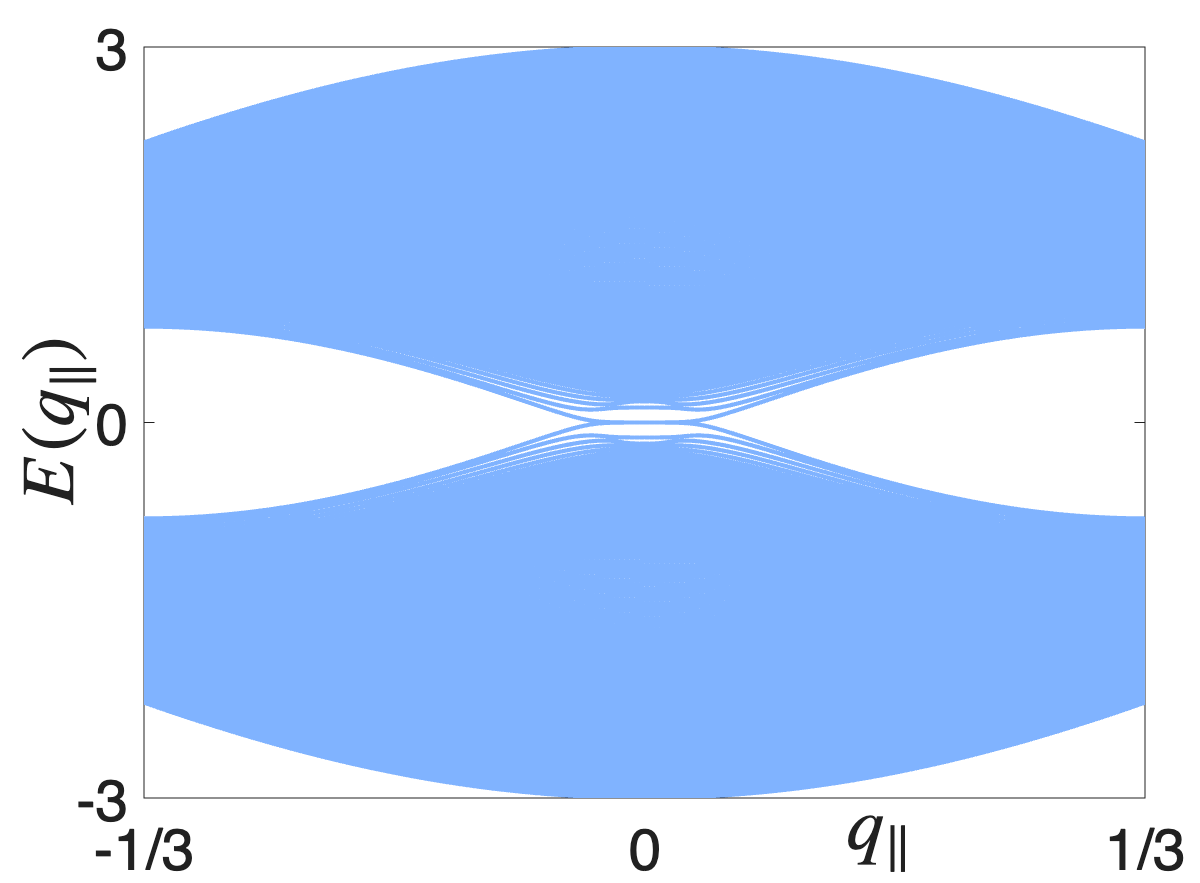}\label{fig:ac-bands-lindef-1}
	}\hfil
    \subfloat[]{
		\includegraphics[height=1.4in]{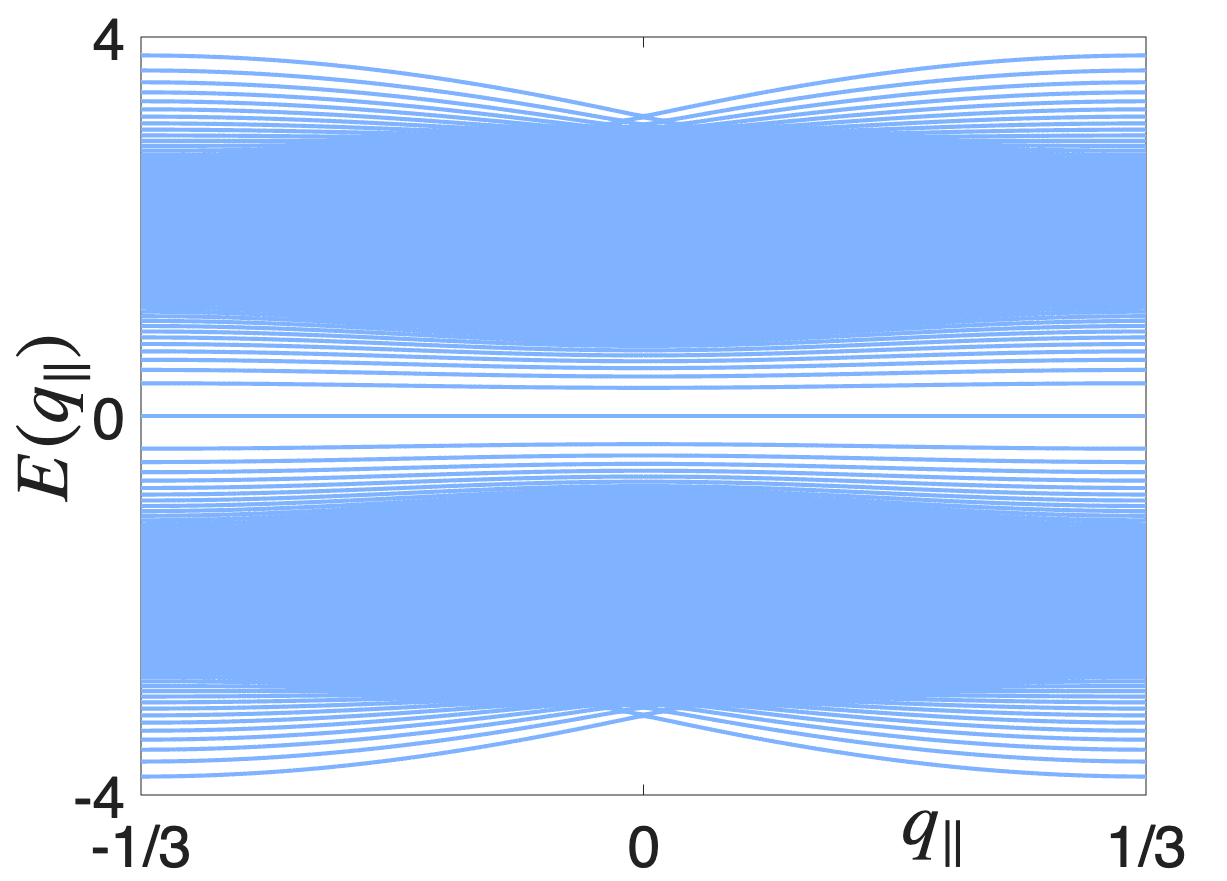}\label{fig:ac-bands-def-2}
	}\\
    \subfloat[]{
		\includegraphics[height=1.4in]{fig/ac-def-middle-band}\label{fig:ac-def-middle-band}
	}\hfil
    \subfloat[]{
		\includegraphics[height=1.4in]{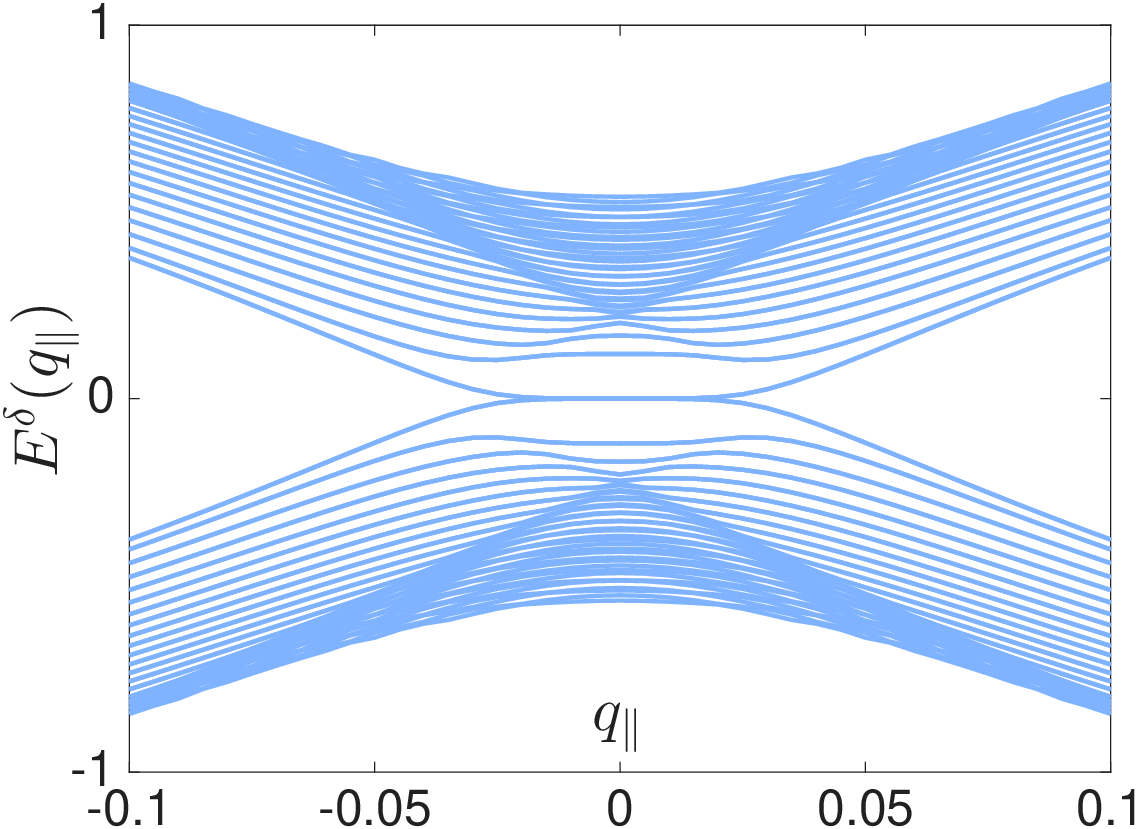}\label{fig:ac-lindef-middle-band}
	}\hfil
    \subfloat[]{
		\includegraphics[height=1.4in]{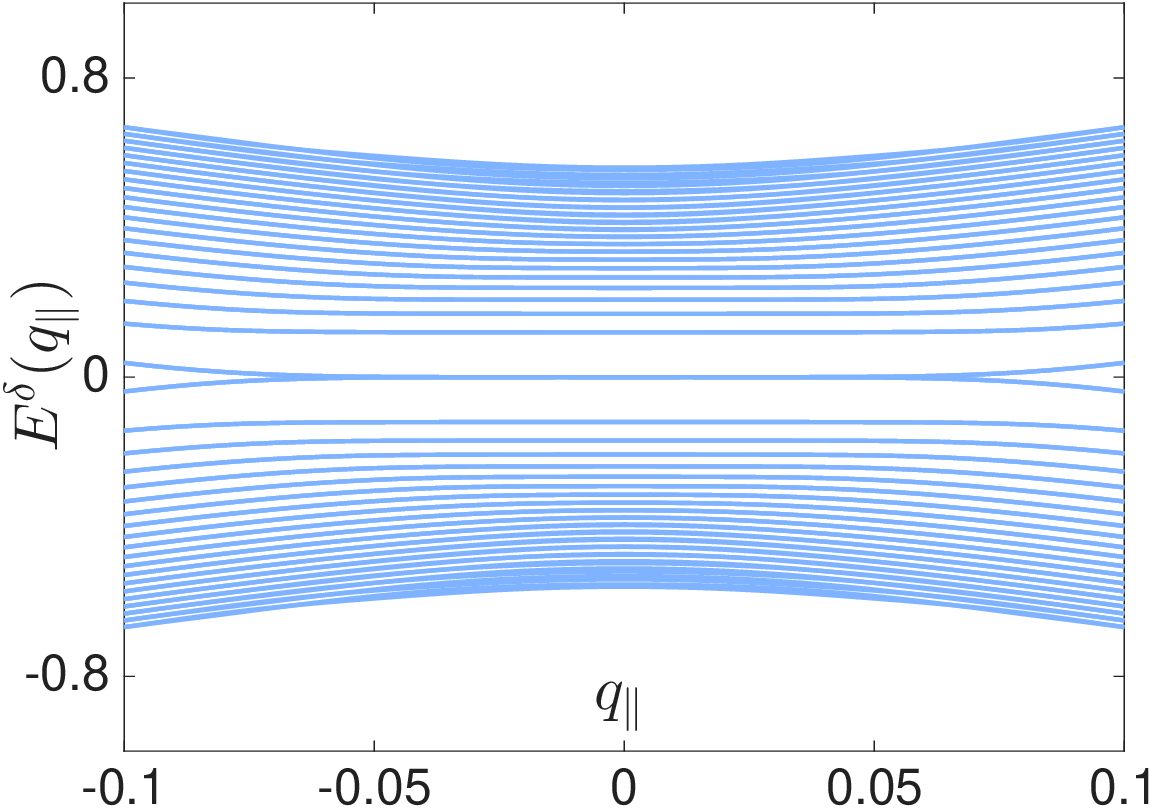}\label{fig:ac-def-middle-band-008}
	}
	\caption{Numerical eigenvalue curves of the honeycomb lattice under $\bm{u}^{\mathrm{AC}}$ and $\bm{u}^{\mathrm{AC}}_{\mathrm{reg}}$ with different $\delta$ and truncation size $N_T=200$: (a) $\bm{u}^{\mathrm{AC}}$ with $\delta=0.04$; (b) $\bm{u}^{\mathrm{AC}}_{\mathrm{reg}}$ with $\delta=0.04$ and $\alpha_\text{reg} = 0.5$; (c) $\bm{u}^{\mathrm{AC}}$ with $\delta=0.08$; (d)-(f) are zoomed-in numerical eigenvalue curves near $q_\parallel=0$ with the 40 numerical eigenvalue curves of smallest magnitude for (a)-(c), and markers in (d) correspond to eigenstates plotted in Figure \ref{fig:ac-defmodes}. The value of $q_\parallel$ is 0.02 for the dot and square in (d).}
	\label{fig:num-band-ac}
\end{figure}

We also plot the eigenstates of the truncated deformed structure in Figure \ref{fig:ac-defmodes}. For the quadratic deformation $\bm{u}^\text{AC}$, zero eigenstates shown in Figure \ref{fig:ac-defmodes} have Gaussian shapes; as $q_\parallel$ varies, the Gaussian shifts its center accordingly, consistent with the behavior of $\Psi_0^A$ and $\Psi_0^B$ in \eqref{eqn:zero-b}. Eigenstates with non-zero eigenvalues (see Figure \ref{fig:ac-def-nonzero-osc-a} and Figure \ref{fig:ac-def-nonzero-osc-b}) also exhibit rapid decay, though their profiles are not Gaussian.

The eigenstates for the linearly regularized deformation $\bm{u}^\text{AC}_\text{reg}$ look almost identical to the ones plotted in Figure \ref{fig:ac-defmodes}. To compare them, we plot them in semilogy plots in Figure \ref{fig:ac-comparison}\footnote{We set $N_b=800$ and $\alpha_{\mathrm{reg}}=0.9$, so that the quadratic deformation acts on only 10 percent of the unit cell. This choice ensures that the difference between the eigenstates under the two deformations remains above machine precision.}. We observe that eigenstates under $\bm{u}^{\mathrm{AC}}$ (with nearly Gaussian decay rate) decay faster than those under $\bm{u}^{\mathrm{AC}}_{\mathrm{reg}}$, which is Gaussian in the middle but transitions to exponential decay.

\begin{figure}[!htb]
	\centering
	\subfloat[]{
		\includegraphics[height=1.44in]{fig/ac-def-zero-osc}\label{fig:ac-def-zero-osc}
	}\hfil
	\subfloat[]{
		\includegraphics[height=1.4in]{fig/ac-def-nonzero-osc-a}\label{fig:ac-def-nonzero-osc-a}
	}\hfil
	\subfloat[]{
		\includegraphics[height=1.48in]{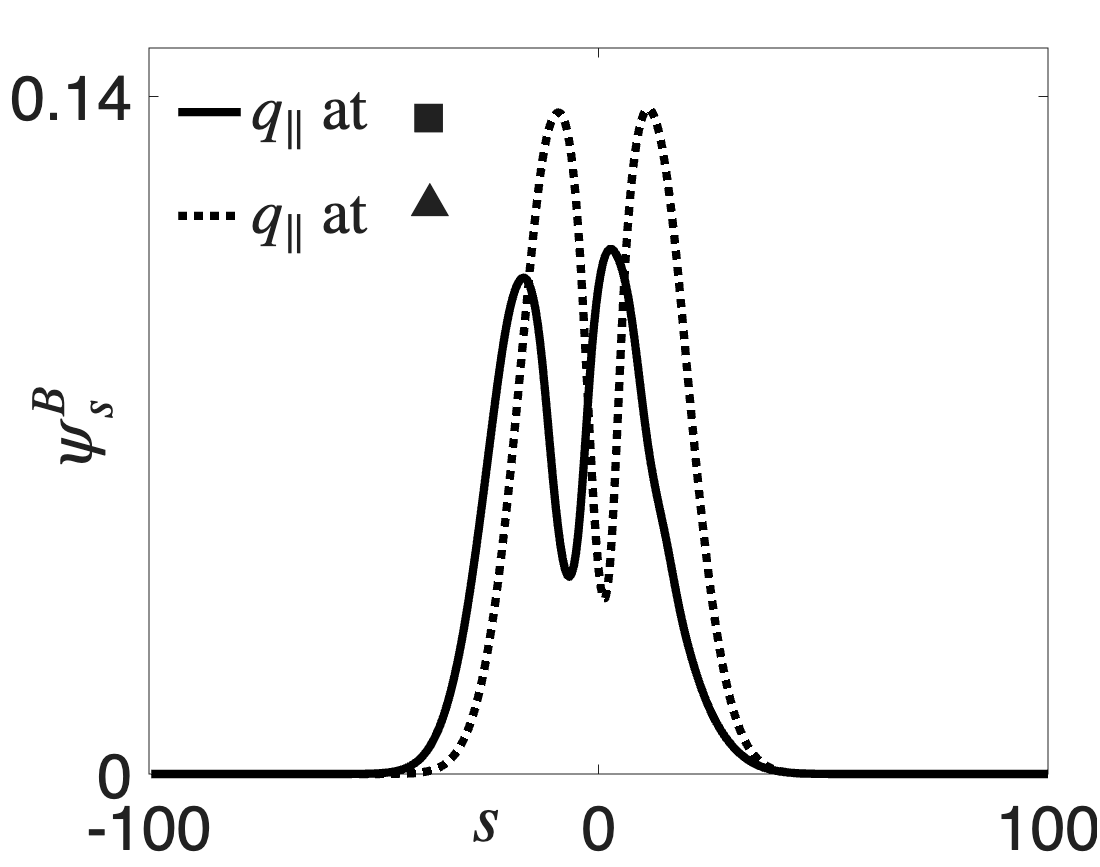}\label{fig:ac-def-nonzero-osc-b}
	}
	\caption{Eigenstates of the deformed honeycomb lattice under $\bm{u}^{\mathrm{AC}}$ with $\delta=0.04$: (a) $|\psi_s^B|$ as a function of $s$ of zero eigenstates for $q_\parallel=0$ at the diamond in Figure \ref{fig:ac-def-middle-band} and $q_\parallel$ at the circle in Figure \ref{fig:ac-def-middle-band}; note that the zero eigenstates vanish on the $A$ nodes; (b)-(c) eigenstates corresponding to the triangle and square markers in Figure \ref{fig:ac-def-middle-band}, representing the smallest nonzero eigenvalues in magnitude; (b) shows $|\psi_s^A|$, oscillations on the $A$ nodes, while (c) shows $|\psi_s^B|$, oscillations on the nodes.}
	\label{fig:ac-defmodes}
\end{figure}

\begin{figure}[!htb]
	\centering
	\subfloat[]{
		\includegraphics[height=1.4in]{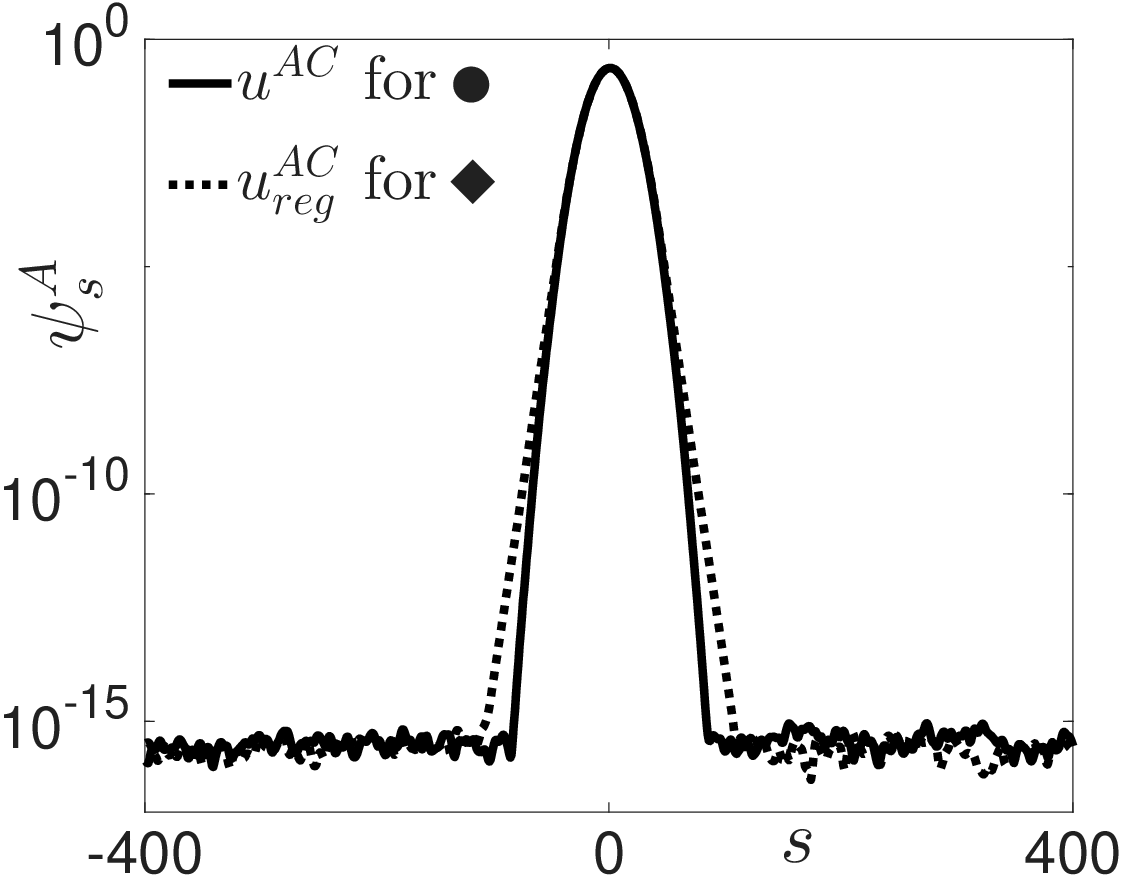}\label{fig:ac-comp-zero-osc}
	}\hfil
	\subfloat[]{
		\includegraphics[height=1.4in]{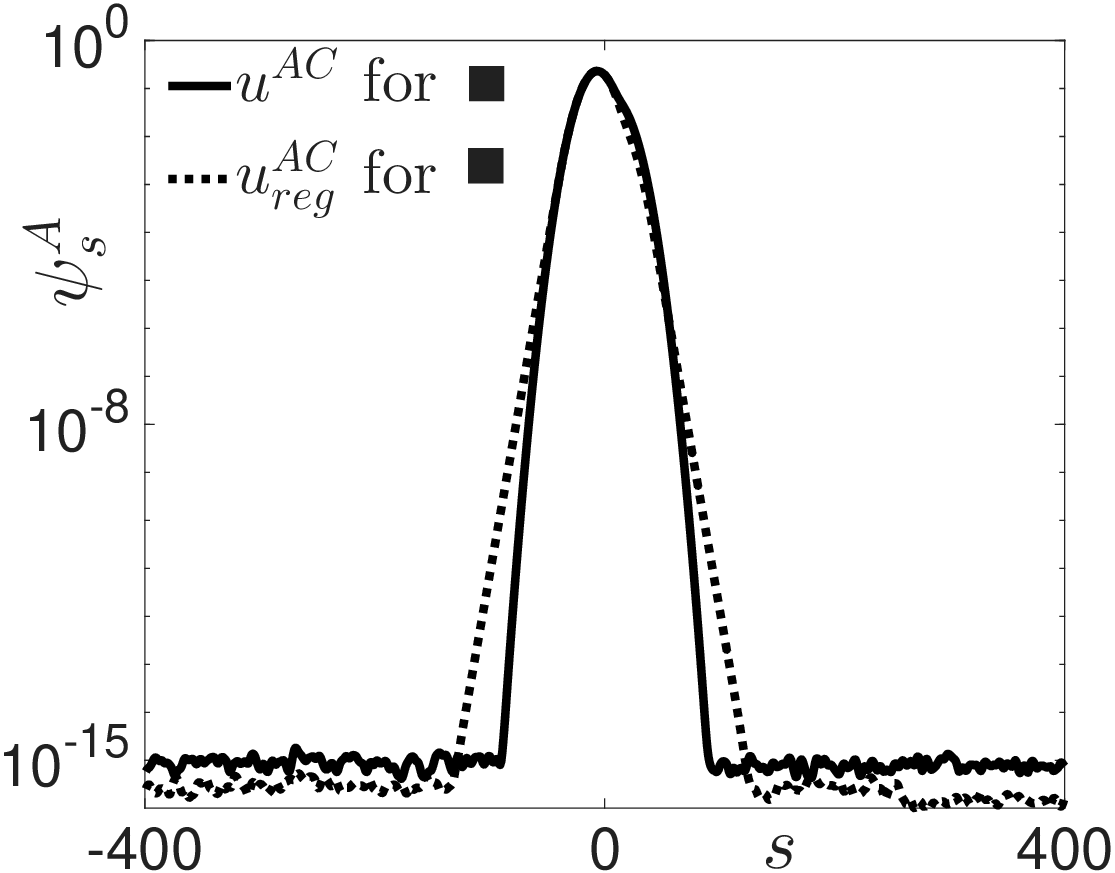}\label{fig:ac-comp-nonzero-osc-a}
	}\hfil
	\subfloat[]{
		\includegraphics[height=1.4in]{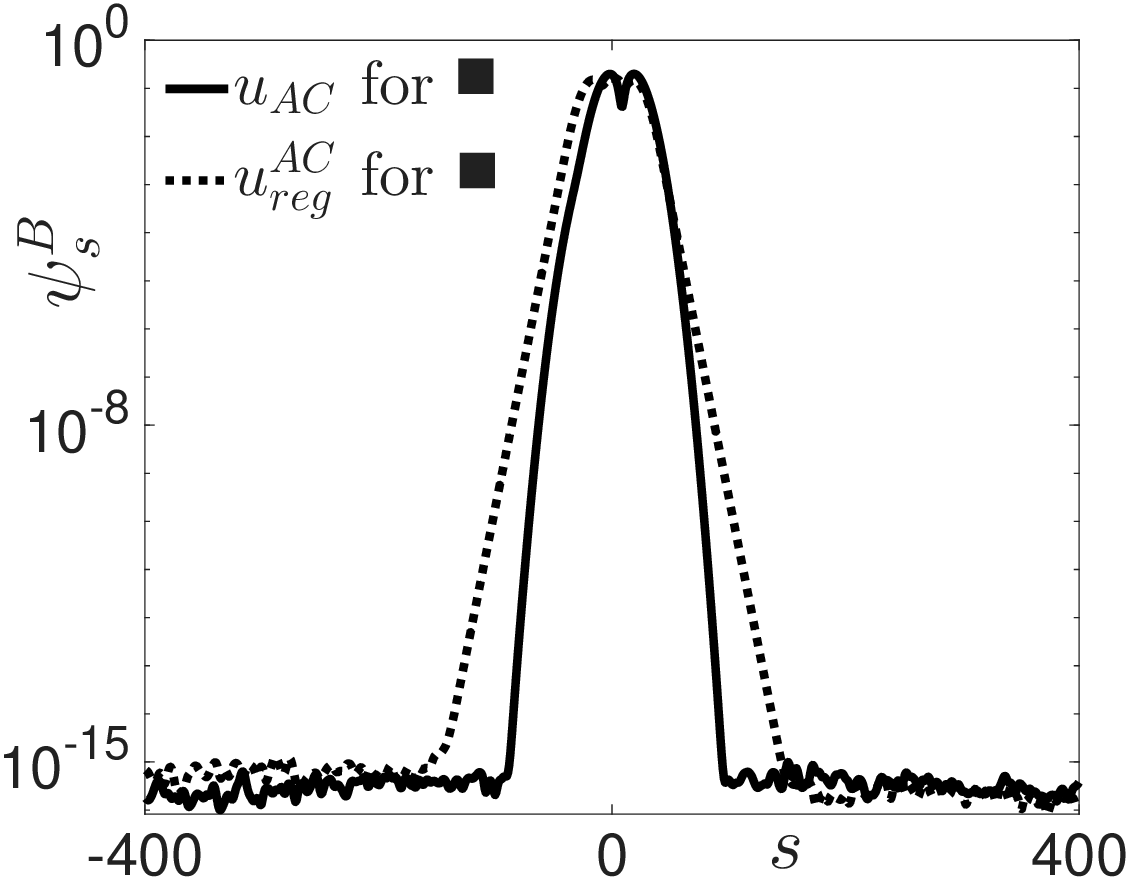}\label{fig:ac-comp-nonzero-osc-b}
	}
	\caption{Comparison between eigenstates for $\bm{u}^{\mathrm{AC}}$ and $\bm{u}^{\mathrm{AC}}_\text{reg}$ with $\alpha_\text{reg} = 0.9$, $\delta = 0.04$ and $N_T = 400$ using semilogy plots: (a) the zero eigenstates at $q_\parallel=0$ plotted at $B$ nodes; (b)-(c) eigenstates for the smallest nonzero eigenvalues, with (b) on $A$ nodes and (c) on $B$ nodes.}
	\label{fig:ac-comparison}
\end{figure}

    

\paragraph{{No spatial localization for deformations with ZZ orientation}} \edit{We also perform the quadratic deformation along the zigzag orientation, which corresponds to the displacement $\bm u^{\rm ZZ} = (X_2^2, 0)^T$. The induced effective pseudo-magnetic potential and magnetic field are
\begin{align*}
    \bm A_\text{eff}^{\rm ZZ} = (0, -t_1 X_2), \qquad \bm B_{\rm eff} = 0 \widehat{\bm z}.
\end{align*}
The magnetic Dirac operator has purely continuous spectrum with $\sigma (\mathcal{D}^{\rm ZZ}) = \mathbb{R}$; see the discussion below \eqref{eqn:zz-dirac}. 
}

To compute the numerical eigenvalue curves for the ZZ edge, we use the same method described for the AC edge with a row of unit cell $A_s, B_s, C_s, D_s$ shown in Figure \ref{fig:zz-ref-nodes}. These nodes are associated with wave functions $\bm\psi_s = (\psi_s^A, \psi_s^B, \psi_s^C, \psi_s^D)^T$. For simplicity, here we choose the lattice spacing to be $a=2/\sqrt{3}$ so that the translation vector is $\bm{v}_\parallel = (0,2)^T$ as shown in Figure \ref{fig:zz-ref-nodes}. Consequently, wave functions on nodes outside the row unit cell differ by the associated Bloch phase, i.e. the oscillation at node $A_s \pm \bm v_\parallel$ is related to $\psi_s^A$ by 
\begin{align}\label{eqn:num-bloch-zz}
    \psi_s^{A \pm \bm v_\parallel} = e^{\pm i2\pi q_\parallel}\psi_s^A,
\end{align}
where $q_\parallel$ is the quasi-momentum related to $\bm{v}_\parallel = (0,2)^T$ and now lives in $q_\parallel \in [-1/2,1/2)$. 

For the undeformed honeycomb along the ZZ edge, the Dirac points occur at $q_\parallel = -1/3$ and $q_\parallel = 1/3$. Indeed, we recall that the honeycomb along the ZZ edge is obtained by a $90^\circ$ rotation of the honeycomb along the AC edge (see Figure \ref{fig:honeycomb-edges}), i.e. the vertical direction plotted in Figure \ref{fig:zz-ref-nodes} is actually the horizontal direction in Figure \ref{fig:honeycomb}, which corresponds to the vector $2\bm v_1 + \bm v_2 = (\sqrt{3},0)^T$. Therefore, by projecting the Dirac points to $2 \bm{a}_1 + \bm a_2$, we obtain that $\bm{K}$ and $\bm{K}'$ in \eqref{eqn:dirac-pt-honeycomb} correspond to $1/3$ and $-1/3$ respectively. 

Unlike the AC edge, the discrete honeycomb along the ZZ edge admits zero eigenstates for $q_\parallel \in [-1/2, -1/3] \cup [1/3, 1/2] $ (see the two zero intervals in Figure \ref{fig:zz-band-ref}), which corresponds to edge states along the ZZ direction in truncated discrete honeycomb in \cite{fefferman2024discrete}.

\begin{figure}[!htb]
	\centering
	\subfloat[]{
		\includegraphics[height=1.5in]{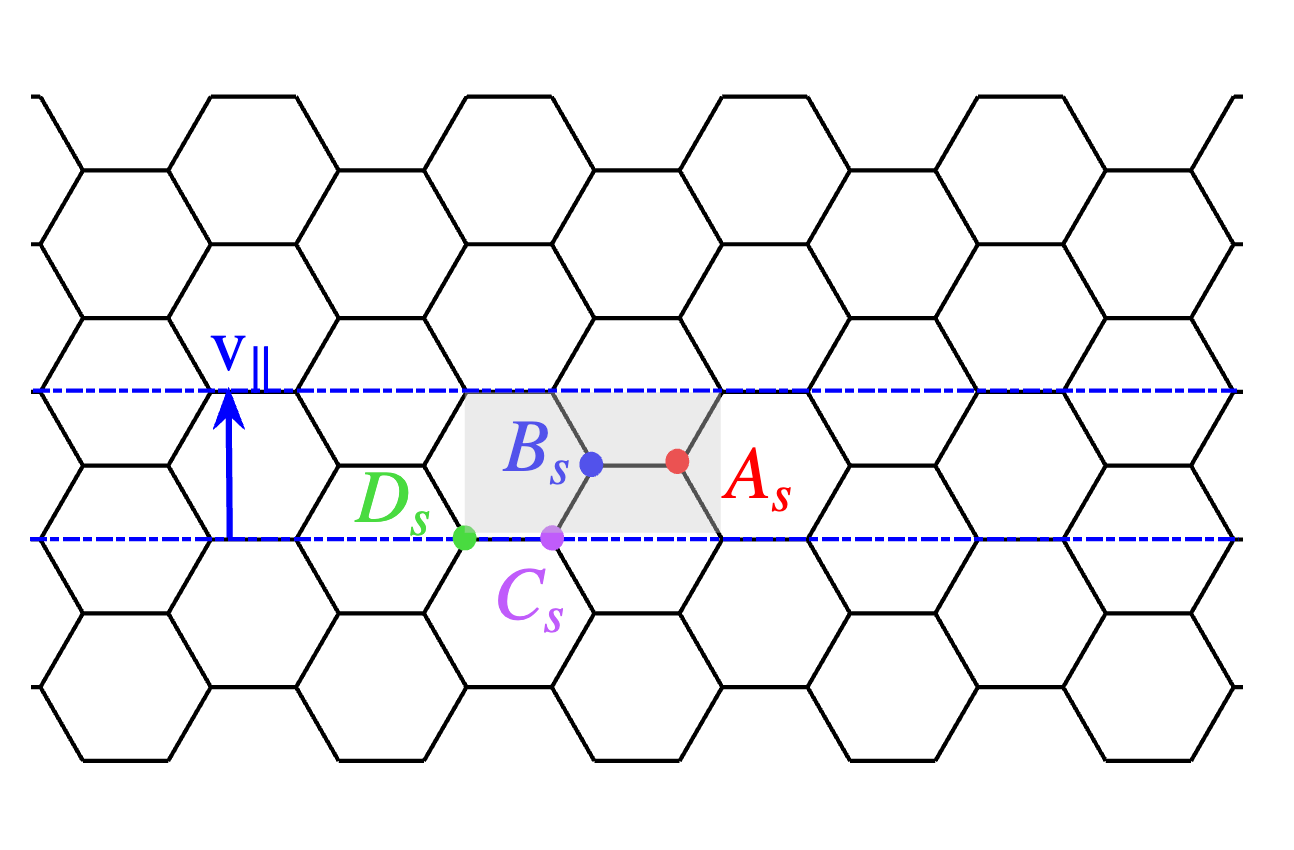}\label{fig:zz-ref-nodes}
	}\hfil
	\subfloat[]{
		\includegraphics[height=1.3in]{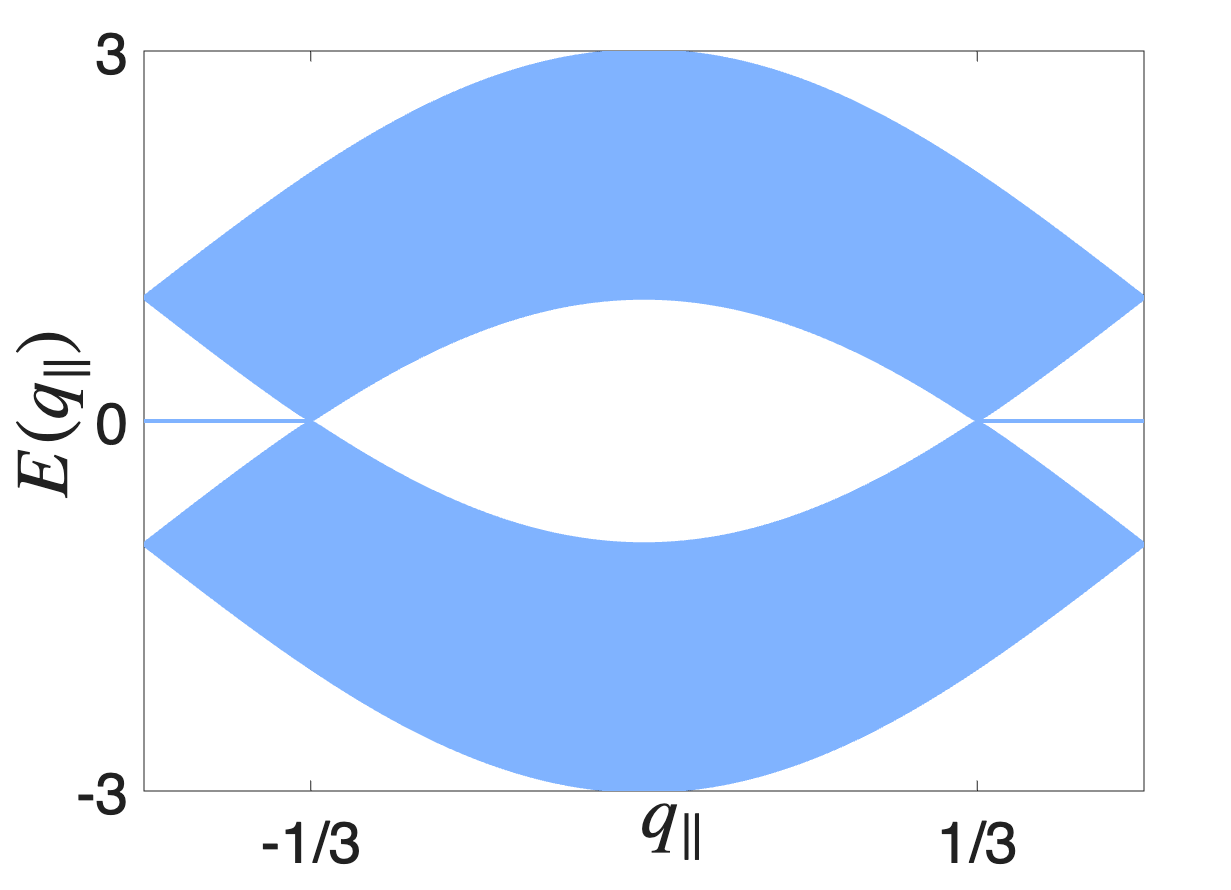}\label{fig:zz-band-ref}
	}\hfil
	\subfloat[]{
		\includegraphics[height=1.25in]{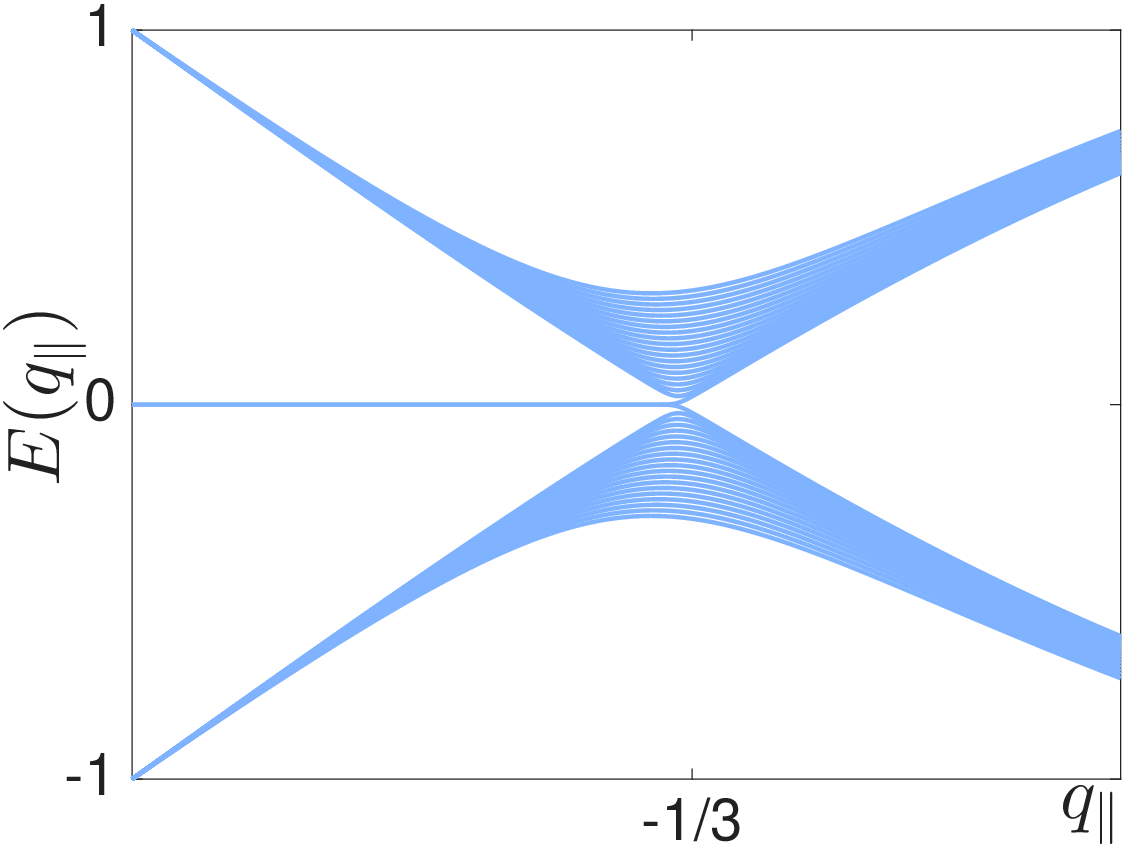}\label{fig:zz-middle-band-ref}
	}
	\caption{The numerical eigenvalue curves along the ZZ edge for the undeformed honeycomb: (a) our numerical model of the honeycomb along the ZZ edge; (b) the numerical eigenvalue curves for the undeformed honeycomb with truncation size $N_T = 200$. Dirac points are observed at $q_\parallel = -1/3$ ($\bm{K}$) and $q_\parallel = 1/3$ ($\bm{K}'$); (c) a zoomed-in version of (b) near the Dirac point $q_\parallel = -1/3$ with the smallest 40 in magnitude.}
	\label{fig:zz-ref}
\end{figure}

When we deform the honeycomb along the ZZ edge with $\bm{u}^\text{ZZ} = (X_2^2, 0)^T$, the numerical eigenvalue curve does not change the double degeneracy near the two Dirac points (see the numerical eigenvalue curves in Figure \ref{fig:zz-band-def-100} with $\delta = 0.04$ and $N_T = 100$). We also observe that in Figure \ref{fig:zz-band-def-100}, the numerical eigenvalue curves near $q_\parallel = \pm \frac{1}{2}$ become flat (the flatness can also be viewed in the zoomed-in numerical eigenvalue curves in Figure \ref{fig:zz-middle-band-def-100}). Although our theory does not capture these flattened eigenvalue curves (our ansatz only captures the behavior of numerical eigenvalue curves near the Dirac point $q_\parallel=\pm1/3$), we do not expect spectral gaps since the operator $\mathcal{D}^{\mathbf{ZZ}}$ in \eqref{eqn:zz-dirac} has purely continuous spectrum with $\sigma(\mathcal{D}^{\mathbf{ZZ}})=\mathbb{R}$. In fact, as the numerical sample size $N_b$ increases, the numerical eigenvalue curves in the range $q_\parallel\in[-1/2,-1/3]\cup[1/3,1/2]$ become denser and gaps become smaller, as shown in Figure \ref{fig:zz-band-def-100} with $N_T=200$ and Figure \ref{fig:zz-band-def} with $N_T=400$ for the same $\delta=0.04$ (the gap is visibly smaller for larger $N_T$ -- the gap is approximately $0.04$ for $N_T=200$ in Figure \ref{fig:zz-middle-band-def-100}, and approximately $0.02$ for $N_T = 400$ in Figure \ref{fig:zz-middle-band-def-200}). We expect that, as $N_T \rightarrow \infty$, the numerical eigenvalue curves eventually fill in the apparent gaps and converge to the continuous spectrum.

\begin{figure}[!htb]
	\centering
	\subfloat[]{
		\includegraphics[height=1.5in]{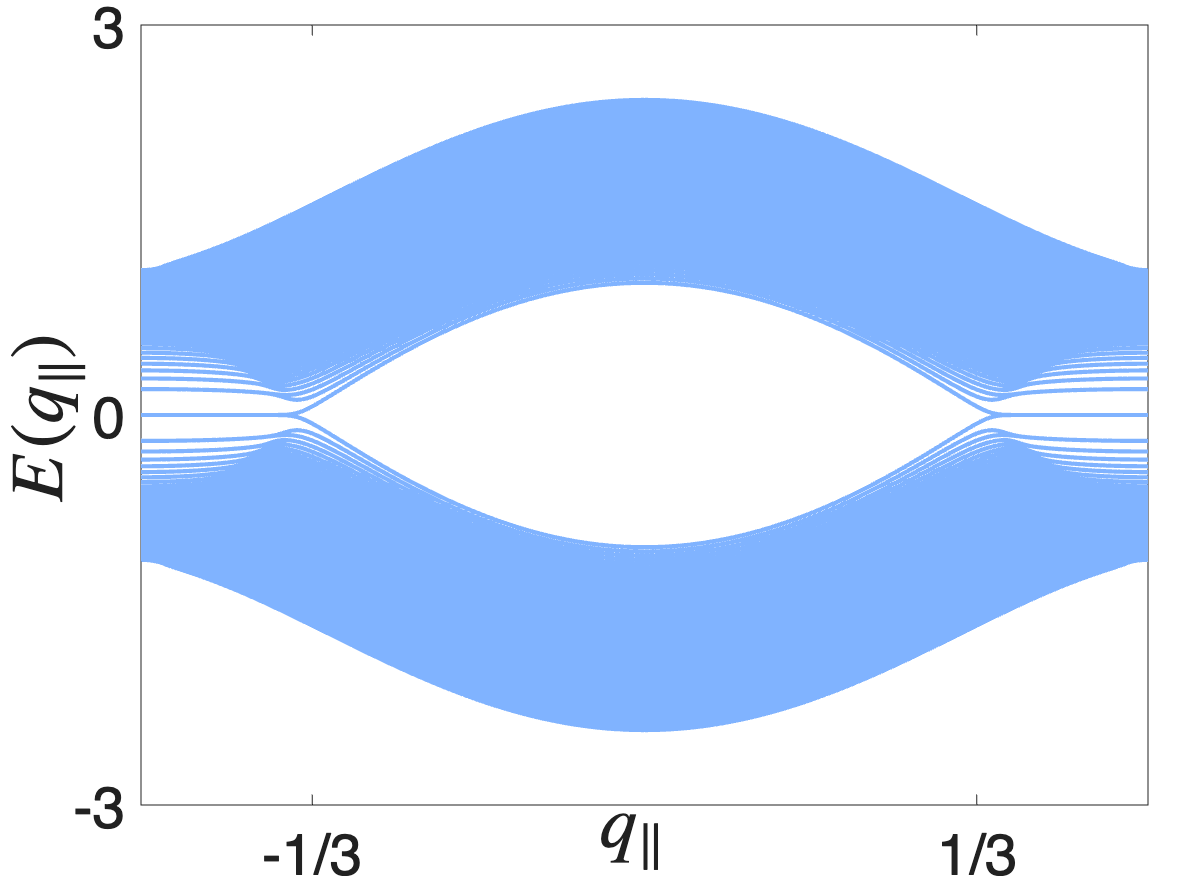}\label{fig:zz-band-def-100}
	}
	\subfloat[]{
		\includegraphics[height=1.45in]{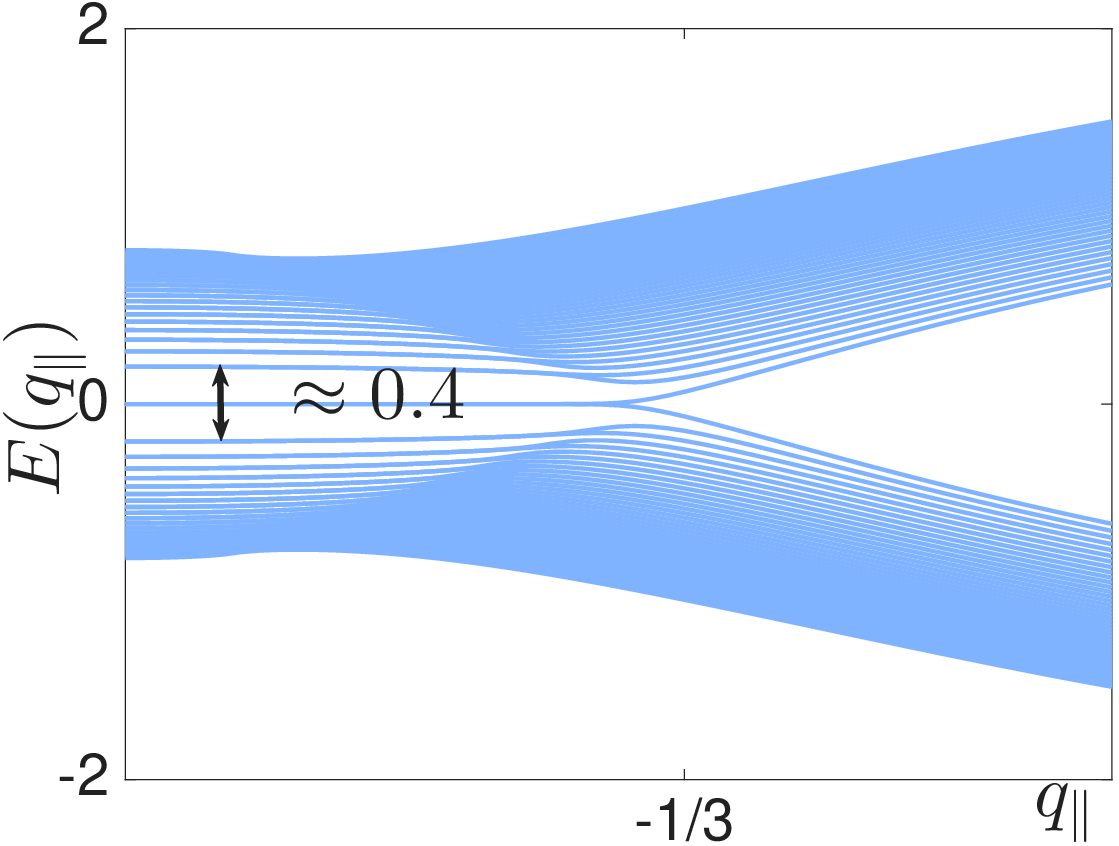}\label{fig:zz-middle-band-def-100}
	}\\
    \subfloat[]{
		\includegraphics[height=1.5in]{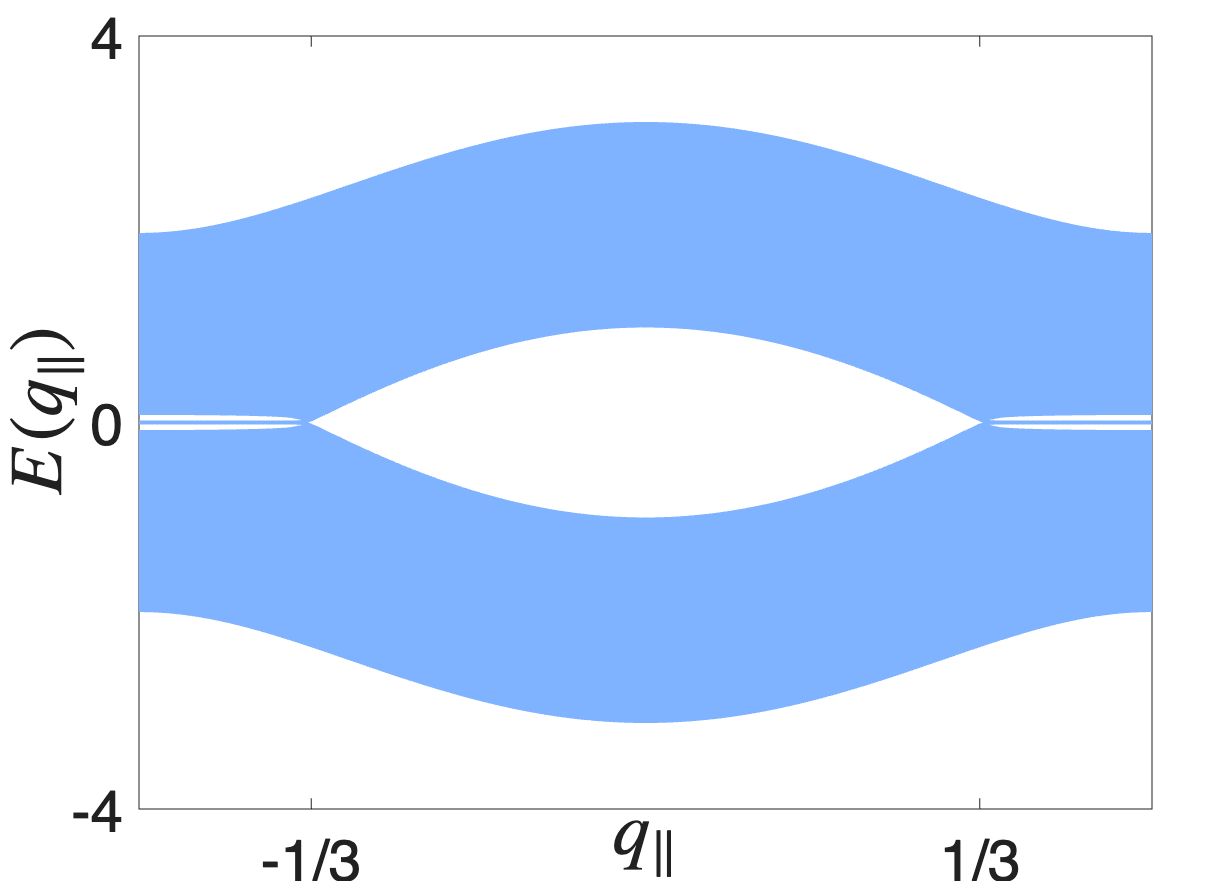}\label{fig:zz-band-def}
	}
	\subfloat[]{
		\includegraphics[height=1.45in]{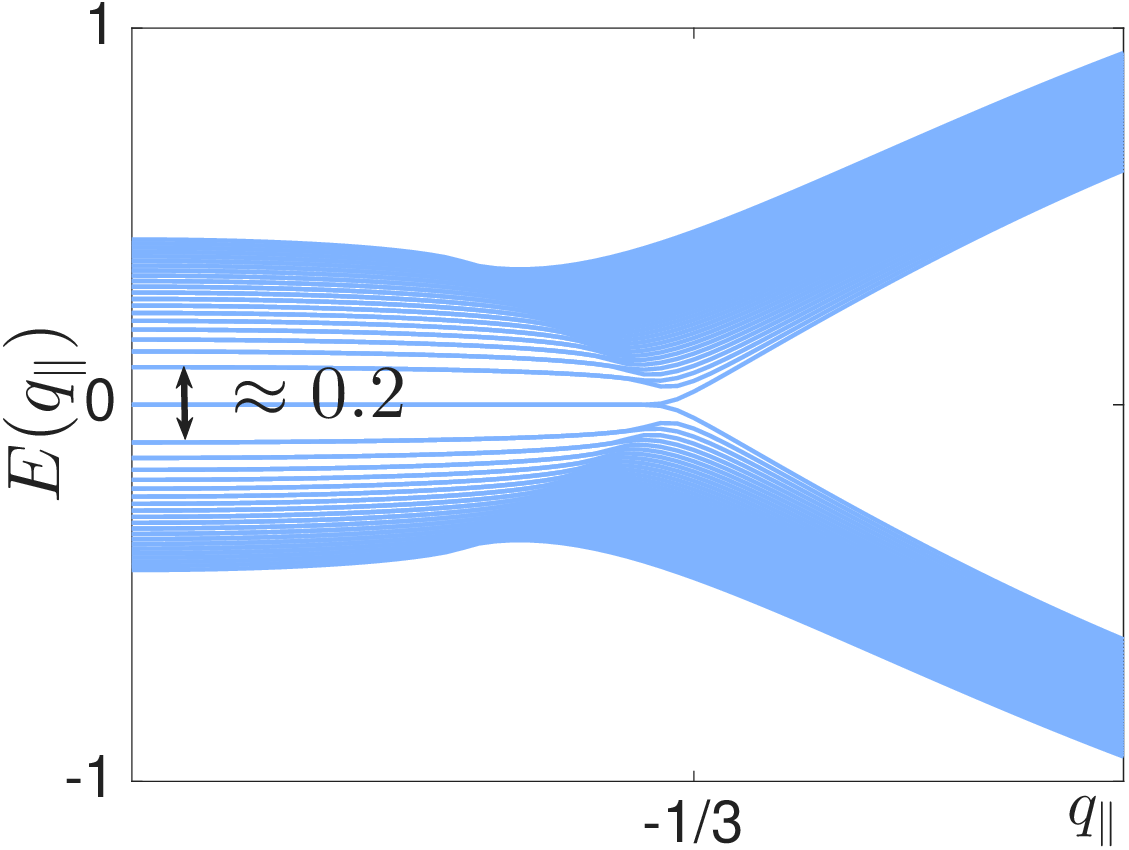}\label{fig:zz-middle-band-def-200}
	}
	\caption{The numerical eigenvalue curves for the deformed honeycomb along the ZZ edge with $\delta = 0.04$: (a) and (c) are the numerical eigenvalue curves with truncation size $N_T = 200$ and $N_T = 400$; (b) and (d) are zoomed-in versions of (a) and (c) near the Dirac point $q_\parallel = -1/3$ with the smallest 40 in magnitude.}
	\label{fig:zz-def}
\end{figure}


\section{Proof of Theorem \ref{thm:main}}\label{sec:proof}

We now prove Theorem \ref{thm:main} by following the steps mentioned in Section \ref{subsec:main-thm}. The remainder of this section develops these steps in detail and derives bounds for the correctors $\widetilde{\bm{\eta}}(k)$. All Lemmas and Propositions in this section are proved under the assumptions of Theorem \ref{thm:main}, and no additional conditions are imposed. For simplicity, we suppress the dependence on $k_\parallel$ for the eigenvalues and Dirac operator; for example, we write $E_1(k_\parallel), E_2(k_\parallel)$ as $E_1, E_2$, $\kappa(X_1;k_\parallel)$ as $\kappa(X_1)$, and $\mathcal{D}(k_\parallel)$ as $\mathcal{D}$.

\subsection{Step 1: Equation for the corrector, $\eta$, and its DFT}\label{subsec:eta-M}

Substituting our eigenstate \eqref{eqn:ansatz-1d-1} into the eigenvalue problem \eqref{eqn:eig-def-hon}, we obtain the equation for $\bm \eta_m$ (the derivation is presented in Appendix \ref{app:eqn-eta-AB}; see \eqref{eqn:app-sim-eig}). To solve for $\bm \eta_m$, we first apply DFT to $\bm \eta_m$ to derive the equations for these ccorrectors in the Fourier space. We use a scaled DFT, which is defined as follows 
\begin{equation}\label{eqn:scale-dft}
	\widetilde{\bm \eta}(k) = \big(\widetilde{\eta}^A(k), \widetilde{\eta}^B(k)\big)^T :=\sum_{m \in \mathbb{Z}} \bm \eta_m e^{-i k\frac{\sqrt{3}}{2} m}, 
\end{equation}
with quasi-momentum $k \in [-\frac{2\pi}{\sqrt{3}}, \frac{2\pi}{\sqrt{3}})$. We choose the scaling with factor $\frac{\sqrt{3}}{2}$ so that the leading order terms related to $\bm \Psi_0$ directly satisfy the Dirac equation \eqref{eqn:1st-expansion} (see the discussion near \eqref{eqn:app-psi-cancel}). Correspondingly, the inverse discrete Fourier transform (IDFT) is
\begin{equation}\label{eqn:scale-idft}
    \bm \eta_m = \frac{\sqrt{3}}{4\pi}\int_{-\frac{2\pi}{\sqrt{3}}}^{\frac{2\pi}{\sqrt{3}}} \widetilde{\bm \eta}(k) e^{ik\frac{\sqrt{3}}{2}m}\: dk.
\end{equation}
Using a scaled version of the Parseval's inequality, we obtain the following norm relation
\begin{equation}\label{eqn:parseval-norm}
    \|\bm \eta_m\|_{l^2(\mathbb{Z};\mathbb{C}^2)}^2 = \;\sum_{m\in\mathbb Z} |\bm \eta_m|^2
=\frac{\sqrt3}{4\pi}\int_{-\frac{2\pi}{\sqrt{3}}}^{-\frac{2\pi}{\sqrt{3}}}|\widetilde{\bm \eta}(k)|^2\,dk\; = \frac{\sqrt3}{4\pi} \|\widetilde{\bm \eta}\|_{L^2_k\left(\left[-\frac{2\pi}{\sqrt{3}}, \frac{2\pi}{\sqrt{3}}\right]\right)}^2.
\end{equation}

To derive the equation for $\widetilde{\bm \eta}(k)$, we multiply both sides of the eigenvalue problem \eqref{eqn:eig-def-hon} by $e^{-i k \frac{\sqrt{3}}{2} m}$. Since \eqref{eqn:eig-def-hon} is linear in the correctors $\bm \eta_m$, the resulting equation for $\widetilde{\bm{\eta}}(k)$ takes the form of a linear system. The detailed derivation is deferred to Appendix \ref{app:eqn-eta-AB}; here we record the resulting linear system in its simplest form:
\begin{subequations}\label{eqn:eta-full}
\begin{align}
    & -\Gamma_1(k,k_\parallel,\delta) \: \widetilde{\eta}^B(k) + \delta E_1 \widetilde{\eta}^A(k) + \delta^2 \mu \widetilde{\eta}^A(k) + \delta \widetilde{F}_1[k;\widetilde{\bm{\eta}}] = \widetilde{I}_1[k;\bm \Psi_0, \mu, \delta], \label{eqn:eta-full-1} \\
    & -\: \Gamma_2(k,k_\parallel,\delta) \widetilde{\eta}^A(k) + \delta E_1 \widetilde{\eta}^B(k) + \delta^2 \mu \widetilde{\eta}^B(k) + \delta\widetilde{F}_2[k;\widetilde{\bm{\eta}}] = \widetilde{I}_2[k;\bm \Psi_0, \mu, \delta] ,\label{eqn:eta-full-2}
\end{align}
\end{subequations}
where $\Gamma_1$ and $\Gamma_2$ are complex Fourier multipliers given by
\begin{subequations}\label{eqn:gamma-cst} 
\begin{align}
    \Gamma_1(k,k_\parallel,\delta) &=e^{i \frac{\pi}{3}}
        \Big(1 + e^{i \frac{4\pi}{3}} e^{i \sqrt{3} k} + e^{i \frac{2\pi}{3}} e^{i \frac{3}{2}  \delta k_\parallel} e^{i \frac{\sqrt{3}}{2} k}\Big), \label{eqn:gamma-cst-1} \\
    \Gamma_2(k,k_\parallel,\delta) &=e^{-i \frac{\pi}{3}}
        \Big(1 + e^{-i \frac{4\pi}{3}} e^{-i \sqrt{3} k} + e^{-i \frac{2\pi}{3}} e^{-i \frac{3}{2}  \delta k_\parallel} e^{-i \frac{\sqrt{3}}{2} k}\Big).\label{eqn:gamma-cst-2} 
\end{align}
\end{subequations}
The terms $\widetilde{F}_1[k;\widetilde{\bm{\eta}}]$ and $\widetilde{F}_2[k;\widetilde{\bm{\eta}}]$ collect the contributions from the slowly varying component of the hopping coefficients; they depend linearly on $\widetilde{\bm{\eta}}$ and are given by
\begin{subequations}\label{eqn:widetilde-F}
\begin{align}
    \widetilde{F}_1[k;\widetilde{\bm{\eta}}] &= - e^{i \frac{\pi}{3}} t_1 \sum_{m \in \mathbb{Z}} \ff_1\left(\frac{\sqrt{3}}{2}\delta m\right) \eta_m^B e^{-ik\frac{\sqrt{3}}{2}m} - e^{i \frac{5\pi}{3}} t_1 \sum_{m \in \mathbb{Z}} \ff_2\left(\frac{\sqrt{3}}{2}\delta m\right) \eta_{m+2}^B e^{-ik\frac{\sqrt{3}}{2}m}, \label{eqn:widetilde-F1}\\
    \widetilde{F}_2[k;\widetilde{\bm{\eta}}] &= - e^{-i \frac{\pi}{3}} t_1 \sum_{m \in \mathbb{Z}} \ff_1\left(\frac{\sqrt{3}}{2}\delta m\right) \eta_m^A e^{-ik\frac{\sqrt{3}}{2}m} - e^{-i \frac{5\pi}{3}} t_1  \sum_{m \in \mathbb{Z}}\ff_2\left(\frac{\sqrt{3}}{2}\delta (m-2)\right) \eta_{m-2}^A e^{-ik\frac{\sqrt{3}}{2}m},\label{eqn:widetilde-F2}
\end{align}
\end{subequations}
where $\ff_1(X_1)$ and $\ff_2(X_1)$ are defined as
\begin{align}
    \ff_1(X_1) := \frac{\sqrt{3}}{4} d'(X_1), \qquad \ff_2(X_1) := -\frac{\sqrt{3}}{4} d'(X_1), \label{eqn:mathcalf}
\end{align}
such that $\ff_\nu\left(\frac{\sqrt{3}}{2}\delta m\right) = f_\nu(\delta \Cell_{m_n})$ for $\nu = 1,2$ as shown in \eqref{eqn:ham-uni-dir} (we omit $\ff_3$ since $f_3 = 0$). The full expressions for $\widetilde{I}_1[k; \bm \Psi_0, \mu, \delta]$ and $\widetilde{I}_2[k; \bm \Psi_0, \mu, \delta]$ are postponed in \eqref{eqn:I1-full} and \eqref{eqn:I2-full}. However, since the $\mu$-dependent terms in both will be central to later analysis, we extract and present them explicitly below:
\begin{subequations}\label{eq:I-simp}
\begin{align}
    \widetilde{I}_1[k;\bm \Psi_0, \mu, \delta] &= -\frac{2}{\sqrt{3}} \mu \sum_{m \in \mathbb{Z}} \widehat{\Psi_0^A}\left(\frac{k + \frac{4\pi}{\sqrt{3}}m}{\delta}\right) + \widetilde{I}_{1,\text{ind}}[k;\bm \Psi_0, \delta]\label{eqn:I1-simp},\\
    \widetilde{I}_2[k;\bm \Psi_0, \mu, \delta] &= -\frac{2}{\sqrt{3}} \mu \sum_{m \in \mathbb{Z}} \widehat{\Psi_0^B}\left(\frac{k + \frac{4\pi}{\sqrt{3}}m}{\delta}\right) + \widetilde{I}_{2,\text{ind}}[k; \bm \Psi_0, \delta]\label{eqn:I2-simp},
\end{align}
\end{subequations}
where $\widetilde{I}_{1,\text{ind}}[k;\bm \Psi_0, \delta]$ and $\widetilde{I}_{2,\text{ind}}[k;\bm \Psi_0, \delta]$ depend linearly on $\bm \Psi_0$ when $\delta$ is given. 

\subsection{Step 2: \edit{Decomposition of the corrector, $\eta$ into  ``near''- and ``far'' quasi-momentum components}}\label{subsec:momentum-sepration}

We now introduce the near- and far  \edit{quasi-momentum decomposition of} functions in $L^2\left(\left[-\frac{2\pi}{\sqrt{3}}, \frac{2\pi}{\sqrt{3}}\right]\right)$. For a given $\delta$ and $\tau \in (0,1)$, we separate $L^2\left(\left[-\frac{2\pi}{\sqrt{3}}, \frac{2\pi}{\sqrt{3}}\right]\right)$ into two complete subspaces as follows
\begin{subequations}\label{eqn:near-far-def}
\begin{align}
    L^2_{\text{near},\delta^\tau}\left(\left[-\frac{2\pi}{\sqrt{3}}, \frac{2\pi}{\sqrt{3}}\right]\right) &:= \Big\{\widetilde{f}(k)\in L^2\left(\left[-\frac{2\pi}{\sqrt{3}}, \frac{2\pi}{\sqrt{3}}\right]\right) \: \Big| \: \widetilde{f}(k) = \widetilde{f}(k) \chi(|k| \leq \delta^\tau)\Big\}, \label{eqn:near-far-def-1}\\
	L^2_{\text{far},\delta^\tau}\left(\left[-\frac{2\pi}{\sqrt{3}}, \frac{2\pi}{\sqrt{3}}\right]\right) &:= \Big\{\widetilde{f}(k)\in L^2\left(\left[-\frac{2\pi}{\sqrt{3}}, \frac{2\pi}{\sqrt{3}}\right]\right)  \: \Big| \: \widetilde{f}(k) = \widetilde{f}(k) \chi(|k| \geq \delta^\tau)\Big\},\label{eqn:near-far-def-2}
\end{align}
\end{subequations}
where $\chi(k \in \Omega)$ is a characteristic function of the set $\Omega$. 

Then we decompose $\widetilde{\bm \eta}(k)$ into the near- and far-momentum parts by multiplying $\chi(|k| \leq \delta^\tau)$ and $\chi(|k| \geq \delta^\tau)$, i.e. $\widetilde{\bm{\eta}}(k) = \widetilde{\bm{\eta}}_\text{near}(k) + \widetilde{\bm{\eta}}_\text{far}(k)$ with
\begin{align}
    \widetilde{\bm{\eta}}_\text{near}(k) := \widetilde{\bm{\eta}}(k) \chi(|k| \leq \delta^\tau), \qquad \widetilde{\bm{\eta}}_\text{far}(k) := \widetilde{\bm{\eta}}(k) \chi(|k| \geq \delta^\tau), \label{eqn:eta-near-far-k}
\end{align}
where $\widetilde{\bm{\eta}}_\text{near}(k) = \big(\widetilde{{\eta}}^A_\text{near}(k), \widetilde{{\eta}}^B_\text{near}(k)\big)^T$ and $\widetilde{\bm{\eta}}_\text{far}(k) = \big(\widetilde{{\eta}}^A_\text{far}(k), \widetilde{{\eta}}^B_\text{far}(k)\big)^T$. We can also write the near- and far-momentum components in the real space 
\begin{align}
	\bm \eta_m &= \bm \eta_m^{\text{near}} + \bm \eta_m^{\text{far}},\label{eqn:eta-near-full-M}
\end{align}
where $\bm \eta_m^{\text{near}}$ and $\bm \eta_m^{\text{far}}$ are the IDFT (see \eqref{eqn:scale-idft}) of $\widetilde{\bm \eta}_\text{near}(k)$ and  $\widetilde{\bm \eta}_\text{far}(k)$.

Before writing the near- and far-momentum equations for $\widetilde{\bm{\eta}}_\text{near}(k)$ and $\widetilde{\bm{\eta}}_\text{far}(k)$, we note a simple consequence of linearity for the terms $\widetilde{F}_1[k;\widetilde{\bm{\eta}}]$ and $\widetilde{F}_2[k;\widetilde{\bm{\eta}}]$ in \eqref{eqn:widetilde-F}: they decompose additively into near and far contributions, i.e.
\begin{equation}\label{eqn:tilde-F-separation}
    \widetilde{F}_i[k;\widetilde{\bm{\eta}}] = \widetilde{F}_i[k;\widetilde{\bm{\eta}}_\text{near}] + \widetilde{F}_i[k;\widetilde{\bm{\eta}}_\text{far}], \qquad i=1,2.
\end{equation}
It is worth noting that $\widetilde{F}_i[k;\widetilde{\bm{\eta}}_\text{near}]$ with $i=1,2$ need not lie in $L^2_{\text{near},\delta^\tau}\left(\left[-\frac{2\pi}{\sqrt{3}}, \frac{2\pi}{\sqrt{3}}\right]\right)$ since multiplication by the slowly-varying coefficients $\ff_\nu\left(\frac{\sqrt{3}}{2}\delta m\right)$ in \eqref{eqn:widetilde-F} can generate far-momentum components. Similarly, the term $\widetilde{F}_i[k;\widetilde{\bm{\eta}}_\text{far}]$ may not belong to $L^2_{\text{far},\delta^\tau}\left(\left[-\frac{2\pi}{\sqrt{3}}, \frac{2\pi}{\sqrt{3}}\right]\right)$.

\paragraph{The far-momentum equation} To obtain the equation for the far-momentum part $\widetilde{\bm{\eta}}_\text{far}$, we multiply \eqref{eqn:eta-full} by $\chi(|k| \geq \delta^\tau)$ respectively. We shall solve the far-momentum component with prescribed near-momentum component. Therefore, we treat the near-momentum parts as source terms and collect them on one side of the equations by using \eqref{eqn:tilde-F-separation}, i.e.
\begin{subequations}\label{eqn:eta-far}
\begin{align}
    & -\Gamma_1(k,k_\parallel,\delta) \: \widetilde{\eta}^B_\text{far}(k) + \delta E_1 \widetilde{\eta}^A_\text{far}(k) + \delta^2 \mu \widetilde{\eta}^A_\text{far}(k) + \delta \widetilde{F}_1[k;\widetilde{\bm{\eta}}_\text{far}] \: \chi(|k|\geq \delta^\tau) \label{eqn:eta-far-1}\\
    & = \bigg(\widetilde{I}_1[k;\bm \Psi_0, \mu, \delta] - \delta\widetilde{F}_1[k;\widetilde{\bm{\eta}}_\text{near}]\bigg) \: \chi(|k|\geq \delta^\tau),  \nonumber \\
    & -\Gamma_2(k,k_\parallel,\delta) \: \widetilde{\eta}^A_\text{far}(k) + \delta E_1 \widetilde{\eta}^B_\text{far}(k) + \delta^2 \mu \widetilde{\eta}^B_\text{far}(k) + \delta \widetilde{F}_2[k;\widetilde{\bm{\eta}}_\text{far}] \: \chi(|k|\geq \delta^\tau) \label{eqn:eta-far-2}\\
    & = \bigg(\widetilde{I}_2[k;\bm \Psi_0, \mu, \delta] -\delta \widetilde{F}_2[k;\widetilde{\bm{\eta}}_\text{near}]\bigg) \: \chi(|k|\geq \delta^\tau) . \nonumber
\end{align}
\end{subequations}

\paragraph{The near-momentum equation} To obtain the equation for the near-momentum part $\widetilde{\bm{\eta}}_\text{near}$, we multiply \eqref{eqn:eta-full} by $\chi(|k| \leq \delta^\tau)$ respectively. Here we treat the far-momentum parts as functions of the near-momentum parts and group them with the near-momentum terms on the same side, i.e.
\begin{subequations}\label{eqn:eta-near}
\begin{align}
    &-\Gamma_1(k,k_\parallel,\delta) \: \widetilde{\eta}^B_\text{near}(k) + \delta E_1 \widetilde{\eta}^A_\text{near}(k) + \delta^2 \mu \widetilde{\eta}^A_\text{near}(k)  + \delta \widetilde{F}_1[k;\bm{\widetilde{\eta}}_\text{near}] \: \chi(|k|\leq \delta^\tau) \label{eqn:eta-near-1} \\
    &= \bigg(\widetilde{I}_1[k;\bm \Psi_0, \mu, \delta] - \widetilde{F}_1[k;\bm{\widetilde{\eta}}_\text{far}] \bigg)\:\chi(|k|\leq \delta^\tau) ,  \nonumber\\
    &- \Gamma_2(k,k_\parallel,\delta) \widetilde{\eta}^A(k) + \delta E_1 \widetilde{\eta}^B(k) + \delta^2 \mu \widetilde{\eta}^B(k) + \delta \widetilde{F}_2[k;\bm{\widetilde{\eta}}_\text{near}] \: \chi(|k|\leq \delta^\tau),\label{eqn:eta-near-2}\\
    &= \bigg(\widetilde{I}_2[k; \bm \Psi_0, \mu, \delta] -\widetilde{F}_2[k;\bm{\widetilde{\eta}}_\text{far}]\bigg)\:\chi(|k|\leq \delta^\tau).\nonumber
\end{align}
\end{subequations}

Before we solve the near- and far-momentum equations, we provide some useful bounds on terms $\Gamma_i$, $\widetilde{F}_i, \widetilde{I}_i$ with $i=1,2$ and summarize them as the following Lemma:
\begin{lemma}\label{lemma:bds-gamma-F-I}
    Fix $M > 0$ and $\delta_0 > 0$. Consider $|\mu| \leq M$ and $\delta \in (0,\delta_0)$. We have:\\
    (i) for any bounded $k_\parallel$, the constants $\Gamma_i(k_,k_\parallel,\delta)$ with $i=1,2$ in \eqref{eqn:gamma-cst} satisfy:
    \begin{subequations}\label{eqn:bd-gamma}
        \begin{align}
       &|\Gamma_i(k,k_\parallel,\delta)| \gtrsim |k| \geq \delta^\tau, \qquad & &\text{when}\quad |k| \geq \delta^\tau \text{ and } k\in\Icell, \label{eqn:bd-gamma-far}\\
       &|\Gamma_i(k,k_\parallel,\delta)| \lesssim |k|, \qquad & &\text{when}\quad |k| \leq \delta^\tau;\label{eqn:bd-gamma-near}
    \end{align}
    \end{subequations}
    (ii) for any $\bm{\eta}_m \in l^2(\mathbb{Z},\mathbb{C}^2)$, its DFT $\bm{\widetilde{\eta}}(k)$ in \eqref{eqn:scale-dft} satisfies
    \begin{align}\label{eqn:bd-F-main}
        \norm{\widetilde{F}_i[k;\widetilde{\bm{\eta}}]}_{\LtwoBrill} \lesssim \norm{\widetilde{\bm{\eta}}}_{\LtwoBrill},\qquad i=1,2;
    \end{align}
    (iii) for any $\bm \Psi_0\in \mathcal{S}(\mathbb{R})$, the far-momentum contributions of $\widetilde{I}_i[k;\bm \Psi_0, \mu, \delta]$ with $i=1,2$ are bounded by
    \begin{align}
        \norm{\widetilde{I}_i[k; \bm \Psi_0, \mu, \delta] \: \chi(|k|\geq \delta^\tau)}_{\LtwoBrill} \leq \delta^{\frac{1}{2}} \norm{\bm \Psi_0}_{H^2(\mathbb{R})} .\label{eqn:bd-I-far-main}
    \end{align}
\end{lemma}
\begin{proof}
    The proof of \eqref{eqn:bd-gamma} is straightforward since for $i=1,2$, the constants satisfy that (1) $\Gamma_i(k,k_\parallel,\delta)$ vanishes only at $k=0$ when $\delta$ is small; and (2) $\partial_k \Gamma_i(0,k_\parallel,\delta) \neq 0$. The proof of \eqref{eqn:bd-F-main} follows directly from the fact that the slowly-varying coefficients $\ff_1\left(\frac{\sqrt{3}}{2} \delta m\right)$ and $\ff_2\left(\frac{\sqrt{3}}{2} \delta m\right)$ in \eqref{eqn:widetilde-F} are uniformly bounded when $d'(X_1) \in L^\infty(\mathbb{R})$. The proof of \eqref{eqn:bd-I-far-main} is shown in Appendix \ref{app:some-bds} (see the discussion near \eqref{eqn:bd-I-far}) after we present the explicit formulas for $\widetilde{I}_i[k;\bm \Psi_0,\mu, \delta]$. 
    

\end{proof} 


\subsection{Step 3: solving the far-momentum equation}\label{subsec:far-sol}

Since the far-momentum equations \eqref{eqn:eta-far} form a linear system in $\widetilde{\bm \eta}(k)$ once the near-momentum parts and $\mu, \delta$ are given, it suffices to show that the associated linear operator in the far-momentum equation \eqref{eqn:eta-far} is invertible in order to solve for the far-momentum parts of the form
\begin{equation}
    \widetilde{\bm{\eta}}_\text{far}(k) = \widetilde{\bm{\eta}}_\text{far}[k; \widetilde{\bm{\eta}}_\text{near}, \mu, \delta].
\end{equation}

We now explain why the system \eqref{eqn:eta-far} is invertible. We divide both sides of \eqref{eqn:eta-far-1} and \eqref{eqn:eta-far-2} by $\Gamma_1$ and $\Gamma_2$ respectively. Then we reorganize the resulting equations and obtain
\begin{subequations}\label{eqn:eta-far-gamma}
\begin{align}
    & \widetilde{\eta}^A_\text{far}(k) - \frac{\delta E_1}{\Gamma_2(k,k_\parallel,\delta)} \widetilde{\eta}^B_\text{far}(k) - \frac{\delta^2 \mu}{\Gamma_2(k,k_\parallel,\delta)} \widetilde{\eta}^B_\text{far}(k) - \delta \frac{\chi(|k|\geq \delta^\tau)}{\Gamma_2(k,k_\parallel,\delta)} \widetilde{F}_2[k;\bm{\widetilde{\eta}}_\text{far}] \label{eqn:eta-far-gamma-2}\\
    =\:& -\frac{\chi(|k|\geq \delta^\tau)}{\Gamma_2(k,k_\parallel,\delta)} \bigg(\widetilde{I}_2[k;\bm \Psi_0, \mu, \delta] -\delta \widetilde{F}_2[k;\bm{\widetilde{\eta}}_\text{near}]\bigg), \nonumber\\
    & \widetilde{\eta}^B_\text{far}(k) - \frac{\delta E_1}{\Gamma_1(k,k_\parallel,\delta)} \widetilde{\eta}^A_\text{far}(k) - \frac{\delta^2 \mu}{\Gamma_1(k,k_\parallel,\delta)} \widetilde{\eta}^A_\text{far}(k) - \delta \frac{\chi(|k|\geq \delta^\tau)}{\Gamma_1(k,k_\parallel,\delta)} \widetilde{F}_1[k;\bm{\widetilde{\eta}}_\text{far}] \label{eqn:eta-far-gamma-1}\\
    =\:& -\frac{\chi(|k|\geq \delta^\tau)}{\Gamma_1(k,k_\parallel,\delta)} \bigg(\widetilde{I}_1[k;\bm \Psi_0, \mu, \delta] -\delta \widetilde{F}_1[k;\bm{\widetilde{\eta}}_\text{near}]\bigg). \nonumber
\end{align}
\end{subequations}
Using \eqref{eqn:bd-gamma-far}, we have $\delta/\Gamma_i \lesssim \delta^{1-\tau}$ and $\delta^2/\Gamma_i \lesssim \delta^{2-\tau}$ are small for $i=1,2$. Therefore, to invert the system \eqref{eqn:eta-far-gamma}, it suffices to show that the remaining terms related to $\widetilde{F}_i$ on the right-hand side are also small. In fact, we have the following bound for $i=1,2$
\begin{equation}\label{eqn:bd-F12-far}
    \left\| \delta \frac{\chi(|k|\geq \delta^\tau)}{\Gamma_i(k,k_\parallel,\delta)} \widetilde{F}_i[k;\bm{\widetilde{\eta}}_\text{far}] \right\|_{\LtwoBrill} \lesssim \delta^{1-\tau} \|\widetilde{\bm{\eta}}_\text{far}\|_{\LtwoBrill}.
\end{equation}
The proof of \eqref{eqn:bd-F12-far} comes directly from \eqref{eqn:bd-gamma-far} and \eqref{eqn:bd-F-main}. Thus, the linear operator in \eqref{eqn:eta-far-gamma} is a small perturbation of the identity, and hence the system is invertible. Consequently, we can solve for the far-momentum components as $\widetilde{\bm{\eta}}_\text{far}[k; \widetilde{\bm{\eta}}_\text{near}, \mu, \delta]$.

We also provide the estimate for $\widetilde{\bm{\eta}}_\text{far}[k; \widetilde{\bm{\eta}}_\text{near}, \mu, \delta]$: when $\mu$ is bounded and $\delta$ is sufficient small, we have
\begin{align}
    \norm{\widetilde{\bm{\eta}}_\text{far}[k; \widetilde{\bm{\eta}}_\text{near}, \mu, \delta]}_{L^2_k\left(\left[-\frac{2\pi}{\sqrt{3}}, \frac{2\pi}{\sqrt{3}}\right]\right)} & \lesssim \delta^{1 - \tau} \norm{\widetilde{\bm{\eta}}_\text{near}}_{L^2_k\left(\left[-\frac{2\pi}{\sqrt{3}}, \frac{2\pi}{\sqrt{3}}\right]\right)} + \delta^{\frac{1}{2} - \tau}. \label{eqn:far-bound}
\end{align}
The proof of \eqref{eqn:far-bound} comes directly from the following bounds on the source terms in the far-momentum equation \eqref{eqn:eta-far-gamma}:
\begin{subequations}\label{eqn:bd-far}
\begin{align}
    \left\| \frac{\chi(|k|\geq \delta^\tau)}{\Gamma_2(k,k_\parallel,\delta)} \widetilde{I}_i[k;\bm \Psi_0, \mu, \delta] \right\|_{\LtwoBrill} &\lesssim \delta^{\frac{1}{2}-\tau},\label{eqn:bd-I12-far}\\
    \left\| \delta \frac{\chi(|k|\geq \delta^\tau)}{\Gamma_2(k,k_\parallel,\delta)} \widetilde{F}_i[k;\bm{\widetilde{\eta}}_\text{near}] \right\|_{\LtwoBrill} &\lesssim \delta^{1-\tau} \:\|\widetilde{\bm{\eta}}_\text{near}\|_{\LtwoBrill}.\label{eqn:bd-F12-far-near}
\end{align}
\end{subequations}
The proof of \eqref{eqn:bd-I12-far} comes directly from \eqref{eqn:bd-I-far-main}, and the proof of \eqref{eqn:bd-F12-far-near} is analogous to that of \eqref{eqn:bd-F12-far}. 

Since the far-momentum equation \eqref{eqn:eta-far-gamma} is linear in $\widetilde{\bm \eta}_\text{far}$ and $ \widetilde{\bm \eta}_\text{near}$, and Lipschitz in $\mu$, we expect the solution $\widetilde{\bm{\eta}}_\text{far}[k; \widetilde{\bm{\eta}}_\text{near}, \mu, \delta]$ of \eqref{eqn:eta-far-gamma} satisfies that (1) $\widetilde{\bm{\eta}}_\text{far}[k; \widetilde{\bm{\eta}}_\text{near}, \mu, \delta]$ is \textit{affine} in $\widetilde{\bm{\eta}}_\text{near}$; and (2) $\widetilde{\bm{\eta}}_\text{far}[k; \widetilde{\bm{\eta}}_\text{near}, \mu, \delta]$ is \textit{Lipschitz} in $\mu$. The above arguments are stated in detail in Proposition \ref{prop:far-energy}.

Before we state Proposition \ref{prop:far-energy}, we define the following closed balls with radius $R$ in the near- and far-momentum space $L^2_{\text{near},\delta^\tau}\left(\Icell\right)$ and $L^2_{\text{far},\delta^\tau}\left(\Icell\right)$, i.e.
\begin{align*}
    B_{\text{near}, \delta^\tau}(R) &:= \left\{\widetilde{f} \in L^2_{\text{near},\delta^\tau}\left(\Icell\right) \: \bigg| \: \|\widetilde{f}\|_{\LtwoBrill} \leq R\right\},\\
    B_{\text{far}, \delta^\tau}(R) &:= \left\{\widetilde{f} \in L^2_{\text{far},\delta^\tau}\left(\Icell\right) \: \bigg| \: \|\widetilde{f}\|_{\LtwoBrill} \leq R\right\}.
\end{align*}
We summarize the above arguments for the solution $\widetilde{\bm{\eta}}_\text{far}[k; \widetilde{\bm{\eta}}_\text{near}, \mu, \delta]$ in the following Proposition:
\begin{proposition}\label{prop:far-energy}
	\begin{enumerate}[(a)]
		\item For any fixed $M > 0, R > 0$, there exists a $0 < \delta_0 < 1$ such that for all $0 < \delta < \delta_0$, the system \eqref{eqn:eta-far} has a \textbf{unique} solution $\widetilde{\bm \eta}_\text{far}[k; \widetilde{\bm \eta}_\text{near}, \mu, \delta]$ as a map
		\begin{equation}
			(\widetilde{\bm{\eta}}_\text{near}, \mu, \delta) \in B_{\text{near}, \delta^\tau}(R) \times \{|\mu| < M\} \times (0,\delta_0) \longmapsto \widetilde{\bm{\eta}}_\text{far}[\cdot; \widetilde{\bm{\eta}}_\text{near}, \mu, \delta] \in B_{\text{far}, \delta^\tau} (\rho_\delta),
		\end{equation}
		with the radius $\rho_\delta = O(\delta^{\frac{1}{2}-\tau})$. 
		
		\item The mapping $(\widetilde{\bm{\eta}}_\text{near}, \mu, \delta) \longmapsto \widetilde{\bm{\eta}}_\text{far}[\cdot; \widetilde{\bm{\eta}}_\text{near}, \mu, \delta]$ is affine in $\widetilde{\bm{\eta}}_\text{near}$, Lipschitz in $\mu$, and satisfies \eqref{eqn:far-bound}.
        Moreover, the mapping $(\widetilde{\bm{\eta}}_\text{near}, \mu, \delta) \longmapsto \widetilde{\bm{\eta}}_\text{far}[\cdot; \widetilde{\bm{\eta}}_\text{near}, \mu, \delta]$ can be expressed as
		\begin{align}
			\widetilde{\bm{\eta}}_\text{far}[k; \widetilde{\bm{\eta}}_\text{near}, \mu, \delta] &= [\mathcal{A}\widetilde{\bm{\eta}}_\text{near}](k;\mu, \delta) + \mathcal{B}(k;\mu,\delta), \label{eqn:far-affine}
		\end{align}
        where $[\mathcal{A}\widetilde{\bm{\eta}}_\text{near}](k;\mu, \delta)$ acts linearly on $\widetilde{\bm{\eta}}_\text{near}$ for fixed $\mu, \delta$. For $\widetilde{\bm{\eta}}_\text{near} \in B_{\text{near},\delta^\tau}(R)$, $|\mu|\leq M$, $0 < \delta < \delta_0$, we have the following bounds on $[\mathcal{A}\widetilde{\bm{\eta}}_\text{near}](k;\mu, \delta)$ and $\mathcal{B}(k;\mu,\delta)$:
        \begin{subequations}\label{eqn:AB-map-bd}
        \begin{align}
        &\norm{[\mathcal{A}\widetilde{\bm{\eta}}_\text{near}](k;\mu, \delta)}_{\LtwoBrill} \lesssim \delta^{1-\tau} \norm{\widetilde{\bm{\eta}}_\text{near}}_{\LtwoBrill},\label{eqn:A-map-bd}\\
        & \norm{\mathcal{B}(k; \mu,\delta)}_{\LtwoBrill} \lesssim \delta^{\frac{1}{2} - \tau},\label{eqn:B-map-bd}\\
        & \norm{[\mathcal{A}\widetilde{\bm{\eta}}_\text{near}](k;\mu_1, \delta) - [\mathcal{A}\widetilde{\bm{\eta}}_\text{near}](k;\mu_2, \delta)}_{\LtwoBrill} \lesssim \delta^{1-\tau} |\mu_1 - \mu_2|,\label{eqn:A-map-lip-mu}\\
		& \norm{\mathcal{B}(k; \mu_1,\delta) - \mathcal{B}(k; \mu_2,\delta)}_{\LtwoBrill} \lesssim \delta^{\frac{1}{2} - \tau} |\mu_1 - \mu_2|.\label{eqn:B-map-lip-mu}
        \end{align}
        \end{subequations}

            \item Define the extension of $\widetilde{\bm \eta}_\text{far}[k; \widetilde{\bm \eta}_\text{far}, \mu, \delta]$ to the half-open interval $\delta \in [0,\delta_0)$ by setting $\widetilde{\bm \eta}_\text{far}[k; \widetilde{\bm \eta}_\text{far}, \mu, 0] = 0$. Then, by \eqref{eqn:far-bound}, $\widetilde{\bm \eta}_\text{far}[k; \widetilde{\bm \eta}_\text{far}, \mu, \delta]$ is continuous at $\delta = 0$.
	\end{enumerate}
\end{proposition}

\begin{proof}
The proof adapts the arguments of Proposition 6.3 and Corollary 6.4 in \cite{fefferman2017topologically} and uses the fixed point theorem. To use a fixed point argument, we reorganize the system \eqref{eqn:eta-far-gamma} in the following form (by moving the last three terms in \eqref{eqn:eta-far-gamma} to the left hand side)
\begin{equation}\label{eqn:far-contraction}
    \widetilde{\bm{\eta}}_\text{far} = \mathcal{E}[\widetilde{\bm{\eta}}_\text{far}; \widetilde{\bm{\eta}}_\text{near}, \mu, \delta].
\end{equation}
We observe that when $|\mu| \leq M$, the mapping $\mathcal{E}[\widetilde{\bm{\eta}}_\text{far}; \widetilde{\bm{\eta}}_\text{near}, \mu, \delta]$ is a contraction by the following estimate
\begin{equation}\label{eqn:map-E-bd}
    \Big\|\mathcal{E}[\widetilde{\bm{\eta}}_\text{far}; \widetilde{\bm{\eta}}_\text{near}, \mu, \delta] \Big\|_{\LtwoBrill} \lesssim \delta^{\frac{1}{2}-\tau} + \delta^{1-\tau} \:\Big(\|\widetilde{\bm{\eta}}_\text{near}\|_{\LtwoBrill} + \|\widetilde{\bm{\eta}}_\text{far}\|_{\LtwoBrill}\Big).
\end{equation}
The estimate \eqref{eqn:map-E-bd} follows from \eqref{eqn:bd-F12-far} and \eqref{eqn:far-bound}. Thus, the fixed point theorem guarantees a unique solution of \eqref{eqn:eta-far-gamma} when $\|\widetilde{\bm{\eta}}_\text{far}\| \lesssim \delta^{\frac{1}{2}-\tau}$ and $\|\widetilde{\bm{\eta}}_\text{near}\| \leq R$.

Lastly, we briefly explain the bounds on $[\mathcal{A}\widetilde{\bm{\eta}}_\text{near}](k;\mu, \delta)$ and $\mathcal{B}(k;\mu,\delta)$ in \eqref{eqn:far-affine}. The first two bounds \eqref{eqn:A-map-bd}-\eqref{eqn:B-map-bd} follow directly from \eqref{eqn:bd-F12-far-near}. The Lipschitz dependence of $\widetilde{\bm{\eta}}_\text{far}[k; \widetilde{\bm{\eta}}_\text{near}, \mu, \delta]$ on $\mu$ follows directly from the $\mu$-Lipschitz continuity of the contraction map $\mathcal{E}[\widetilde{\bm{\eta}}_\text{far}$; $\widetilde{\bm{\eta}}_\text{near}, \mu, \delta]$. However, proving the bounds \eqref{eqn:A-map-lip-mu}-\eqref{eqn:B-map-lip-mu} requires additional detail since the Lipschitz constants in them differ; we therefore postpone these proofs to Appendix \ref{app:some-bds} (see near \eqref{eqn:bd-map-E}).
\end{proof}

\subsection{The inverse operator of $\mathcal{D} - E_1$}\label{subsec:inv-Dk}


One major step in solving the near-momentum equation \eqref{eqn:eta-near} involves inverting $\mathcal{D} - E_1$ on the orthogonal space of $\bm \Psi_0$. In this section, we properly define the inverse operator $\left(\mathcal{D} - E_1\right)^{-1}$ and provide its estimate in Proposition \ref{prop:inv-Dk}. 


We start by defining the inverse operator b considering the following system with $\bm{u} = (u^A,u^B)^T$ and $\bm{f} = (f^A,f^B)^T$
\begin{equation}\label{eqn:inv-f-def}
    \left(\mathcal{D} - E_1\right) \bm{u} = \bm{f}, \qquad \bm u \perp \bm{\Psi_0}.
\end{equation}
If the system \eqref{eqn:inv-f-def} has a unique solution $\bm u$, then we define $\left(\mathcal{D} - E_1\right)^{-1} \bm{f}:= \bm{u}$. 

We also provide the solvability condition on $\bm f$. Given that the displacement $d(X_1)$ has a bounded gradient, using Lemma \ref{lemma:1d}, we know that each eigenvalue $E_1$ of $\mathcal{D}$ is simple with an associated eigenfunction $\bm{\Psi}_0$. Therefore, the null space of $\mathcal{D} - E_1$ is spanned by $\bm{\Psi_0}$ and the Fredholm solvability condition requires $\langle \bm{\Psi_0}, \bm{f}\rangle_{L^2(\mathbb{R})} = 0$. 


With the inverse operator defined, we present the conditions on the displacement $d(X_1)$ such that $\mathcal{D} - E_1$ admits a bounded inverse on the orthogonal complement of $\bm{\Psi_0}$ in the following Proposition:
\begin{proposition}\label{prop:inv-Dk}
    Consider a displacement $d(X_1) \in C^1(\mathbb{R})$ that satisfies the asymptotic condition \eqref{eqn:d-condition} and $k_\parallel$ satisfies \eqref{eqn:k-para-cond}. For any $\bm{f} \in L^2(\mathbb{R})$ such that $\langle \bm{\Psi_0}, \bm{f}\rangle_{L^2(\mathbb{R})} = 0$, there is a unique solution $\bm{u} \perp \bm{\Psi_0}$ of \eqref{eqn:inv-f-def}. Moreover, the solution $\bm u$ satisfies the following estimate
    \begin{equation}\label{eqn:est-inv-Dk-u}
        \|\bm{u}\|_{H^{1}(\mathbb{R})} \lesssim \|\bm{f}\|_{L^2(\mathbb{R})}.
    \end{equation}
\end{proposition}

Notice that Proposition \ref{prop:inv-Dk} does not require the integrability condition \eqref{eqn:main-d-condition-2}. In fact, the existence of a bounded solution to \eqref{eqn:inv-f-def} follows by applying classical results in the exponential dichotomy theory in \cite{palmer1984exponential}, which only requires the asymptotic condition \eqref{eqn:d-condition} and \eqref{eqn:k-para-cond} to make $\kappa(X_1;k_\parallel)$ approach limits at $\pm \infty$ with opposite signs. A self-contained review, including the definition of an exponential dichotomy for a given ODE system, is provided in Appendix \ref{app:proof-inv-dk}. Here we state only the particular results from the exponential dichotomy theory that will be used in the proof of Proposition \ref{prop:inv-Dk}.

\begin{lemma}\label{lemma:main-exp-dich-bdd}
    For a bounded and continuous matrix $A(X_1) \in \mathbb{R}^{2\times 2}$ on $\mathbb{R}$, the linear first-order ODE system
    \begin{equation}\label{eqn:main-exp-dich-0}
        \partial_{X_1} \bm{u} = A(X_1) {\bm{u}}
    \end{equation}
    admits \textit{exponential dichotomy} on $(-\infty,0]$ and $[0,\infty)$, provided that $A(X_1)$ satisfies $A(X_1) \rightarrow A_\pm$ as $X_1 \rightarrow \pm \infty$ and $\det A_\pm < 0$. 
\end{lemma}

\begin{lemma}\label{lemma:palmer-lemma42}
    For a bounded and continuous matrix $A(X_1)$ on $\mathbb{R}$, if the system \eqref{eqn:main-exp-dich-0} has exponential dichotomy on both half lines, then the system 
    \begin{equation}\label{eqn:main-exp-dich-1}
        \partial_{X_1} \bm{u} = A(X_1) {\bm{u}} + \bm{g}(X_1)
    \end{equation}
    has a bounded solution if and only if 
    \begin{equation}
        \int_{\mathbb{R}} \bm{\psi}^*(s) \bm{g}(s) ds = 0,\label{eqn:palmer-fredhom-alt}
    \end{equation}
    for all bounded solutions $\bm{\psi}(X_1)$ of the adjoint system
    \begin{equation}
        \partial_{X_1}\bm{v} = -A^*(X_1) \bm{v}.\label{eqn:palmer-adj-system}
    \end{equation}
\end{lemma}
Lemma \ref{lemma:main-exp-dich-bdd} follows from Lemma 3.4 of \cite{palmer1984exponential} (a proof is given in Appendix \ref{app:proof-inv-dk}; see Corollary \ref{lemma:app-exp-dich-asy}), and Lemma \ref{lemma:palmer-lemma42} is Lemma 4.2 of \cite{palmer1984exponential} specialized to our setting. With Lemmas \ref{lemma:main-exp-dich-bdd} and \ref{lemma:palmer-lemma42}, we are ready to prove Proposition \ref{prop:inv-Dk}.

\begin{proof}[Proof of Proposition \ref{prop:inv-Dk}]
    We prove the existence and uniqueness of a bounded solution to \eqref{eqn:inv-f-def}. The uniqueness of a bounded solution to \eqref{eqn:inv-f-def} comes directly from the fact that $E_1$ is a simple eigenvalue of $\mathcal{D}$ and the compatibility condition $\langle \bm{\Psi_0}, \bm{f}\rangle_{L^2(\mathbb{R})} = 0$. The key step is to prove the existence of a bounded solution via exponential dichotomy theory, whereas the estimate \eqref{eqn:est-inv-Dk-u} requires an explicit representation of the bounded solution together with an estimate on its norm and is therefore deferred to Appendix \ref{app:proof-inv-dk} (see Lemma \ref{lemma:exp-dich-bd}).

    To show the existence of a bounded solution, we write the system \eqref{eqn:inv-f-def} as a two-dimensional linear ODE system, where the exponential dichotomy theory applies. We observe that the complex Dirac operator $\mathcal{D}$ in \eqref{eqn:dirac-1d} is, in fact, unitarily equivalent to a real Dirac operator $\mathcal{D}^r$, i.e.
    \begin{equation}\label{eqn:real-Dk-relation}
        \frac{3}{2} V \mathcal{D}^r V^* = \mathcal{D}, \qquad \text{where} \quad \mathcal{D}^r := \begin{pmatrix}
            0 & -\partial_{X_1} - \kappa(X_1)\\
            \partial_{X_1} - \kappa(X_1) &0
        \end{pmatrix}.
    \end{equation}
    where $V = \text{diag}(e^{i\frac{\pi}{4}}, e^{-i\frac{\pi}{4}})$ is a constant unitary matrix. The equivalence of the two systems is a straightforward calculation (see \eqref{eqn:sm-dirac-real} in the Supplementary Material). Therefore, the system \eqref{eqn:inv-f-def} is equivalent to 
    \begin{equation}
        \left(\mathcal{D}^r(k_\parallel) - \frac{2}{3}E_1(k_\parallel)\right) \bm v = \bm g, \qquad \text{where}\quad \bm v = \frac{3}{2} V^* \bm{u} \quad \text{and} \quad \bm g= V^*\bm{f}.\label{eqn:tilde-Dk-sys}
    \end{equation}
    Denote $\bm v = (v^A, v^B)^T$ and $\bm{g} = (g^A, g^B)^T$. Then \eqref{eqn:tilde-Dk-sys} can be written in matrix-vector form as
    \begin{equation}\label{eqn:app-ode}
         \begin{pmatrix}
            - \frac{2}{3} E_1 & -\partial_{X_1} - \kappa(X_1) \\
            \partial_{X_1} - \kappa(X_1) & -\frac{2}{3} E_1
        \end{pmatrix} \begin{pmatrix}
            v^A\\
            v^B
        \end{pmatrix} = \begin{pmatrix}
            g^A\\
            g^B
        \end{pmatrix},
    \end{equation}
    which is equivalent to the following linear ODE system
    \begin{equation}\label{eqn:app-A-def}
        \partial_{X_1} \begin{pmatrix}
            v^A\\
            v^B
        \end{pmatrix} = \begin{pmatrix}
            \kappa(X_1) & \frac{2}{3}E_1\\
            -\frac{2}{3} E_1 & -\kappa(X_1)
        \end{pmatrix} \begin{pmatrix}
            v^A\\
            v^B
        \end{pmatrix} + \begin{pmatrix}
            g^B\\
            -g^A
        \end{pmatrix}.
    \end{equation}
    For simplicity, we denote the matrix in \eqref{eqn:app-A-def} as $A(X_1)$\footnote{Here we slightly abuse the notation $\widetilde{A}(X_1)$, which now denotes a matrix-valued function rather than the effective magnetic potential in \eqref{eqn:eff-mag}.} and the system \eqref{eqn:app-A-def} becomes
    \begin{equation}
        \partial_{X_1} \bm v = A(X_1) \bm v + \bm{g}^\text{swap}, \qquad \text{where} \quad \bm{g}^\text{swap} = (g^B, -g^A)^T.\label{eqn:exp-dich-ode-g}
    \end{equation}

    We now use the exponential dichotomy theory to show the existence of a bounded solution to \eqref{eqn:exp-dich-ode-g}. First, we observe that the system $\partial_{X_1} \bm v = A(X_1) \bm v$ has an exponential dichotomy on $[0, \infty)$ and $(-\infty, 0]$. To see why, using \eqref{eqn:d-condition} and \eqref{eqn:k-para-cond}, the function $\kappa(X_1)$ in \eqref{eqn:dirac-1d} approximate two constants with opposite sign at $\pm \infty$, i.e. we have $\kappa_+ \kappa_- < 0$ with $\kappa_\pm$ defined in \eqref{eqn:kappa-pm}. Consequently, the matrix ${A}(X_1)$ approximate the following two constant matrices at $\pm \infty$
    \begin{equation}\label{eqn:mat-a-pm}
        A_\pm := \lim_{X_1 \rightarrow \pm\infty} A(X_1) = \begin{pmatrix}
            \kappa_\pm & \frac{2}{3} E_1\\
            -\frac{2}{3} E_1 & -\kappa_\pm
        \end{pmatrix}.
    \end{equation}
    Since $E_1$ is an eigenvalue in the spectral gap of $\mathcal{D}(k_\parallel)$, i.e. 
    \begin{align*}
        \frac{2}{3}|E_1| \leq \frac{2}{3}a = \min\left\{|\kappa_+|, |\kappa_-|\right\}.
    \end{align*}
    where the gap $a$ is given in \eqref{eqn:kappa-pm}. Therefore $A_\pm$ both have negative determinant, i.e. 
    \begin{equation}\label{eqn:det-A-pm}
        \det\Big({A}_\pm\Big) < 0, 
    \end{equation}
    Using Lemma \ref{lemma:main-exp-dich-bdd}, we have proved the exponential dichotomy for the ODE system $\partial_{X_1} \bm v = A(X_1) \bm v$.

    Then we apply Lemma \ref{lemma:palmer-lemma42} to show the existence of a bounded solution to \eqref{eqn:exp-dich-ode-g}. It is sufficient to show that the condition \eqref{eqn:palmer-fredhom-alt} is equivalent to the compatibility condition $\langle \bm{\Psi_0}, \bm{f}\rangle_{L^2(\mathbb{R})} = 0$. We observe that the adjoint system $\partial_{X_1} \bm{\widetilde{v}} = -{A}(X_1)^* \bm{\widetilde{v}}$ in matrix-vector form is
    \begin{equation}\label{eqn:ode-adj-v}
        \partial_{X_1} \bm{\widetilde{v}} = \begin{pmatrix}
            -\kappa(X_1) & \frac{2}{3} E_1\\
            -\frac{2}{3} E_1 & \kappa(X_1)
        \end{pmatrix} \bm{\widetilde{v}}.
    \end{equation}
    One can easily check that $\bm{\phi} = (\phi^A, \phi^B)^T$ is a solution to \eqref{eqn:ode-adj-v} if and only if $\bm{\phi}^{\text{swap}}=(\phi^B, -\phi^A)^T$ is a solution to $\partial_{X_1} \bm v = A(X_1) \bm v$. We also observe that $\widetilde{\bm{\Psi}}_0 :=V^* {\bm{\Psi}}_0$ spans $\ker(\mathcal{D}^r - E_1)$, since ${\bm{\Psi}}_0$ spans $\ker(\mathcal D-E_1)$. Consequently, up to scalar multiplication, $\widetilde{\bm{\Psi}}_0$ is the unique solution of $\partial_{X_1} \bm v = A(X_1) \bm v$ and $\widetilde{\bm \Psi}_0^\text{swap} := \left((\widetilde{\bm{\Psi}}_0)^B, -(\widetilde{\bm{\Psi}}_0)^A\right)^T$ is the unique solution to the adjoint system $\partial_{X_1} \bm{\widetilde{v}} = -A(X_1)^* \bm{\widetilde{v}}$. Thus, the compatibility condition in \eqref{eqn:palmer-fredhom-alt} becomes
    \begin{align*}
        & \left\langle \widetilde{\bm \Psi}_0^\text{swap}, \bm g^{\text{swap}}\right\rangle_{L^2(\mathbb{R})} = \left\langle -(\widetilde{\bm{\Psi}}_0)^B, g^B \right\rangle_{L^2(\mathbb{R})}  + \left\langle (\widetilde{\bm{\Psi}}_0)^A, -g^A \right\rangle_{L^2(\mathbb{R})} \\
        =& -\left\langle (V^* {\bm{\Psi}}_0)^B, (V^*\bm{f})^B \right\rangle_{L^2(\mathbb{R})}  - \left\langle (V^* {\bm{\Psi}}_0)^A, (V^*\bm{f})^A \right\rangle_{L^2(\mathbb{R})}  = -\langle \bm{\Psi_0}, \bm{f}\rangle_{L^2(\mathbb{R})} = 0,
    \end{align*}
    which completes the proof of existence of a bounded solution to \eqref{eqn:inv-f-def}.
\end{proof}

\subsection{Step 4: solving $\widetilde{\bm \eta}_\text{near}$ and $\mu(\delta)$ by Lyapunov-Schmidt reduction}\label{subsec:near-sol}

We now substitute $\widetilde{\bm{\eta}}_\text{far}[k; \widetilde{\bm{\eta}}_\text{near}, \mu, \delta]$ into 
the near-momentum equation \eqref{eqn:eta-near} and solve for $\widetilde{\bm{\eta}}_\text{near}$. We first show that $\widetilde{\bm{\eta}}_\text{near}$ corresponds to a solution of a perturbed Dirac system. 

Recall that $\widetilde{\bm{\eta}}_\text{near}$ is supported in the small region $|k|\leq \delta^\tau$. To exploit this localization, we introduce the following rescaling and define
\begin{align}
	  \widehat{\bm{\beta}}_\text{near}(\xi) &= \big(\widehat{\beta}^A_\text{near}(\xi),\widehat{\beta}^B_\text{near}(\xi)\big)^T :=\delta \widetilde{\bm{\eta}}_\text{near}(k), \qquad \text{where} \quad \xi = \frac{k}{\delta} \quad \text{and} \quad |\xi|\leq \delta^{\tau-1},\label{eqn:near-hat}\\
      &= \delta \chi(|\xi|\leq \delta^{\tau - 1}) \widetilde{\bm{\eta}}_\text{near}(k).\nonumber
\end{align}

We provide a few important observations related to the scaling \eqref{eqn:near-hat}. We first observe that the near-momentum components $\bm{\eta}_m^\text{near}$ of the correctors can now be viewed as point evaluations of the function $\bm{\beta}_\text{near}(X_1) = \big({\beta}^A_\text{near}(X_1),{\beta}_\text{near}^B(X_1)\big)^T$ at $X_1 = \frac{\sqrt{3}}{2} \delta m$, where $\bm{\beta}_\text{near}(X_1)$ is the inverse Fourier transform of $\widehat{\bm{\beta}}_\text{near}$; that is,
\begin{align}
\bm{\eta}_m^\text{near} &= \frac{\sqrt{3}}{2} \bm{\beta}_\text{near}\left(\frac{\sqrt{3}}{2} \delta m\right).\label{eqn:eta-idft}
\end{align}
The proof of \eqref{eqn:eta-idft} comes from a direct calculation by using the scaled IDFT in \eqref{eqn:scale-idft}, i.e.
\begin{align*}
\bm{\eta}_m^\text{near} &= \frac{\sqrt{3}}{4\pi} \int_{-\frac{2\pi}{\sqrt{3}}}^{\frac{2\pi}{\sqrt{3}}} \widetilde{\bm{\eta}}_\text{near}(k) e^{ik \frac{\sqrt{3}}{2}m} \: dk = \frac{\sqrt{3}}{4\pi\delta} \int_{-\frac{2\pi}{\sqrt{3}}}^{\frac{2\pi}{\sqrt{3}}} \widehat{\bm{\beta}}_\text{near}\left(\frac{k}{\delta}\right) e^{ik \frac{\sqrt{3}}{2}m} \: dk \\
		&=  \frac{\sqrt{3}}{4\pi} \int_\mathbb{R} \widehat{\bm{\beta}}_\text{near}\left(\xi\right) e^{i \xi \frac{\sqrt{3}}{2}m} \: d\xi = \frac{\sqrt{3}}{2} \bm{\beta}_\text{near}\left(\frac{\sqrt{3}}{2} \delta m\right),
\end{align*}
where the second equality holds due to the scaling \eqref{eqn:near-hat} and the third equality holds since $\widehat{\bm{\beta}}_\text{near}(\xi)$ is band limited with $\widehat{\bm{\beta}}_\text{near}(\xi) = \widehat{\bm{\beta}}_\text{near}(\xi) \chi(|\xi|\leq \delta^{\tau-1})$. We also observe that the $l^2$-norm of $\bm{\eta}_m^\text{near}$ and the $L^2$-norm of $\widehat{\bm{\beta}}_\text{near}$ are related by
\begin{align}
    \norm{\bm{\eta}_m^\text{near}}_{l^2(\mathbb{Z};\mathbb{C}^2)} \lesssim \delta^{-\frac{1}{2}} \norm{\widehat{\bm{\beta}}_\text{near}}_{L^2_\xi(\mathbb{R})},\label{eqn:eta-idft-bd}
\end{align}
and the $L^2$-norm of $\widetilde{\bm{\eta}}_\text{near}$ and the $L^2$-norm of $\widehat{\bm{\eta}}_\text{near}$ are related by
\begin{align}
     \norm{\widehat{\bm{\beta}}_\text{near}}_{L^2_\xi(\mathbb{R})} = \delta^{\frac{1}{2}}\norm{\widetilde{\bm{\eta}}_\text{near}}_{\LtwoBrill}.\label{eqn:eta-idft-bd-2}
\end{align}
The bound \eqref{eqn:eta-idft-bd} comes from \eqref{eqn:eta-idft} and \eqref{eqn:l2-norm}, and the bound \eqref{eqn:eta-idft-bd-2} comes from a direct calculation.

Using \eqref{eqn:eta-idft}, we notice that the near-momentum component $\delta \bm{\eta}_m^\text{near}$ plays a role analogous to the next-order term $\delta \Psi_1$ in \eqref{eqn:ansatz-1d}. More precisely, we have
\begin{align*}
    \delta \bm{\eta}_m^\text{near} = \frac{\sqrt{3}}{2} \bm{\beta}_\text{near}\left(\frac{\sqrt{3}}{2} \delta m\right) \approx \delta \bm \Psi_1\left(\frac{\sqrt{3}}{2}\delta m\right).
\end{align*}
Therefore, we expect that the equation for $\bm{\beta}_\text{near}$ is a perturbed equation of \eqref{eqn:2nd-expansion} for $\bm \Psi_1$. In fact, by a detailed calculation of substituting \eqref{eqn:near-hat} and \eqref{eqn:eta-idft} into \eqref{eqn:eta-near}, we obtain the equation for $\widehat{\bm{\eta}}_\text{near}$ (the detailed derivation is presented in Appendix \ref{app:dev-near})
\begin{equation}\label{eqn:beta-system}
	\left[\widehat{\mathcal{D}^\delta} - E_1 + \widehat{\mathcal{L}^\delta}(\mu) - \delta \mu\right] \widehat{\bm{\beta}}_\text{near}(\xi) = \mu \widehat{\bm{\mathcal{M}}}(\xi;\delta) + \widehat{\bm{\mathcal{N}}}(\xi;\mu,\delta),
\end{equation}
where the operator $\widehat{\mathcal{D}^\delta}$ is a band-limited version of $\mathcal{D}$ in \eqref{eqn:dirac-1d} in Fourier representation, i.e.
\begin{equation}
    \widehat{\mathcal{D}^\delta} = \chi(|\xi|\leq \delta^{\tau-1}) \widehat{\mathcal{D}}, \qquad \text{where} \quad \widehat{\mathcal{D}} \widehat{\bm{\beta}}_\text{near}(\xi)= \frac{3}{2} \Big[\xi \sigma_1 \widehat{\bm{\beta}}_\text{near}(\xi)+ \sigma_2 \widehat{\kappa \bm{\beta}_\text{near}}(\xi)\Big].\label{eqn:dirac-1d-fourier}
\end{equation}
The operator $\widehat{\mathcal{L}^\delta}(\mu)$ in \eqref{eqn:beta-system} captures the far-momentum contributions in terms of $\bm{\widehat{\beta}}_\text{near}$, and $\widehat{\bm{\mathcal{M}}}(\xi;\delta)$ and $\widehat{\bm{\mathcal{N}}}(\xi;\mu, \delta)$ represent the source terms related to $\bm \Psi_0$ (see their detailed expressions in \eqref{eqn:Ldelta} and \eqref{eqn:MN-12}). Furthermore, the linear operator $\widehat{\mathcal{D}^\delta} - E_1 + \widehat{\mathcal{L}^\delta}(\mu) - \delta \mu$ in \eqref{eqn:beta-system} is a small perturbation of the Dirac operator $\widehat{\mathcal{D}}$, which is formalized through the following bounds:
\begin{lemma}\label{lemma:eta-near-bds}
    Fix $M > 0$, $\delta \in (0,1)$ and $\tau \in (0,1)$. For $|\mu|< M$, the operators in \eqref{eqn:beta-system} satisfy:
    \begin{align}
        \left\|\widehat{\mathcal{D}^\delta} - \widehat{\mathcal{D}}\right\|_{L^{2,1}_\xi(\mathbb{R}) \rightarrow L^2(\mathbb{R})} &\lesssim \delta^{1-\tau}, \qquad \left\|\widehat{\mathcal{L}^\delta}(\mu)\right\|_{L^{2,1}_\xi(\mathbb{R}) \rightarrow L^2(\mathbb{R})} \lesssim \delta^{\tau} + \delta^{1-\tau},\label{eqn:near-op-bd}
    \end{align}
    where $L^{2,1}(\mathbb{R})$ is the equivalent space of $H^1(\mathbb{R})$ defined in the Fourier space (see \eqref{eqn:L21-space} for its definition). The operator $\widehat{\mathcal{L}^\delta}(\mu)$ is Lipschitz in $\mu$ with
    \begin{align}
        \norm{\widehat{\mathcal{L}^\delta}(\mu_1) - \widehat{\mathcal{L}^\delta}(\mu_2)}_{L^{2,1}_\xi(\mathbb{R}) \rightarrow L^2(\mathbb{R})} \lesssim \delta^{1-\tau} \: |\mu_1 - \mu_2|.\label{eqn:lip-L-delta}
    \end{align}
    Moreover, the source terms in \eqref{eqn:beta-system} satisfy $\left\| \widehat{\bm{\mathcal{M}}}(\xi;\delta) \right\|_{L^2_\xi(\mathbb{R})} \lesssim 1$ and $\left\| \widehat{\bm{\mathcal{N}}}(\xi;\mu,\delta) \right\|_{L^2_\xi(\mathbb{R})} \lesssim 1$. More specifically, as $\delta \rightarrow 0$, we have the following limits
    \begin{subequations}\label{eqn:MN-bds-main}
        \begin{align}
        & \lim_{\delta \rightarrow 0} \left\langle \bm{\widehat{\Psi_0}}(\xi), \widehat{\bm{\mathcal{M}}}(\xi; \delta) \right\rangle_{L^2_\xi(\mathbb{R})} = 1, \label{eqn:M-limit}\\
        &\lim_{\delta \rightarrow 0} \left\langle \bm{\widehat{\Psi_0}}(\xi), \widehat{\bm{\mathcal{N}}}(\xi;\mu,\delta) \right\rangle_{L^2_\xi(\mathbb{R})} = -E_2. \label{eqn:N-limit}
    \end{align}
    \end{subequations}
\end{lemma}
The proof of Lemma \ref{lemma:eta-near-bds} is presented in Appendix \ref{app:some-bds} (see the paragraph after \eqref{eqn:app-map-E}). 

Using the bounds in \eqref{eqn:near-op-bd}, we know that the linear operator in \eqref{eqn:beta-system} is a small perturbation of the Dirac operator $\widehat{\mathcal{D}} - E_1$ in \eqref{eqn:dirac-1d-fourier}. Since $\widehat{\mathcal{D}}-E_1$ has a one-dimensional kernel spanned by $\bm{\widehat{\Psi}_0}$, the linear part of \eqref{eqn:beta-system} is not invertible, and thus we cannot solve \eqref{eqn:beta-system} for $\widehat{\bm{\beta}}_\text{near}$ by a direct perturbative method. To solve such a perturbed system with a non-trivial kernel space, we apply \textit{Lyapunov-Schmidt reduction} by decomposing $\widehat{\bm{\beta}}_\text{near}$ into its component along $\bm{\Psi}_0$ and its orthogonal complement.

\paragraph{Lyapunov-Schmidt reduction} We decompose $\widehat{\bm{\beta}}_\text{near}$ via the projection operators $\widehat{P_\parallel}$ and $\widehat{P_\perp}$: for any $\bm{\widehat{f}} \in L^2(\mathbb{R})$, we define
\begin{equation}\label{eqn:proj-ops}
	\widehat{P_\parallel} \widehat{\bm{f}} := \langle \bm{\Psi_0}, \bm{f}\rangle_{L^2(\mathbb{R})} \widehat{\bm{\Psi_0}}, \qquad \widehat{P_\perp} \widehat{\bm{f}} := (I - \widehat{P_\parallel}) \widehat{\bm{f}},
\end{equation}
where $\widehat{P_\parallel}$ denotes the projection onto $\widehat{\bm{\Psi}_0}$ and $\widehat{P_\perp}$ denotes the projection onto its orthogonal complement. These projections satisfy
\begin{equation}\label{eqn:projection}
    \widehat{P_\parallel} \Big(\widehat{\mathcal{D}} - E_1\Big) \widehat{\bm{f}} = 0, \qquad \widehat{P_\perp} \Big(\widehat{\mathcal{D}} - E_1\Big) \widehat{\bm{f}} = \Big(\widehat{\mathcal{D}} - E_1\Big) \widehat{\bm{f}}.
\end{equation}
Therefore, applying $\widehat{P_\parallel}$ and $\widehat{P_\perp}$ to \eqref{eqn:beta-system} yields
\begin{subequations}\label{eqn:beta}
\begin{align}
	&\widehat{P_\parallel} \Big[\left(\widehat{\mathcal{D}^\delta} - \widehat{\mathcal{D}} + \widehat{\mathcal{L}^\delta}(\mu) - \delta \mu\right) \widehat{\bm{\beta}}_\text{near}(\xi) - \mu \widehat{\bm{\mathcal{M}}}(\xi;\delta) - \widehat{\bm{\mathcal{N}}}(\xi;\mu,\delta)\Big] = 0, \label{eqn:beta-parallel}\\
	& \Big(\widehat{\mathcal{D} } - E_1\Big) \widehat{\bm{\beta}}_\text{near}(\xi) + \widehat{P_\perp} \Big[\left(\widehat{\mathcal{D} ^\delta} - \widehat{\mathcal{D} }\right)+ \widehat{\mathcal{L}^\delta}(\mu) - \delta \mu \Big] \widehat{\bm{\beta}}_\text{near}(\xi) = \widehat{P_\perp} \Big[\mu \bm{\widehat{\mathcal{M}}}(\xi; \delta) + \widehat{\bm{\mathcal{N}}}(\xi;\mu,\delta) \Big]. \label{eqn:beta-perp}
\end{align}
\end{subequations}
We shall first solve \eqref{eqn:beta-perp} for $\widehat{\bm{\beta}}_\text{near}$, and then substitute the solution into \eqref{eqn:beta-parallel} to determine $\mu(\delta)$.

\paragraph{Solving for $\widehat{\bm{\beta}}_\text{near}$ from \eqref{eqn:beta-perp}} We notice that $\widehat{\mathcal{D} } - E_1$ is now invertible after the $\widehat{P_\perp}$-projection. Recall from Proposition \ref{prop:inv-Dk} that, for any $\widehat{\bm{f}} \in L^2(\mathbb{R})$, since $\widehat{P_\perp} \widehat{\bm{f}} \perp \widehat{\bm{\Psi_0}}$ automatically satisfies the Fredholm condition, the operator $\left(\widehat{\mathcal{D} } - E_1\right)^{-1} \widehat{P_\perp}$ is well-defined and the estimate \eqref{eqn:est-inv-Dk-u} becomes
\begin{equation}\label{eqn:Dk-inv-bd-1}
    \left\|\left(\widehat{\mathcal{D} } - E_1\right)^{-1} \widehat{P_\perp} \widehat{\bm{f}} \right\|_{L^{2,1}(\mathbb{R})} \lesssim \left\|\bm{\widehat{f}} \right\|_{L^2(\mathbb{R})}.
\end{equation}
Equivalently, the operator $\left(\widehat{\mathcal{D} } - E_1\right)^{-1} \widehat{P_\perp}$ satisfies
\begin{equation}\label{eqn:Dk-inv-bd}
    \norm{\left(\widehat{\mathcal{D} } - E_1\right)^{-1} \widehat{P_\perp}}_{L^2(\mathbb{R}) \rightarrow L^{2,1}(\mathbb{R})} \lesssim 1.
\end{equation}Applying $\left(\widehat{\mathcal{D} } - E_1\right)^{-1} \widehat{P_\perp}$ to both sides of \eqref{eqn:beta-perp} yields
\begin{align}
    \Big[I + \widehat{\mathcal{C}^\delta}(\mu)\Big] \widehat{\bm{\beta}}_\text{near}(\xi) &= \Big(\widehat{\mathcal{D} } - E_1\Big)^{-1} \widehat{P_\perp} \Big[\mu \bm{\widehat{\mathcal{M}}}(\xi; \delta) + \widehat{\bm{\mathcal{N}}}(\xi;\mu,\delta) \Big], \label{eqn:perp-inv}\\
    \text{where}\qquad\widehat{\mathcal{C}}^\delta(\mu) &:= \Big(\widehat{\mathcal{D} } - E_1\Big) ^{-1} \widehat{P_\perp} \Big[(\widehat{\mathcal{D} ^\delta} - \widehat{\mathcal{D} )} + \widehat{\mathcal{L}}^\delta(\mu) - \delta \mu\Big].\label{eqn:op-C-delta}
\end{align}
Using the bounds in \eqref{eqn:near-op-bd} and \eqref{eqn:Dk-inv-bd}, we have 
\begin{align}
    \left\|\widehat{\mathcal{C}}^\delta(\mu)\right\|_{L^{2,1}(\mathbb{R})\rightarrow L^{2,1}(\mathbb{R})} \lesssim \delta^\tau + \delta^{1-\tau}. \label{eqn:bd-C-delta}
\end{align}
Therefore, for sufficiently small $\delta$, the linear operator $I + \widehat{\mathcal{C}^\delta}(\mu)$ is invertible in $L^{2,1}(\mathbb{R})$, and $\widehat{\bm{\beta}}_\text{near}$ can be solved in $L^{2,1}(\mathbb{R})$ from \eqref{eqn:perp-inv}. We summarize these arguments in the following proposition:
\begin{proposition}\label{prop:beta}
	For any $M > 0$, there exists $\delta_0 > 0$ and a unique mapping $(\mu, \delta) \in \{|\mu| < M\} \times (0,\delta_0) \longmapsto \widehat{\bm{\beta}}_\text{near}(\xi; \mu, \delta) \in L^{2,1}(\mathbb{R})$, which is \textbf{Lipschitz} in $\mu$ and solves \eqref{eqn:beta-perp}. Moreover, we have a bound for this mapping, i.e.
	\begin{equation}\label{eqn:eta-near-scaled-bd}
		\norm{\widehat{\bm{\beta}}_\text{near}(\xi; \mu, \delta)}_{L^{2,1}(\mathbb{R})} \lesssim 1.
	\end{equation}
\end{proposition}
We have outlined the main proof strategy, and a unique solution $\widehat{\bm{\beta}}_\text{near}(\xi; \mu, \delta)$ exists due to the invertibility of $I + \widehat{\mathcal{C}^\delta}(\mu)$ (alternatively, by a similar contraction mapping argument used in Proposition \ref{prop:far-energy}). The 
bound \eqref{eqn:eta-near-scaled-bd} follows from \eqref{eqn:MN-bds-main}, and the Lipschitz dependence on $\mu$ is a direct consequence of the Lipschitz continuity of $\widehat{\mathcal{L}^\delta}(\mu)$. Thus, we complete the proof of Proposition \ref{prop:beta}.

\paragraph{Solving $\mu$ from \eqref{eqn:beta-parallel}} We have now solved \eqref{eqn:beta-perp} and obtained $\widehat{\bm{\beta}}_\text{near}(\xi; \mu, \delta)$ as a function of $\mu, \delta$. It remains to substitute $\widehat{\bm{\beta}}_\text{near}(\xi; \mu, \delta)$ into \eqref{eqn:beta-parallel} and solve for $\mu$ as a function of $\delta$.

By substituting $\widehat{\bm{\beta}}_\text{near}(\xi; \mu, \delta)$ into \eqref{eqn:beta-parallel}, we obtain $\mathcal{J}_+[\mu,\delta] = 0$, where $\mathcal{J}_+[\mu,\delta]$ is defined as
\begin{align*}
    &\mathcal{J}_+[\mu,\delta] := \mu \left\langle \bm{\widehat{\Psi_0}}(\xi), \widehat{\bm{\mathcal{M}}}(\xi, \delta) \right\rangle_{L^2_\xi(\mathbb{R})} +  \left\langle \bm{\widehat{\Psi_0}}(\xi),  \bm{\widehat{\mathcal{N}}}(\xi,\mu, \delta) \right\rangle_{L^2_\xi(\mathbb{R})} \\
		&- \left\langle \bm{\widehat{\Psi_0}}(\xi), \Big(\widehat{\mathcal{D} ^\delta} - \widehat{\mathcal{D} }\Big) \widehat{\bm{\beta}}_\text{near}(\xi; \mu, \delta) \right\rangle_{L^2_\xi(\mathbb{R})} - \left\langle \bm{\widehat{\Psi_0}}(\xi), \widehat{\mathcal{L}^\delta}(\mu) \widehat{\bm{\beta}}_\text{near}(\xi; \mu, \delta) \right\rangle_{L^2_\xi(\mathbb{R})} \nonumber\\
		&+\delta \left\langle \bm{\widehat{\Psi_0}}(\xi), \widehat{\bm{\beta}}_\text{near}(\xi; \mu, \delta)  \right\rangle_{L^2_\xi(\mathbb{R})}. 
\end{align*}
We notice that $\mathcal{J}_+[\mu,\delta]$ is a well-defined function in the range $|\mu| < M$ and $0<\delta < \delta_0$, and is also Lipschitz in $\mu$ since all terms related to $\mu$ are Lipschitz in $\mu$. Due to the Lipschitz continuity in $\mu$, the function $\mathcal{J}_+[\mu,\delta]$ can be continuously extended from $(0,\delta_0)$ to $[0,\delta_0)$. We state this argument as:
\begin{lemma} \label{lemma:J-cont-ext}
	For any $\delta_0 > 0$ and $M>0$,  we define $\mathcal{J}[\mu,\delta]$ on $\{|\mu|<M\}\times[0,\delta_0)$ by
	\begin{equation}
		\mathcal{J}[\mu, \delta] := \begin{cases}
			\mathcal{J}_+[\mu, \delta], & 0 < \delta < \delta_0,\\
			\mu - E_2, & \delta = 0,
		\end{cases}
	\end{equation}
    where $E_2$ is given by \eqref{eqn:E2}. Then $\mathcal{J}[\mu,\delta]$ is well-defined and continuous on its domain.
\end{lemma}
\begin{proof}
    It is equivalent to show $\lim_{\delta \rightarrow 0^+} \mathcal{J}_+[\mu, \delta] = \mu - E_2$ converges uniformly for $|\mu| < C_\mu$, which comes directly from \eqref{eqn:near-op-bd} and \eqref{eqn:MN-bds-main}.
\end{proof}

We now solve \eqref{eqn:beta-parallel} for $\mu(\delta)$ by a continuity argument. The result is stated as follows:
\begin{proposition}\label{prop:solve-Jplus}
	For any $M > |E_2|$, there exists $\delta_0 > 0$ and a function $\delta \mapsto \mu(\delta)$, defined for all $0\leq \delta < \delta_0$ such that (1) $|\mu(\delta)| < M$ for all $0\leq \delta < \delta_0$; (2) $\lim_{\delta \rightarrow 0} \mu(\delta) = E_2$; and (3) $\mathcal{J}[\mu(\delta), \delta] = 0$ for all  $0\leq \delta < \delta_0$.
\end{proposition}

The proof is the same as that of Proposition 6.17 in \cite{fefferman2017topologically} and is omitted here. 
    

\paragraph{Completion of the Proof of Theorem \ref{thm:main}} We now assemble the preceding estimates and conclude the proof of Theorem \ref{thm:main} with the following steps:
\begin{enumerate}[(i)]
    \item We first observe that it suffices to prove the existence of $M > 0$ and $\delta_0 > 0$ such that (1) the correctors $\delta^{\frac{3}{2}} \bm{\eta}_m = \delta^{\frac{3}{2}} \bm{\eta}_m^\text{near} + \delta^{\frac{3}{2}} \bm{\eta}_m^\text{far}$ in \eqref{eqn:ansatz-1d} is bounded by $\delta$ when $\delta \in (0,\delta_0)$; and (2) $|\mu(\delta)| \leq M$ and $\lim_{\delta\rightarrow 0} \mu(\delta) = E_2$.
    
    \item We start with the near-momentum part. Using Propositions \ref{prop:beta} and \ref{prop:solve-Jplus}, we obtain that for any $M > |E_2|$, there exists $\delta_0 > 0$ such that for any $\delta \in (0,\delta_0)$, we have a pair $\widehat{\bm{\beta}}_\text{near}(\xi)$ and $\mu(\delta)$ as a solution to \eqref{eqn:beta} with a bound $\norm{\widehat{\bm{\beta}}_\text{near}(\xi; \mu, \delta)}_{L^{2,1}(\mathbb{R})} \lesssim 1$ (see \eqref{eqn:eta-near-scaled-bd}) and a limit $\lim_{\delta \rightarrow 0} \mu(\delta) = E_2$. Using $\norm{\widehat{\bm{\beta}}_\text{near}(\xi; \mu, \delta)}_{L^{2,1}(\mathbb{R})} \lesssim 1$ and the scaling between $\widehat{\bm{\beta}}_\text{near}$ and $\bm{\eta}_m^\text{near}$ in \eqref{eqn:eta-idft-bd}, we obtain 
    \begin{align*}
        \norm{\bm{\eta}_m^\text{near}}_{l^2(\mathbb{Z};\mathbb{C}^2)} \lesssim \delta^{-\frac{1}{2}} \qquad \Rightarrow \qquad \norm{\delta^{\frac{3}{2}}\bm{\eta}_m^\text{near}}_{l^2(\mathbb{Z};\mathbb{C}^2)} \lesssim \delta;
    \end{align*}

    \item Lastly, for the far-momentum part, we use the relation between the near- and far-momentum parts in \eqref{eqn:far-bound}: the far-momentum part $\delta \bm{\eta}_m^\text{far}$ is bounded by
    \begin{align*}
        \norm{\delta^{\frac{3}{2}} \bm{\eta}_m^\text{far}}_{l^2(\mathbb{Z};\mathbb{C}^2)} \lesssim \delta^{\frac{3}{2}} \Big(\delta^{1-\tau} \norm{\bm{\eta}_m^\text{near}}_{l^2(\mathbb{Z};\mathbb{C}^2)} + \delta^{\frac{1}{2}-\tau}\Big) \lesssim \delta^{2 - \tau} \lesssim \delta.
    \end{align*}
    Thus, we complete our proof of Theorem \ref{thm:main}.
    
\end{enumerate}


\section{Summary and conclusion}\label{sec:conclusion}

\edit{For tight-binding models of slowly strained honeycomb lattices, we derive a continuum magnetic Dirac Hamiltonian via a formal discrete multiscale expansion; see Section \ref{sec:derivation}). For unidirectional deformations that preserve periodicity in one direction, we reduce the eigenvalue problem to a one-dimensional $l^2_{k_\parallel}$-eigenvalue problem (see \eqref{eqn:lk-evp}) and prove that for bounded deformation gradients that eigenstates of an effective Dirac operator give rise to $l^2_{k_\parallel}$-eigenstates of $H^\delta$; see Theorem \ref{thm:main}. Numerical simulations for quadratic deformations and their linear regularizations in both armchair and zigzag orientations support corroborate our analysis; see Section \ref{sec:num-comparison}.}

\appendix
\section*{Appendix}
\addcontentsline{toc}{section}{Appendix}

\renewcommand{\thesubsection}{\Alph{subsection}}

\renewcommand{\theequation}{\thesubsection.\arabic{equation}}
\renewcommand{\theHequation}{app.\thesubsection.\arabic{equation}}

\renewcommand{\thetheorem}{\thesubsection.\arabic{theorem}}
\renewcommand{\theHtheorem}{app.\thesubsection.\arabic{theorem}}

\makeatletter
\@addtoreset{theorem}{subsection} 
\makeatother

\setcounter{equation}{0}
\subsection{Quadratic deformations with AC and ZZ orientations}\label{sec:ac-zz-edge}

In this section, we consider the spectral properties of $\mathcal{D}_{\bm A}$ in \eqref{eqn:dirac-mag} for two special cases: quadratic deformations with armchair (AC) orientation and quadratic deformations with zigzag (ZZ) orientation. We show that (1) for the quadratic deformation with AC orientation, the spectrum of $\mathcal{D}_{\bm A}$ is purely discrete with each eigenvalue of infinite multiplicity; (2) for the quadratic deformation with ZZ orientation, $\mathcal{D}_{\bm A}$ has purely continuous spectrum, and $\sigma(\mathcal{D}_{\bm A}) = \mathbb{R}$.


\paragraph{AC and ZZ orientation} Notice that the discrete honeycomb lattice is anisotropic, so the spectrum of $\mathcal{D}_{\bm A}$ in \eqref{eqn:dirac-mag} depends on the orientation in which the quadratic deformation is applied -- namely, the AC or ZZ orientation. Accordingly, we consider two quadratic deformations: $\bm u^{\mathrm{AC}}=(0,X_1^2)^T$, associated with the AC orientation, and $\bm u^{\mathrm{ZZ}}=(X_2^2,0)^T$, associated with the ZZ orientation.

Let us first explain why these two deformations correspond to quadratic deformations along AC and ZZ orientations. We start by introducing the AC and ZZ edges in the honeycomb structure: the armchair edge corresponds to cutting the honeycomb along a direction where the boundary alternates like the shape of an armchair (the left and right boundaries of Figure \ref{fig:honeycomb-edges} are AC edges), while the zigzag edge follows a direction where hexagons align in a straight row of nodes, giving a zigzag pattern (the top and bottom boundaries of Figure \ref{fig:honeycomb-edges} are ZZ edges). 

These two orientations are related by a $90^\circ$ rotation: in Figure \ref{fig:honeycomb-edges}, the AC edge is oriented vertically and the ZZ edge horizontally, so rotating Figure \ref{fig:honeycomb-edges} by $90^\circ$ interchanges their roles. Guided by this geometric relation, we take the quadratic deformation along the AC orientation to be $\bm u^{\mathrm{AC}}=(0,X_1^2)^T$ and, correspondingly, the quadratic deformation along the ZZ orientation to be $\bm u^{\mathrm{ZZ}}=(X_2^2,0)^T$. The resulting deformed configurations are shown in Figures \ref{fig:ac-edge-def} and \ref{fig:zz-edge-def}.


\begin{figure}[!htb]
	\centering
	\subfloat[]{
		\includegraphics[height=1.3in]{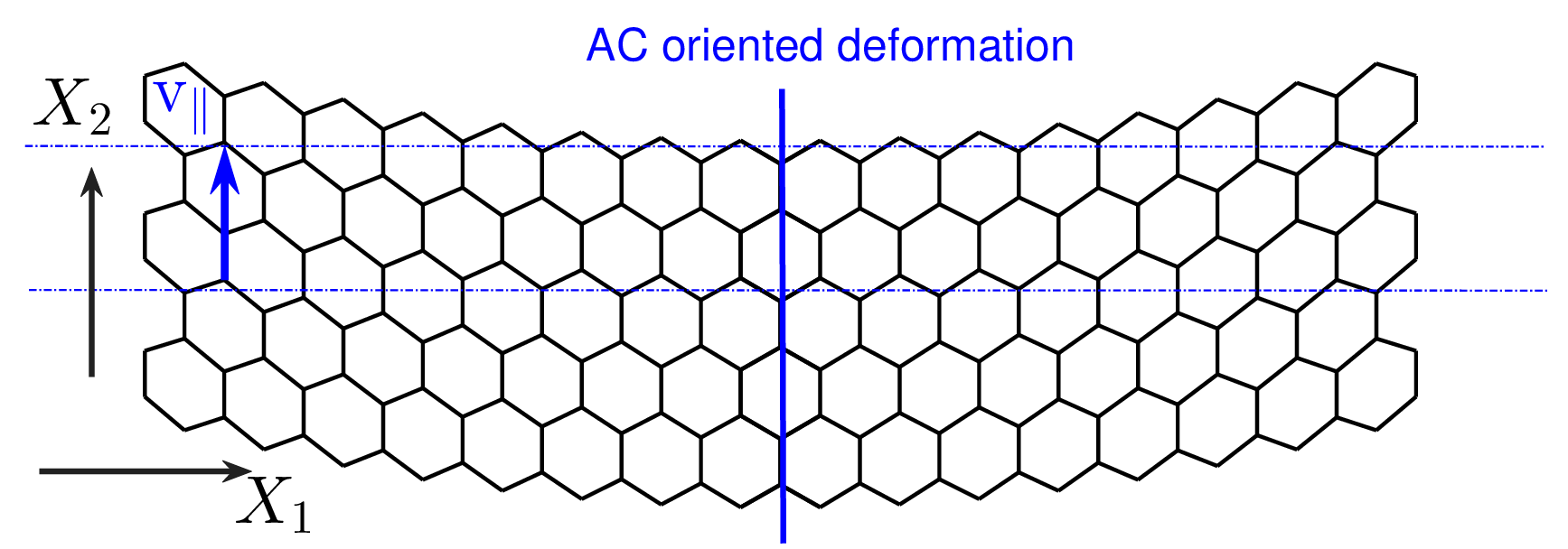}\label{fig:ac-edge-def}
	}\hfil
     \subfloat[]{
		\includegraphics[height=1.3in]{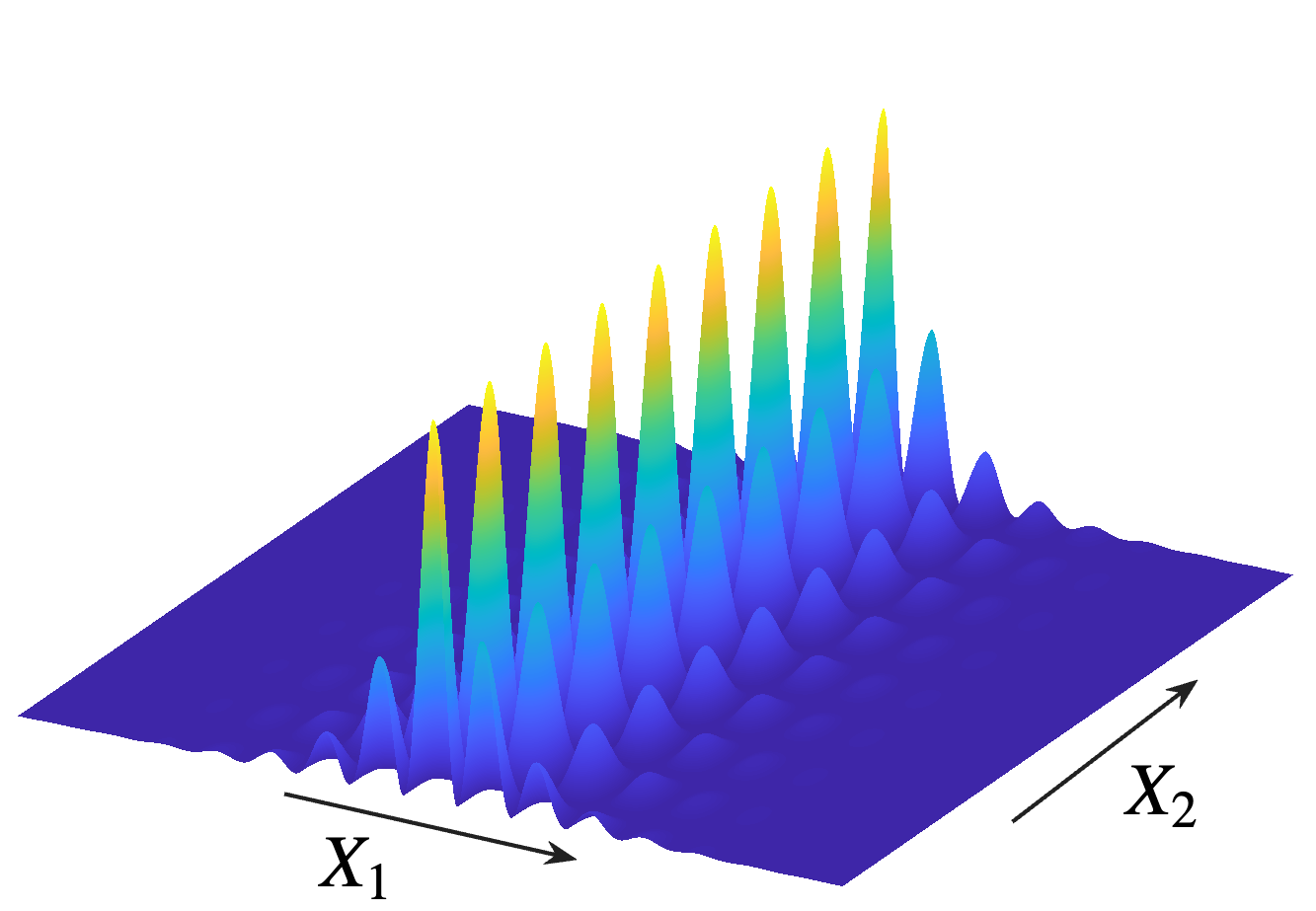}\label{fig:eigmode}
	}\\
	\subfloat[]{
		\includegraphics[height=1.28in]{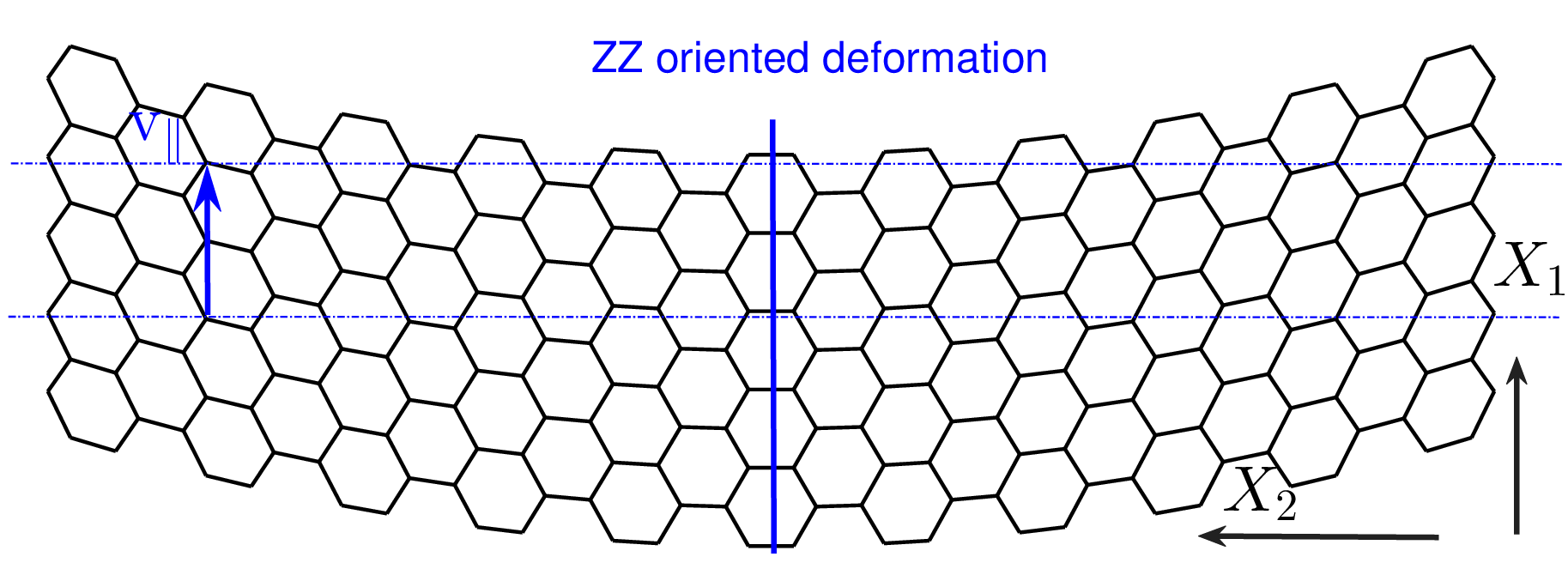}\label{fig:zz-edge-def}
	}\hfil
    \subfloat[]{
		\includegraphics[height=1.22in]{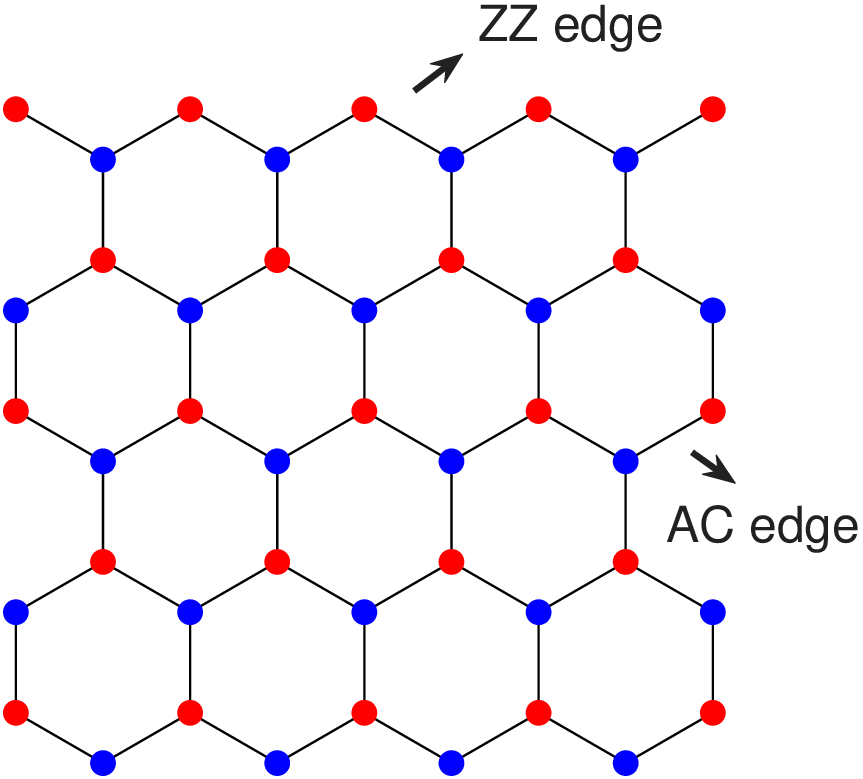}\label{fig:honeycomb-edges}
	}
	\caption{Unidirectional deformations with AC and ZZ orientations: (a) the deformed honeycomb by $\bm{u}_{\text{AC}}$ with $X_1$ horizontal and $X_2$ vertical; (b) eigenmode of the form $e^{i k_\parallel X_2} \Psi(X_1;k_\parallel)$, where $\Psi(X_1;k_\parallel)$ decays in the $X_1$ direction; (c) the deformed honeycomb by $\bm{u}_{\text{ZZ}}$ with $X_1$ vertical and $X_2$ horizontal; (d) the undeformed honeycomb lattice with marked AC and ZZ edges.}
\end{figure}

\paragraph{\edit{Deformations with AC orientation and constant perpendicular effective magnetic field}} For the quadratic deformation $\bm{u}^\text{AC}$ along the AC edge in Figure \ref{fig:ac-edge-def}, the effective magnetic potential in \eqref{eqn:eff-mag}, denoted by $\bm{A}^\text{AC}_\text{eff} = (A_1^\text{AC}, A_2^\text{AC})$, and the magnetic Dirac operator in \eqref{eqn:dirac-mag}, denoted by $\mathcal{D}^\text{AC}$, become
\begin{equation}\label{eqn:ac-dirac}
	A_1^\text{AC}(\bm{X}) = 0, \quad A_2^\text{AC}(\bm{X})  = t_1 X_1,\quad \mathcal{D}^\text{AC} =  \frac{3}{2}\Big[-i \partial_{X_1} \sigma_1 + (-i\partial_{X_2} -t_1 X_1) \sigma_2\Big].
\end{equation}
The corresponding pseudo-magnetic field \eqref{eqn:mag-field} is a constant in the out-of-plane direction with $\bm{B}_\text{eff} = t_1 \widehat{\bm{z}}$. It is worth mentioning that the effective operator $\mathcal{D}^\text{AC}$ is translation invariant in the $X_2$ direction, allowing it to be reduced to the one-dimensional Dirac operator $\mathcal{D}^\text{AC}(k_\parallel)$ in \eqref{eqn:dirac-1d} with
\begin{equation}\label{eqn:ac-1d-dirac}
    \mathcal{D}^\text{AC}(k_\parallel) := \frac{3}{2} \Big[-i\partial_{X_1} \sigma_1 + (k_\parallel - t_1 X_1) \sigma_2\Big], 
\end{equation}
where $k_\parallel - t_1 X_1$ is obtained by substituting $d(X_1) = X_1^2$ into \eqref{eqn:dirac-1d}.

The operator $\mathcal{D}^\text{AC}$ in \eqref{eqn:ac-dirac} is a magnetic Dirac operator with Landau gauge and its spectrum has been well-studied (see e.g. Appendix E of \cite{barsukova2024direct}). Here we briefly review the spectrum of $\mathcal{D}^\text{AC}$ and its 1-dimensional reduction $\mathcal{D}^\text{AC}(k_\parallel)$ as the following Lemma:

\begin{lemma}[The spectrum of $\mathcal{D}^\text{AC}$ and $\mathcal{D}^\text{AC}(k_\parallel)$]\label{lemma:spectrum-landau-dirac-mag-op}
(a) The spectrum of $\mathcal{D}^\text{AC}$ in \eqref{eqn:ac-dirac} is purely discrete of the form
\begin{equation}\label{eqn:flat-spectrum}
	\sigma(\mathcal{D}^\text{AC}) = \left\{ \omega^\pm_s = \pm \frac{3\sqrt{2}}{2} \sqrt{|t_1 s|} \: \Bigg|\: s = 0, 1,2,\dots \right\}.
\end{equation}
The discrete energy levels are the so-called \textbf{Landau levels}.\\

\noindent (b) Each discrete eigenvalue $\omega_{s}^\pm = \pm \frac{3\sqrt{2}}{2} \sqrt{|t_1 s|}$ of $\mathcal{D}^\text{AC}$ has infinite multiplicity and the corresponding eigenfunctions $\bm \phi(\bm{X};\omega_{s}^\pm) = \big(\phi^A(\bm{X};\omega_{s}^\pm), \phi^B(\bm{X};\omega_{s}^\pm)\big)^T$ are of the form 
\begin{equation}\label{eqn:quantum-osc}
	\bm \phi(\bm{X};\omega_{s}^\pm) = e^{ik_\parallel X_2} \bm \psi(X_1;k_\parallel,\omega_{s}^\pm), \quad \bm \psi(X_1;k_\parallel,\omega_{s}^\pm) = \Big(\psi^A(X_1;k_\parallel,\omega_{s}^\pm),\psi^B(X_1;k_\parallel,\omega_{s}^\pm)\Big)^T,
\end{equation}
where $k_\parallel$ denotes the quasi-momentum along the $X_2$-direction and $\bm \psi(X_1;k_\parallel,\omega_{s}^\pm)$ are the eigenfunctions of the reduced 1D Dirac operator $\mathcal{D}^\text{AC}(k_\parallel)$ in \eqref{eqn:ac-1d-dirac} associated with the eigenvalue $\omega_{s}^\pm$ with
\begin{equation}\label{eqn:ac-1d-eigval}
	\mathcal{D}^\text{AC}(k_\parallel) \ \bm \psi(X_1;k_\parallel,\omega_{s}^\pm) = \omega_{s}^\pm \ \bm \psi(X_1;k_\parallel,\omega_{s}^\pm),\quad \text{for any}\quad s =0,1,2,\dots.
\end{equation}
In fact, the 1D Dirac operator $\mathcal{D}^\text{AC}(k_\parallel)$ also has eigenvalues $\omega_{s}^\pm = \pm \frac{3\sqrt{2}}{2} \sqrt{|t_1 s|}$ (independent of $k_\parallel$).\\

\noindent(c) For a given $k_\parallel$ and $s=0,1,2,\dots$, the eigenfunctions $\bm \psi(X_1;k_\parallel,\omega_{s}^\pm)$ of $\mathcal{D}^\text{AC}(k_\parallel)$ in \eqref{eqn:ac-1d-eigval} have \textbf{Gaussian decay rate}. 
\end{lemma}

The proof of Lemma \ref{lemma:spectrum-landau-dirac-mag-op}, including a detailed calculation of the eigenvalues and eigenfunctions of $\mathcal{D}^\text{AC}$ and $\mathcal{D}^\text{AC}(k_\parallel)$, can be found in Appendix E of \cite{barsukova2024direct}. The spectrum of $\mathcal{D}^\text{AC}$ and $\mathcal{D}^\text{AC}(k_\parallel)$, given by $\omega^{\pm}_s = \pm \frac{3\sqrt{2}}{2} \sqrt{|t_1 s|}$ with $s = 0,1,2,\dots$, follows from Equations (S211) and (S228) of \cite{barsukova2024direct} by setting $v_D = 3/2$ and $B_0 = |t_1|$. The eigenfunctions $\bm \psi(X_1;k_\parallel,\omega_{s}^\pm)$ are the standard quantum harmonic oscillators (see Equations (S230)-(S231) and (S237) in \cite{barsukova2024direct}), i.e. a Gaussian multiplied by a Hermite polynomial, and therefore exhibit Gaussian decay.

\paragraph{Gaussian decay of eigenstates} Although we refer to \cite{barsukova2024direct} for the full proof of Lemma \ref{lemma:spectrum-landau-dirac-mag-op}, we briefly explain how the Gaussian decay of these eigenmodes mentioned in Lemma \ref{lemma:spectrum-landau-dirac-mag-op}(c) arises by examining the zero eigenfunction as a representative example; the remaining eigenmodes exhibit similar behavior. For the zero eigenmodes with $\omega_0 = 0$, the eigenvalue problem in \eqref{eqn:ac-1d-eigval} becomes two decoupled 1st-order ODEs: 
\begin{equation*}
	\Big(\partial_{X_1} - (t_1 X_1 - k_\parallel)\Big) \psi^B(X_1;k_\parallel,\omega_0) = 0, \quad \Big(\partial_{X_1} + (t_1 X_1 - k_\parallel)\Big) \psi^A(X_1;k_\parallel,\omega_0) = 0,
\end{equation*}
where the zero eigenfunctions can be solved explicitly by an integrating factor. Among all solutions, we select the physically meaningful one that stays bounded, i.e.
\begin{subequations}\label{eqn:zero-quantum-osc}
\begin{align}
    \psi^A(X_1;k_\parallel,\omega_0) &= 0, &\psi^B(X_1;k_\parallel,\omega_0) &= \sqrt{\frac{t_1}{2\pi}} \exp \left(\frac{t_1}{2}\bigg[\left(X_1 - \frac{k_\parallel}{t_1}\right)^2\bigg]\right), & &t_1 < 0, \label{eqn:zero-quantum-osc-1}\\
	\psi^B(X_1;k_\parallel,\omega_0) &= 0, & \psi^A(X_1;k_\parallel,\omega_0) &= \sqrt{\frac{t_1}{2\pi}}\exp \left(-\frac{t_1}{2} \bigg[\left(X_1 - \frac{k_\parallel}{t_1}\right)^2\bigg]\right) , & & t_1 > 0. \label{eqn:zero-quantum-osc-2}
\end{align}
\end{subequations}
where $\sqrt{\frac{2\pi}{t_1}}$ is the normalization constant such that $\norm{\bm \psi}_{L^2(\mathbb{R})}=1$. We notice that the zero eigenfunction has Gaussian decay rate with respect to $X_1$ and the center of these Gaussian modes $k_\parallel / t_1$ depends on $k_\parallel$ and $t_1$.

We also revisit our approximation of solutions \eqref{eqn:uni-ansatz-approx} to the eigenvalue problem \eqref{eqn:eig-def-hon} by plugging in the eigenvalue and eigenfunction of $\mathcal{D}^\text{AC}$. For each eigenvalue $\omega_{s}^\pm$, its corresponding eigenfunction $\bm \psi(X_1;k_\parallel,\omega_s^\pm)$ in \eqref{eqn:quantum-osc} gives rise to an approximation \eqref{eqn:uni-ansatz-approx-2}, i.e.
\begin{align}\label{eqn:ac-ansatz}
    \bm \psi_{m,n}^{(0)} &= e^{i \bm K \cdot \bm x}e^{ik_\parallel X_2} U \bm \psi(X_1;k_\parallel,\omega_s^\pm)  \Big|_{\substack{\bm x = \Cell_{m,n}\\
    \bm X = \delta \Cell_{m,n}}}.
\end{align}
Since $\bm \psi(X_1;k_\parallel,\omega_s^\pm)$ has Gaussian decay, the approximation \eqref{eqn:ac-ansatz} oscillates in the $X_2$-direction and decays in Gaussian rate in the transverse direction. In fact, our numerical simulations for the AC edge confirm this Gaussian decay predicted by this approximation \eqref{eqn:ac-ansatz} (see e.g. Figure \ref{fig:ac-comp-zero-osc} in Section \ref{sec:num-comparison}).


\paragraph{\edit{Deformations with ZZ orientation induce no effective magnetic field}} For the quadratic deformation $\bm{u}^\text{ZZ} = (X_2^2, 0)^T$ along the ZZ edge, the effective magnetic potential $\bm{A}^\text{ZZ}_\text{eff} = (A_1^\text{ZZ}, A_2^\text{ZZ})$ and magnetic Dirac operator $\mathcal{D}^\text{ZZ}$ become
\begin{equation}\label{eqn:zz-dirac}
	A_1^\text{ZZ}(\bm{X}) = 0, \quad A_2^\text{ZZ}(\bm{X})  = -t_1 X_2,\quad \mathcal{D}^\text{ZZ} =  \frac{3}{2}\Big[(-i \partial_{X_1}) \sigma_1 + (-i\partial_{X_2} +t_1 X_2) \sigma_2\Big].
\end{equation}
The corresponding pseudo-magnetic field in \eqref{eqn:mag-field} vanishes, i.e. $\bm{B}_\text{eff} = 0 \widehat{\bm{z}}$. 

Unlike $\mathcal{D}^\text{AC}$, which has purely discrete spectrum, the effective magnetic Dirac operator $\mathcal{D}^\text{ZZ}$ only has continuous spectrum with no discrete eigenvalue. To see why $\mathcal{D}^\text{ZZ}$ only has continuous spectrum, we square the operator $\mathcal{D}^\text{ZZ}$ and obtain
\begin{align}
    (\mathcal{D}^\text{ZZ})^2 &= \Big(-\Delta + t_1^2 X_2^2 - it_1 - 2it_1 X_2 \partial_{X_2}\Big) \sigma_0 \label{eqn:squared-ZZ}\\
    &= \Big(-\partial_{X_1}^2 + (-i\partial_{X_2}+t_1X_2)^2\Big) \sigma_0,\nonumber
\end{align}
where $\sigma_0 = I$. We observe that the part $(-i\partial_{X_2}+t_1X_2)^2$ is unitarily equivalent to $-\partial_{X_2}^2$. To see why, we take the mapping $f(X_2) \mapsto e^{\frac{it_1X_2^2}{2}} f(X_2)$ and use the following calculation
\begin{align*}
     -\partial^2_{X_2} \left(e^{\frac{it_1X_2^2}{2}} f\right) &= e^{\frac{it_1X_2^2}{2}} \left(-\partial_{X_2}^2 + t_1^2 X_2^2 - it_1 - 2it_1 X_2 \partial_{X_2}\right) f\\
        &= e^{\frac{it_1X_2^2}{2}} (-i\partial_{X_2}+t_1X_2)^2 f.
\end{align*}
Therefore, the spectrum of $(\mathcal{D}^\text{ZZ})^2$ is equivalent to the spectrum of $-\Delta$ with $\sigma((\mathcal{D}^\text{ZZ})^2) = \sigma(-\Delta) = [0,\infty)$. Since $\mathcal{D}^\text{ZZ}$ has chiral symmetry (it is self-adjoint and has no diagonal entries), the spectrum of $\mathcal{D}^\text{ZZ}$ is symmetric about zero, and hence $\mathcal{D}^\text{ZZ}$ only has continuous spectrum with $\sigma(\mathcal{D}^\text{ZZ}) = \mathbb{R}$. 

Returning to the quadratically deformed honeycomb lattice along the ZZ edge, this deformation does not open the conical touching at the Dirac point, since the effective operator has spectrum $\sigma(\mathcal{D}^\text{ZZ}) = \mathbb{R}$ with no gap and discrete eigenvalues. This conclusion is consistent with our numerical simulations for the ZZ edge (see e.g. Figure \ref{fig:zz-def} in Section \ref{sec:num-comparison}).


\setcounter{equation}{0}
\subsection{Discrete Fourier transform and Poisson summation formula}\label{app:prelim-dft}

We introduce two fundamental tools for analyzing the spectrum of a discrete Hamiltonian operator: the discrete Fourier transform and the Poisson summation formula. These tools connect the discrete formulation in real space to a continuous representation in momentum space.

We start by introducing the 1D discrete Fourier transform. Consider a sequence $(f_n)_{n \in \mathbb{Z}} \in l^2(\mathbb{Z})$, its discrete Fourier transform (DFT) for $(f_n)_{n \in \mathbb{Z}}$ is defined as
\begin{equation}\label{eqn:dft-definition}
	\widetilde{f}(k) := \sum_{n \in \mathbb{Z}} f_n e^{-ikn}, \qquad k \in [-\pi, \pi),
\end{equation}
where $k$ is known as the quasi-momentum. We use the tilde notation for DFT to distinct from the Fourier transform. We observe that $\widetilde{f}(k)$ is a $L^2$ function on $[-\pi,\pi]$ when $(f_n)_{n \in \mathbb{Z}} \in l^2(\mathbb{Z})$ due to the following equality
\begin{equation}\label{eqn:dft-norm}
	\|\widetilde{f}(k)\|^2_{L^2([-\pi,\pi])} = 2\pi\|f_n\|_{l^2(\mathbb{Z})}^2.
\end{equation}
The sequence $(f_n)_{n \in \mathbb{Z}} \in l^2(\mathbb{Z})$ can be recovered from its discrete Fourier transform $\widetilde{f}(k)$ via the inverse discrete Fourier transform (IDFT), given by:
\begin{equation}\label{eqn:idft}
	f_n = \frac{1}{2\pi}\int_{-\pi}^{\pi} \widetilde{f}(k) e^{ikn} dk.
\end{equation}


We are particularly interested in dealing with slowly-varying perturbations, especially sequences $(\Psi(n\delta))_{n\in \mathbb{Z}}$ with a small parameter $\delta$. For these slowly-varying sequences, their DFT can be represented by the Fourier transform $\widehat{\Psi}$ by applying the Poisson summation formula, which is stated as follows.
\begin{lemma}[Poisson summation formula]\label{lemma:psf}
	For a Schwartz function $f(x) \in \mathcal{S}(\mathbb{R})$ and its Fourier transform $\widehat{f}(\xi) = \int_{-\infty}^\infty f(x) e^{-i\xi x} dx$, the following identity holds
	{\small
    \begin{equation}\label{eqn:psf}
		\sum_{n \in \mathbb{Z}} f(n) = \sum_{n \in \mathbb{Z}} \widehat{f}(2\pi n). 
	\end{equation}}
\end{lemma}

\paragraph{Poisson's summation formula for slowly-varying functions} By choosing $f(x)= \Psi(x) e^{-ikx}$,  we obtain the DFT for the sequence $\Psi(n\delta)$: for a Schwartz function $\Psi(x) \in \mathcal{S}(\mathbb{R})$, the DFT of the sequence $\psi_n = \Psi(n\delta)$ can be written as
\begin{equation}\label{eqn:psf-fast-osc}
    \widetilde{\psi}(k) = \sum_{n \in \mathbb{Z}} \Psi(n\delta) e^{-ikn} = \frac{1}{\delta} \sum_{n \in \mathbb{Z}} \widehat{\Psi}\left(\frac{k + 2\pi n}{\delta}\right),
\end{equation}
where $\widehat{\Psi}(\xi) = \int_{-\infty}^\infty \Psi(x) e^{-i \xi x} dx$ is the Fourier transform of $\Psi(x)$. 

We use the following scaled Poisson's summation formula frequently used in Section \ref{sec:proof}: for a given Schwartz function $\Psi(x) \in \mathcal{S}(\mathbb{R})$, we have 
\begin{equation}\label{eqn:psf-scaled}
    \sum_{m \in \mathbb{Z}} \Psi\left(\frac{\sqrt{3}}{2} \delta m\right) e^{-ik \frac{\sqrt{3}}{2} m} = \frac{2}{\sqrt{3} \delta}\sum_{m \in \mathbb{Z}} \widehat{\Psi}\left(\frac{k + \frac{4\pi}{\sqrt{3}}m}{\delta}\right),
\end{equation}
which is obtained by replacing $\delta$ and $k$ in \eqref{eqn:psf-fast-osc} with $\frac{\sqrt{3}}{2} \delta$ and $\frac{\sqrt{3}}{2}k$. By a standard translation, we also obtain
\begin{subequations}\label{eqn:psf-scaled-1}
    \begin{align}
        \sum_{m \in \mathbb{Z}} \Psi\left(\frac{\sqrt{3}}{2} \delta (m+1) \right) e^{-ik \frac{\sqrt{3}}{2} m} &= e^{i\frac{\sqrt{3}}{2} k} \frac{2}{\sqrt{3} \delta}\sum_{m \in \mathbb{Z}} \widehat{\Psi}\left(\frac{k + \frac{4\pi}{\sqrt{3}}m}{\delta}\right),\label{eqn:psf-scaled-1a}\\
        \sum_{m \in \mathbb{Z}} \Psi\left(\frac{\sqrt{3}}{2} \delta (m+2) \right) e^{-ik \frac{\sqrt{3}}{2} m} &= e^{i\sqrt{3} k} \frac{2}{\sqrt{3} \delta}\sum_{m \in \mathbb{Z}} \widehat{\Psi}\left(\frac{k + \frac{4\pi}{\sqrt{3}}m}{\delta}\right).\label{eqn:psf-scaled-1b}
    \end{align}
\end{subequations}

\paragraph{The $l^2$ norm of sequences $\psi_m = \Psi\left(\frac{\sqrt{3}}{2}\delta m\right)$} Given a scalar function $\Psi \in \mathcal{S}(\mathbb{R})$, the $l^2$ norm of $\psi_m$ is of order $\delta^{-\frac{1}{2}}$, i.e. for sufficiently small $\delta$, we have
\begin{equation}\label{eqn:l2-norm}
	\|\psi_m\|_{l^2(\mathbb{Z};\mathbb{C})} \lesssim \delta^{-\frac{1}{2}} \|\widehat{\Psi}\|_{L^2(\mathbb{R})}.
\end{equation}
Specifically, we have the following limit
\begin{equation}\label{eqn:app-l2-L2-norm}
	\lim_{\delta \rightarrow 0}\quad \delta \|\psi_m\|^2_{l^2(\mathbb{Z};\mathbb{C})} = \frac{2}{\sqrt{3}} \|\Psi\|^2_{L^2(\mathbb{R})}.
\end{equation}
To show these, we observe that $|\Psi^2| \in \mathcal{S}(\mathbb{R})$ since $\Psi \in \mathcal{S}(\mathbb{R})$ and perform PSF in \eqref{eqn:psf-scaled} with $k=0$:
\begin{equation*}
	\|\psi_m\|^2_{l^2(\mathbb{Z};\mathbb{C})} = \sum_{m \in \mathbb{Z}} \left|\Psi\left(\frac{\sqrt{3}}{2}\delta m\right)\right|^2 = \sum_{m \in \mathbb{Z}} g\left(\frac{\sqrt{3}}{2}\delta m\right)= \frac{2}{\sqrt{3}\delta} \sum_{m \in \mathbb{Z}} \widehat{g}\left(\frac{4\pi m}{\sqrt{3}\delta}\right),
\end{equation*}
where $g(x) := |\Psi^2(x)|$. By separating the above sum with $m=0$ and $m \neq 0$, we obtain 
\begin{align*}
	&\delta \|\psi_m\|^2_{l^2(\mathbb{Z};\mathbb{C})} - \frac{2}{\sqrt{3}} \widehat{g}(0) = \frac{2}{\sqrt{3}} \sum_{m \neq 0} \widehat{g}\left(\frac{2\pi m}{\delta}\right) \lesssim \sum_{m \neq 0} \frac{\delta }{m^2}.
\end{align*}
The last inequality holds since $\widehat{g} \in \mathcal{S}(\mathbb{R})$ and we can bound $\widehat{g}$ by $|\widehat{g}(k)| \lesssim 1/|k|^2$. Since $\widehat{g}(0) = \|g\|_{L^1(\mathbb{R})} = \|\Psi\|^2_{L^2(\mathbb{R})}$, we obtain the desired bounds \eqref{eqn:l2-norm} and \eqref{eqn:app-l2-L2-norm}, i.e.
\begin{align*}
    \lim_{\delta \rightarrow 0} \quad \delta \|\psi_m\|^2_{l^2(\mathbb{Z};\mathbb{C})} = \frac{2}{\sqrt{3}} \|\Psi\|_{L^2(\mathbb{R})}^2.
\end{align*}


\setcounter{equation}{0}

\subsection{A review of classical results in exponential dichotomy theory}\label{app:proof-inv-dk}

In this section, we provide a self-contained review of some classical results from exponential dichotomy theory in \cite{palmer1984exponential} and prove the bound \eqref{eqn:inv-f-def} in Proposition \ref{prop:inv-Dk}.

Let us first review the definition of exponential dichotomy in terms of a 2-dimensional linear ODE system $\bm{\varphi}' = \bm{A}(t) \bm{\varphi}$, where $\bm{\varphi} \in \mathbb{R}^2$ and $\bm{A}(t) \in \mathbb{R}^{2\times 2}$. 
\begin{definition}[Exponential dichotomy on half lines]\label{def:exp-dich}
We say the system $\bm{\varphi}' = \bm{A}(t) \bm{\varphi}$ has an exponential dichotomy on $I_+ = [0,\infty)$ if there exists a projection matrix $\bm{P_+} \in \mathbb{R}^{2\times 2}$ and constants $K_+\geq 1, \alpha_+>0$ such that
\begin{subequations}\label{eqn:exp-dich-cst-pos}
    \begin{align}
        \|\bm{\Phi}(t) \bm{P}_+ \bm{\Phi}^{-1}(s)\| \leq K_+ e^{-\alpha_+(t-s)}, \qquad 0\leq s \leq t,\label{eqn:exp-dich-cst-1}\\
        \|\bm{\Phi}(t) (\bm{I} - \bm{P_+}) \bm{\Phi}^{-1}(s)\| \leq K_+ e^{-\alpha_+(s-t)}, \qquad s \geq t\geq 0,\label{eqn:exp-dich-cst-2}
    \end{align}
\end{subequations}
where $\bm{\Phi}(t)$ is the fundamental matrix for the ODE system. 

Similarly, we say the system $\bm{\varphi}' = \bm{A}(t) \bm{\varphi}$ has an exponential dichotomy on $(-\infty, 0]$ if there exists a projection matrix $\bm{P}_- \in \mathbb{R}^{2\times 2}$ and constants $K_-\geq 1, \alpha_->0$ such that
\begin{subequations}\label{eqn:exp-dich-cst-neg}
    \begin{align}
        \|\bm{\Phi}(t) \bm{P}_- \bm{\Phi}^{-1}(s)\| \leq K_- e^{-\alpha_-(t-s)}, \qquad s \leq t\leq 0,\label{eqn:exp-dich-cst-3}\\
        \|\bm{\Phi}(t) (\bm{I} - \bm{P_-}) \bm{\Phi}^{-1}(s)\| \leq K_- e^{-\alpha_-(s-t)}, \qquad 0\geq s \geq t.\label{eqn:exp-dich-cst-4}
    \end{align}
\end{subequations}

Moreover, we say the system $\bm{\varphi}' = \bm{A}(t) \bm{\varphi}$ has exponential dichotomy on both half lines if it has an exponential dichotomy on $[0,\infty)$ and $(-\infty,0]$ respectively.
\end{definition}

A simple but well-known example of a system $\bm{\varphi}' = \bm{A}(t) \bm{\varphi}$ exhibiting exponential dichotomy on both half lines is the constant-coefficient case $\bm{A}(t)\equiv \bm{A} \in \mathbb{R}^{2 \times 2}$ with $\det \bm A < 0$. In fact, one can check \eqref{eqn:exp-dich-cst-pos} and \eqref{eqn:exp-dich-cst-neg} by using $\bm{P}_+$ which projects onto the stable (negative-eigenvalue) subspace on $[0,\infty)$, and $\bm{P}_-$ which projects onto the unstable (positive-eigenvalue) subspace on $(-\infty,0]$. The constants $-\alpha_+$ is the negative eigenvalue of $\bm{A}$ and $\alpha_+$ is the positive eigenvalue of $\bm{A}$. Here we omit the details.

The property of exponential dichotomy is robust against certain small perturbations, which is described in Lemma 3.4 in \cite{palmer1984exponential}. Here we list it in terms of our setting:
\begin{lemma}\label{lemma:app-exp-dich-asy-0}
    Let $\bm{A}(t)$ and $\bm{B}(t)$ be two bounded and continuous $2\times 2$ matrix on $[0,\infty)$. If $\bm{\varphi}' = \bm{A}(t)\bm{\varphi}$ has an exponential dichotomy on $[0,\infty)$ and $\lim_{t \rightarrow \infty} \|\bm{B}(t)\| = 0$, then the perturbed system $\bm{\varphi}' = [\bm{A}(t) + \bm{B}(t)]\bm{\varphi}$ also has an exponential dichotomy on $[0,\infty)$.
\end{lemma}

As mentioned above, a constant matrix with negative determinant has exponential dichotomy on both half lines $[0,\infty)$ and $(-\infty, 0]$. Therefore, a direct application of Lemma \ref{lemma:app-exp-dich-asy-0} on $[0,\infty)$ is the following Corollary:
\begin{corollary}\label{lemma:app-exp-dich-asy}
    Let $\bm{A} \in \mathbb{R}^{2\times 2}$ be a constant matrix with $\det \bm A < 0$ and $\bm{B}(t) \in \mathbb{R}^{2\times 2}$ be bounded and continuous on $[0,\infty)$. If $\lim_{t \rightarrow +\infty} \|\bm{B}(t)\| = 0$, then the perturbed system $\bm{\varphi}' = [\bm{A} + \bm{B}(t)]\bm{\varphi}$ has an exponential dichotomy on $[0,\infty)$.
\end{corollary}

A similar result of Corollary \ref{lemma:app-exp-dich-asy} holds for $(-\infty, 0]$ when $\lim_{t \rightarrow -\infty} \|\bm{B}(t)\| = 0$. Furthermore, when a linear system $\bm{\varphi}' = \bm{A}(t) \bm{\varphi}$ has exponential dichotomy on both half lines, we can explicitly solve for 
\begin{equation}\label{eqn:exp-f}
    \bm{\varphi}' - \bm{A}(t) \bm{\varphi} = \bm{f}
\end{equation}
for suitable $\bm{f} \in L^2(\mathbb{R})$. This argument is stated as
\begin{lemma}\label{lemma:exp-dich-sol}
    If $\bm A(t) \in \mathbb{R}^{2\times 2}$ is a bounded and continuous matrix on $t \in \mathbb{R}$ and the linear system $\bm{\varphi}' = \bm{A}(t) \bm{\varphi}$ has exponential dichotomy on both half lines, i.e. \eqref{eqn:exp-dich-cst-pos}-\eqref{eqn:exp-dich-cst-neg} hold, then the system \eqref{eqn:exp-f} has a bounded solution if $\bm{f}$ satisfies that
    \begin{equation}\label{eqn:exp-f-cond}
        \langle \bm{\Psi}, \bm{f}\rangle_{L^2(\mathbb{R})} = 0
    \end{equation}
    holds for all bounded solutions $\bm{\Psi}(t)$ satisfying $\bm{\Psi}' = -\bm{A}^*(t) \bm{\Psi}$.
\end{lemma}
Lemma \ref{lemma:exp-dich-bd} is a restatement of Lemma 4.2 in \cite{palmer1984exponential}. Moreover, any bounded solution $\bm \varphi$ to \eqref{eqn:exp-f} satisfies the following properties
\begin{lemma}\label{lemma:exp-dich-bd}
   Under the assumptions of Lemma \ref{lemma:exp-dich-sol}, any bounded solution $\bm{\varphi}$ of \eqref{eqn:exp-f} can be written as
    \begin{subequations}\label{eqn:exp-f-sol}
        \begin{align}
            \bm{\varphi}(t) &= \bm{\Phi}(t) \bm{P}_+ \bm{\xi} + \int_0^t \bm{\Phi}(t) \bm{P}_+ \bm{\Phi}^{-1}(s) \bm{f}(s) \: ds - \int_t^\infty \bm{\Phi}(t) (\bm{I} - \bm{P}_+) \bm{\Phi}^{-1}(s) \bm{f}(s) \: ds, \ t \geq 0,\label{eqn:exp-f-sol-1}\\
            \bm{\varphi}(t) &= \bm{\Phi}(t) (\bm{I} - \bm{P}_-) \bm{\xi} + \int_{-\infty}^t \bm{\Phi}(t) \bm{P}_- \bm{\Phi}^{-1}(s) \bm{f}(s) \: ds - \int_t^0 \bm{\Phi}(t) (\bm{I} - \bm{P}_-) \bm{\Phi}^{-1}(s) \bm{f}(s) \: ds, \ t \leq 0,\label{eqn:exp-f-sol-2}
        \end{align}
    \end{subequations}
    where the constant vector $\bm{\xi} \in \mathbb{R}^2$ is a solution to
    \begin{equation}\label{eqn:xi-exp-dich}
        \Big[\bm P_+ - (\bm I - \bm P_-)\Big] \bm \xi = \int_{-\infty}^0 \bm P_- \bm\Phi^{-1}(s) \bm f(s) \: ds + \int_0^{\infty} (\bm I - \bm P_+) \bm\Phi^{-1}(s) \bm f(s) \: ds.
    \end{equation}
    Moreover, the bounded solution $\bm{\varphi}$ in \eqref{eqn:exp-f-sol} satisfies the following estimate in terms of $\bm{f}$:
    \begin{equation}\label{eqn:exp-f-bd}
        \|\bm{\varphi}\|_{H^1(\mathbb{R})} \lesssim \|\bm{f}\|_{L^2(\mathbb{R})}.
    \end{equation}
\end{lemma}

\begin{proof}
    The explicit formula \eqref{eqn:exp-f-sol} for $\bm \varphi$ is presented on the top few lines of page 247 of \cite{palmer1984exponential} and the equation for $\bm \xi$ is presented on the bottom line of page 246 of \cite{palmer1984exponential}. To prove \eqref{eqn:exp-f-bd}, it suffices to show
    \begin{equation}\label{eqn:sm-l2}
        \|\bm{\varphi}\|_{L^2(\mathbb{R})} \lesssim \|\bm{f}\|_{L^2(\mathbb{R})},
    \end{equation}
    since the derivative $\bm \varphi'$ in \eqref{eqn:exp-f} is bounded by
    \begin{equation*}
        \|\bm\varphi'\|_{L^2(\mathbb{R})} \lesssim \|\bm A(t)\|_{L^\infty(\mathbb{R})} \|\bm\varphi\|_{L^2(\mathbb{R})} + \|\bm f\|_{L^2(\mathbb{R})} \lesssim \|\bm f\|_{L^2(\mathbb{R})}.
    \end{equation*}

    \paragraph{Proof of \eqref{eqn:sm-l2}:} We observe that using \eqref{eqn:exp-f-sol-1}, the solution $\bm\varphi(t)$ on $[0,\infty)$ becomes
    \begin{align}
        \bm\varphi(t)-\bm{\Phi}(t)\bm{P_+}\bm{\xi}
    &= \int_0^t \bm\Phi(t) \bm P_+\bm \Phi^{-1}(s)\bm f(s)\,ds
      -\int_t^\infty \bm \Phi(t)(\bm I-\bm P_+)\bm \Phi^{-1}(s)\bm f(s)\,ds \nonumber\\
      &=: \int_0^\infty \bm G_+(t,s)\bm f(s)\,ds, \qquad \forall \quad t\geq 0. \label{eqn:G-plus}
    \end{align}
    Similarly, using \eqref{eqn:exp-f-sol-2}, the solution $\bm\varphi(t)$ on $(-\infty,0]$ satisfies
    \begin{align}
        \bm\varphi(t)  - \bm\Phi(t)(\bm I-\bm P_-) \bm \xi &= \int_{-\infty}^t \bm \Phi(t)\bm P_-\bm \Phi^{-1}(s)\bm f(s)\,ds
                 -\int_t^0 \bm \Phi(t)(\bm I-\bm P_-)\bm \Phi^{-1}(s)\bm f(s)\,ds \nonumber\\
        &=: \int_{-\infty}^0 \bm G_-(t,s) \bm f(s)\,ds, \qquad \forall \quad t \leq 0.\label{eqn:G-minus}
    \end{align}
    Therefore, to prove \eqref{eqn:sm-l2}, it suffices to show the following bounds
    \begin{subequations}\label{eqn:sm-l2-eq}
        \begin{align}
            & \norm{\int_0^\infty \bm G_+(t,s)\bm f(s)\,ds}_{L^2([0,\infty))} \lesssim \norm{\bm f}_{L^2([0,\infty))}, \label{eqn:sm-l2-eq-1}\\
            & \norm{\int_0^\infty \bm G_-(t,s)\bm f(s)\,ds}_{L^2((-\infty,0])} \lesssim \norm{\bm f}_{L^2((-\infty,0])}, \label{eqn:sm-l2-eq-2}\\
            & \norm{\bm\Phi(t) \bm P_+ \bm \xi}_{L^2([0,\infty))} \lesssim \norm{\bm f}_{L^2(\mathbb{R})}, \qquad \norm{\bm\Phi(t)(\bm I-\bm P_-) \bm \xi}_{L^2((-\infty,0])} \lesssim \norm{\bm f}_{L^2(\mathbb{R})}.\label{eqn:sm-l2-eq-3}
        \end{align}
    \end{subequations}

    \paragraph{Proof of \eqref{eqn:sm-l2-eq-1} and \eqref{eqn:sm-l2-eq-2}:} To simplify our calculation, we define a scalar kernel $k_+(t,s)$
    \begin{equation*}
        k_+(t,s):=\|\bm G_+(t,s)\|, \qquad t,s\geq0.
    \end{equation*}
    Using the exponential dichotomy \eqref{eqn:exp-dich-cst-pos} on $[0,\infty)$, we obtain a bound on the kernel $k_+(t,s)$, i.e.
    \begin{equation*}
        k_+(t,s) \le
    \begin{cases}
    K_+e^{-\alpha_+(t-s)}, & 0\le s\le t,\\[1mm]
    K_+e^{-\alpha_+(s-t)}, & t\le s<\infty.
    \end{cases}
    \end{equation*}
    Therefore, for fixed $t\ge0$, we have
    \begin{align*}
    \|k_+(t,\cdot)\|_{L^1([0,\infty))}
    &\le K_+\!\left(\int_0^t e^{-\alpha_+(t-s)}\,ds
                  +\int_t^\infty e^{-\alpha_+(s-t)}\,ds\right)  \\
    &= K_+\!\left(\int_0^t e^{-\alpha_+u}\,du
                  +\int_0^\infty e^{-\alpha_+u}\,du\right)
     \le \frac{2K_+}{\alpha_+}.
    \end{align*}
    Similarly, for fixed $s\ge0$, we have
    \begin{equation*}
        \|k_+(\cdot,s)\|_{L^1([0,\infty))} \le \frac{2K_+}{\alpha_+}.
    \end{equation*}
    Let us take $h(s):=|\bm f(s)|$ for $s\ge0$. Then, for every $t\ge0$, we have
    \begin{align*}
        &\|\bm \varphi(t)\|
    = \Big\|\int_0^\infty \bm G_+(t,s) \bm f(s)\,ds\Big\|
    \le \int_0^\infty k_+(t,s)h(s)\,ds\\
    \Rightarrow \qquad & \|\bm \varphi\|_{L^2([0,\infty))}
    \le \Big\|\int_0^\infty k_+(\cdot,s)h(s)\,ds\Big\|_{L^2([0,\infty))}.
    \end{align*}
    By applying Young's inequality\footnote{The Young's inequality for kernels used in the above proof is stated as follows (a version with $r=1$ that we used here can be found in on page 9 in \cite{folland1995introduction}). } for kernels with $p=q=2$, $r=1$ and $C_x = C_y = \frac{2K_+}{\alpha_+}$, we obtain the desired bound \eqref{eqn:sm-l2-eq-1} with
    \begin{equation*}
        \Big\|\int_0^\infty k_+(\cdot,s)h(s)\,ds\Big\|_{L^2([0,\infty))}
    \le C_y^{1/2} C_x^{1/2}\|h\|_{L^2([0,\infty))}
    \le \frac{2K_+}{\alpha_+}\,\|\bm f\|_{L^2([0,\infty))}.
    \end{equation*}
    Similarly, we define the scalar kernel $k_-(t,s):=\|\bm G_-(t,s)\|$ with $t,s \in (-\infty,0]$ and it satisfies
    \begin{align*}
        \sup_{t\le0}\|k_-(t,\cdot)\|_{L^1((-\infty,0])}\le\frac{2K_-}{\alpha_-}, \qquad
    \sup_{s\le0}\|k_-(\cdot,s)\|_{L^1((-\infty,0])}\le\frac{2K_-}{\alpha_-}.
    \end{align*}
    Applying Young's inequality for kernels yields the desired bound \eqref{eqn:sm-l2-eq-2} with
    \begin{equation*}
    \|\bm \varphi\|_{L^2((-\infty,0])} \leq \frac{2K_-}{\alpha_-} \|\bm f\|_{L^2((-\infty,0])}.
    \end{equation*}

    \paragraph{Proof of \eqref{eqn:sm-l2-eq-3}:} Since $\bm \Phi(0) = \bm I$, by setting $s=0$ in \eqref{eqn:exp-dich-cst-1} and \eqref{eqn:exp-dich-cst-4}, we have
    \begin{align*}
        \norm{\bm\Phi(t) \bm P_+ \bm \xi}_{L^2([0,\infty))} \lesssim \norm{\bm \xi}, \qquad \norm{\bm\Phi(t) (\bm I - \bm P_-) \bm \xi}_{L^2([0,\infty))} \lesssim \norm{\bm \xi}.
    \end{align*}
    Therefore, it is sufficient to show 
    \begin{equation}
        \norm{\bm \xi} \lesssim \norm{\bm f}_{L^2(\mathbb{R})}. \label{eqn:bd-f-xi}
    \end{equation}
    Using \eqref{eqn:xi-exp-dich}, we obtain
    \begin{align*}
        \norm{\bm \xi} &\lesssim \Bigg|\int_{-\infty}^0 \bm P_- \bm\Phi^{-1}(t) \bm f(t) \: dt \Bigg| + \Bigg| \int_0^{\infty} (\bm I - \bm P_+) \bm\Phi^{-1}(t) \bm f(t) \: dt \Bigg|\\
        &\lesssim \norm{e^{\alpha_-s}}_{L^2_s((-\infty,0])} \norm{\bm f}_{L^2((-\infty,0])} + \norm{e^{-\alpha_+ s}}_{L^2_s([0,\infty))} \norm{\bm f}_{L^2([0,\infty))} \lesssim \norm{\bm f}_{L^2(\mathbb{R})},
    \end{align*}
    where the first inequality on the second line comes from applying $t=0$ in \eqref{eqn:exp-dich-cst-2} and \eqref{eqn:exp-dich-cst-3}. Thus, we complete the proof of Lemma \ref{lemma:exp-dich-bd}.
\end{proof}



\setcounter{equation}{0}
\subsection{Numerical scheme and numerical artifacts}\label{app:numerical-zero}

In this section, we present our numerical scheme in detail and explain the apparent numerical artifacts. 

\paragraph{Numerical scheme} We begin by describing the numerical scheme for the AC edge, using the unit cell shown in Figure \ref{fig:ac-ref-nodes}. To reproduce the Hamiltonian in \eqref{eqn:ham-hon-delta-approx} on the associated supercell with nodes $A_s, B_s, C_s, D_s$, it suffices to determine the hopping coefficients between neighboring nodes. Following the definition in \eqref{eqn:hopping-def-honeycomb}, for any pair of undeformed nearest-neighbor nodes $X,Y$, the hopping coefficient $t^\delta(X,Y)$ after applying the displacement $\bm u$ is
\begin{align*}
    t^\delta(X,Y) = \left \langle \frac{\bm{u}(X) - \bm{u}(Y)}{|X-Y|}, \frac{X-Y}{|X-Y|} \right \rangle.
\end{align*}
Then the Hamiltonian operator in \eqref{eqn:ham-hon-delta-approx} now becomes $H^\delta(q_\parallel)$ with
\begin{align*}
    (H^\delta(q_\parallel) {\Psi})_s^A &= t^\delta(A_s,B_{s+1}) \psi^B_{s+1} + t^\delta(A_s,B_s) \psi^B_s + t^\delta(A_s,D_{s+1} + \bm{v}_{\parallel}) e^{i 3 \pi q_\parallel} \psi^D_{s+1},\\
    (H^\delta(q_\parallel) {\Psi})^B_s &= t^\delta(A_s,B_s) \psi^A_s + t^\delta(A_{s-1},B_s) \psi^A_{s-1} + t^\delta(B_s,C_s) \psi^C_s,\\
    (H^\delta(q_\parallel) {\Psi})^C_s &= t^\delta(B_s,C_s) \psi^B_s + t^\delta(C_s,D_s) \psi^D_s + t^\delta(C_s,D_{s+1}) \psi^D_{s+1},\\
    (H^\delta(q_\parallel) {\Psi})^D_s &= t^\delta(C_s,D_s) \psi^C_s + t^\delta(C_{s-1},D_s) \psi^C_{s-1} + t^\delta(D_s,A_{s-1}-\bm{v}_{\parallel}) e^{-i 3 \pi q_\parallel} \psi^A_{s-1},
\end{align*}
where the factor $e^{\pm i 3 \pi q_\parallel}$ comes from the translation along the $\bm{v}_\parallel$ direction (see \eqref{eqn:num-bloch-ac}). The model on the supercell for the ZZ edge is obtained analogously, hence we omit the detail here.

\paragraph{Numerical artifacts} We now explain why there are six zero numerical eigenvalues when $q_\parallel$ is small. This arises from (i) our choice of unit cell in Figure \ref{fig:ac-ref-nodes} and (ii) boundary-localized modes that enter as numerical artifacts.

Let us first explain why our choice of unit cell leads to a doubling of all eigenvalues. This effect is a consequence of an additional symmetry that is introduced when we double the unit-cell size. The additional symmetry is easy to see in the unperturbed case $\delta=0$: the Hamiltonian $H^0(q_\parallel)$ on the infinite lattice takes a simple form since all hopping coefficients are equal to 1, i.e.
\begin{align*}
    (H^0(q_\parallel) {\Psi})_s^A &= \psi^B_{s+1} +  \psi^B_s + e^{i 3 \pi q_\parallel} \psi^D_{s+1},\\
    (H^0(q_\parallel) {\Psi})^B_s &= \psi^A_s + \psi^A_{s-1} + \psi^C_s,\\
    (H^0(q_\parallel) {\Psi})^C_s &= \psi^B_s + \psi^D_s + \psi^D_{s+1},\\
    (H^0(q_\parallel) {\Psi})^D_s &= \psi^C_s + \psi^C_{s-1} +  e^{-i 3 \pi q_\parallel} \psi^A_{s-1},
\end{align*}
where $\Psi_s = (\psi_s^A, \psi_s^B, \psi_s^C, \psi_s^D)$ with $s \in \mathbb{Z}$ is the wave function.

Since nodes $(C_s,D_s)$ are translated copy of $(A_s, B_s)$, we construct the following left-shift operator $T(q_\parallel)$, which shifts the oscillations to the left by mapping $(C_s,D_s)\mapsto(A_s,B_s)$ and $(A_{s-1},B_{s-1})\mapsto(C_s,D_s)$, together with phase factors $e^{\pm i\frac{3}{2}\pi q_\parallel}$, i.e.
\begin{align*}
    (T(q_\parallel)\psi)_s^A &= e^{i \frac{3}{2} \pi q_\parallel} \psi^C_s, & (T(q_\parallel)\psi)_s^B &= e^{i \frac{3}{2} \pi q_\parallel} \psi^D_s,\\
    (T(q_\parallel)\psi)_s^C &= e^{-i \frac{3}{2} \pi q_\parallel} \psi^A_{s-1}, & (T(q_\parallel)\psi)_s^D &= e^{-i \frac{3}{2} \pi q_\parallel} \psi^B_{s-1}.
\end{align*}
One can check that $H^0(q_\parallel)$ commutes with the left-shift operator $T(q_\parallel)$, i.e.
\begin{align}
    T(q_\parallel) H^0(q_\parallel) = H^0(q_\parallel) T(q_\parallel).\label{eqn:app-extra-sym}
\end{align}
Here, we only check $T(q_\parallel) H^0(q_\parallel) - H^0(q_\parallel)T(q_\parallel)$ on the $A_s$ node as an illustrative example:
\begin{align*}
    \Big(T(q_\parallel) H^0(q_\parallel) \psi\Big)_s^A &= e^{i \frac{3}{2} \pi q_\parallel} \Big(H^0(q_\parallel) \psi\Big)_s^C = e^{i \frac{3}{2} \pi q_\parallel} \Big(\psi^B_s + \psi^D_s + \psi^D_{s+1}\Big),\\
    \Big(H^0(q_\parallel) T(q_\parallel) \psi\Big)_s^A &= \Big(T(q_\parallel) \psi\Big)^B_{s+1} + \Big(T(q_\parallel) \psi\Big)^B_s + e^{i 3 \pi q_\parallel} \Big(T(q_\parallel) \psi\Big)^D_{s+1}\\
    &=e^{i \frac{3}{2} \pi q_\parallel} \psi^D_{s+1} + e^{i \frac{3}{2} \pi q_\parallel} \psi^D_s + e^{i \frac{3}{2} \pi q_\parallel} \psi^B_s =  \Big(T(q_\parallel) H^0(q_\parallel) \psi\Big)_s^A.
\end{align*}
Therefore, due to the symmetry \eqref{eqn:app-extra-sym}, we know that if $\psi$ is an eigenvector of $H^0(q_\parallel)$ with corresponding eigenvalue $E$, then $T(q_\parallel) \psi$ is also an eigenvector of $H^0(q_\parallel)$ associated with eigenvalue $E$. The two eigenvectors $\psi$ and $T(q_\parallel) \psi$ are linearly independent and share the same envelope behavior, which explains the doubled multiplicity of each eigenvalue of $H^0(q_\parallel)$.

We now explain the ``doubled'' multiplicity of each eigenvalue of $H^\delta(q_\parallel)$. In fact, a more accurate description is that each eigenvalue of $H^\delta(q_\parallel)$ is accompanied by another eigenvalue at an $O(\delta)$ distance. This behavior is a perturbative effect: at $\delta=0$, each eigenvalue of $H^0(q_\parallel)$ is exactly doubled due to the symmetry \eqref{eqn:app-extra-sym}. For small $\delta$, since the hopping coefficients vary slowly by an $O(\delta)$ perturbation\footnote{The $O(\delta)$ perturbation viewpoint does not apply to the quadratic deformation on the infinite lattice. Nevertheless, we numerically truncate the lattice at a large finite size with zero boundary conditions, and under this truncation the hopping coefficient differs from the $\delta=0$ case by $O(\delta)$.}, we have $\|H^\delta(q_\parallel) - H^0(q_\parallel)\| \lesssim \delta$
and the double eigenvalues generically split into two distinct eigenvalues separated by $O(\delta)$.

So far, we have explained how our choice of unit cell results in an apparent ``doubling'' of the eigenvalues of the infinite-lattice Hamiltonian $H^\delta(q_\parallel)$. We recall that our numerical scheme imposes a zero truncation on a large but finite piece and computes the eigenvalues of the truncated operator. For sufficiently large truncation size, the eigenvalues of the truncated problem approximate the discrete eigenvalues of $H^\delta(q_\parallel)$ with the same multiplicities. This explains the eigenvalue ``doubling'' in the computed band structures: each numerical band is accompanied by another band that lies extremely close to it (often on top of it).

\paragraph{Boundary modes at zero energy} We now turn to the numerical observation that the zero eigenvalue appears with multiplicity six. After factoring out the doubling effect, the zero eigenvalue has multiplicity three, and the two extra zero eigenmodes are in fact numerical artifacts of the zero truncation. We shall explain that the zero truncation generates a pair of boundary localized edge modes, one concentrated near the left boundary and the other near the right.

To find the edge modes among the zero eigenspace, we use an optimization scheme. We know that for a given small $q_\parallel$, the zero eigenspace for the zero-truncated version of $H^\delta(q_\parallel)$ has dimensional six. We name the zero eigenspace as $\mathcal{V}_6$. To find the edge modes in $\mathcal{V}_6$ that concentrate most on the left boundary, we use the following optimization problem
\begin{align}
    \max_{\psi \in \mathcal{V}_6} \quad \|P_\text{left} \psi\|_2^2,\label{eqn:max-left-bdry}
\end{align}
where $P_\text{left}$ denotes the projection onto a narrow boundary layer adjacent to the left edge. The optimizer of \eqref{eqn:max-left-bdry} is shown in the left panel of Figure \ref{fig:ac-bdry-modes}. 
We know that due to the doubling effect, there are two eigenmodes localized near the left boundary. To find the second left-edge mode, we maximize \eqref{eqn:max-left-bdry} over $\mathcal{V}_6$ after modulating out the optimizer of \eqref{eqn:max-left-bdry} (the second left-edge mode is very close to the first left-edge mode, so we omit its plot). By the same procedure, we obtain two right-edge modes and plot one in the right panel of Figure \ref{fig:ac-bdry-modes}. The remaining two modes in the zero eigenspace $\mathcal{V}_6$ are bulk modes, and we plot one of them in the middle panel of Figure \ref{fig:ac-bdry-modes}.

\begin{figure}[!htb]
	\centering
	\subfloat[]{
		\includegraphics[width=0.7\linewidth]{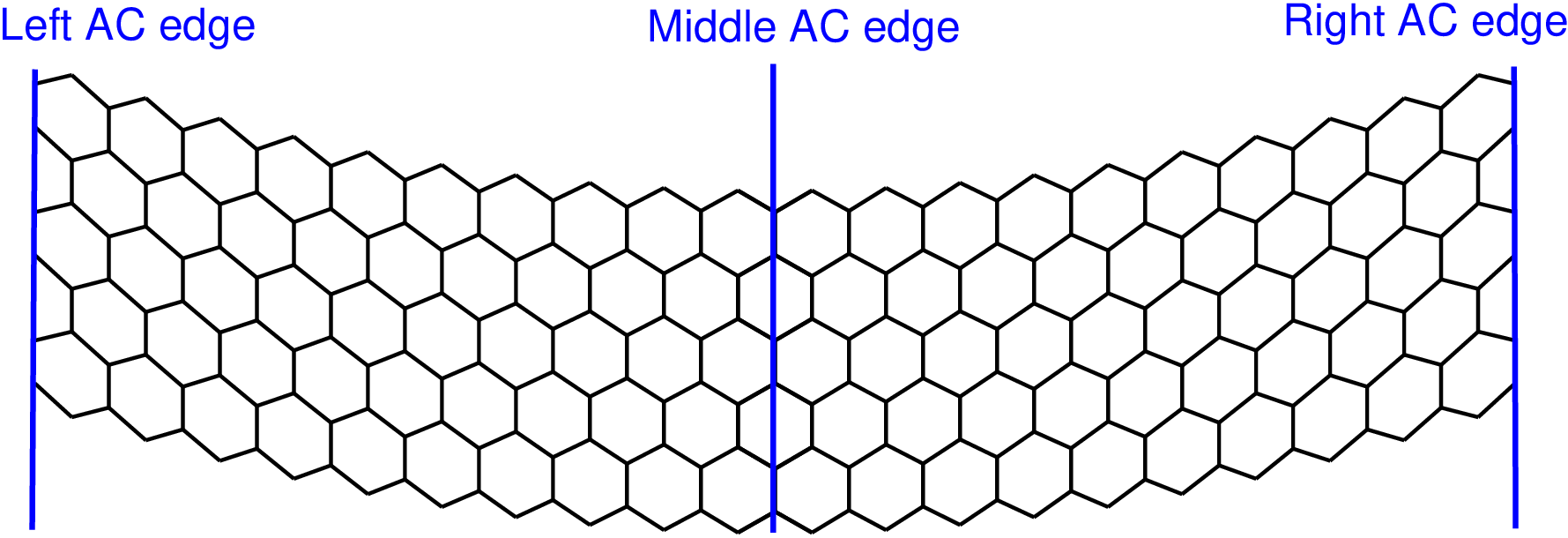}\label{fig:ac-bdry}
	}\\
	\subfloat[]{
		\includegraphics[width=0.7\linewidth]{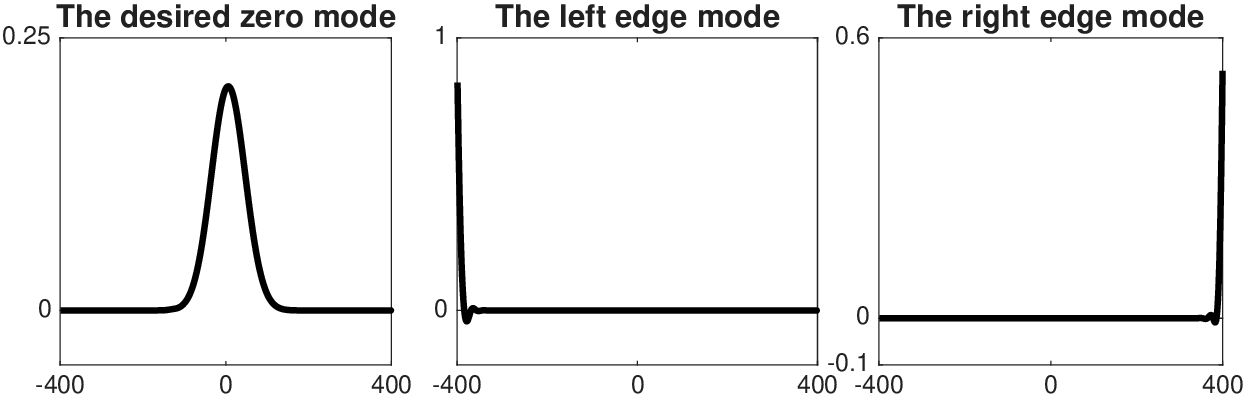}\label{fig:ac-bdry-modes}
	}
	\caption{Numerical artifacts in the zero eigenmodes for the quadratic deformation with $\delta = 0.04$ and $N_T = 400$: (a) the truncated honeycomb with left/right AC edges. (b) three computed zero modes: one bulk mode (physical) and two edge-modes that concentrate near boundaries (numerical artifacts).}
	\label{fig:numerical-artifacts}
\end{figure}


\setcounter{equation}{0}


\subsection{Derivation of the 2D effective envelope equations}\label{app:dev-2d-eff-eqn}

We derive the following hierarchy effective envelope equations for $\bm \Phi_j(\bm X)$ with $j \geq 0$ in \eqref{eqn:EXPAND}:
\begin{itemize}
    \item At order $\delta^0$: we obtain an algebraic equation \eqref{eqn:main-order-0}, i.e.
    \begin{align*}
        \widetilde{H}^0(\bm k) \begin{pmatrix}
            \Phi_0^A\big(\bm X \big)\\
            \Phi_0^B\big(\bm X \big)
        \end{pmatrix} = E_0 \begin{pmatrix}
            \Phi_0^A\big(\bm X \big)\\
            \Phi_0^B\big(\bm X \big)
        \end{pmatrix}
    \end{align*}
    where $\widetilde{H}^0(\bm k)$ is displayed in \eqref{eqn:tb-bloch-hamiltonian}.

    \item At order $\delta^1$: the equation \eqref{eqn:main-order-1} for $\bm \Phi_0(\bm X)$ and $E_1$ is 
    \begin{align*}
        \mathcal{H}_\text{eff}\begin{pmatrix}
            \Phi^A_0\big(\bm X \big)\\
            \Phi^B_0\big(\bm X \big)
        \end{pmatrix} = E_1\ \begin{pmatrix}
            \Phi^A_0\big(\bm X \big)\\
            \Phi^B_0\big(\bm X\big)
        \end{pmatrix}\ , 
    \end{align*}
    where $\mathcal{H}_\text{eff}$ is the effective Hamiltonian given in \eqref{eqn:eff-ham}.

    \item At order $\delta^2$: the equation \eqref{eqn:next-order} for $\bm \Phi_1(\bm X)$ and $E_2$ is
    \begin{equation*}
    	 \Big(\mathcal{H}_{\text{eff}} - E_1\Big) \bm{\Phi}_1 = E_2 \bm{\Phi}_0 + \bm{R}_2[\bm{\Phi}_0], \qquad \bm{R}_2[\bm{\Phi}_0] = (R_2^A, R_2^B)^T,
    \end{equation*}
    where $\bm{R}_2[\bm{\Phi}_0]$ is given in \eqref{eqn:2d-rem-terms}.
\end{itemize}

To obtain these equations, we substitute \eqref{eqn:main-exp-ansatz-2} into \eqref{eqn:multi-calc}. By expanding every term in \eqref{eqn:multi-calc} around $\bm X = \delta \Cell_{m,n}$ to order $\delta^2$, we obtain
\begin{subequations}\label{eqn:app-multi-calc}
    \begin{align}
        & \Big(\delta E_1 + \delta^2 E_2 \Big) \sum_{j=0}^2 \delta^j \Phi^A_j(\bm X)+ O(\delta^3) \label{eqn:app-multi-calc-1}\\
        =\:& \sum_{\nu = 1}^3 e^{i \bm{k} \cdot \bm w_\nu}  \sum_{j=0}^2 \delta^j \bigg(\Phi^B_j (\bm X) + \delta \nabla_{\bm{X}} \Phi^B_j (\bm X) \cdot \bm w_\nu + \frac{\delta^2}{2} \left \langle\bm w_\nu, \ \Big(D^2_{\bm X} \Phi^B_j (\bm X)\Big) \bm w_\nu \right \rangle \bigg) \nonumber\\
        +\:& \delta t_1 \sum_{\nu = 1}^3 e^{i \bm{k} \cdot \bm w_\nu} f_\nu(\bm X) \sum_{j=0}^2 \delta^j \bigg(\Phi^B_j (\bm X) + \delta \nabla_{\bm{X}} \Phi^B_j (\bm X) \cdot \bm w_\nu + \frac{\delta^2}{2} \left \langle\bm w_\nu, \ \Big(D^2_{\bm X} \Phi^B_j (\bm X)\Big) \bm w_\nu \right \rangle \bigg),\nonumber\\
        & \Big(\delta E_1 + \delta^2 E_2 \Big) \sum_{j=0}^2 \delta^j \Phi^B_j(\bm X)+ O(\delta^3) \label{eqn:app-multi-calc-2}\\
        =\:& \sum_{\nu = 1}^3 e^{-i \bm{k} \cdot \bm w_\nu}  \sum_{j=0}^2 \delta^j \bigg(\Phi^A_j(\bm X) - \delta \nabla_{\bm{X}} \Phi^A_j(\bm X) \cdot \bm w_\nu + \frac{\delta^2}{2} \left \langle\bm w_\nu, \Big(D^2_{\bm X} \Phi^A_j(\bm X)\Big) \bm w_\nu \right \rangle\bigg) \nonumber\\
        +\:& \delta t_1 \sum_{\nu = 1}^3 e^{-i \bm{k} \cdot \bm w_\nu} \bigg(f_\nu(\bm X) - \delta \nabla_{\bm{X}} f_\nu(\bm X) \cdot \bm w_\nu + \frac{\delta^2}{2} \left \langle\bm w_\nu, \Big(D^2_{\bm X} f_\nu(\bm X)\Big) \bm w_\nu \right \rangle\bigg) \nonumber\\
        &\times \sum_{j=0}^2 \delta^j \bigg(\Phi^A_j(\bm X) - \delta \nabla_{\bm{X}} \Phi^A_j(\bm X) \cdot \bm w_\nu + \frac{\delta^2}{2} \left \langle\bm w_\nu, \Big(D^2_{\bm X} \Phi^A_j(\bm X)\Big) \bm w_\nu \right \rangle\bigg) .\nonumber
    \end{align}
\end{subequations}

\paragraph{At order $\delta^0$:} By matching terms at order $\delta^0$, we obtain
\begin{align*}
    \sum_{\nu = 1}^3 e^{i \bm k \cdot \bm w_\nu} \Phi_0^B(\bm X)  = 0, \qquad \sum_{\nu = 1}^3 e^{-i \bm k \cdot \bm w_\nu} \Phi_0^A(\bm X) = 0,
\end{align*}
which is equivalent to its matrix-vector form
\begin{align}
    \widetilde{H}^0(\bm k) \ \bm \Phi_0(\bm X)  = 0, \label{eqn:app-2d-eff-order0}
\end{align}
where $\widetilde{H}^0(\bm k)$ is defined in \eqref{eqn:tb-bloch-hamiltonian}. Since $\widetilde{H}^0(\bm k)$ has zero eigenvalues only at $\bm k = \bm K$ or $\bm k = \bm K'$, we henceforth choose $\bm k = \bm K$. Notice that we cannot determine $\bm \Phi_0$ at this stage since $\widetilde{H}^0(\bm k) = 0$.

\paragraph{At order $\delta^1$:} By matching terms at order $\delta^1$, we obtain 
\begin{align*}
    E_1 \bm \Phi_0^A(\bm X) &= \sum_{\nu = 1}^3 e^{i \bm K \cdot \bm w_\nu} \Big(\nabla_{\bm{X}} \Phi_0^B(\bm X) \cdot \bm w_\nu + t_1 f_\nu(\bm X) \Phi_0^B(\bm X) \Big),\\
    E_1 \bm \Phi_0^B(\bm X) &= \sum_{\nu = 1}^3 e^{-i \bm K \cdot \bm w_\nu} \Big(-\nabla_{\bm{X}} \Phi_0^A(\bm X) \cdot \bm w_\nu + t_1 f_\nu(\bm X) \Phi_0^A(\bm X) \Big).
\end{align*}
By canceling out the terms related to $\bm \Phi_1$, we obtain an equivalent matrix-vector form 
\begin{align}
    \mathcal{H}_\text{eff} \ \bm \Phi_0(\bm X)= E_1 \ \bm \Phi_0(\bm X), \label{eqn:app-2d-eff-order1}
\end{align}
where $\mathcal{H}_\text{eff}$ is the same as the one in \eqref{eqn:eff-ham}.

\paragraph{At order $\delta^2$:} By matching all terms at order $\delta^2$, we obtain 
\begin{align*}
    &\Big(E_1 \Phi_1^A(\bm X) + E_2 \Phi_0^A(\bm X)\Big) \\
    =& \sum_{\nu = 1}^3 e^{i \bm K \cdot \bm w_\nu} \left(\frac{1}{2}  \left \langle\bm w_\nu, \ \Big(D^2_{\bm X} \Phi^B_0 (\bm X)\Big) \bm w_\nu \right \rangle + \nabla_{\bm X} \Phi^B_1(\bm X) \cdot \bm w_\nu + \Phi^B_2(\bm X)\right)\\
    +& t_1 \sum_{\nu = 1}^3 e^{i \bm K \cdot \bm w_\nu} f_\nu(\bm X) \Big(\nabla_{\bm X}\Phi_0^B(\bm X) \cdot \bm w_\nu + \Phi_1^B(\bm X)\Big),\\
    &\Big(E_1 \Phi_1^B(\bm X) + E_2 \Phi_0^B(\bm X)\Big) \\
    =& \sum_{\nu = 1}^3 e^{-i \bm K \cdot \bm w_\nu} \left(\frac{1}{2} \left \langle\bm w_\nu, \ \Big(D^2_{\bm X} \Phi^A_0 (\bm X)\Big) \bm w_\nu \right \rangle - \nabla_{\bm X} \Phi^A_1(\bm X) \cdot \bm w_\nu + \Phi^A_2(\bm X)\right)\bigg|_{\bm X = \delta \Cell_{m,n}}\\
    +& t_1 \sum_{\nu = 1}^3 e^{-i \bm K \cdot \bm w_\nu} \Big(-f_\nu(\bm X) \big(\nabla_{\bm X} \Phi_1^A(\bm X) \cdot \bm w_\nu\big) - \big(\nabla_{\bm X} f_\nu(\bm X) \cdot \bm w_\nu \big)\Phi_0^A(\bm X)  \Big),
\end{align*}
where the last line can be simplified to
\begin{align*}
    -t_1 \sum_{\nu = 1}^3 e^{-i \bm K \cdot \bm w_\nu} \Big(\nabla_{\bm X} \big(f_\nu(\bm X) \Phi_1^A(\bm X)\big) \cdot \bm w_\nu\big) \Big).
\end{align*}
By canceling out the terms related to $\bm \Phi_2$, we reorganize and obtain the equation for $\bm \Phi_1$
\begin{align}
    \Big(\mathcal{H}_{\text{eff}} - E_1\Big) \ \bm \Phi_1 (\bm X) = \bm E_2 \Phi_0 (\bm X) + \bm R_2[\Phi_0](\bm X) , \label{eqn:app-2d-eff-order2}
\end{align}
where $\bm R_2[\Phi_0] = \big(R_2^A[\bm \Phi_0], R_2^B[\bm \Phi_0]\big)^T$ is given by
\begin{subequations}\label{eqn:app-2d-R2}
    \begin{align}
        R_2^A[\bm \Phi_0] &= - \sum_{\nu = 1}^3 e^{i \bm K \cdot \bm w_\nu} \left(\frac{1}{2}  \left \langle\bm w_\nu, \ \Big(D^2_{\bm X} \Phi^B_0 (\bm X)\Big) \bm w_\nu \right \rangle + t_1 f_\nu(\bm X) \big(\nabla_{\bm X}\Phi_0^B(\bm X) \cdot \bm w_\nu \big) \right), \label{eqn:app-2d-R2-a}\\
        R_2^B[\bm \Phi_0] &= \sum_{\nu = 1}^3 e^{-i \bm K \cdot \bm w_\nu} \left(-\frac{1}{2} \left \langle\bm w_\nu, \ \Big(D^2_{\bm X} \Phi^A_0 (\bm X)\Big) \bm w_\nu \right \rangle + t_1 \nabla_{\bm X} \big(f_\nu(\bm X)\Phi_1^A(\bm X)\big) \cdot \bm w_\nu\right). \label{eqn:app-2d-R2-b}
    \end{align}
\end{subequations}
where $\bm R_2[\Phi_0]$ is the same as the one in \eqref{eqn:2d-rem-terms}. Thus, we complete our derivation for the effective envelope equations appeared in Section \ref{sec:derivation}.

\subsection{Derivation of the 1D effective envelope equations}\label{app:dev-1d-eff-eqn}

We derive the following 1D effective envelope equations for $\bm \Psi_0(X_1;k_\parallel), \bm \Psi_1(X_1;k_\parallel)$ in Section \ref{sec:uni-displacement}:
\begin{itemize}
    \item At order $\delta^1$: the equation for $\bm \Psi_0(X_1;k_\parallel)$ and $E_1(k_\parallel)$ is \eqref{eqn:1st-expansion}, i.e.
    \begin{align*}
        \mathcal{D}_{\bm A}(k_\parallel) \bm \Psi_0 = E_1(k_\parallel) \bm \Psi_0,
    \end{align*}
    where $\mathcal{D}_{\bm A}(k_\parallel)$ is given in \eqref{eqn:dirac-1d}.

    \item At order $\delta^2$: the equation for $\bm \Psi_1(X_1;k_\parallel)$ and $E_2(k_\parallel)$ is \eqref{eqn:2nd-expansion}, i.e.
    \begin{align*}
        \Big(\mathcal{D}_{\bm A}(k_\parallel) - E_1(k_\parallel)\Big) \bm \Psi_1 = E_2(k_\parallel) \ \bm \Psi_0 + \bm{\mathcal{R}}_2[\bm \Psi_0;k_\parallel],
    \end{align*}
    where $\bm{\mathcal{R}}_2[\bm \Psi_0;k_\parallel]$ are the remaining terms at $O(\delta^2)$ given in \eqref{eqn:app-rem-delta-2}.
\end{itemize}

To obtain these equations, we substitute the relation \eqref{eqn:phi-psi-relation} between $\bm \Psi_i(X_1;k_\parallel)$ and $\bm \Phi_i(\bm X)$ into the 2D evelope equations in Section \ref{sec:derivation}.

\paragraph{At order $\delta^1$:} We substitute \eqref{eqn:phi-psi-relation} into \eqref{eqn:main-order-0} and obtain
\begin{align*}
    &\mathcal{H}_{\text{eff}} \ \bm \Phi_0 = U \mathcal{D}_{\bm A} U^* \Big(U e^{i \frac{k_\parallel}{3} X_2} \bm \Psi_0\Big) = U \mathcal{D}_{\bm A} \Big(e^{i \frac{k_\parallel}{3} X_2} \bm \Psi_0\Big) = U e^{i \frac{k_\parallel}{3} X_2}\mathcal{D}_{\bm A}(k_\parallel) \bm \Psi_0,\\
    &E_1(k_\parallel) \ \bm \Phi_0 = E_1(k_\parallel) \ U e^{i \frac{k_\parallel}{3} X_2} \bm \Psi_0.
\end{align*}
Therefore \eqref{eqn:main-order-0} becomes $\mathcal{D}_{\bm A}(k_\parallel) \bm \Psi_0 = E_1(k_\parallel) \bm \Psi_0$, which is the same as \eqref{eqn:1st-expansion}.

\paragraph{At order $\delta^2$:} We substitute \eqref{eqn:phi-psi-relation} into \eqref{eqn:next-order} and \eqref{eqn:2d-rem-terms}. Following the same procedure, we obtain
\begin{align*}
    &\Big(\mathcal{D}_{\bm A}(k_\parallel) - E_1(k_\parallel)\Big) \bm \Psi_1 = E_2(k_\parallel) \ \bm \Psi_0 + \bm{\mathcal{R}}_2[\bm \Psi_0;k_\parallel],
\end{align*}
where $\bm{\mathcal{R}}_2[\bm \Psi_0;k_\parallel]$ is given by
\begin{align}
    \bm{\mathcal{R}}_2[\bm \Psi_0;k_\parallel] := U^* e^{-i \frac{k_\parallel}{3} X_2} \bm R_2[U e^{i \frac{k_\parallel}{3} X_2}\bm \Psi_0].\label{eqn:app-rem-psi}
\end{align}

We now simplify $\bm R_2[U e^{i \frac{k_\parallel}{3} X_2}\bm \Psi_0]$ in \eqref{eqn:2d-rem-terms}. Notice that for unidirectional displacement $\bm u = (0,d(X_1))^T$, $f_\nu(\bm X)$ in \eqref{eqn:f-u-relation} depends only on $X_1$ with
\begin{subequations}\label{eqn:ham-uni-dir}
    \begin{align}
        f_1(\bm X) &= \frac{\sqrt{3}}{4}\partial_{X_1} u_2 = \frac{\sqrt{3}}{4} d'(X_1), \label{eqn:ham-uni-dir-1}\\
        f_2(\bm X) &= -\frac{\sqrt{3}}{4}\partial_{X_1} u_2  = -\frac{\sqrt{3}}{4} d'(X_1), \label{eqn:ham-uni-dir-2} \qquad f_3(\bm X) = \partial_{X_2} u_2 = 0.
\end{align}
\end{subequations}
Therefore, we replace $f_1(\bm X)$ with $-f_2(\bm X)$ and write $\bm R_2[\bm \Phi_0]$ in \eqref{eqn:2d-rem-terms} as follows
\begin{subequations}\label{eqn:app-2d-rem-terms}
\begin{align}
    R_2^A[\bm \Phi_0] =& -\frac{ e^{i \frac{4\pi}{3}}}{2} \left \langle\bm w_2, \ \Big(D^2_{\bm X} \Phi_0^B(\bm X)\Big) \bm w_2 \right \rangle - \frac{e^{i \frac{2\pi}{3}}}{2}  \left \langle\bm w_3, \ \Big(D^2_{\bm X} \Phi_0^B(\bm X)\Big) \bm w_3 \right \rangle \nonumber \\
    &- t_1 e^{i \frac{4\pi}{3}} f_2(\bm{X}) \ \Big(\nabla_{\bm X} \Phi_0^B \cdot \bm w_2\Big), \label{eqn:app-2d-rem-terms-a}\\
	R_2^B[\bm \Phi_0] =& -\frac{ e^{-i \frac{4\pi}{3}}}{2} \left \langle\bm w_2, \ \Big(D^2_{\bm X} \Phi_0^A(\bm X)\Big) \bm w_2 \right \rangle - \frac{e^{-i \frac{2\pi}{3}}}{2}  \left \langle\bm w_3, \ \Big(D^2_{\bm X} \Phi_0^A(\bm X)\Big) \bm w_3 \right \rangle \nonumber\\
    &+ t_1 e^{-i \frac{4\pi}{3}} \nabla_{\bm X} \Big(f_2(\bm{X}) \Phi_0^A\Big) \cdot \bm w_2.\label{eqn:app-2d-rem-terms-b}
\end{align}
\end{subequations}

To further simplify $\bm{\mathcal{R}}_2[\Psi_0;k_\parallel]$ in \eqref{eqn:app-rem-psi} with $\bm \Phi_0(\bm X) = U e^{i\frac{k_\parallel}{3} X_2} \Psi_0(X_1;k_\parallel)$, we provide some useful calculations for a scalar function $G(\bm X) = e^{i \widetilde{k} X_2}g(X_1)$ with $\bm w = (w_1,w_2)^T$:
\begin{subequations}\label{eqn:app-calc-G}
    \begin{align}
        \left\langle\bm w, \ \Big(D^2_{\bm X} G(\bm X)\Big) \bm w \right\rangle &= w_1^2 \partial_{X_1}^2 G + 2w_1w_2 \partial_{X_1 X_2} G + w_2^2 \partial_{X_2}^2 G \nonumber\\
        &= e^{i \widetilde{k} X_2} \Big(w_1^2 \partial_{X_1}^2 g + 2w_1 w_2 (i \widetilde{k}) \partial_{X_1} g  - w_2^2 \widetilde{k}^2 g\Big),\label{eqn:app-calc-G1}\\
        \Big(\nabla_{\bm X} G(\bm X)\Big) \cdot \bm w &= e^{i \widetilde{k} X_2} \Big(w_1 \partial_{X_1} g + w_2 (i \widetilde{k}) g\Big).\label{eqn:app-calc-G2}
    \end{align}
\end{subequations}
Therefore, applying \eqref{eqn:app-calc-G} into \eqref{eqn:app-2d-R2} with $\widetilde{k} = \frac{k_\parallel}{3}$, we obtain
\begin{align*}
    R_2^A [U e^{i \frac{k_\parallel}{3} X_2} \bm \Psi_0]=&-e^{i\frac{\pi}{6}} e^{i \frac{k_\parallel}{3} X_2} \left[\frac{ e^{i \frac{4\pi}{3}}}{2} \Big( 3\partial_{X_1}^2 \Psi_0^B\Big) +\sqrt{3} t_1 e^{i\frac{4\pi}{3}} f_2(\bm X) \partial_{X_1}\Psi_0^B\right]\\
    &- e^{i\frac{\pi}{6}} e^{i \frac{k_\parallel}{3} X_2}\left[\frac{ e^{i \frac{2\pi}{3}}}{2} \Big(\frac{3}{4}\partial_{X_1}^2 \Psi_0^B + \frac{3\sqrt{3}}{2}\left(i\frac{k_\parallel}{3}\right)\partial_{X_1}\Psi_0^B - \frac{9}{4}\left(\frac{k_\parallel}{3}\right)^2\Psi_0^B\Big)\right],\\
    R_2^B [U e^{i \frac{k_\parallel}{3} X_2} \bm \Psi_0]=&-e^{-i\frac{\pi}{6}} e^{i \frac{k_\parallel}{3} X_2} \left[\frac{ e^{-i \frac{4\pi}{3}}}{2} \Big( 3\partial_{X_1}^2 \Psi_0^A\Big) - \sqrt{3} t_1 e^{-i\frac{4\pi}{3}} \partial_{X_1} \Big(f_2(\bm X)\ \Psi_0^A\Big)\right]\\
    &- e^{-i\frac{\pi}{6}} e^{i \frac{k_\parallel}{3} X_2}\left[\frac{ e^{-i \frac{2\pi}{3}}}{2} \Big(\frac{3}{4}\partial_{X_1}^2 \Psi_0^A + \frac{3\sqrt{3}}{2}\left(i\frac{k_\parallel}{3}\right)\partial_{X_1}\Psi_0^A - \frac{9}{4} \left(\frac{k_\parallel}{3}\right)^2\Psi_0^A\Big)\right].
\end{align*}
Therefore $\bm{\mathcal{R}}_2[\bm \Psi_0;k_\parallel] = U^* e^{-i \frac{k_\parallel}{3} X_2} \bm R_2[U e^{i\frac{k_\parallel}{3} X_2}\bm \Phi_0]$ can be simplified as
\begin{align*}
    &\mathcal{R}^A_2[\bm \Psi_0;k_\parallel] = e^{i\frac{\pi}{6}}e^{-i \frac{k_\parallel}{3} X_2} R_2^A [U e^{i \frac{k_\parallel}{3} X_2} \bm \Psi_0] \\
    =& \Big(-\frac{3}{2}e^{i\frac{5\pi}{3}} + \frac{3}{8}\Big) \partial_{X_1}^2 \Psi_0^B - \frac{1}{8}k_\parallel^2 \Psi_0^B + \Big(-\sqrt{3}t_1 e^{i\frac{5\pi}{3}} f_2(\bm X) + \frac{\sqrt{3}}{4} (ik_\parallel)\Big) \partial_{X_1} \Psi_0^B,\\
    & \mathcal{R}^B_2[\bm \Psi_0;k_\parallel] =  e^{-i\frac{\pi}{6}}e^{-i \frac{k_\parallel}{3} X_2} R_2^B [U e^{i \frac{k_\parallel}{3} X_2} \bm \Psi_0]\\
    =&\Big(-\frac{3}{2}e^{-i\frac{5\pi}{3}} + \frac{3}{8}\Big) \partial_{X_1}^2 \Psi_0^A - \frac{1}{8}k_\parallel^2 \Psi_0^A + \frac{\sqrt{3}}{4} (ik_\parallel) \partial_{X_1} \Psi_0^A + \sqrt{3}t_1 e^{-i\frac{5\pi}{3}} \partial_{X_1} \Big( f_2(\bm X) \Psi_0^A \Big),
\end{align*}
which are the same as the ones in \eqref{eqn:app-rem-delta-2}.

\subsection{The equation for $\widetilde{\eta}(k)$, the corrector}\label{app:eqn-eta-AB}

We present the detailed derivation of \eqref{eqn:eta-full} for the correctors $\widetilde{\bm \eta}(k)$. We first recall the eigenvalue problem \eqref{eqn:eig-def-hon}:
\begin{align*}
        \begin{pmatrix}
            \sum_{\nu=1}^3  \Big(1 + \delta t_1 f_\nu(\delta \Cell_{m,n}) \Big)\psi_{m+m_\nu,n+n_\nu}^B\\
            \sum_{\nu=1}^3  \Big(1 + \delta t_1 f_\nu(\delta \Cell_{m-m_\nu,n-n_\nu})\Big)\psi_{m-m_\nu,n-n_\nu}^A
        \end{pmatrix} &= E \begin{pmatrix}
            \psi_{m,n}^A\\
            \psi_{m,n}^B
        \end{pmatrix}.
\end{align*}
For simplicity, we denote
\begin{align}
    \ff_1(X_1) = \frac{\sqrt{3}}{4}d'(X_1), \qquad \ff_2(X_1) = -\frac{\sqrt{3}}{4} d'(X_1), \qquad \ff_3(X_1) = 0.\label{eqn:app-f-nu}
\end{align}
Therefore, we have $f_\nu(\delta \Cell_{m,n})= \ff_\nu\left(\frac{\sqrt{3}}{2}\delta m\right)$, where $\frac{\sqrt{3}}{2}m$ comes from the $x_1$-coordinate of $\Cell_{m,n}$ (see \eqref{eqn:ref-ab-honeycomb}). Then our eigenvalue problem becomes
\begin{align}
    \begin{pmatrix}
        \sum_{\nu=1}^3 \Big(1 + \delta t_1 \ff_\nu\left(\frac{\sqrt{3}}{2}\delta m\right) \Big)\psi_{m+m_\nu,n+n_\nu}^B\\
        \sum_{\nu=1}^3 \Big(1 + \delta t_1 \ff_\nu\left(\frac{\sqrt{3}}{2}\delta (m-m_\nu) \right)\Big)\psi_{m-m_\nu,n-n_\nu}^A
    \end{pmatrix} &= E \begin{pmatrix}
        \psi_{m,n}^A\\
        \psi_{m,n}^B
    \end{pmatrix}.\label{eqn:app-eig-pb}
\end{align}
We also recall the eigenstate and eigenvalue in \eqref{eqn:ansatz-1d}
\begin{align*}
    \bm \psi_{m,n} = \delta^{\frac{1}{2}} e^{i \bm K \cdot \bm x} e^{ik_\parallel X_2} U \Bigg(\bm \Psi_0(X_1) + \delta \bm \eta_m\Bigg)\Bigg|_{\substack{\bm x = \Cell_{m,n}\\ \bm X = \delta \Cell_{m,n}}}, \qquad E = \delta E_1 + \delta^2 \mu,
\end{align*}
where the $k_\parallel$-dependence in $\bm \Psi_0, E, E_1$ and $\mu$ are omitted for brevity. Substituting the explicit coordinates of $\Cell_{m,n}$ in \eqref{eqn:ref-ab-honeycomb}, the eigenstate \eqref{eqn:ansatz-1d} becomes 
\begin{align}
    \bm \psi_{m,n} = \delta^{\frac{1}{2}} e^{i \frac{2\pi m}{3}} e^{i k_\parallel \delta n} e^{-i k_\parallel \frac{1}{2} \delta m} U \left[\bm \Psi_0\left(\frac{\sqrt{3}}{2}\delta m\right) + \delta \bm \eta_m\right].\label{eqn:app-eta-eqn}
\end{align}

To derive the equation for $\bm \eta_m$, we substitute \eqref{eqn:app-eta-eqn} into the eigenvalue problem \eqref{eqn:app-eig-pb}. Since both sides of \eqref{eqn:app-eig-pb} have the factor $\delta^{\frac{1}{2}} e^{i \frac{2\pi m}{3}}  e^{i k_\parallel \delta n} e^{-i k_\parallel \frac{1}{2} \delta m}$, we cancel it out and obtain (we also multiple $e^{\pm i\frac{\pi}{6}}$ on both sides so that no phase factor appears in front of $E_1$)
\begin{subequations}\label{eqn:app-sim-eig}
    \begin{align}
        \bigg[\delta E_1 + \delta^2 \mu\bigg] &\Bigg[\Psi_0^A \left(\frac{\sqrt{3}}{2} \delta m\right) + \delta \eta_m^A \Bigg] =e^{i \frac{\pi}{3}} \Bigg[1 + \delta t_1 \ff_1 \left(\frac{\sqrt{3}}{2}\delta m\right)\Bigg] \Bigg[\Psi_0^B \left(\frac{\sqrt{3}}{2} \delta m\right) + \delta \eta_m^B\Bigg]\nonumber\\
         & + e^{i \frac{\pi}{6}} e^{i \frac{4\pi}{3}} \Bigg[1 + \delta t_1 \ff_2 \left(\frac{\sqrt{3}}{2}\delta m\right)\Bigg] \Bigg[\Psi_0^B \left(\frac{\sqrt{3}}{2} \delta (m+2)\right) + \delta \eta_{m+2}^B\Bigg] \nonumber\\
		&+ \:e^{i \frac{\pi}{6}} e^{i \frac{2\pi}{3}} e^{i \frac{1}{2} k_\parallel \delta} \Bigg[\Psi_0^B \left(\frac{\sqrt{3}}{2} \delta (m+1) \right) + \delta \eta_{m+1}^B\Bigg],\label{eqn:app-sim-eig-1a}\\
        \bigg[\delta E_1 + \delta^2 \mu\bigg] &\Bigg[\Psi_0^B \left(\frac{\sqrt{3}}{2} \delta m\right) + \delta \eta_m^B \Bigg] = e^{-i \frac{\pi}{3}} \Bigg[1 + \delta t_1 \ff_1 \left(\frac{\sqrt{3}}{2}\delta m\right)\Bigg] \Bigg[\Psi_0^A \left(\frac{\sqrt{3}}{2} \delta m\right) + \delta \eta_m^A\Bigg] \nonumber\\
		& + \: e^{-i \frac{\pi}{6}} e^{i \frac{4\pi}{3}} \Bigg[1 + \delta t_1 \ff_2 \left(\frac{\sqrt{3}}{2}\delta (m-2)\right)\Bigg] \Bigg[\Psi_0^A \left(\frac{\sqrt{3}}{2} \delta (m-2) \right) + \delta \eta_{m-2}^A\Bigg] \nonumber\\
		& +\:e^{-i \frac{\pi}{6}} e^{i \frac{2\pi}{3}} e^{i \frac{1}{2} k_\parallel \delta} \Bigg[\Psi_0^A \left(\frac{\sqrt{3}}{2} \delta (m-1) \right) + \delta \eta_{m-1}^A\Bigg].\label{eqn:app-sim-eig-1b}
    \end{align}
\end{subequations}

\paragraph{Derivation of \eqref{eqn:eta-full-1}} We first derive \eqref{eqn:eta-full-1} from the first equation in the eigenvalue problem \eqref{eqn:app-sim-eig-1a}. To seek the equation for $\widetilde{\bm\eta}(k)$, the scaled DFT \eqref{eqn:scale-dft} of $\bm \eta_m$, we multiply $e^{-ik\frac{\sqrt{3}}{2}m}$ on both sides of \eqref{eqn:app-sim-eig} and sum over $m$. We observe that for terms like $\bm \Psi_0\left(\frac{\sqrt{3}}{2} \delta m\right)$ and $\ff_\nu \left(\frac{\sqrt{3}}{2} \delta m\right) \bm \Psi_0\left(\frac{\sqrt{3}}{2} \delta m\right)$ with $\nu=1,2$, their scaled DFT simplify via the scaled Poisson summation formula presented in \eqref{eqn:psf-scaled} and \eqref{eqn:psf-scaled-1} to simplify them. Therefore \eqref{eqn:app-sim-eig-1a} becomes
\begin{align}
    & \big(\delta E_1 + \delta^2 \mu\big) \Bigg[\frac{2}{\sqrt{3} \delta}\sum_{m \in \mathbb{Z}} \widehat{\Psi_0^A}\left(\frac{k + \frac{4\pi}{\sqrt{3}}m}{\delta}\right) + \delta \widetilde{\eta}^A(k) \Bigg] \label{eqn:app-sim-eig-2}\\
    =\: &  e^{i \frac{\pi}{3}}
    \Big(1 + e^{i \frac{4\pi}{3}} e^{i \sqrt{3} k} + e^{i \frac{2\pi}{3}} e^{i \frac{1}{2}  \delta k_\parallel} e^{i \frac{\sqrt{3}}{2} k}\Big) \Bigg[ \frac{2}{\sqrt{3} \delta}\sum_{m \in \mathbb{Z}} \widehat{\Psi_0^B}\left(\frac{k + \frac{4\pi}{\sqrt{3}}m}{\delta}\right) + \delta \widetilde{\eta}^B(k)\Bigg] \nonumber\\
    +\: & \frac{2}{\sqrt{3}} e^{i \frac{\pi}{3}} t_1 \sum_{m \in \mathbb{Z}} \widehat{\ff_1\Psi_0^B}\left(\frac{k + \frac{4\pi}{\sqrt{3}}m}{\delta}\right) + \frac{2}{\sqrt{3}} e^{i \frac{\pi}{3}} e^{i \frac{4\pi}{3}} e^{i \sqrt{3} k} t_1 \sum_{m \in \mathbb{Z}} \widehat{\ff_2\Psi_0^B}\left(\frac{k + \frac{4\pi}{\sqrt{3}}m}{\delta}\right) \nonumber\\
    +\: & \delta^2 e^{i \frac{\pi}{3}} t_1 \sum_{m \in \mathbb{Z}} \ff_1\left(\frac{\sqrt{3}}{2}\delta m\right) \eta_m^B e^{-ik\frac{\sqrt{3}}{2}m} + \delta^2 e^{i \frac{\pi}{3}} e^{i \frac{4\pi}{3}} t_1 \sum_{m \in \mathbb{Z}} \ff_2\left(\frac{\sqrt{3}}{2}\delta m\right) \eta_{m+2}^B e^{-ik\frac{\sqrt{3}}{2}m} \nonumber\\
    +\: & \delta e^{i \frac{\pi}{3}} e^{i \frac{4\pi}{3}} t_1 \sum_{m \in \mathbb{Z}} \Bigg(\ff_2\left(\frac{\sqrt{3}}{2}\delta m\right) - \ff_2\left(\frac{\sqrt{3}}{2}\delta (m+2)\right)\Bigg) \Psi_0^B\left(\frac{\sqrt{3}}{2} \delta (m+2) \right)e^{-ik\frac{\sqrt{3}}{2}m}, \nonumber
\end{align}
where the last line of \eqref{eqn:app-sim-eig-2} comes from the mismatch of $\ff_2$ and $\Psi_0^B$ on $\frac{\sqrt{3}}{2} \delta m$ and $\frac{\sqrt{3}}{2} \delta (m+2)$ (see the second line of \eqref{eqn:app-sim-eig-1a}). Notice that the last line of \eqref{eqn:app-sim-eig-2} is of order $O(\delta^2)$ when $\ff_2' \in L^\infty(\mathbb{R})$.

\underline{\textit{Canceling leading order $\bm \Psi_0$ terms}:} Before we further simplify \eqref{eqn:app-sim-eig-2}, we provide a useful observation on the Fourier modes of $\bm \Psi_0$ with $|m|\geq 1$ in the sum -- they are arbitrarily small since $\bm \Psi_0$ decay fast. Let us first separate the Fourier modes related to $\bm \Psi_0$ in \eqref{eqn:app-sim-eig-2} as follows
\begin{equation*}
    \sum_{m \in \mathbb{Z}} \widehat{\bm \Psi_0}\left(\frac{k + \frac{4\pi}{\sqrt{3}}m}{\delta}\right) =\widehat{\bm \Psi_0}\left(\frac{k}{\delta}\right) + \sum_{|m| \geq 1} \widehat{\bm \Psi_0}\left(\frac{k + \frac{4\pi}{\sqrt{3}}m}{\delta}\right).
\end{equation*}
Since $\bm \Psi_0$ has exponential decay and $\bm \Psi_0\in H^s(\mathbb{R})$ for any $s \geq 1$ (see Lemma \ref{lemma:1d}), we have
\begin{equation}\label{eqn:sum-M-ge1-bd}
    \left\| \sum_{|m| \geq 1} \widehat{\bm \Psi_0}\left(\frac{k + \frac{4\pi}{\sqrt{3}}m}{\delta}\right) \right\|_{L^2\left(\left[-\frac{2\pi}{\sqrt{3}}, \frac{2\pi}{\sqrt{3}}\right]\right)}^2
\lesssim
\delta^{(2s+1)}\|\bm \Psi_0\|_{H^s(\mathbb{R})}^2, \quad \forall \: s \geq 1.
\end{equation}
The proof of \eqref{eqn:sum-M-ge1-bd} is presented in Appendix \ref{app:some-bds} (see \eqref{eqn:app-sum-M-bd}). Similarly, analogous bounds hold for $\ff_1 \bm \Psi_0$ and $\ff_2 \bm \Psi_0$, since $\ff_1, \ff_2$ and all their derivatives are $L^\infty(\mathbb{R})$.

We now simplify \eqref{eqn:app-sim-eig-2} by canceling the leading order terms related to $\bm \Psi_0$. We observe that the leading order contributions in \eqref{eqn:app-sim-eig-2} are $O(\delta^{-1})$, i.e. the terms appearing in the second line of \eqref{eqn:app-sim-eig-2} with $m=0$). We claim that both the $O(\delta^{-1})$ and $O(1)$ terms involving $\bm\Psi_0$ vanish since $\bm\Psi_0$ satisfies the effective equation \eqref{eqn:1st-expansion}. We formalize this claim in the following proposition:
\begin{proposition}\label{prop:app-cancel-psi0}
    Assume that $\bm \Psi_0$ satisfies \eqref{eqn:1st-expansion}. The $O(\delta^{-1})$ and $O(1)$ terms in \eqref{eqn:app-sim-eig-2} vanish, i.e.
    \begin{align}
        O(\delta)=&-\frac{2}{\sqrt{3}} E_1 \widehat{\Psi_0^A}\left(\frac{k}{\delta}\right) + e^{i \frac{\pi}{3}}
    \Big(1 + e^{i \frac{4\pi}{3}} e^{i \sqrt{3} k} + e^{i \frac{2\pi}{3}} e^{i \frac{1}{2}  \delta k_\parallel} e^{i \frac{\sqrt{3}}{2} k}\Big) \frac{2}{\sqrt{3} \delta}\widehat{\Psi_0^B}\left(\frac{k}{\delta}\right) \label{eqn:app-cancel-psi0}\\
    & \ + \frac{2}{\sqrt{3}} e^{i \frac{\pi}{3}} t_1 \widehat{\ff_1\Psi_0^B}\left(\frac{k}{\delta}\right) + \frac{2}{\sqrt{3}} e^{i \frac{\pi}{3}} e^{i \frac{4\pi}{3}} e^{i \sqrt{3} k} t_1 \widehat{\ff_2\Psi_0^B}\left(\frac{k}{\delta}\right),\nonumber
    \end{align}
    for sufficiently small $\delta $ and $k \in \Icell$.
\end{proposition}

\begin{proof}
Since $k/\delta$ appears frequently in \eqref{eqn:app-cancel-psi0}, we simplify \eqref{eqn:app-cancel-psi0} by taking $k = \delta \xi$. Then we simplify the $O(\delta^{-1})$ term in \eqref{eqn:app-cancel-psi0} by the expansions of $e^{i \sqrt{3} k}$ and $e^{i \frac{\sqrt{3}}{2}k}$
\begin{equation}
    e^{i \sqrt{3} k} \approx 1 + i \sqrt{3}\delta \xi - \frac{3}{2} \delta^2 \xi^2 + \dots, \qquad
        e^{i \frac{\sqrt{3}}{2} k} \approx 1 + i \frac{\sqrt{3}}{2} \delta \xi - \frac{3}{8} \delta^2 \xi^2 + \dots,\label{eqn:app-expo-exp}
\end{equation}
and the $O(\delta^{-1})$ term becomes $O(1)$ with
\begin{align}
    & e^{i \frac{\pi}{3}}
    \Big(1 + e^{i \frac{4\pi}{3}} e^{i \sqrt{3} k} + e^{i \frac{2\pi}{3}} e^{i \frac{1}{2} \delta k_\parallel} e^{i \frac{\sqrt{3}}{2} k}\Big) \frac{2}{\sqrt{3} \delta}\widehat{\Psi_0^B}\left(\frac{k}{\delta}\right) \label{eqn:sim-leading}\\
    =\:& \frac{2}{\sqrt{3}\delta} \Bigg[e^{i \frac{\pi}{3}} +e^{i \frac{5\pi}{3}} \left(1 + i \sqrt{3}\delta \xi + \dots\right) - \left(1+i\frac{1}{2}\delta k_\parallel + \dots\right) \left(1 + i \frac{\sqrt{3}}{2} \delta \xi + \dots\right)\Bigg] \widehat{\Psi_0^B}(\xi) \nonumber\\
    = \: & \frac{2}{\sqrt{3} \delta} \Bigg[ e^{i \frac{5\pi}{3}} i \sqrt{3} \delta \xi- i \frac{1}{2} \delta k_\parallel - i \frac{\sqrt{3}}{2} \delta \xi \Bigg] \widehat{\Psi_0^B}(\xi) + O(\delta) = \frac{2}{\sqrt{3}} \Bigg[ \frac{3}{2} \xi- i \frac{1}{2} k_\parallel \Bigg] \widehat{\Psi_0^B}(\xi) + O(\delta).\nonumber
\end{align}
The $O(\delta^{-1})$ term vanishes due to the equality $-1 + e^{i \frac{\pi}{3}} + e^{i \frac{5\pi}{3}} = 0$. 

Then we further simplify the $O(1)$ contributions in \eqref{eqn:app-cancel-psi0}. We notice that in addition to the $O(1)$ term in \eqref{eqn:sim-leading}, there are three extra $O(1)$ terms in \eqref{eqn:app-sim-eig-2}. In fact, using $\ff_1 = -\ff_2 = \frac{\sqrt{3}}{4} d'$ in \eqref{eqn:app-f-nu}, these four $O(1)$ contributions cancel since $\bm \Psi_0$ is the zero eigenfunction of $\mathcal{D}  - E_1$, i.e.
\begin{align}
    &-E_1 \frac{2}{\sqrt{3}} \widehat{\Psi_0^A}(\xi) + \frac{2}{\sqrt{3}} \Big( \frac{3}{2} \xi- i \frac{1}{2} k_\parallel \Big) \widehat{\Psi_0^B}(\xi) + \frac{2}{\sqrt{3}} e^{i \frac{\pi}{3}} t_1 \widehat{\ff_1\Psi_0^B}(\xi) +  \frac{2}{\sqrt{3}} e^{i \frac{5\pi}{3}} t_1 \widehat{\ff_2\Psi_0^B}(\xi)\nonumber\\
    =\:& \frac{2}{\sqrt{3}} \Bigg[-E_1 \widehat{\Psi_0^A} + \Big( \frac{3}{2} \xi- i \frac{1}{2} k_\parallel \Big) \widehat{\Psi_0^B} + \frac{\sqrt{3}}{4} t_1 \Big(e^{i \frac{\pi}{3}} - e^{i \frac{5\pi}{3}}\Big) \widehat{d'\Psi_0^B}\Bigg]\nonumber\\
    =\:& \frac{2}{\sqrt{3}} \Bigg[-E_1 \widehat{\Psi_0^A} + \Big( \frac{3}{2} \xi- i \frac{1}{2} k_\parallel \Big) \widehat{\Psi_0^B} + i \frac{3}{4} t_1 \widehat{d'\Psi_0^B}\Bigg] \nonumber\\
    = \: & \frac{2}{\sqrt{3}} \Bigg[-E_1 \Psi_0^A + \frac{3}{2}\Big(-i \partial_{X_1} - i \frac{k_\parallel}{3} + i\frac{t_1}{2} d'(X_1)\Big)\Psi_0^B\Bigg]^\wedge \qquad (\text{apply } \kappa(X_1) = \frac{k_\parallel}{3}- \frac{t_1}{2} d'(X_1))\nonumber\\
    = \: & \frac{2}{\sqrt{3}} \Bigg[-E_1 \Psi_0^A + \frac{3}{2}\Big(-i \partial_{X_1} - i\kappa(X_1) \Big)\Psi_0^B\Bigg]^\wedge = 0.\label{eqn:app-psi-cancel}
\end{align}
Thus, the $O(\delta^{-1})$ and $O(1)$ terms in \eqref{eqn:app-cancel-psi0} vanish, and we complete our proof of \eqref{eqn:app-cancel-psi0}.
\end{proof}

\underline{\textit{After canceling the leading-order $\bm \Psi_0$ terms:}} Using Proposition \ref{prop:app-cancel-psi0}, we cancel the $O(\delta^{-1})$ and $O(1)$ terms in \eqref{eqn:app-sim-eig-2} and obtain
\begin{align}
    &\delta^2 E_1 \widetilde{\eta}^A(k) + \delta^3 \mu \widetilde{\eta}^A(k) + \frac{2}{\sqrt{3}} E_1 \sum_{|m| \geq 1} \widehat{\Psi_0^A}\left(\frac{k + \frac{4\pi}{\sqrt{3}}m}{\delta}\right) + \frac{2}{\sqrt{3}} \delta \mu \sum_{m \in \mathbb{Z}} \widehat{\Psi_0^A}\left(\frac{k + \frac{4\pi}{\sqrt{3}}m}{\delta}\right) \label{eqn:app-eta-A}\\
        =\: & \delta e^{i \frac{\pi}{3}}
        \Big(1 + e^{i \frac{4\pi}{3}} e^{i \sqrt{3} k} + e^{i \frac{2\pi}{3}} e^{i \frac{1}{2}  \delta k_\parallel} e^{i \frac{\sqrt{3}}{2} k}\Big)  \widetilde{\eta}^B(k) \nonumber\\
        +\: & \delta^2 e^{i \frac{\pi}{3}} t_1 \sum_{m \in \mathbb{Z}} \ff_1\left(\frac{\sqrt{3}}{2}\delta m\right) \eta_m^B e^{-ik\frac{\sqrt{3}}{2}m} + \delta^2 e^{i \frac{\pi}{3}} e^{i \frac{4\pi}{3}} t_1 \sum_{m \in \mathbb{Z}} \ff_2\left(\frac{\sqrt{3}}{2}\delta m\right) \eta_{m+2}^B e^{-ik\frac{\sqrt{3}}{2}m} \nonumber\\
        + \: & \frac{2}{\sqrt{3} \delta} e^{i \frac{\pi}{3}} 
        \Bigg[e^{i \frac{4\pi}{3}} \left(e^{i \sqrt{3} k}-1- i \sqrt{3}k\right) + e^{i \frac{2\pi}{3}} \left(e^{i \frac{1}{2}  \delta k_\parallel} e^{i \frac{\sqrt{3}}{2} k} -1 - i\frac{1}{2}\delta k_\parallel - i \frac{\sqrt{3}}{2}k\right)\Bigg] \widehat{\Psi_0^B}\left(\frac{k}{\delta}\right) \nonumber\\
        + \: &\frac{2}{\sqrt{3} \delta} e^{i \frac{\pi}{3}}
        \Big(1 + e^{i \frac{4\pi}{3}} e^{i \sqrt{3} k} + e^{i \frac{2\pi}{3}} e^{i \frac{1}{2}  \delta k_\parallel} e^{i \frac{\sqrt{3}}{2} k}\Big) \sum_{|m| \geq 1} \widehat{\Psi_0^B}\left(\frac{k + \frac{4\pi}{\sqrt{3}}m}{\delta}\right) \nonumber\\
        +\: & \frac{2}{\sqrt{3}} e^{i \frac{\pi}{3}} t_1 \sum_{|m| \geq 1} \widehat{\ff_1\Psi_0^B}\left(\frac{k + \frac{4\pi}{\sqrt{3}}m}{\delta}\right) + \frac{2}{\sqrt{3}} e^{i \frac{\pi}{3}} e^{i \frac{4\pi}{3}} t_1 \sum_{|m| \geq 1} \widehat{\ff_2\Psi_0^B}\left(\frac{k + \frac{4\pi}{\sqrt{3}}m}{\delta}\right) \nonumber\\
        +\:& \delta e^{i \frac{\pi}{3}} e^{i \frac{4\pi}{3}} t_1 \sum_{m \in \mathbb{Z}} \Bigg(\ff_2\left(\frac{\sqrt{3}}{2}\delta m\right) - \ff_2\left(\frac{\sqrt{3}}{2}\delta (m+2)\right)\Bigg) \Psi_0^B\left(\frac{\sqrt{3}}{2} \delta (m+2) \right)e^{-ik\frac{\sqrt{3}}{2}m}.\nonumber
\end{align}
Notice that the leading order terms in \eqref{eqn:app-eta-A} are $O(\delta)$. To simplify \eqref{eqn:app-eta-A}, we denote
\begin{subequations}
    \begin{align}
        & \Gamma_1[k,k_\parallel,\delta] := e^{i \frac{\pi}{3}}
        \Big(1 + e^{i \frac{4\pi}{3}} e^{i \sqrt{3} k} + e^{i \frac{2\pi}{3}} e^{\frac{i}{2}  \delta k_\parallel} e^{i \frac{\sqrt{3}}{2} k}\Big), \label{eqn:app-gamma-1}\\
        &\Gamma_3[k,k_\parallel,\delta] := \Gamma_1[k,k_\parallel,\delta] - k\partial_k\Gamma_1[0,k_\parallel,\delta] \label{eqn:gamma-3}\\
    =\:&e^{i \frac{\pi}{3}} 
        \Bigg[e^{i \frac{4\pi}{3}} \left(e^{i \sqrt{3} k}-1- i \sqrt{3}k\right) + e^{i \frac{2\pi}{3}} \left(e^{ \frac{i}{2}  \delta k_\parallel} e^{i \frac{\sqrt{3}}{2} k} -1 - \frac{i}{2}\delta k_\parallel - i \frac{\sqrt{3}}{2}k\right)\Bigg].\nonumber
    \end{align}
\end{subequations}
Then we obtain \eqref{eqn:eta-full-1} by grouping all terms related to $\bm \Psi_0$ together and dividing $\delta$ on both sides, i.e.
\begin{align}
    &- \Gamma_1[k,k_\parallel,\delta] \widetilde{\eta}^B(k) + \delta E_1 \widetilde{\eta}^A(k) + \delta^2 \mu \widetilde{\eta}^A(k) + \delta \widetilde{F}_1[k;\widetilde{\bm \eta}] = \widetilde{I}_1[k;\bm \Psi_0, \mu, \delta], \label{eqn:app-eta-full-1}
\end{align}
where $\widetilde{F}_1[k;\widetilde{\bm \eta}]$ and $\widetilde{I}_1[k;\bm \Psi_0, \mu, \delta]$ are defined as
\begin{align}
    &\widetilde{F}_1[k;\widetilde{\bm \eta}] = -e^{i \frac{\pi}{3}} t_1\sum_{m \in \mathbb{Z}} \ff_1\left(\frac{\sqrt{3}}{2}\delta m\right) \eta_m^B e^{-ik\frac{\sqrt{3}}{2}m} -  e^{i \frac{5\pi}{3}} t_1 \sum_{m \in \mathbb{Z}} \ff_2\left(\frac{\sqrt{3}}{2}\delta m\right) \eta_{m+2}^B e^{-ik\frac{\sqrt{3}}{2}m},\label{eqn:app-F1}\\
    &\widetilde{I}_1[k;\bm \Psi_0, \mu, \delta] :=  -\frac{2}{\sqrt{3}} \mu \sum_{m \in \mathbb{Z}} \widehat{\Psi_0^A}\left(\frac{k + \frac{4\pi}{\sqrt{3}}m}{\delta}\right) + I_{1,\text{ind}} [k;\bm \Psi_0,\mu, \delta] \label{eqn:I1-full}\\
     &I_{1,\text{ind}} [k;\bm \Psi_0, \mu, \delta]:= - \frac{2}{\sqrt{3}\delta} E_1 \sum_{|m| \geq 1} \widehat{\Psi_0^A}\left(\frac{k + \frac{4\pi}{\sqrt{3}}m}{\delta}\right) \nonumber\\
     & \quad + \frac{2}{\sqrt{3} \delta^2} \Gamma_3[k,k_\parallel,\delta] \widehat{\Psi_0^B}\left(\frac{k}{\delta}\right) + \frac{2}{\sqrt{3} \delta^2} \Gamma_1[k,k_\parallel,\delta] \sum_{|m| \geq 1} \widehat{\Psi_0^B}\left(\frac{k + \frac{4\pi}{\sqrt{3}}m}{\delta}\right) \nonumber\\
    & \quad + \frac{2}{\sqrt{3}\delta} e^{i \frac{\pi}{3}} t_1 \sum_{|m| \geq 1} \widehat{\ff_1\Psi_0^B}\left(\frac{k + \frac{4\pi}{\sqrt{3}}m}{\delta}\right) + \frac{2}{\sqrt{3}\delta} e^{i \frac{5\pi}{3}} t_1 \sum_{|m| \geq 1} \widehat{\ff_2\Psi_0^B}\left(\frac{k + \frac{4\pi}{\sqrt{3}}m}{\delta}\right) \nonumber\\
    & \quad + e^{i \frac{5\pi}{3}} t_1 \sum_{m \in \mathbb{Z}} \Bigg(\ff_2\left(\frac{\sqrt{3}}{2}\delta m\right)- \ff_2\left(\frac{\sqrt{3}}{2}\delta (m+2)\right)\Bigg) \Psi_0^B\left(\frac{\sqrt{3}}{2} \delta (m+2) \right)e^{-ik\frac{\sqrt{3}}{2}m}. \nonumber
\end{align}
Thus, we complete our derivation of \eqref{eqn:eta-full-1}.

\paragraph{Derivation of \eqref{eqn:eta-full-2}} Following a similar procedure, we derive \eqref{eqn:eta-full-2} for $\widetilde{\bm \eta}(k)$ from the second line of the eigenvalue problem \eqref{eqn:app-sim-eig-1b}. Here we omit the details and present the final result
\begin{align}
    & - \Gamma_2(k,k_\parallel,\delta) \ \widetilde{\eta}^A(k) + \delta E_1 \widetilde{\eta}^B(k) + \delta^2 \mu \widetilde{\eta}^B(k) + \delta \widetilde{F}_2[k;\widetilde{\bm \eta}] = \widetilde{I}_2[k;\bm \Psi_0, \mu, \delta](k),
\end{align}
where the constant $\Gamma_2(k,k_\parallel,\delta)$ is given by
\begin{align}
    \Gamma_2(k,k_\parallel,\delta) = e^{-i \frac{\pi}{3}}
        \Big(1 + e^{-i \frac{4\pi}{3}} e^{-i \sqrt{3} k} + e^{-i \frac{2\pi}{3}} e^{- \frac{i}{2} \delta k_\parallel} e^{-i \frac{\sqrt{3}}{2} k}\Big).
\end{align}
The terms $\widetilde{F}_2[k;\widetilde{\bm \eta}]$ and $\widetilde{I}_2[k; \bm \Psi_0, \mu, \delta]$ are defined as
\begin{align}
    & \widetilde{F}_2[k;\widetilde{\bm \eta}] = - e^{-i \frac{\pi}{3}} t_1 \sum_{m \in \mathbb{Z}} \ff_1\left(\frac{\sqrt{3}}{2}\delta m\right) \eta_m^B e^{-ik\frac{\sqrt{3}}{2}m}\label{eqn:app-F2}\\
    &\qquad \qquad \ - e^{-i \frac{5\pi}{3}} t_1 \sum_{m \in \mathbb{Z}} \ff_2\left(\frac{\sqrt{3}}{2}\delta (m-2)\right) \eta_{m-2}^B e^{-ik\frac{\sqrt{3}}{2}m}\nonumber \\
     &\widetilde{I}_2[k;\bm \Psi_0, \mu, \delta] := -\frac{2}{\sqrt{3}} \mu \sum_{m \in \mathbb{Z}} \widehat{\Psi_0^B}\left(\frac{k + \frac{4\pi}{\sqrt{3}}m}{\delta}\right) + \widetilde{I}_{2,\text{ind}}[k; \bm \Psi_0, \mu, \delta]\label{eqn:I2-full}\\
    &\widetilde{I}_{2,\text{ind}}[k;\bm \Psi_0, \mu, \delta] := - \frac{2}{\sqrt{3}\delta} E_1 \sum_{|m| \geq 1} \widehat{\Psi_0^B}\left(\frac{k + \frac{4\pi}{\sqrt{3}}m}{\delta}\right) \nonumber\\
    & \quad + \frac{2}{\sqrt{3} \delta^2} \Gamma_4(k,k_\parallel,\delta) \widehat{\Psi_0^A}\left(\frac{k}{\delta}\right) + \frac{2}{\sqrt{3} \delta^2} \Gamma_2(k,k_\parallel, \delta) \sum_{|m| \geq 1} \widehat{\Psi_0^A}\left(\frac{k + \frac{4\pi}{\sqrt{3}}m}{\delta}\right) \nonumber\\
    & \quad + \frac{2}{\sqrt{3}\delta} e^{-i \frac{\pi}{3}} t_1 \sum_{|m| \geq 1} \widehat{\ff_1\Psi_0^A}\left(\frac{k + \frac{4\pi}{\sqrt{3}}m}{\delta}\right) + \frac{2}{\sqrt{3}\delta} e^{-i \frac{5\pi}{3}} t_1 \sum_{|m| \geq 1} \widehat{\ff_2\Psi_0^A}\left(\frac{k + \frac{4\pi}{\sqrt{3}}m}{\delta}\right) \nonumber\\
    & \quad + e^{-i \frac{5\pi}{3}} t_1 \Big(e^{-ik\sqrt{3}} - 1\Big) \widehat{\ff_2 \Psi_0^A}\left(\frac{k}{\delta}\right),  \nonumber
\end{align}
where $\Gamma_4(k,k_\parallel,\delta)$ is given by
\begin{align}
    &\Gamma_4(k,k_\parallel,\delta) := \Gamma_2(k,k_\parallel,\delta) - k\partial_k\Gamma_2(0,k_\parallel,\delta) \label{eqn:gamma-4}\\
    =\:&e^{-i \frac{\pi}{3}} 
        \Bigg[e^{-i \frac{4\pi}{3}} \left(e^{-i \sqrt{3} k}-1+ i \sqrt{3}k\right) + e^{-i \frac{2\pi}{3}} \left(e^{-\frac{i}{2}  \delta k_\parallel} e^{-i \frac{\sqrt{3}}{2} k} -1 + \frac{i}{2}\delta k_\parallel + i \frac{\sqrt{3}}{2}k\right)\Bigg].\nonumber
\end{align}
Thus, we complete our derivation of \eqref{eqn:eta-full-2}.

\subsection{The derivation for the near-momentum equation $\widehat{\bm \beta}_\text{near}$}\label{app:dev-near}

In this appendix, we provide the detailed derivation of \eqref{eqn:beta-system} by substituting \eqref{eqn:near-hat} and \eqref{eqn:eta-idft} into \eqref{eqn:eta-near}. We start with \eqref{eqn:eta-near} and change $k=\delta \xi$
\begin{subequations}\label{eqn:app-near-full}
\begin{align}
    &\delta^{-1} \Gamma_1(\delta \xi,k_\parallel,\delta) \: \widehat{\beta}^B_\text{near}(\xi) - E_1 \widehat{\beta}^A_\text{near}(\xi) - \delta \mu \widehat{\beta}^A_\text{near}(\xi) \label{eqn:app-near-full-1}\\
    =\ &\delta \bigg(\widetilde{F}_1[\delta \xi;\bm{\widetilde{\eta}}_\text{near}] + \widetilde{F}_1[\delta \xi;\bm{\widetilde{\eta}}_\text{far}]\bigg) \: \chi(|\xi|\leq \delta^{\tau-1}) -\widetilde{I}_1[\delta \xi;\bm \Psi_0, \mu, \delta]\:\chi(|\xi|\leq \delta^{\tau-1}), \nonumber\\
    &\delta^{-1} \Gamma_2(\delta \xi,k_\parallel,\delta) \ \widehat{\beta}^A_\text{near}(\xi) - E_1 \widehat{\beta}^B_\text{near}(\xi) - \delta \mu \widehat{\beta}^B_\text{near}(\xi)\label{eqn:app-near-full-2}\\
    = \ & \delta \bigg(\widetilde{F}_2[\delta \xi;\bm{\widetilde{\eta}}_\text{near}] + \widetilde{F}_2[\delta \xi;\bm{\widetilde{\eta}}_\text{far}]\bigg) \chi(|\xi|\leq \delta^{\tau-1}) -\widetilde{I}_2[\delta \xi;\bm \Psi_0, \mu, \delta]\chi(|\xi|\leq \delta^{\tau-1}).\nonumber
\end{align}
\end{subequations}
We claim that the leading-order terms in \eqref{eqn:app-near-full} coincide with the band-limited Dirac operator $\widehat{\mathcal{D}} ^\delta$ introduced in \eqref{eqn:dirac-1d-fourier}. We formulate this claim as the following Proposition:
\begin{proposition}\label{prop:app-simp-beta}
    Assume $\widehat{\bm \beta}_\text{near}(\xi) \in L^{2,1}(\mathbb{R})$. For sufficiently small $\delta$ and $k \in \Icell$, we have
    \begin{align}
        \begin{pmatrix}
            \delta^{-1}\Gamma_1(\delta \xi,k_\parallel,\delta) \: \widehat{\beta}^B_\text{near}(\xi) - E_1 \widehat{\beta}^A_\text{near}(\xi) - \delta \widetilde{F}_1[\delta \xi;\bm{\widetilde{\eta}}_\text{near}] \: \chi(|\xi|\leq \delta^{\tau-1})\\
            \delta^{-1}\Gamma_2(\delta \xi,k_\parallel,\delta) \ \widehat{\beta}^A_\text{near}(\xi) - E_1\widehat{\beta}^B_\text{near}(\xi) - \delta \widetilde{F}_2[\delta \xi;\bm{\widetilde{\eta}}_\text{near}]\: \chi(|\xi|\leq \delta^{\tau-1})
        \end{pmatrix} = \widehat{\mathcal{D}} ^\delta \widehat{\bm \beta}_\text{near}(\xi) + O(\delta).\label{eqn:app-simp-beta}
    \end{align}
\end{proposition}
\begin{proof}
We first review the band-limited Dirac operator $\widehat{\mathcal{D}} ^\delta$ in \eqref{eqn:dirac-1d-fourier}
\begin{align}
    \widehat{\mathcal{D}} ^\delta \widehat{\bm \beta}_\text{near} = \chi(|\xi|\leq \delta^{\tau-1}) \widehat{\mathcal{D}}  \widehat{\bm \beta}_\text{near} = \chi(|\xi|\leq \delta^{\tau-1}) \begin{pmatrix}
        \frac{3}{2}\xi \widehat{\beta}_\text{near}^B - \frac{3}{2}i  \widehat{\kappa \beta}_\text{near}^B - E_1 \widehat{\beta}_\text{near}^A\\
    \frac{3}{2}\xi \widehat{\beta}_\text{near}^A + \frac{3}{2}i \widehat{\kappa \beta}_\text{near}^A - E_1 \widehat{\beta}_\text{near}^B
    \end{pmatrix},\label{eqn:app-dirac-band-limited}
\end{align}
where the expression of $\widehat{\mathcal{D}} $ is given in \eqref{eqn:dirac-1d-fourier}. 

We now compute the terms in \eqref{eqn:app-simp-beta}, beginning with the $O(\delta^{-1})$ terms:
\begin{align}
    &\delta^{-1} \Gamma_1(\delta \xi,k_\parallel,\delta) \: \widehat{\beta}^B_\text{near}(\xi) = \frac{1}{\delta}e^{i\frac{\pi}{3}} \left(e^{i\frac{4\pi}{3}} (i\sqrt{3}\delta \xi) + e^{i\frac{2\pi}{3}} \left(\frac{i}{2} \delta k_\parallel + \frac{\sqrt{3}}{2} i \delta \xi\right)\right) \widehat{\beta}^B_\text{near}(\xi) + O(\delta) \label{eqn:app-simp-beta-1}\\
    =&\Big[\frac{3}{2}\xi \widehat{\beta}_\text{near}^B - \frac{i}{2} k_\parallel \widehat{\beta}_\text{near}^B \Big] + O(\delta) = \Big[\frac{3}{2}\xi \widehat{\beta}_\text{near}^B - \frac{3}{2}i \widehat{\kappa \beta}_\text{near}^B \Big]\chi(|\xi|\leq \delta^{\tau-1}) + O(\delta). \nonumber
\end{align}
Similarly, we have 
\begin{align}
    \delta^{-1}\Gamma_2(\delta \xi,k_\parallel,\delta) \ \widehat{\beta}^A_\text{near}(\xi) = \Big[\frac{3}{2}\xi \widehat{\beta}_\text{near}^A + \frac{3}{2}i \widehat{\kappa \beta}_\text{near}^A\Big] \chi(|\xi|\leq \delta^{\tau-1}) + O(\delta).\label{eqn:app-simp-beta-2}
\end{align}

Then we simplify the terms related to $\widetilde{F}_i[\delta \xi;\bm{\widetilde{\eta}}_\text{near}]$ in \eqref{eqn:app-simp-beta}. Notice that $\widetilde{F}_i[\delta \xi;\bm{\widetilde{\eta}}_\text{near}]$ in \eqref{eqn:widetilde-F} can be reformulated by using $\bm \beta_\text{near}$ introduced in \eqref{eqn:eta-idft}: for $\widetilde{F}_1[\delta \xi;\bm{\widetilde{\eta}}_\text{near}]$, we have
\begin{align*}
    &\widetilde{F}_1[\delta \xi;\widetilde{\bm{\eta}}_\text{near}] = - e^{i \frac{\pi}{3}} t_1 \sum_{m \in \mathbb{Z}} \frac{\sqrt{3}}{2} \ff_1\left(\frac{\sqrt{3}}{2}\delta m\right) \beta_\text{near}^B \left(\frac{\sqrt{3}}{2}\delta m\right) e^{-i\delta \xi\frac{\sqrt{3}}{2}m} \\
    &- e^{i \frac{5\pi}{3}} t_1 \sum_{m \in \mathbb{Z}} \frac{\sqrt{3}}{2} \ff_2\left(\frac{\sqrt{3}}{2}\delta m\right) \beta_\text{near}^B \left(\frac{\sqrt{3}}{2}\delta (m+2)\right) e^{-i \delta \xi\frac{\sqrt{3}}{2}m}.
\end{align*}
Using the scaled PSF in \eqref{eqn:psf-scaled}-\eqref{eqn:psf-scaled-1}, the term related to $\widetilde{F}_1[\delta \xi;\widetilde{\bm{\eta}}_\text{near}]$ in \eqref{eqn:app-simp-beta} becomes
\begin{align}
    &\delta \widetilde{F}_1[\delta \xi;\widetilde{\bm{\eta}}_\text{near}] \chi(|\xi|\leq \delta^{\tau-1}) = - e^{i \frac{\pi}{3}} t_1 \chi(|\xi|\leq \delta^{\tau-1}) \sum_{m \in \mathbb{Z}} \widehat{\ff_1 \beta_\text{near}^B} \left(\frac{\delta \xi + \frac{4\pi}{\sqrt{3}}m}{\delta}\right) \label{eqn:app-F1}  \\
    &- e^{i \frac{5\pi}{3}} t_1 \chi(|\xi|\leq \delta^{\tau-1}) \sum_{m \in \mathbb{Z}} \widehat{\ff_2 \beta_\text{near}^B} \left(\frac{\delta \xi + \frac{4\pi}{\sqrt{3}}m}{\delta}\right) \nonumber\\
    & - e^{i \frac{5\pi}{3}} \Big(e^{i\sqrt{3}\delta \xi}-1\Big)t_1 \chi(|\xi|\leq \delta^{\tau-1}) \sum_{m \in \mathbb{Z}} \widehat{\ff_2 \beta_\text{near}^B} \left(\frac{\delta \xi + \frac{4\pi}{\sqrt{3}}m}{\delta}\right) \nonumber\\
    &- \delta e^{i \frac{5\pi}{3}} t_1 \chi(|\xi|\leq \delta^{\tau-1}) \sum_{m \in \mathbb{Z}} \frac{\sqrt{3}}{2} \left[ \ff_2\left(\frac{\sqrt{3}}{2}\delta m\right) - \ff_2\left(\frac{\sqrt{3}}{2}\delta (m+2)\right) \right] \beta_\text{near}^B \left(\frac{\sqrt{3}}{2}\delta (m+2)\right) e^{-i \delta \xi\frac{\sqrt{3}}{2}m},\nonumber
\end{align}
where the last two lines are $O(\delta)$. Similarly, the term related to $\widetilde{F}_2[\delta \xi;\widetilde{\bm{\eta}}_\text{near}]$ in \eqref{eqn:app-simp-beta} becomes
\begin{align}
    &\delta \widetilde{F}_2[\delta \xi;\widetilde{\bm{\eta}}_\text{near}] \chi(|\xi|\leq \delta^{\tau-1})= - e^{-i \frac{\pi}{3}} t_1 \chi(|\xi|\leq \delta^{\tau-1}) \sum_{m \in \mathbb{Z}} \widehat{\ff_1 \beta_\text{near}^A} \left(\frac{\delta \xi + \frac{4\pi}{\sqrt{3}}m}{\delta}\right) \label{eqn:app-F2}\\
    &- e^{-i \frac{5\pi}{3}} t_1 \chi(|\xi|\leq \delta^{\tau-1}) \sum_{m \in \mathbb{Z}} \widehat{\ff_2 \beta_\text{near}^A} \left(\frac{\delta \xi + \frac{4\pi}{\sqrt{3}}m}{\delta}\right) \nonumber\\
    &- e^{-i \frac{5\pi}{3}} \Big(e^{-i\sqrt{3}\delta \xi} - 1\Big) t_1 \chi(|\xi|\leq \delta^{\tau-1}) \sum_{m \in \mathbb{Z}} \widehat{\ff_2 \beta_\text{near}^A} \left(\frac{\delta \xi + \frac{4\pi}{\sqrt{3}}m}{\delta}\right), \nonumber
\end{align}
where the last line is $O(\delta)$.

Lastly, we collect the $O(1)$ terms in \eqref{eqn:app-simp-beta-1}-\eqref{eqn:app-F2} to complete our proof of \eqref{eqn:app-simp-beta}. Notice that $\widehat{\bm \beta}_{\text{near}}, \widehat{\ff_i\bm \beta}_{\text{near}}\in L^{2,1}(\mathbb{R})$ with $i=1,2$, since $\bm \beta_{\text{near}} \in H^1(\mathbb{R})$ and $\ff_i \in C_b^\infty(\mathbb{R})$. Using \eqref{eqn:sum-M-ge1-bd}, the Fourier modes of $\widehat{\bm \beta}_{\text{near}}, \widehat{\ff_i\bm \beta}_{\text{near}}$ with $|m|\geq 1$ are $O(\delta)$ (by taking $s=1$ in \eqref{eqn:sum-M-ge1-bd}). Therefore, for the first line of \eqref{eqn:app-simp-beta}, the $O(1)$ terms are
\begin{align}
    &\delta^{-1}\Gamma_1(\delta \xi,k_\parallel,\delta) \: \widehat{\beta}^B_\text{near}(\xi) - E_1 \widehat{\beta}^A_\text{near}(\xi) - \delta \widetilde{F}_1[\delta \xi;\bm{\widetilde{\eta}}_\text{near}] \: \chi(|\xi|\leq \delta^{\tau-1}) \label{eqn:app-simp-beta-f1}\\
    =& \Big[\frac{3}{2}\xi \widehat{\beta}_\text{near}^B - \frac{3}{2}i k_\parallel \widehat{\beta}_\text{near}^B - E_1 \widehat{\beta}^A_\text{near}(\xi)\Big]\chi(|\xi|\leq \delta^{\tau-1}) \nonumber\\
    -& t_1\Big[e^{i \frac{\pi}{3}} \widehat{\ff_1 \beta_\text{near}^B} (\xi) - e^{i \frac{5\pi}{3}} \widehat{\ff_2 \beta_\text{near}^B} (\xi)\Big] \chi(|\xi|\leq \delta^{\tau-1})  + O(\delta)\nonumber\\
    =&\Big[\frac{3}{2}\xi \widehat{\beta}_\text{near}^B - \frac{3}{2}i  \widehat{\kappa \beta^B_\text{near}} - E_1 \widehat{\beta}_\text{near}^A\Big] \chi(|\xi|\leq \delta^{\tau-1}) + O(\delta),\nonumber
\end{align}
where the last line holds since $\ff_1 = -\ff_2 = -\frac{\sqrt{3}}{4}d'$ and $\kappa = k_\parallel - \frac{t_1}{2}d'$. Similarly, for the second line of \eqref{eqn:app-simp-beta}, the $O(1)$ terms are
\begin{align}
    &\delta^{-1}\Gamma_2(\delta \xi,k_\parallel,\delta) \ \widehat{\beta}^A_\text{near}(\xi) - E_1\widehat{\beta}^B_\text{near}(\xi) - \delta \widetilde{F}_2[\delta \xi;\bm{\widetilde{\eta}}_\text{near}]\: \chi(|\xi|\leq \delta^{\tau-1})\label{eqn:app-simp-beta-f2}\\
    =&\Big[\frac{3}{2}\xi \widehat{\beta}_\text{near}^A + \frac{3}{2}i \widehat{\kappa \beta_\text{near}^A} - E_1 \widehat{\beta}_\text{near}^B\Big] \chi(|\xi|\leq \delta^{\tau-1}) + O(\delta).\nonumber
\end{align}
Thus, we complete the proof of \eqref{eqn:app-simp-beta}.
\end{proof}

So far, Proposition \ref{prop:app-simp-beta} shows that \eqref{eqn:app-near-full} is a band-limited Dirac equation with an $O(\delta)$ perturbation. Moreover, we have identified all $O(\delta)$ perturbative contributions except for the term involving $\widetilde{\bm \eta}_{\mathrm{far}}$. We now first write down the $O(\delta)$ perturbation arising from $\widetilde{\bm \eta}_{\mathrm{far}}$.
Since $\bm{\widetilde{\eta}}_\text{far}$ has representation \eqref{eqn:far-affine} and $\widetilde{F}_i$ are linear, we can write the far-momentum part as
\begin{align}
    \widetilde{F}_i[\delta \xi;\bm{\widetilde{\eta}}_\text{far}] = \widetilde{F}_i[\delta \xi;\mathcal{A}\bm{\widetilde{\eta}}_\text{near}] + \widetilde{F}_i[\delta \xi;\mathcal{B}], \qquad i=1,2.\label{eqn:tildeF-affine-far}
\end{align}
Therefore, the $O(\delta)$ contributions from \eqref{eqn:app-near-full} are
\begin{align} \label{eqn:tildeF-affine-far-1}
    \delta \widetilde{F}_i[\delta \xi;\bm{\widetilde{\eta}}_\text{far}] \: \chi(|\xi|\leq \delta^{\tau-1}) = \delta \bigg(\widetilde{F}_i[\delta \xi;\mathcal{A}\bm{\widetilde{\eta}}_\text{near}] + \widetilde{F}_i[\delta \xi;\mathcal{B}]\bigg) \: \chi(|\xi|\leq \delta^{\tau-1}), \qquad i=1,2.
\end{align}

We complete the derivation of \eqref{eqn:app-near-full} by grouping all $O(\delta)$ perturbation in \eqref{eqn:app-near-full} together. We omit the detailed calculations and present only the final result:
\begin{align*}
    &\Big[\frac{3}{2}\xi \widehat{\beta}_\text{near}^B - \frac{3}{2}i  \widehat{\kappa \beta_\text{near}^B} - E_1 \widehat{\beta}_\text{near}^A\Big] \chi(|\xi|\leq \delta^{\tau-1}) + [\widehat{\mathcal{L}_1^\delta}(\mu) \bm{\widehat{\beta}_\text{near}}](\xi) - \delta \mu \widehat{\beta}_\text{near}^A(\xi) = \mu \widehat{\mathcal{M}_1}(\xi;\delta) + \widehat{\mathcal{N}_1}(\xi;\mu,\delta),\\
    &\Big[\frac{3}{2}\xi \widehat{\beta}_\text{near}^A + \frac{3}{2}i \widehat{\kappa \beta_\text{near}^A} - E_1 \widehat{\beta}_\text{near}^B\Big] \chi(|\xi|\leq \delta^{\tau-1}) + [\widehat{\mathcal{L}_2^\delta}(\mu) \bm{\widehat{\beta}_\text{near}}](\xi) - \delta \mu \widehat{\beta}_\text{near}^B(\xi) = \mu \widehat{\mathcal{M}_2}(\xi;\delta) + \widehat{\mathcal{N}_2}(\xi;\mu,\delta),
\end{align*}
where $[\widehat{\mathcal{L}_1^\delta}(\mu) \bm{\widehat{\beta}_\text{near}}](\xi)$ and $[\widehat{\mathcal{L}_2^\delta}(\mu) \bm{\widehat{\beta}_\text{near}}](\xi)$ are
\begin{subequations}\label{eqn:Ldelta}
\begin{align}
    &[\widehat{\mathcal{L}_1^\delta}(\mu) \bm{\widehat{\beta}_\text{near}}](\xi) := \frac{1}{\delta} \Gamma_3(\delta\xi,k_\parallel,\delta) \widehat{\beta}_\text{near}^B(\xi) + \:t_1 e^{i \frac{5\pi}{3}} \Big(e^{i\sqrt{3}\delta \xi} - 1\Big) \chi(|\xi|\leq \delta^{\tau-1}) \widehat{\ff_2 \beta_\text{near}^B} (\xi) \label{eqn:Ldelta-1}\\
    & + \: t_1 e^{i \frac{\pi}{3}} \chi(|\xi|\leq \delta^{\tau-1}) \sum_{|m|\geq 1} \widehat{\ff_1 \beta_\text{near}^B} \left(\frac{\delta \xi + \frac{4\pi}{\sqrt{3}}m}{\delta}\right) - \delta \chi(|\xi|\leq \delta^{\tau-1}) \widetilde{F}_1[\delta \xi;\mathcal{A}\bm{\widetilde{\eta}}_\text{near}(\mu,\delta)] \nonumber\\
    & +\:t_1 e^{i \frac{5\pi}{3}} e^{i\sqrt{3}\delta \xi} \chi(|\xi|\leq \delta^{\tau-1}) \Bigg(\sum_{|m|\geq 1} \widehat{\ff_2 \beta_\text{near}^B} \left(\frac{\delta \xi + \frac{4\pi}{\sqrt{3}}m}{\delta}\right)\nonumber\\
    &\qquad + \sum_{m \in \mathbb{Z}} \frac{\sqrt{3}}{2} \left[ \ff_2\left(\frac{\sqrt{3}}{2}\delta m\right) - \ff_2\left(\frac{\sqrt{3}}{2}\delta (m+2)\right) \right] \beta_\text{near}^B \left(\frac{\sqrt{3}}{2}\delta (m+2)\right) e^{-i \delta \xi\frac{\sqrt{3}}{2}m} \Bigg), \nonumber\\
    &[\widehat{\mathcal{L}_2^\delta}(\mu) \bm{\widehat{\beta}_\text{near}}](\xi) := \frac{1}{\delta} \Gamma_4(\delta\xi,k_\parallel,\delta)\widehat{\beta}_\text{near}^A(\xi) + t_1 e^{-i \frac{5\pi}{3}}  \Big(e^{-i\sqrt{3}\delta \xi} -1 \Big) \chi(|\xi| \leq \delta^{\tau-1})  \widehat{\ff_2 \beta_\text{near}^A}(\xi) \label{eqn:Ldelta-2}\\
    & + \:  t_1 e^{-i \frac{\pi}{3}} \chi(|\xi| \leq \delta^{\tau-1})\sum_{|m|\geq 1} \widehat{\ff_1 \beta_\text{near}^A} \left(\frac{\delta \xi + \frac{4\pi}{\sqrt{3}}m}{\delta}\right) - \delta \chi(|\xi|\leq \delta^{\tau-1}) \widetilde{F}_2[\delta \xi;\mathcal{A}\bm{\widetilde{\eta}}_\text{near}(\mu,\delta)] \nonumber\\
    & +\: t_1 e^{-i \frac{5\pi}{3}} e^{-i\sqrt{3}\delta \xi} \chi(|\xi| \leq \delta^{\tau-1}) \sum_{|m|\geq 1} \widehat{\ff_2 \beta_\text{near}^A} \left(\frac{\delta \xi + \frac{4\pi}{\sqrt{3}}m}{\delta}\right) \nonumber\\
    &+ \: \delta t_1 e^{-i \frac{5\pi}{3}} \Big(e^{-i\sqrt{3}\delta \xi} -1\Big)\chi(|\xi| \leq \delta^{\tau-1}) \widehat{\ff_2 \beta_\text{near}^A} (\xi).\nonumber
\end{align}
\end{subequations}
The terms $\widehat{\mathcal{M}_i}(\xi;\delta)$ and $\widehat{\mathcal{N}_i}(\xi;\delta)$ with $i=1,2$ are
\begin{subequations}\label{eqn:MN-12}
\begin{align}
    \widehat{\mathcal{M}_1}(\xi;\delta) &= \frac{2}{\sqrt{3}} \chi(|\xi|\leq \delta^{\tau-1}) \sum_{m \in \mathbb{Z}} \widehat{\Psi_0^A}\left(\frac{\delta \xi + \frac{4\pi}{\sqrt{3}}m}{\delta}\right)\label{eqn:M1},\\
    \widehat{\mathcal{M}_2}(\xi;\delta)  &= \frac{2}{\sqrt{3}} \chi(|\xi|\leq \delta^{\tau-1})\sum_{m \in \mathbb{Z}} \widehat{\Psi_0^B}\left(\frac{\delta \xi + \frac{4\pi}{\sqrt{3}}m}{\delta}\right)\label{eqn:M2},\\
    \widehat{\mathcal{N}_1}(\xi;\mu,\delta) &= \bigg(-\widetilde{I}_{1,\text{ind}}[\delta \xi;\bm \Psi_0, \delta] + \delta \widetilde{F}_1[\delta \xi;\mathcal{B}(\mu,\delta)]\bigg) \chi(|\xi|\leq \delta^{\tau-1}),\label{eqn:N1}\\
    \widehat{\mathcal{N}_2}(\xi;\mu,\delta) &= \bigg(-\widetilde{I}_{2,\text{ind}}[\delta \xi;\bm \Psi_0, \delta] + \delta \widetilde{F}_2[\delta \xi;\mathcal{B}(\mu,\delta)]\bigg) \chi(|\xi|\leq \delta^{\tau-1}).\label{eqn:N2}
\end{align}
\end{subequations}

\bibliography{ref}
\bibliographystyle{plain}

\clearpage
\section*{Supplementary Material} 

\setcounter{section}{0}
\renewcommand{\thesection}{S\arabic{section}}

\setcounter{equation}{0}
\makeatletter
\@addtoreset{equation}{section}
\makeatother

\renewcommand{\theequation}{\thesection.\arabic{equation}}

\renewcommand{\thefigure}{S\arabic{figure}}
\setcounter{figure}{0}

\setcounter{theorem}{0}
\makeatletter
\@addtoreset{theorem}{section}
\makeatother
\renewcommand{\thetheorem}{\thesection.\arabic{theorem}}

This supplementary information section acts as a companion to the main text, where we provide additional details on calculations and proofs.

\setcounter{equation}{0}
\section{A general choice of hopping coefficients}\label{app:hopping}

We provide a general choice of hopping coefficients for the strained honeycomb lattice. Since the nearest-neighbor distances in the strained honeycomb are no longer uniform, the hopping coefficients need to be modified to reflect these variations. We define the hopping coefficient as a function of the distance between deformed nearest-neighbor nodes (see Equation (172) in \cite{castro2009electronic}), i.e. the hopping coefficient $t(\widetilde{X}, \widetilde{Y})$ for a pair of deformed nearest neighbors $\widetilde{X}$ and $\widetilde{Y}$ is given by
\begin{equation}\label{eqn:def-hopping}
	t(\widetilde{X}, \widetilde{Y}) = h(|\widetilde{X} - \widetilde{Y}|),
\end{equation}
where $h:\mathbb{R}_+ \rightarrow \mathbb{R}_+$ is a smooth function satisfying $h(1) = 1$, which ensures that the hopping coefficient is uniformly 1 in the undeformed honeycomb with nearest-neighbor distance equal to 1. 

In fact, for the slowly varying deformation $\bm{u}(\delta \bm{X})$, the hopping coefficients defined in \eqref{eqn:def-hopping} are slowly perturbed away from 1 and can be expressed as 1 plus an $O(\delta)$ correction. As a representative example, we expand the hopping coefficient $t(\widetilde{A}_{m,n}, \widetilde{B}_{m,n})$ using the Taylor series of $h$ (the other hopping coefficients are be expanded analogously): 
\begin{align*}
    &t(\widetilde{A}_{m,n} , \widetilde{B}_{m,n}) = h(|\widetilde{A}_{m,n} - \widetilde{B}_{m,n}|) \\
        =&\:  h(1) + t_1\Big(|\widetilde{A}_{m,n} - \widetilde{B}_{m,n}|-1\Big) + o\Big(\big|\big(|\widetilde{A}_{m,n} - \widetilde{B}_{m,n}|-1\big)\big|\Big),
\end{align*}
where $t_1$ is defined as $t_1 := h'(1)$. To further expand $|\widetilde{A}_{m,n} - \widetilde{B}_{m,n}|-1$, we use the Taylor series expansion of the norm function at $\bm{X}_0$:
\begin{equation}\label{eqn:norm-exp}
	|\bm{X}| = |\bm{X}_0| + \frac{1}{|\bm{X}_0|} \bm{X}_0 \cdot (\bm{X} -\bm{X}_0) + o(|\bm{X} - \bm{X}_0|).
\end{equation}
Since $A_{m,n} - B_{m,n} = \bm{e}_1$ is a unit vector, we use \eqref{eqn:norm-exp} with $\bm{X}_0 = {A}_{m,n} - {B}_{m,n}$ and obtain
\begin{align*}
    &|\widetilde{A}_{m,n} - \widetilde{B}_{m,n}|-1 = |A_{m,n} - B_{m,n} + \bm{u}(\delta A_{m,n}) - \bm{u}(\delta B_{m,n})| - 1 \\
		=\: & \big \langle \bm{u}(\delta A_{m,n}) - \bm{u}(\delta B_{m,n}), A_{m,n} - B_{m,n}\big \rangle + o\Big(|\bm{u}(\delta A_{m,n}) - \bm{u}(\delta B_{m,n})|\Big)\\
		=\: & \delta \big\langle \nabla \bm{u}(\delta A_{m,n}) (A_{m,n} - B_{m,n}), A_{m,n} - B_{m,n}\big\rangle + o(\delta) = \delta \bm{e}_1^T\nabla \bm{u}(\delta A_{m,n}) \bm{e}_1 + o(\delta).
\end{align*}
Therefore, the hopping coefficient $t(\widetilde{A}_{m,n} , \widetilde{B}_{m,n})$ is perturbed away from 1 by 
\begin{equation}
    t(\widetilde{A}_{m,n} , \widetilde{B}_{m,n}) = 1 + \delta t_1 \bm{e}_1^T\nabla \bm{u}(\delta A_{m,n}) \bm{e}_1 + O(\delta^2).\label{eqn:app-hopping-AB}
\end{equation}
We observe that the $O(\delta)$-order term coincides with that defined in \eqref{eqn:hopping-def-honeycomb}. Hence, the same effective operator $\mathcal{H}_\text{eff}$ in \eqref{eqn:eff-ham} is obtained under this general choice of hopping coefficients.


\setcounter{equation}{0}
\section{The unitary equivalence between $\mathcal{H}_{\text{eff}}$ and $\mathcal{D}_{\mathbf{A}}$}\label{app:dirac-op}

We show that $\mathcal{H}_\text{eff}$ is unitarily equivalent to the magnetic Dirac operator $\mathcal{D}_{\bm{A}}$ as presented in \eqref{eqn:dirac-mag}. We start with the detailed form of $\mathcal{H}_{\text{eff}}$ by substituting $\bm w_\nu = \mathcal{C}_{m_\nu,n_\nu}$ in \eqref{eqn:c-vec} into \eqref{eqn:eff-ham} and obtain
\begin{subequations}\label{eqn:app-eff-ham}
\begin{align}
	(\mathcal{H}_\text{eff} \bm{\Phi})_A&= \Big(e^{i \frac{4\pi}{3}} \bm w_2 \cdot \nabla_{\bm{X}} + e^{i \frac{2\pi}{3}} \bm w_3 \cdot \nabla_{\bm{X}}\Big) \Phi^B \nonumber \\
    &\: +t_1 \Big(f_1(\bm X) + e^{i \frac{4\pi}{3}} f_2(\bm X) + e^{i \frac{2\pi}{3}} f_3(\bm X)\Big) \Phi^B,\label{eqn:app-eff-ham-a}\\
    (\mathcal{H}_\text{eff} \bm{\Phi})_B &= \Big(-e^{-i \frac{4\pi}{3}} \bm w_2 \cdot \nabla_{\bm X} - e^{-i \frac{2\pi}{3}} \bm w_3 \cdot \nabla_{\bm X}\Big) \Phi^A \nonumber \\
    &\: + t_1 \Big(f_1(\bm X) + e^{-i \frac{4\pi}{3}} f_2(\bm X) + e^{-i \frac{2\pi}{3}}f_3(\bm X) \Big) \Phi^A,\label{eqn:app-eff-ham-b}
\end{align}
\end{subequations} 
where the phase factors come from $e^{\pm i \bm K \cdot \bm w_1} = 1, e^{\pm i \bm K \cdot \bm w_2} = e^{\pm i \frac{4\pi}{3}}$ and $e^{\pm i \bm K \cdot \bm w_3} = e^{\pm i \frac{2\pi}{3}}$. Then we simplify the differential operator in the first row of \eqref{eqn:app-eff-ham-a} using $\bm w_\nu = \mathcal{C}_{m_\nu,n_\nu}$ in \eqref{eqn:c-vec} and obtain
\begin{align}
    &\Big(e^{i \frac{4\pi}{3}} \bm w_2 \cdot \nabla_{\bm{X}} + e^{i \frac{2\pi}{3}} \bm w_3 \cdot \nabla_{\bm{X}}\Big)\Phi_0^B =\sqrt{3} e^{i\frac{4}{3}\pi} \partial_{X_1}  \Phi_0^B+ \frac{\sqrt{3}}{2} e^{i\frac{2}{3}\pi} \partial_{X_1}  \Phi_0^B+  \frac{3}{2} e^{i\frac{2}{3}\pi} \partial_{X_2} \Phi_0^B \label{eqn:eff-sim-cal-1}\\
    = \:& \sqrt{3} \left( e^{i\frac{4}{3}\pi} + \frac{1}{2}  e^{i\frac{2}{3}\pi}\right) \partial_{X_1}  \Phi_0^B+ \frac{3}{2} e^{i\frac{2}{3}\pi}  \partial_{X_2}  \Phi_0^B \nonumber\\
    = \: &\frac{3}{2} e^{i\frac{7}{6}\pi} \partial_{X_1}  \Phi_0^B  + \frac{3}{2} e^{i\frac{2}{3}\pi}  \partial_{X_2}  \Phi_0^B = \frac{3}{2} e^{-i\frac{\pi}{3}} \Big[-i \partial_{X_1}  \Phi_0^B - \partial_{X_2}  \Phi_0^B\Big].\nonumber
\end{align}
Similarly, the differential operator in the second row of \eqref{eqn:eff-ham} can be simplified as
\begin{equation}\label{eqn:eff-sim-cal-2}
	\Big(-e^{-i \frac{4\pi}{3}} \bm w_2 \cdot \nabla_{\bm X} - e^{-i \frac{2\pi}{3}} \bm w_3 \cdot \nabla_{\bm X} \Big) \Phi_0^A = \frac{3}{2} e^{i\frac{\pi}{3}} \Big[-i \partial_{X_1} \Phi_0^A + \partial_{X_2} \Phi_0^A\Big].
\end{equation}

Then we replace terms related to $f_\nu(\bm X)$ with $\nu = 1,2,3$ in \eqref{eqn:f-u-relation} as combinations of $\widetilde{A}_\nu(\bm X)$, which are defined as
\begin{align*}
    \widetilde{A}_1 = \partial_{X_1} u_1 - \partial_{X_2} u_2, \quad \widetilde{A}_2 = -(\partial_{X_1} u_2 + \partial_{X_2} u_1), \quad \widetilde{A}_3 = \partial_{X_1} u_1 + \partial_{X_2} u_2.
\end{align*}
We observe that $\widetilde{A}_1, \widetilde{A}_2$ are scaled versions of $A_1, A_2$ in \eqref{eqn:eff-mag}, i.e.
\begin{align}\label{eqn:tilde-A-relation}
    A_1 = -\frac{t_1}{2}\widetilde{A}_1, \qquad A_2 = -\frac{t_1}{2}\widetilde{A}_2.
\end{align}
Since $\partial_{X_1} u_1 = (\widetilde{A}_1 + \widetilde{A}_3)/2$ and $\partial_{X_2} u_2 = (\widetilde{A}_3 - \widetilde{A}_1)/2$, we can represent $f_\nu$ in terms of $\widetilde{A}_\nu$
\begin{align*}
    f_1&= \frac{3}{8} (\widetilde{A}_1 + \widetilde{A}_3) - \frac{\sqrt{3}}{4} \widetilde{A}_2 + \frac{1}{8} (\widetilde{A}_3 - \widetilde{A}_1) = \frac{1}{4} \widetilde{A}_1 - \frac{\sqrt{3}}{4} \widetilde{A}_2 + \frac{1}{2} \widetilde{A}_3,\\
	f_2& = \frac{3}{8} (\widetilde{A}_1 + \widetilde{A}_3) + \frac{\sqrt{3}}{4} \widetilde{A}_2 + \frac{1}{8} (\widetilde{A}_3 - \widetilde{A}_1) = \frac{1}{4} \widetilde{A}_1 + \frac{\sqrt{3}}{4} \widetilde{A}_2 + \frac{1}{2} \widetilde{A}_3,\quad f_3= \frac{1}{2} (\widetilde{A}_3 - \widetilde{A}_1).
\end{align*}

By writing the $t_1$ related terms in the first row of \eqref{eqn:eff-ham} using $\widetilde{A}_\nu$ with $\nu = 1,2,3$, we obtain
\begin{align*}
    &t_1 \Big(f_1 +e^{i\frac{4}{3}\pi} f_2 + e^{i\frac{2}{3}\pi} f_3\Big) \Phi_0^B \\
    =\:& t_1 \left(\Big[\frac{1}{4} \widetilde{A}_1 - \frac{\sqrt{3}}{4} \widetilde{A}_2 + \frac{1}{2} \widetilde{A}_3\Big] + e^{i\frac{2}{3}\pi}  \Big[\frac{1}{2}\widetilde{A}_3 -\frac{1}{2} \widetilde{A}_1\Big] + e^{i\frac{4}{3}\pi}\Big[\frac{1}{4} \widetilde{A}_1 + \frac{\sqrt{3}}{4} \widetilde{A}_2 + \frac{1}{2} \widetilde{A}_3\Big] \right) \Phi_0^B\\
		=\: &t_1 \left(\left[\frac{1}{4} - \frac{1}{2} e^{i\frac{2}{3}\pi} + \frac{1}{4} e^{i\frac{4}{3}\pi}\right] \widetilde{A}_1 + \left[- \frac{\sqrt{3}}{4} + \frac{\sqrt{3}}{4} e^{i\frac{4}{3}\pi}\right] \widetilde{A}_2 + \left[\frac{1}{2} + \frac{1}{2} e^{i\frac{2}{3}\pi}  + \frac{1}{2} e^{i\frac{4}{3}\pi}\right] \widetilde{A}_3\right) \Phi_0^B \\
    =\:& \frac{3}{4} t_1  e^{-i\frac{\pi}{3}}\Big[\widetilde{A}_1 - i \widetilde{A}_2\Big] \Phi_0^B.
\end{align*}
Using the relationship between $\widetilde{A}_1, \widetilde{A}_2$ and $A_1, A_2$ in \eqref{eqn:tilde-A-relation}, we have
\begin{equation}\label{eqn:eff-sim-cal-3}
	t_1 \Big(f_1 +e^{i\frac{4}{3}\pi} f_2 + e^{i\frac{2}{3}\pi} f_3\Big) \Phi_0^B = \frac{3}{2} e^{-i\frac{\pi}{3}} \Big[- A_1 +  i A_2\Big] \Phi_0^B
\end{equation}
Similarly, we can write the $t_1$ related terms in the second row of \eqref{eqn:eff-ham} as
\begin{equation}\label{eqn:eff-sim-cal-4}
	t_1 \Big(f_1 + e^{-i\frac{4}{3}\pi} f_2 + e^{-i\frac{2}{3}\pi} f_3\Big) \Phi_0^A = \frac{3}{2} e^{i\frac{\pi}{3}} \Big[- A_1 -  i A_2\Big] \Phi_0^A.
\end{equation}
Combining \eqref{eqn:eff-sim-cal-1}-\eqref{eqn:eff-sim-cal-2} with \eqref{eqn:eff-sim-cal-3}-\eqref{eqn:eff-sim-cal-4}, the effective operator $\mathcal{H}_\text{eff}$ in \eqref{eqn:eff-ham} becomes
\begin{align}
    \mathcal{H}_\text{eff} &= \frac{3}{2} \begin{pmatrix}
				e^{-i\frac{\pi}{3}} & 0\\
				0 & e^{i\frac{\pi}{3}}
			\end{pmatrix} \begin{pmatrix}
				0 & -i \partial_{X_1} + \partial_{X_2} - A_1 + iA_2\\
				-i \partial_{X_1} - \partial_{X_2} -A_1 - iA_2 & 0
			\end{pmatrix} \nonumber\\
			&= \begin{pmatrix}
				e^{-i\frac{\pi}{3}} & 0\\
				0 & e^{i\frac{\pi}{3}}
			\end{pmatrix} \mathcal{D}_{\bm{A}} =   \begin{pmatrix}
				e^{-i\frac{\pi}{6}} & 0\\
				0 & e^{i\frac{\pi}{6}}
			\end{pmatrix} \mathcal{D}_{\bm{A}}  \begin{pmatrix}
				e^{i\frac{\pi}{6}} & 0\\
				0 & e^{-i\frac{\pi}{6}}
			\end{pmatrix} = U \mathcal{D}_{\bm{A}} U^*,\label{eqn:uni-similar}
\end{align}
where $U = \text{diag}(e^{-i\frac{\pi}{6}}, e^{i\frac{\pi}{6}})$ is a constant unitary matrix. Thus, we complete our proof of \eqref{eqn:dirac-mag}.

\section{Proof of Lemma \ref{lemma:1d}}\label{sm:proof-Lemma6-1}

In this section, we prove Lemma \ref{lemma:1d}. First, we review that the one-dimensional Dirac operator $\mathcal{D}$. For simplicity, we replace $X_1$ by $x$ and $\mathcal{D}$ becomes
\begin{equation}\label{eqn:sm-dirac-k}
    \mathcal{D} = \frac{3}{2}\begin{pmatrix}
        0 & -i \partial_x - i\kappa(x)\\
        -i\partial_{x} + i\kappa(x) & 0
    \end{pmatrix}.
\end{equation}
We observe that this operator is closely related to another one-dimensional Dirac operator $\mathcal{D}^r$ with real coefficients
\begin{align}
    \mathcal{D}  &= \frac{3}{2} \begin{pmatrix}
        i & 0\\
        0 & -i
    \end{pmatrix} \mathcal{D}^r , \qquad \mathcal{D}^r =\begin{pmatrix}
        0 & -\partial_{x} - \kappa(x)\\
        \partial_{x} - \kappa(x) & 0
    \end{pmatrix}, \label{eqn:sm-dirac-real}\\
    \Rightarrow \qquad \mathcal{D} &= \frac{3}{2} \begin{pmatrix}
        e^{i\frac{\pi}{4}} & 0\\
        0 & e^{-i\frac{\pi}{4}} 
    \end{pmatrix} \mathcal{D}^r  \begin{pmatrix}
        e^{-i\frac{\pi}{4}} & 0\\
        0 & e^{i\frac{\pi}{4}} 
    \end{pmatrix}= \frac{3}{2} V \mathcal{D}^r  V^*, \qquad V = \text{diag}(e^{i\frac{\pi}{4}}, e^{-i\frac{\pi}{4}}).\nonumber
\end{align}
Therefore, Lemma \ref{lemma:1d} is equivalent to the following Lemma, which lists the spectral properties of $\mathcal{D}^r $ in \eqref{eqn:sm-dirac-real} for general real-valued $\kappa(x)$.

\begin{lemma}\label{lemma:sm-spec-Dk}
    Suppose $\kappa(x)$ is continuous and has bounded values with
	\begin{equation}\label{eqn:sm-kappa-condition}
		\lim_{x \rightarrow +\infty} \kappa(x) = \kappa_+, \qquad \lim_{x \rightarrow -\infty} \kappa(x) = \kappa_-, \qquad \kappa_+, \kappa_- \neq 0.
	\end{equation}
    We further assume that
	\begin{equation}\label{eqn:d-condition-2}
		\kappa(x) - \kappa_+ \in L^1([0,\infty)), \qquad \kappa(x) - \kappa_- \in L^1((-\infty,0]),
	\end{equation}
	then the following statements hold for $\mathcal{D}^r $ in \eqref{eqn:sm-dirac-real}:\\[1ex]
    (a) The operator $\mathcal{D}^r $ has essential spectrum $(-\infty, -a] \cup [a, \infty)$ with
		\begin{equation}\label{eqn:sm-spec-gap}
			a = \min \Big\{|\kappa_+|, |\kappa_-|\Big\}.
		\end{equation}
    (b) If $\mathcal{D}^r $ has an eigenvalue $E$ in between the gap $(-a,a)$, then it is \textbf{simple} and the corresponding bounded zero eigenstates $\bm \psi = \big(\psi^A, \psi^B\big)^T$ have \textbf{exponential decay}, i.e. there exists $\lambda_- < 0 < \lambda_+$ and $\bm p_\pm = \big(p_\pm^A, p_\pm^B\big)^T \in \mathbb{R}^2$ such that
		\begin{equation}\label{eqn:psi0-decay}
			\begin{aligned}
			    &\lim_{x +\infty} \bm \psi(x) e^{\lambda_+ x} = \bm p_+,\qquad \lim_{x \rightarrow -\infty} \bm \psi(x) e^{\lambda_- x} = \bm p_-,
			\end{aligned}
		\end{equation}
        More specifically, $-\lambda_+$ is the negative eigenvalue of $\bm{A}_+$ and $-\lambda_-$ is the positive eigenvalue of $\bm{A}_-$, where $\bm{A}_+, \bm{A}_-$ are defined as
        \begin{equation}\label{eqn:sm-Amat-pm}
            \bm{A_+} :=\begin{pmatrix}
            \kappa_+ & E\\
            -E & \kappa_+
        \end{pmatrix}, \qquad \bm{A_-} :=\begin{pmatrix}
            \kappa_- & E\\
            -E & \kappa_-
        \end{pmatrix}
        \end{equation}
        and $\bm p_\pm$ form eigenvectors of $\bm{A}_\pm$, i.e.
        \begin{equation}\label{eqn:sm-A-eig-vec}
            \bm{A_+} \bm p_+ = -\lambda_+ \bm p_+, \qquad \bm{A_-} \bm p_- = -\lambda_- \bm p_-.
        \end{equation}
		  Furthermore, if all derivatives of $\kappa(x)$ are also bounded, i.e. $\kappa^{(j)}(x) \in L^\infty(\mathbb{R})$ with $j\geq 1$, all derivatives of $\bm \psi$ have exponential decay and we have
		\begin{equation}\label{eqn:psi0-Hs-bd}
			\| \bm \psi\|_{H^s(\mathbb{R})} < \infty, \qquad \forall s \in \mathbb{N}.
		\end{equation}
    (c) If $\kappa_+ \kappa_- < 0$, then there exists a simple zero eigenvalue in the spectral gap $(-a, a)$. 
\end{lemma}

\begin{proof}
    We start with (a). We observe that $\mathcal{D}^r $ is a compact perturbation of $\mathcal{D}_\pm$, given by
    \begin{equation*}
        \mathcal{D}_+ = \begin{pmatrix}
            0 & -\partial_x - \kappa_+\\
            \partial_x - \kappa_+ & 0
        \end{pmatrix}, \qquad \mathcal{D}_- = \begin{pmatrix}
            0 & -\partial_x - \kappa_-\\
            \partial_x - \kappa_- & 0
        \end{pmatrix}.
    \end{equation*}
    Therefore, the essential spectrum of $\mathcal{D}^r $ can be computed by $\sigma_{\text{ess}}(\mathcal{D}^r ) = \sigma_{\text{ess}}(\mathcal{D}_+) \cup \sigma_{\text{ess}}(\mathcal{D}_-)$, where the essential spectrum of $\mathcal{D}_+$ and $\mathcal{D}_-$ are well-known:
    \begin{equation*}
        \sigma_{\text{ess}}(\mathcal{D}_+) = (-\infty, -|\kappa_+|] \cup [|\kappa_+|,\infty), \quad \sigma_{\text{ess}}(\mathcal{D}_-) = (-\infty, -|\kappa_-|] \cup [|\kappa_-|,\infty).
    \end{equation*}
    Therefore, we finish the proof of (a).\\
    
    We now prove (b). For a given eigenvalue $E$ of $\mathcal{D} ^r$, we know the corresponding eigenfunction $\bm \psi(x)$ is already bounded in $H^1(\mathbb{R})$, hence $\bm \psi(x) \in L^\infty(\mathbb{R})$. Therefore, it is sufficient to show that $\bm \psi(x)$ has exponential decay at $\pm \infty$. 
    
    We prove the exponential decay of $\bm \psi(x)$ by relating it to an ODE system. We observe that the eigenvalue problem 
    \begin{equation}\label{eqn:sm-ode-full}
        \mathcal{D}^r  \begin{pmatrix}
            \psi^A\\
            \psi^B
        \end{pmatrix} = E \begin{pmatrix}
            \psi^A\\
            \psi^B
        \end{pmatrix} \quad \Leftrightarrow \quad \begin{cases}
            \partial_x \psi^B + E \psi^A + \kappa(x) \psi^B = 0\\
            \partial_x \psi^A - \kappa(x) \psi^A - E \psi^B = 0
        \end{cases}
    \end{equation}
    is equivalent to the following 1st-order ODE system
    \begin{equation}\label{eqn:sm-ode}
        \begin{aligned}
            \partial_x \begin{pmatrix}
                \psi^A\\
                \psi^B
            \end{pmatrix} = \begin{pmatrix}
                \kappa(x) & E\\
                -E & -\kappa(x)
            \end{pmatrix} \begin{pmatrix}
                \psi^A\\
                \psi^B
            \end{pmatrix}.
        \end{aligned}
    \end{equation}
    
    In fact, we can construct exact bounded solutions for \eqref{eqn:sm-ode} on $[0, \infty)$ and $(-\infty,0]$ separately by viewing \eqref{eqn:sm-ode} as a perturbation of constant-coefficient ODE system. On $[0, \infty)$, the linear ODE system \eqref{eqn:sm-ode} can be viewed as
    \begin{align*}
        \partial_x \begin{pmatrix}
            \psi^A\\
            \psi^B
        \end{pmatrix} = \left[\begin{pmatrix}
            \kappa_+ & E\\
            -E & -\kappa_+
        \end{pmatrix} + \bm{R}_+(x)\right] \begin{pmatrix}
            \psi^A\\
            \psi^B
        \end{pmatrix},
    \end{align*}
    where $\bm{R}_+(x)$ is the perturbative term that vanishes at $\infty$, given by
    \begin{align}\label{eqn:sm-ode-pos-1}
        \bm{R}_+(x) := \begin{pmatrix}
                \kappa(x) - \kappa_+ & 0\\
                0 & -\kappa(x)+\kappa_+
            \end{pmatrix}.
    \end{align}
    Equivalently, we have
    \begin{equation}\label{eqn:sm-ode-pos-2}
        \bm{\psi}' = \bm{A}_+ \bm{\psi} + \bm{R}_+(x) \bm{\psi},
    \end{equation}
    where $\bm{A}_+$ is defined in \eqref{eqn:sm-Amat-pm}. Similarly, on $(-\infty, 0]$, we can write our ODE system as
    \begin{equation}\label{eqn:sm-ode-neg-2}
        \bm{\psi}' = \bm{A}_- \bm{\psi} + \bm{R}_-(x) \bm{\psi},
    \end{equation}
    where $\bm{A}_-$ is defined in \eqref{eqn:sm-Amat-pm} and $\bm{R}_-(x)$ is defined as
    \begin{equation}\label{eqn:sm-ode-neg-1}
        \bm{R}_-(x) := \begin{pmatrix}
                \kappa(x) - \kappa_- & 0\\
                0 & -\kappa(x)+\kappa_-
            \end{pmatrix}.
    \end{equation}
    Since $\kappa(x)$ satisfies \eqref{eqn:sm-kappa-condition}, the two perturbation matrices $\bm{R}_+(x)$ and $\bm{R}_-(x)$ are $L^1$-integrable, i.e.
    \begin{equation}\label{eqn:sm-A-l1-int}
        \int_{0}^\infty \|\bm{R}_+(x)\| \, dx < \infty, \qquad \int_{-\infty}^0 \|\bm{R}_-(x)\| \, dx < \infty.
    \end{equation}
    
    From classical ODE theory, we know that for a linear system obtained by adding an $L^1$-integrable perturbation to a constant-coefficient system, the solutions of the perturbed system have the same asymptotic behavior at infinity as those of the unperturbed system. This statement is presented in detail in Lemma \ref{lemma:sm-lemma-ode-asy}. Given Lemma \ref{lemma:sm-lemma-ode-asy}, it suffices to study the solutions of the following ODE systems with exponential decay
    \begin{align}
        \bm{\varphi}'_+ &= \bm{A}_+ \bm{\varphi}_+, \qquad x\geq 0,\qquad \bm{\varphi}'_- = \bm{A}_- \bm{\varphi}_-, \qquad x\leq 0. \label{eqn:supp-mat-pm-ode}
    \end{align}
    
    We observe that $\det\bm{A}_\pm < 0$ since $E$ sits in the spectral gap with $|E| < a=\min\{|\kappa_+|, |\kappa_-|\}$ (see \eqref{eqn:sm-spec-gap}). Therefore, for the positive branch $[0,\infty)$ and the negative branch $(-\infty, 0]$, there always exists solutions with exponential decay, i.e.
    \begin{equation*}
        \bm{\varphi}_+(x) = e^{-\lambda_+ x} \bm p_+ \quad \text{at }x \geq 0, \qquad \bm{\varphi}_-(x) = e^{-\lambda_- x} \bm p_- \quad \text{at }x \leq 0,
    \end{equation*}
    where $-\lambda_+$ is the negative eigenvalue of $\bm{A}_+$ and $-\lambda_-$ is the positive eigenvalue of $\bm{A}_-$, and $\bm p_\pm$ form the eigenvectors of $\bm{A}_\pm$ as mentioned in \eqref{eqn:sm-A-eig-vec}. 
    
    By applying Lemma \ref{lemma:sm-lemma-ode-asy}, we obtain a unique (up to scaling) bounded solution $\bm \psi_+(x)$ for \eqref{eqn:sm-ode} on $[0,\infty)$ with
    \begin{subequations}\label{eqn:sm-psi-eqn-1}
        \begin{align}
            &\bm{\psi}_+' =  \Big(\bm{A}_+ + \bm{R}_+(x) \Big) \bm{\psi}_+,\qquad x\geq 0, \label{eqn:sm-psi-eqn-1a}\\
            &\lim_{x\rightarrow \infty}\bm{\psi}_+(x) = \bm p_+.\label{eqn:sm-psi-eqn-1b}
        \end{align}
    \end{subequations}
    Due to the uniqueness, we have
    \begin{align}
        \bm \psi(x) = c_+ \bm{\psi}_+(x), \qquad x\in [0,\infty) \label{eqn:sm-psi-pos-asym}
    \end{align}
    for some constant $c_+ \in \mathbb{R}$. Similarly, we obtain a unique (up to scaling) bounded solution $\bm \psi_-(x)$ for \eqref{eqn:sm-ode} on $(-\infty, 0]$ with
    \begin{subequations}\label{eqn:sm-psi-eqn-2}
        \begin{align}
            &\bm{\psi}_-' =  \Big(\bm{A}_- + \bm{R}_-(x) \Big) \bm{\psi}_-, \qquad x\leq 0, \label{eqn:sm-psi-eqn-2a}\\
            &\lim_{x\rightarrow \infty}\bm{\psi}_-(x) = \bm p_-.\label{eqn:sm-psi-eqn-2b}
        \end{align}
    \end{subequations}
    and due to the uniqueness, we have
    \begin{align}
        \bm \psi(x) = c_- \bm{\psi}_-(x), \qquad x\in (-\infty,0] \label{eqn:sm-psi-neg-asym}
    \end{align}
    on $(-\infty,0]$ for some constant $c_- \in \mathbb{R}$. Thus, from \eqref{eqn:sm-psi-pos-asym} and \eqref{eqn:sm-psi-neg-asym}, we obtain the proof of (b). \\
    
    Lastly, we briefly explain why $k^{(j)}(x) \in L^\infty(\mathbb{R})$ guarantees the integrability of all derivatives of $\bm \psi$. We first observe that $\partial_x \bm \psi\in L^2(\mathbb{R})$ without using the condition on derivatives of $\kappa(x)$ since $\partial_x\psi^A, \partial_x\psi^B$ are linear combinations of $\psi^A, \psi^B$ by using \eqref{eqn:sm-ode-full}. Now we calculate $\partial_x^2 \psi^A$, which is
    \begin{equation*}
        \partial_x^2 \psi^A = \partial_x \left(\kappa(x) \psi^A\right) + E \psi^B.
    \end{equation*}
    Since $\psi^A, \psi^B,\partial_x \psi^A, \partial_x \psi^A \in L^2(\mathbb{R})$, we obtain $\partial_x^2 \psi^A \in L^2(\mathbb{R})$ when $\kappa'(x) \in L^\infty(\mathbb{R})$. Similarly, this process can be continued to all higher order derivatives of $\psi^A, \psi^B$ when all derivatives of $\kappa$ are bounded in $L^\infty(\mathbb{R})$. Therefore, we obtain all derivatives of $\psi^A, \psi^B$ are $L^2$-integrable and \eqref{eqn:psi0-Hs-bd} holds.\\
    
    The proof of (c) is briefly explained in the main text by directly calculating the zero eigenstate. Hence, we omit the proof here. Thus, we complete the proof of Lemma \ref{lemma:sm-spec-Dk}.
\end{proof}

We now state the classical ODE result in \cite{coddington1956theory}, which is an exercise problem (exercise 29) on page 104 in \cite{coddington1956theory}\footnote{The original exercise problem is stated for $\mathbb{R}^n$. Here we choose $n=2$ to fit into our setting.}. Since this is an exercise problem, we also provide the proof. 
\begin{lemma}\label{lemma:sm-lemma-ode-asy}
Let $\bm{A} \in \mathbb{R}^{2\times 2}$ be a constant matrix and $\bm{R}(t)\in \mathbb{R}^{2\times 2}$ a continuous and integrable matrix such that
\begin{equation}
    \int_0^\infty \|\bm{R}(t)\|\,dt < \infty .
\end{equation}
If $\lambda_1,\lambda_2$ are eigenvalues of $\bm{A}$ and
$\bm{p}_1(t),\bm{p}_2(t)$ are the corresponding eigenvectors, then the following ODE problem
\begin{equation}\label{eqn:sm-lemma-ode}
    \bm{x}' = \bm{A} \bm{x} + \bm{R}(t)\bm{x}
\end{equation}
has two solutions $\bm{\varphi}_1(t),\bm{\varphi}_2(t)$ such that
\begin{equation}
    \lim_{t\to +\infty} \bm{\varphi}_i(t)e^{-\lambda_i t} = \bm{p}_i, \qquad i=1,2.\label{eqn:sm-lemma-asym}
\end{equation}
Moreover, if $\det A < 0$ and $\lambda_1 < 0 < \lambda_2$, then every bounded solution of \eqref{eqn:sm-lemma-ode} is a scalar multiple of $\bm \varphi_1(t)$.
\end{lemma}

\begin{proof}
We first prove the existence of $\bm \varphi_1(t)$ and $\bm \varphi_2(t)$. It is sufficient to construct a solution that grows like $e^{\lambda_1 t}$ when $t \to \infty$. Let us write
\begin{equation*}
    e^{\bm{A}t}=\bm{Y}_1(t)+\bm{Y}_2(t),
\end{equation*}
where $\bm{Y}_1(t)$ corresponds to $e^{\lambda_1 t}$ and $\bm{Y}_2(t)$ is the sum of all terms $e^{\lambda_2 t}$. Since there are only two eigenvalues, there exists $\sigma < 0, \delta>0$ such that
$\lambda_1\le\sigma-\delta$ and $\lambda_2 \ge \sigma$.
Hence there exists constants $K_1,K_2>0$ with
\begin{equation}\label{eq:Y-bounds}
\|\bm{Y}_1(t)\|\le K_1 e^{(\sigma-\delta)t}\quad\text{for }t\ge0,\qquad
\|\bm{Y}_2(t)\|\le K_2 e^{\sigma t}\quad \text{for }t\le0.
\end{equation}

\noindent \textbf{Construction of a solution via Picard iteration:} We look for a solution on $[a,\infty)$ for some $a\ge1$ in the form of a
solution of the integral equation
\begin{equation}\label{eq:fixed-point}
\bm{\psi}(t) = e^{\lambda_1 t}\bm{p}_1
  + \int_a^t \bm{Y}_1(t-s)\bm{R}(s)\bm{\psi}(s)\,ds
  - \int_t^\infty \bm{Y}_2(t-s)\bm{R}(s)\bm{\psi}(s)\,ds ,
  \qquad t\ge a .
\end{equation}
One can check that any continuously differentiable solution $\bm{\psi}$ of \eqref{eq:fixed-point} satisfies \eqref{eqn:sm-lemma-ode}. Therefore, a solution to \eqref{eq:fixed-point} is a solution to the desired ODE system \eqref{eqn:sm-lemma-ode}.

\begin{equation*}
    \bm{\psi}_0(t)=e^{\lambda_1 t}\bm{p}_1,\qquad t\ge a\geq 1,
\end{equation*}
and for $\ell\ge0$, we construct
\begin{equation*}
    \bm{\psi}_{\ell+1}(t)
 = e^{\lambda_1 t}\bm{p}_1
 + \int_a^t \bm{Y}_1(t-s)\bm{R}(s)\bm{\psi}_\ell(s)\,ds
 - \int_t^\infty \bm{Y}_2(t-s)\bm{R}(s)\bm{\psi}_\ell(s)\,ds .
\end{equation*}
We can choose $K_0>0$ so that
\begin{equation*}
    \|\bm{\psi}_0(t)\|=\|e^{\lambda_1 t}\bm{p}_1\|\le K_0 e^{\sigma t},\qquad t\ge a \geq 1.
\end{equation*}
Moreover, we can choose sufficiently large $a$ such that
\begin{equation}\label{eq:a-choice}
(K_1+K_2)\int_a^\infty \|\bm{R}(s)\|\,ds < \frac12.
\end{equation}

We now claim that, for all $\ell\ge0$ and all $t\ge a\geq 1$, we have
\begin{equation}\label{eq:iteration-estimate}
\|\bm{\psi}_{\ell+1}(t)-\bm{\psi}_\ell(t)\|
\le \frac{K_0 e^{\sigma t}}{2^{\ell+1}} .
\end{equation}
We prove our claim by induction. For the base case with $\ell=0$, we have
\begin{equation*}
    \bm{\psi}_1(t)-\bm{\psi}_0(t)
 = \int_a^t \bm{Y}_1(t-s)\bm{R}(s)\bm{\psi}_0(s)\,ds
 - \int_t^\infty \bm{Y}_2(t-s)\bm{R}(s)\bm{\psi}_0(s)\,ds .
\end{equation*}
Using \eqref{eq:Y-bounds} and $\|\bm{\psi}_0(s)\|\le K_0 e^{\sigma s}$, we obtain
\begin{align*}
\|\bm{\psi}_1(t)-\bm{\psi}_0(t)\| &\le
\int_a^t K_1 e^{(\sigma-\delta)(t-s)}\|\bm{R}(s)\|K_0 e^{\sigma s}\,ds + \int_t^\infty K_2 e^{\sigma(t-s)}\|\bm{R}(s)\|K_0 e^{\sigma s}\,ds\\
&= K_0 e^{\sigma t} \Bigg[
\int_a^t K_1 e^{-\delta(t-s)}\|\bm{R}(s)\|\,ds
+\int_t^\infty K_2\|\bm{R}(s)\|\,ds
\Bigg]\\
&\le K_0 e^{\sigma t}(K_1+K_2)\int_a^\infty\|\bm{R}(s)\|\,ds \le \frac12 K_0 e^{\sigma t},
\end{align*}
where the last inequality comes from \eqref{eq:a-choice}. Therefore, we prove our claim for $\ell=0$. We now assume \eqref{eq:iteration-estimate} holds for some $\ell\ge0$ and prove for the $\ell + 1$ case. We observe that 
\begin{align*}
\bm{\psi}_{\ell+1}(t)-\bm{\psi}_\ell(t)
&= \int_a^t \bm{Y}_1(t-s)\bm{R}(s)\bigl(\bm{\psi}_\ell(s)-\bm{\psi}_{\ell-1}(s)\bigr)\,ds\\
&\quad -\int_t^\infty \bm{Y}_2(t-s)\bm{R}(s)\bigl(\bm{\psi}_\ell(s)-\bm{\psi}_{\ell-1}(s)\bigr)\,ds .
\end{align*}
Using the induction hypothesis
\begin{equation*}
    \|\bm{\psi}_\ell(s)-\bm{\psi}_{\ell-1}(s)\|
\le \frac{K_0 e^{\sigma s}}{2^\ell},\qquad s\ge a
\end{equation*}
and \eqref{eq:Y-bounds}, we obtain
\begin{align*}
\|\bm{\psi}_{\ell+1}(t)-\bm{\psi}_\ell(t)\|
&\le \frac{K_0 e^{\sigma t}}{2^\ell}
(K_1+K_2)\int_a^\infty\|\bm{R}(s)\|\,ds\le \frac{K_0 e^{\sigma t}}{2^{\ell+1}},
\end{align*}
which completes our induction proof. Thus, by a standard Picard iteration technique, we can show that $\bm{\psi}_\ell(t)$ is a Cauchy sequence and its limit
\begin{equation}\label{eqn:sm-lemma-ode-sol}
    \bm{\varphi}_1(t) = \lim_{\ell\rightarrow \infty} \bm{\psi}_l(t)
\end{equation}
is the solution of \eqref{eq:fixed-point} and \eqref{eqn:sm-lemma-ode}.\\

\noindent \textbf{The asymptotic behavior of $\bm{\varphi}_1(t)$:} Lastly, we present the asymptotic behavior of $\bm{\varphi}_1(t)$. We multiply \eqref{eq:fixed-point} by $e^{-\lambda_1 t}$ and write it as
\begin{equation*}
    e^{-\lambda_1 t}\bm{\varphi}_1(t)-\bm{p}_1 = \bm{I}_1(t)-\bm{I}_2(t),
\end{equation*}
where
\begin{align*}
    \bm{I}_1(t) &= e^{-\lambda_1 t}
\int_a^t \bm{Y}_1(t-s)\bm{R}(s)\bm{\varphi}_1(s)\,ds,\\[1mm]
\bm{I}_2(t) &= e^{-\lambda_1 t}
\int_t^\infty \bm{Y}_2(t-s)\bm{R}(s)\bm{\varphi}_1(s)\,ds .
\end{align*}
We shall show that $\bm{I}_1(t)\to0$ and $\bm{I}_2(t)\to0$ as $t\to\infty$. Using \eqref{eq:Y-bounds} and $\|\bm{\varphi}_1(s)\|\le C e^{\sigma s}$
(we can choose $C=2K_0$),
\begin{align*}
\|\bm{I}_1(t)\|
&\le e^{-\sigma t}\int_a^t
K_1 e^{(\sigma-\delta)(t-s)}\|\bm{R}(s)\| C e^{\sigma s}\,ds\\
&= C K_1 \int_a^t e^{-\delta(t-s)}\|\bm{R}(s)\|\,ds .
\end{align*}
The right-hand side is the convolution of $\|\bm{R}(s)\|$ with $e^{-\delta s}$.
Since $\|\bm{R}\|\in L^1([a,\infty))$ and $0\le e^{-\delta(t-s)}\le e^{-\delta(t-a)}$, by a standard integral bound, we have 
\begin{equation*}
    \lim_{t \rightarrow 0}\int_a^t e^{-\delta(t-s)}\|\bm{R}(s)\|\,ds = 0.
\end{equation*}
Therefore, we obtain $\|\bm{I}_1(t)\|\to0$ as $t \to 0$.\\

Similarly, to bound $\|\bm{I}_2(t)\|$, we use
$\|\bm{Y}_2(t-s)\|\le K_2 e^{\sigma(t-s)}$ for $s\ge t$ and
$\|\bm{\varphi}_1(s)\|\le C e^{\sigma s}$:
\begin{align*}
\|\bm{I}_2(t)\|
&\le e^{-\sigma t}\int_t^\infty
K_2 e^{\sigma(t-s)}\|\bm{R}(s)\| C e^{\sigma s}\,ds = C K_2 \int_t^\infty \|\bm{R}(s)\|\,ds,
\end{align*}
and the tail of the integrable function $\|\bm{R}(s)\|$ tends to $0$ as
$t\to\infty$. Hence $\|\bm{I}_2(t)\|\to0$. Therefore
\begin{equation*}
    \lim_{t\to\infty}\bigl(e^{-\lambda_1 t}\bm{\varphi}_1(t)-\bm{p}_1\bigr)
= \lim_{t\to\infty}\bm{I}_1(t)-\bm{I}_2(t)=0,
\end{equation*}
which completes the proof of the existence of $\bm \varphi_1(t)$. \\

\noindent \textbf{Final construction of $\bm \varphi_1(t)$ and $\bm \varphi_2(t)$:} So far, we have constructed a solution $\bm \varphi_1(t)$ on $[a,\infty)$ such that $\bm \varphi_1(t) \rightarrow e^{\lambda_1 t}\bm p_1$. We can extend the solution between $[0,a]$ to complete our construction of $\bm \varphi_1(t)$. Similarly, using $\bm \psi_0(t) = e^{\lambda_2 t}\bm p_2$, the Picard iteration leads to a solution $\bm \varphi_2(t) \rightarrow e^{\lambda_2 t}\bm p_2$ as $t\rightarrow \infty$.\\

\noindent \textbf{The uniqueness of the bounded solution:} Lastly, we prove the argument related to $\det A < 0$ given the existence of $\bm \varphi_1(t)$ and $\bm \varphi_2(t)$. Since $\lambda_1 < 0 < \lambda_2$, $\bm \varphi_1(t)$ and $\bm \varphi_2(t)$ are linearly independent due to their different asymptotic behavior at $\infty$. Therefore, every bounded solution is a scalar multiple of $\bm \varphi_1(t)$.
\end{proof}

\section{Some useful calculations and bounds}\label{app:some-bds}

In this section, we state and prove several useful bounds and calculations.\\

\noindent \underline{\textit{ $H^\delta$ is self-adjoint}:} We verify that $H^\delta=H^0+\delta P^\delta$ in \eqref{eqn:ham-hon-delta-approx} is self-adjoint, where 
\begin{align*}
   (P^\delta \psi)_{m,n} \begin{pmatrix}
        (P^\delta \psi)_{m,n}^A\\
        (P^\delta \psi)_{m,n}^B
    \end{pmatrix} = \sum_{\nu = 1}^3 \begin{pmatrix}
        f_\nu(\delta \Cell_{m,n}) \psi_{m+m_\nu,n+n_\nu}^B\\
        f_\nu(\delta \Cell_{m-m_\nu,n-n_\nu}) \psi_{m-m_\nu,n-n_\nu}^A
    \end{pmatrix}
\end{align*}

Since $H^0$ is self-adjoint, it suffices to check that  $P^\delta[\bm{u}]$ is. For any $\varphi, \psi \in l^2(\mathbb{H})$, 
\begin{align*}
    \langle \varphi, P^\delta \psi\rangle_{l^2(\mathbb{H})} &= \sum_\nu \sum_{m,n} \overline{\varphi_{m,n}^A}\ f_\nu(\delta\Cell_{mn})\ \psi^B_{m+m_\nu,n+n_\nu}
    + \sum_\nu \sum_{m,n} \overline{\varphi_{m,n}^B}\ f_\nu(\delta\Cell_{m-m_\nu,n-n_\nu})\  \psi^A_{m-m_\nu,n-n_\nu} \\
    &=\sum_\nu \sum_{m,n} \overline{\varphi_{m-m_\nu,n-n_\nu}^A}\ f_\nu(\delta\Cell_{m-m_\nu,n-n_\nu})\ \psi^B_{m,n}
    + \sum_\nu \sum_{m,n } \overline{\varphi_{m+m_\nu,n+n_\nu}^B}\ f_\nu(\delta\Cell_{m,n})\  \psi^A_{m,n}\\
    &=  \sum_\nu \sum_{m,n } \overline{ f_\nu(\delta\Cell_{m,n})\ \varphi_{m+m_\nu,n+n_\nu}^B} \  \psi^A_{m,n} + \sum_\nu \sum_{m,n} \overline{ f_\nu(\delta\Cell_{m-m_\nu,n-n_\nu})\ \varphi_{m-m_\nu,n-n_\nu}^A}\ \psi^B_{m,n}\\
     &= \langle P^\delta  \varphi, \psi\rangle_{l^2(\mathbb{H})}.
\end{align*}

\noindent \underline{\textit{The real valued $E_2(k_\parallel)$ in \eqref{eqn:E2}}:} We prove that $E_2(k_\parallel) \in \mathbb{R}$ by the following calculation: 
\begin{align*}
    & \langle \Psi_0^A, \mathcal{R}_2^A\rangle_{L^2(\mathbb{R})} + \langle \Psi_0^B, \mathcal{R}_2^B\rangle_{L^2(\mathbb{R})}= -\frac{1}{4} k_\parallel^2 \int_{\mathbb{R}} \text{Re}(\overline{\Psi_0^A} \Psi_0^B) \: dX_1 +\frac{3}{4} \int_{\mathbb{R}} \text{Re}\left( \overline{\Psi_0^A} \partial_{X_1}^2 \Psi_0^B\right) \: dX_1 \\
    &\qquad - \frac{3}{2} \int_{\mathbb{R}} e^{i\frac{5}{3}\pi} \overline{\Psi_0^A} \partial_{X_1}^2 \Psi_0^B \: d X_1 - \frac{3}{2} \int_{\mathbb{R}} e^{-i\frac{5}{3}\pi} \overline{\Psi_0^B} \partial_{X_1}^2 \Psi_0^A\: d X_1\\
    &\qquad + \Big(\frac{\sqrt{3}}{4} k_\parallel i\Big)\int_{\mathbb{R}}
    \overline{\Psi_0^A} \partial_{X_1} \Psi_0^B \: dX_1 + \Big(\frac{\sqrt{3}}{4} k_\parallel i\Big)\int_{\mathbb{R}} \overline{\Psi_0^B} \partial_{X_1} \Psi_0^A \: dX_1\\
    &\qquad - \sqrt{3} t_1\int_{\mathbb{R}} e^{i\frac{5}{3}\pi} \overline{\Psi_0^A} \Big(f_2 \partial_{X_1} \Psi_0^B\Big) \: dX_1 + \sqrt{3} t_1\int_{\mathbb{R}} e^{-i\frac{5}{3}\pi} \overline{\Psi_0^B} \partial_{X_1} \Big(f_2\Psi_0^A\Big) \: dX_1.
\end{align*}
Using integration by parts, we obtain $\langle \Psi_0^A, \mathcal{R}_2^A\rangle_{L^2(\mathbb{R})} + \langle \Psi_0^B, \mathcal{R}_2^B\rangle_{L^2(\mathbb{R})} \in \mathbb{R}$ with
\begin{align}
    &\langle \Psi_0^A, \mathcal{R}_2^A\rangle_{L^2(\mathbb{R})} + \langle \Psi_0^B, \mathcal{R}_2^B\rangle_{L^2(\mathbb{R})}= -\frac{1}{4} k_\parallel^2 \int_{\mathbb{R}} \text{Re}(\overline{\Psi_0^A} \Psi_0^B) \: dX_1 + \frac{3}{4} \int_{\mathbb{R}} \text{Re}\left( \overline{\Psi_0^A} \partial_{X_1}^2 \Psi_0^B\right) \: dX_1 \nonumber\\
    &\qquad - 3 \int_{\mathbb{R}} \text{Re} \left(e^{i\frac{5}{3}\pi} \overline{\Psi_0^A} \partial_{X_1}^2 \Psi_0^B\right) \: d X_1 +\frac{\sqrt{3}}{4} k_\parallel \int_{\mathbb{R}} \text{Re} \left(i\overline{\Psi_0^A} \partial_{X_1} \Psi_0^B\right)\: d X_1 \nonumber\\
    &\qquad - 2\sqrt{3} t_1\int_{\mathbb{R}} \text{Re}\left(e^{i\frac{5}{3}\pi} \overline{\Psi_0^A} \Big(f_2 \partial_{X_1} \Psi_0^B\Big)\right) \: dX_1. \label{eqn:app-real-E2}
\end{align}

\noindent \underline{\textit{Bound on $\widehat{\Psi}\left(\frac{k}{\delta}\right) \chi(|k|\geq \delta^\tau)$:}} Given that $\tau \in (0,1)$ and $\Psi\in H^s(\mathbb R)$ with $s\geq 1$. Then the tail of $\widehat{\Psi}\left(\frac{k}{\delta}\right)$ is arbitrarily small, i.e.
\begin{align}\label{eqn:bd-psi-far}
    \left\|\widehat{\Psi}\!\left(\frac{k}{\delta}\right)\chi\!\left(|k|\ge \delta^\tau\right)\right\|_{L^2_k\left(\left[-\frac{2\pi}{\sqrt{3}}, \frac{2\pi}{\sqrt{3}}\right]\right)} \leq \delta^{\frac{1+2s(1-\tau)}{2}}\,\|\Psi\|_{H^s(\mathbb{R})}.
\end{align}
The proof of \eqref{eqn:bd-psi-far} uses a change of parameter $k = \delta q$ and a standard trick by multiplying and dividing $|q|^{2s}$. We write the left hand side of \eqref{eqn:bd-psi-far} as
\begin{align*}
    \int_{-\frac{2\pi}{\sqrt{3}}}^{\frac{2\pi}{\sqrt{3}}}\left|\widehat{\Psi}\!\left(\frac{k}{\delta}\right)\right|^2 \chi\!\left(|k|\ge \delta^\tau\right)\,dk &\leq \delta \int_{|q| \geq \delta^{\tau-1}} \left|\widehat{\Psi}(q)\right|^2 \,dq \leq \delta \int_{|q| \geq \delta^{\tau-1}} \frac{1}{|q|^{2s}} \Big(|q|^{2s} \left|\widehat{\Psi}(q)\right|^2\Big) \,dq.
\end{align*}
Since $|q| \geq \delta^{\tau-1}$, we have $|q|^{-2s} \leq \delta^{2s(1-\tau)}$. Therefore, our desired bound \eqref{eqn:bd-psi-far} holds.\\

\noindent \underline{\textit{Bounds on Fourier modes with $|m|\geq 1$}:} Given that $\Psi \in H^s(\mathbb{R})$ with $s\geq 1$, we have
\begin{equation}\label{eqn:app-sum-M-bd}
    \left\| \sum_{|m| \geq 1} \widehat{\Psi}\left(\frac{k + \frac{4\pi}{\sqrt{3}}m}{\delta}\right) \right\|_{L^2\left(\left[-\frac{2\pi}{\sqrt{3}}, \frac{2\pi}{\sqrt{3}}\right]\right)}
\lesssim
\delta^{\frac{2s+1}{2}}\|\Psi\|_{H^s(\mathbb{R})}.
\end{equation}
To simplify the proof, we define a constant $c:=\cconst$ and the following functions
\begin{equation}\label{eqn:app-sk}
    S(k)\;:=\;\sum_{|m|\ge 1} a_M(k), \qquad a_m(k):=\widehat{\Psi}\big(\frac{k+cm}{\delta}\big),
\qquad k\in \Icell = \left[-\frac{c}{2}, \frac{c}{2}\right].
\end{equation}
Using the new notation, it is equivalent to show that for any $\Psi\in H^s(\mathbb{R})$ with $s\geq 1$, we have
\begin{equation}\label{eqn:app-Hs-bd}
    \|S(k)\|_{L^2\left(\Icell\right)}^2 \leq C_s \delta^{2s+1} \|\Psi\|^2_{H^s(\mathbb{R})}.
\end{equation}
for constant $C_s>0$ that depends only on $s$ and $c$.

We now start our proof of \eqref{eqn:app-sum-M-bd}. We observe that for each $k\in I$, we use Cauchy-Schwarz inequality
\begin{equation}\label{eq:wcs}
|S(k)|^2 = \Big|\sum_{|m|\ge 1} a_m(k)\Big|^2
\le
\Big(\sum_{|m|\ge 1}|k+cm|^{-2s}\Big)
\Big(\sum_{|m|\ge 1}|k+cm|^{2s}|a_m(k)|^2\Big).
\end{equation}
Since $k\in [-c/2,c/2]$, we have $|k+cm|\ge c(|m|-1/2) \geq c|m|/2$ for all $|m|\geq 1$. Hence, we can define the constant
\begin{align*}
    A_s:=\sup_{k\in \Icell} \quad \sum_{|m|\ge 1}|k+cm|^{-2s}<\infty
\end{align*}
for all $s \geq 1$. Therefore, we have
\begin{equation}\label{eq:L2reduce}
\|S\|_{L^2\left(\Icell\right)}^2
\le
A_s \sum_{|m|\ge 1} \int_{-\frac{c}{2}}^{\frac{c}{2}} |k+cm|^{2s}\Big|\widehat{\Psi}\!\left(\frac{k+cm}{\delta}\right)\Big|^2\,dk.
\end{equation}
By choosing $\xi=\frac{k+cm}{\delta}$, we obtain $|k+cm|^{2s}=\delta^{2s}|\xi|^{2s}\le \delta^{2s}(1+|\xi|^2)^s$. For each $m$, we have
\begin{equation*}
    \int_{-\frac{c}{2}}^{\frac{c}{c}} |k+cm|^{2s}\Big|\widehat{\Psi}\!\left(\frac{k+cm}{\delta}\right)\Big|^2\,dk = \delta^{2s+1} \int_{\frac{c\left(m - \frac{1}{2}\right)}{\delta}}^{\frac{c\left(m + \frac{1}{2}\right)}{\delta}} (1+|\xi|^2)^s\Big|\widehat{\Psi}(\xi)\Big|^2\,d\xi.
\end{equation*}
Summing over $|m|\ge 1$ yields a disjoint union of intervals covering $\mathbb{R}\setminus \left(-\frac{c}{2\delta}, \frac{c}{2\delta}\right)$, hence
\begin{equation}\label{eq:cover}
\sum_{|m|\ge 1}  \int_{-\frac{c}{2}}^{\frac{c}{2}} |k+cm|^{2s}\Big|\widehat{\Psi}\!\left(\frac{k+cm}{\delta}\right)\Big|^2\,dk
\leq 
\delta^{2s+1}\int_{\mathbb{R}}(1+|\xi|^2)^s\Big|\widehat{\Psi}(\xi)\Big|^2\,d\xi=
\delta^{2s+1} \|\Psi\|_{H^s(\mathbb{R})}^2.
\end{equation}
Combining \eqref{eq:L2reduce} with \eqref{eq:cover}, we obtain the desired inequality \eqref{eqn:app-Hs-bd}.\\

\noindent \underline{\textit{Bounds on $\widetilde{I}_1[k;\bm \Psi_0, \mu, \delta]$ and $\widetilde{I}_2[k;\bm \Psi_0, \mu,\delta]$:}} Given that $\bm \Psi_0 \in H^s(\mathbb{R})$ for all $s \geq 1$, $|\mu|\leq M$, $\delta \in (0,\delta_0)$ and $\tau \in (0,1)$, We prove the bounds \eqref{eqn:bd-I-far-main} for $\widetilde{I}_i[k;\bm \Psi_0,\mu,\delta]$ with $i=1,2$, i.e.
\begin{align}\label{eqn:bd-I-far}
   \left\|\widetilde{I}_i[k; \bm \Psi_0, \mu, \delta] \: \chi(|k|\geq \delta^\tau)\right\|_{L^2\left(\left[-\frac{2\pi}{\sqrt{3}}, \frac{2\pi}{\sqrt{3}}\right]\right)} \lesssim \delta^{\frac{1}{2}} \norm{\bm \Psi_0}_{H^2(\mathbb{R})}, \qquad i=1,2,
\end{align}
Let us first explain the $\delta^{\frac{1}{2}}$-contributions in \eqref{eqn:bd-I-far}, which actually come from the following term:
\begin{align*}
    &\left\|\frac{2}{\sqrt{3} \delta^2} \left[\Gamma_3[k,k_\parallel,\delta] \: \widehat{\Psi_0^B}\left(\frac{k}{\delta}\right) \: + \Gamma_4[k,k_\parallel,\delta] \: \widehat{\Psi_0^A}\left(\frac{k}{\delta}\right) \right] \: \chi(|k|\geq \delta^\tau)\right\|_{\LtwoBrill}^2 \\
    \lesssim \: & \delta \int_{|\xi|\geq \delta^{\tau-1}} \frac{|k|^4}{\delta^4} \left(\left|\widehat{\Psi_0^B}(\xi)\right|^2 + \left|\widehat{\Psi_0^A}(\xi)\right|^2\right) \: d\xi\lesssim  \delta^{\frac{1}{2}} \Big( \norm{\partial_{X_1}^2 \Psi_0^A}_{L^2{(\mathbb{R})}} + \norm{\partial_{X_1}^2 \Psi_0^B}_{L^2{(\mathbb{R})}}\Big).
\end{align*}
The first inequality comes from the fact that $|\Gamma_3[k,k_\parallel,\delta]| + |\Gamma_4[k,k_\parallel,\delta]|\lesssim |k|^2$ (see near \eqref{eqn:gamma-3} and \eqref{eqn:gamma-4}). The rest terms in $\widetilde{I}_i[k;\bm \Psi_0,\mu,\delta]$ are arbitrarily small in the far-momentum regime: the Fourier modes with $|m|\geq 1$ are controlled by \eqref{eqn:app-sum-M-bd}, and the remaining contributions in \eqref{eqn:I1-full} and \eqref{eqn:I2-full} can be estimated by
\begin{align*}
    \text{(i)} \ & \left\|-\frac{2}{\sqrt{3}}\mu\widehat{\bm \Psi_0}\left(\frac{k}{\delta}\right) \: \chi(|k|\geq \delta^\tau)\right\|_{\LtwoBrill} \lesssim \delta^{1+2s(1-\tau)} \: \|\bm \Psi_0\|_{H^s(\mathbb{R})},\\
    \text{(ii)} \ & \left\|\chi(|k|\geq \delta^\tau) \sum_{m \in \mathbb{Z}} \Bigg[ \ff_2 \left(\frac{\sqrt{3}}{2} \delta m\right)-\ff_2\left(\frac{\sqrt{3}}{2} \delta (m+2) \right)\Bigg] \bm \Psi_0\left(\frac{\sqrt{3}}{2} \delta (m+2) \right)e^{-ik\frac{\sqrt{3}}{2}m} \right\|_{\LtwoBrill} \\
    \lesssim \: & \delta \:\left\|\chi(|k|\geq \delta^\tau) \sum_{m \in \mathbb{Z}} \bm \Psi_0\left(\frac{\sqrt{3}}{2} \delta (m+2) \right)e^{-ik\frac{\sqrt{3}}{2}m} \right\|_{\LtwoBrill} \\
    \lesssim  \: & \norm{\chi(|k|\geq \delta^\tau) \sum_{m \in \mathbb{Z}} \widehat{\bm \Psi_0}\left(\frac{k + \frac{4\pi}{\sqrt{3}}m}{\delta}\right)}_{\LtwoBrill} \lesssim \Big(\delta^{1+2s(1-\tau)} + \delta^{s+\frac{1}{2}}\Big) \: \|\bm \Psi_0\|_{H^s(\mathbb{R})}.
\end{align*}
We first prove (i): for any Schwartz function $\Psi \in H^s(\mathbb{R})$ for any $s \geq 1$, its far-momentum part is sufficiently small and bounded by
\begin{align}
    &\norm{\widehat{\Psi}\left(\frac{k}{\delta}\right) \: \chi(|k|\geq \delta^\tau)}_{\LtwoBrill}^2 \leq \delta^{1+2s(1-\tau)} \|\Psi\|_{H^s(\mathbb{R})}.\label{eqn:main-far-adv}
\end{align}
The proof of \eqref{eqn:main-far-adv} follows by a straightforward calculation: 
\begin{align*}
    &\norm{\widehat{\Psi}\left(\frac{k}{\delta}\right) \: \chi(|k|\geq \delta^\tau)}_{\LtwoBrill}^2 = \delta \int_{|\xi|\geq \delta^{\tau-1}} |\widehat{\Psi}(\xi)|^2 \: d\xi \qquad (\text{take }k = \delta \xi )\nonumber\\
    &\leq \delta^{1+2s(1-\tau)} \int_{|\xi|\geq \delta^{\tau-1}} |\xi|^{2s} \: |\widehat{\Psi}(\xi)|^2 \: d\xi \leq \delta^{1+2s(1-\tau)} \|\Psi\|_{H^s(\mathbb{R})},
\end{align*}
where the first inequality on the second line holds since it uses $|\xi|^{-1} \leq \delta^{1-\tau}$. For (ii), we observe that the first inequality in (ii) comes from the uniform bound on $f_2'$ and the second inequality comes from \eqref{eqn:psf-scaled-1}. The last inequality in (ii) comes from (i) and \eqref{eqn:app-sum-M-bd}. Thus, we obtain the bound \eqref{eqn:bd-I-far}.\\

\noindent \underline{\textit{Lipschitz estimates in Proposition \ref{prop:far-energy}}:} We prove the Lipschitz continuity in $\mu$ for the far-momentum solution in Proposition \ref{prop:far-energy}, i.e. the bounds \eqref{eqn:A-map-lip-mu} and \eqref{eqn:B-map-lip-mu}. It is sufficient to show the Lipschitz estimates for the contraction mapping $\mathcal{E}[\bm{\widetilde{\eta}}_\text{far};\bm{\widetilde{\eta}}_\text{near},\mu, \delta]$ in \eqref{eqn:far-contraction}, i.e. for bounded $\bm{\widetilde{\eta}}_\text{far} \in B_{\text{far},\delta^\tau}(\rho_\delta)$ and $\bm{\widetilde{\eta}}_\text{near} \in B_{\text{near},\delta^\tau}(R)$, we have
\begin{subequations}\label{eqn:bd-map-E}
    \begin{align}
        &\norm{\mathcal{E}[\bm{\widetilde{\eta}}_\text{far};\bm{\widetilde{\eta}}_\text{near},\mu_1, \delta] - \mathcal{E}[\bm{\widetilde{\eta}}_\text{far};\bm{\widetilde{\eta}}_\text{near},\mu_2, \delta]}_{\LtwoBrill} \label{eqn:bd-map-E-1}\\
        &\lesssim  \ |\mu_1 - \mu_2| \left[\delta^{\frac{1}{2}-\tau} \norm{\bm \Psi_0^A}_{L^2(\mathbb{R})} + \delta^{2-\tau} \norm{\bm{\widetilde{\eta}}_\text{far}}_{\LtwoBrill}\right], \nonumber\\
        &\norm{\mathcal{E}[\bm{\widetilde{\eta}}_\text{far}^1;\bm{\widetilde{\eta}}_\text{near},\mu, \delta] - \mathcal{E}[\bm{\widetilde{\eta}}_\text{far}^2;\bm{\widetilde{\eta}}_\text{near},\mu, \delta]}_{\LtwoBrill} \lesssim \delta^{1-\tau} \norm{\bm{\widetilde{\eta}}_\text{far}^1 - \bm{\widetilde{\eta}}^2_\text{far}}_{\LtwoBrill}.\label{eqn:bd-map-E-2}
    \end{align}
\end{subequations}
We briefly explain why \eqref{eqn:bd-map-E} leads to \eqref{eqn:A-map-lip-mu} and \eqref{eqn:B-map-lip-mu}. We fix $\delta, \bm{\widetilde{\eta}}_\text{near}$ and assume $\bm{\widetilde{\eta}}_\text{far}^1,\bm{\widetilde{\eta}}_\text{far}^2$ are the solutions of \eqref{eqn:far-contraction} for $\mu_1, \mu_2$. Substituting the representation \eqref{eqn:far-affine} into these solutions, we obtain
\begin{align}
    \bm{\widetilde{\eta}}_\text{far}^1 - \bm{\widetilde{\eta}}_\text{far}^2 &= [\mathcal{A}\widetilde{\bm{\eta}}_\text{near}](k;\mu_1, \delta) + \mathcal{B}(k;\mu_1,\delta) - [\mathcal{A}\widetilde{\bm{\eta}}_\text{near}](k;\mu_2, \delta) - \mathcal{B}(k;\mu_2,\delta) \nonumber\\
    &=\mathcal{E}[\bm{\widetilde{\eta}}_\text{far}^1;\bm{\widetilde{\eta}}_\text{near},\mu_1, \delta] - \mathcal{E}[\bm{\widetilde{\eta}}_\text{far}^2;\bm{\widetilde{\eta}}_\text{near},\mu_2, \delta].\label{eqn:map-E-equality}
\end{align}
We observe that when $\Psi_0^A = \Psi_0^B$, the source terms in \eqref{eqn:eta-far} vanish. Consequently, we have $\mathcal{B}(k;\mu,\delta) = 0$ and \eqref{eqn:map-E-equality} becomes
\begin{align*}
    &\bm{\widetilde{\eta}}_\text{far}^1 - \bm{\widetilde{\eta}}_\text{far}^2 = [\mathcal{A}\widetilde{\bm{\eta}}_\text{near}](k;\mu_1, \delta)  - [\mathcal{A}\widetilde{\bm{\eta}}_\text{near}](k;\mu_2, \delta)\\
    =& \mathcal{E}[\bm{\widetilde{\eta}}_\text{far}^1;\bm{\widetilde{\eta}}_\text{near},\mu_1, \delta] - \mathcal{E}[\bm{\widetilde{\eta}}_\text{far}^1;\bm{\widetilde{\eta}}_\text{near},\mu_2, \delta] + \mathcal{E}[\bm{\widetilde{\eta}}_\text{far}^1;\bm{\widetilde{\eta}}_\text{near},\mu_2, \delta] - \mathcal{E}[\bm{\widetilde{\eta}}_\text{far}^2;\bm{\widetilde{\eta}}_\text{near},\mu_2, \delta].
\end{align*}
Using \eqref{eqn:bd-map-E-1} and \eqref{eqn:bd-map-E-2}, we obtain \eqref{eqn:A-map-lip-mu}
\begin{align*}
    &\norm{\bm{\widetilde{\eta}}_\text{far}^1 - \bm{\widetilde{\eta}}_\text{far}^2}_{\LtwoBrill} = \norm{[\mathcal{A}\widetilde{\bm{\eta}}_\text{near}](k;\mu_1, \delta)  - [\mathcal{A}\widetilde{\bm{\eta}}_\text{near}](k;\mu_2, \delta)}_{\LtwoBrill} \\
    \lesssim \: & |\mu_1 - \mu_2| \Big(\delta^{2-\tau}\norm{\bm{\widetilde{\eta}}_\text{far}^1}_{\LtwoBrill}\Big) + \delta^{1-\tau} \norm{\bm{\widetilde{\eta}}_\text{far}^1 - \bm{\widetilde{\eta}}_\text{far}^2}_{\LtwoBrill}\\
    \Rightarrow \qquad & \norm{[\mathcal{A}\widetilde{\bm{\eta}}_\text{near}](k;\mu_1, \delta)  - [\mathcal{A}\widetilde{\bm{\eta}}_\text{near}](k;\mu_2, \delta)}_{\LtwoBrill} \lesssim |\mu_1 - \mu_2| \Big(\delta^{2-\tau}\norm{\bm{\widetilde{\eta}}_\text{far}^1}_{\LtwoBrill}\Big) \\
    \lesssim \: & |\mu_1 - \mu_2| \Big(\delta^{3-2\tau} \|\widetilde{\bm{\eta}}_\text{near}\|_{\LtwoBrill}\Big) \lesssim  |\mu_1 - \mu_2| \Big(\delta^{1-\tau} \|\widetilde{\bm{\eta}}_\text{near}\|_{\LtwoBrill}\Big).
\end{align*}
Using a similar procedure by taking $\bm{\widetilde{\eta}}_\text{near} = 0$, we obtain \eqref{eqn:B-map-lip-mu} from the bounds \eqref{eqn:bd-map-E}.

We now prove \eqref{eqn:bd-map-E} by first writing the explicit form of the contraction map $\mathcal{E}[\bm{\widetilde{\eta}}_\text{far};\bm{\widetilde{\eta}}_\text{near},\mu, \delta] = (\mathcal{E}_1[\bm{\widetilde{\eta}}_\text{far};\bm{\widetilde{\eta}}_\text{near},\mu, \delta], \mathcal{E}_2[\bm{\widetilde{\eta}}_\text{far};\bm{\widetilde{\eta}}_\text{near},\mu, \delta])^T$ with
\begin{subequations}\label{eqn:app-map-E}
    \begin{align}
        \mathcal{E}_1[\bm{\widetilde{\eta}}_\text{far};\bm{\widetilde{\eta}}_\text{near},\mu, \delta] &= -\frac{\chi(|k|\geq \delta^\tau)}{\Gamma_2(k,k_\parallel,\delta)} \bigg(\widetilde{I}_2[k;\bm \Psi_0, \mu, \delta] -\delta \widetilde{F}_2[k;\bm{\widetilde{\eta}}_\text{near}]\bigg) \nonumber\\
        &+ \frac{\delta E_1}{\Gamma_2(k,k_\parallel,\delta)} \widetilde{\eta}^B_\text{far}(k) + \frac{\delta^2 \mu}{\Gamma_2(k,k_\parallel,\delta)} \widetilde{\eta}^B_\text{far}(k) + \delta \frac{\chi(|k|\geq \delta^\tau)}{\Gamma_2(k,k_\parallel,\delta)} \widetilde{F}_2[k;\bm{\widetilde{\eta}}_\text{far}]. \label{eqn:map-E1}\\
        \mathcal{E}_2[\bm{\widetilde{\eta}}_\text{far};\bm{\widetilde{\eta}}_\text{near},\mu, \delta] &= -\frac{\chi(|k|\geq \delta^\tau)}{\Gamma_1(k,k_\parallel,\delta)} \bigg(\widetilde{I}_1[k;\bm \Psi_0, \mu, \delta] -\delta \widetilde{F}_1[k;\bm{\widetilde{\eta}}_\text{near}]\bigg) \nonumber\\
        &+ \frac{\delta E_1}{\Gamma_1(k,k_\parallel,\delta)} \widetilde{\eta}^A_\text{far}(k) + \frac{\delta^2 \mu}{\Gamma_1(k,k_\parallel,\delta)} \widetilde{\eta}^A_\text{far}(k) + \delta \frac{\chi(|k|\geq \delta^\tau)}{\Gamma_1(k,k_\parallel,\delta)} \widetilde{F}_1[k;\bm{\widetilde{\eta}}_\text{far}]. \label{eqn:map-E2}
    \end{align}
\end{subequations}
We prove \eqref{eqn:bd-map-E} for $\mathcal{E}_1[\bm{\widetilde{\eta}}_\text{far};\bm{\widetilde{\eta}}_\text{near},\mu,\delta]$; the estimate for $\mathcal{E}_2$ is analogous. Since $|k|^{-1} \leq \delta^{-\tau}$ on the far-momentum regime, the proof for $\mathcal{E}_1$ follows from the bounds below: for $i=1,2$, we have
\begin{align*}
    \text{(i)} \quad &\norm{-\frac{\chi(|k|\geq \delta^\tau)}{\Gamma_i(k,k_\parallel,\delta)} \widetilde{I}_i[k;\bm \Psi_0, \mu_1, \delta] + \frac{\chi(|k|\geq \delta^\tau)}{\Gamma_i(k,k_\parallel,\delta)} \widetilde{I}_i[k;\bm \Psi_0, \mu_2, \delta]}_{\LtwoBrill} \\
    \lesssim \:& |\mu_1 - \mu_2| \: \delta^{-\tau} \: \norm{\chi(|k|\geq \delta^\tau)\sum_{m \in \mathbb{Z}} \widehat{\bm \Psi_0}\left(\frac{k + \frac{4\pi}{\sqrt{3}}m}{\delta}\right)}_{\LtwoBrill} \\
    \lesssim \:& |\mu_1 - \mu_2| \delta^{\frac{1}{2}-\tau} \norm{\bm \Psi_0}_{L^2(\mathbb{R})},\\
    \text{(ii)} \quad & \norm{\frac{\delta^s}{\Gamma_i(k,k_\parallel,\delta)} \widetilde{\bm \eta}_\text{far}(k)}_{\LtwoBrill} \lesssim \delta^{s-\tau}  \norm{\bm{\widetilde{\eta}}_\text{far}}_{\LtwoBrill}, \qquad (\text{use }s=1,2)\\
    \text{(iii)} \quad & \norm{\delta \frac{\chi(|k|\geq \delta^\tau)}{\Gamma_i(k,k_\parallel,\delta)} \widetilde{F}_i[k;\bm{\widetilde{\eta}}_\text{far}]}_{\LtwoBrill} \lesssim \delta^{1-\tau} \norm{\bm{\widetilde{\eta}}_\text{far}}_{\LtwoBrill}.
\end{align*}
The bound (i) comes from \eqref{eqn:main-far-adv} and \eqref{eqn:app-sum-M-bd}. The bound (iii) comes from \eqref{eqn:bd-F-main}. Thus, we finish our proof  of \eqref{eqn:bd-map-E} for $\mathcal{E}_1[\bm{\widetilde{\eta}}_\text{far};\bm{\widetilde{\eta}}_\text{near},\mu,\delta]$.\\

\noindent \underline{\textit{Bounds in Lemma \ref{lemma:eta-near-bds}}:} We now prove \eqref{eqn:near-op-bd} and \eqref{eqn:MN-bds-main}. The bound on $\widehat{\mathcal{D} ^\delta}-\widehat{\mathcal{D} }$ in \eqref{eqn:near-op-bd} follows the same argument as in the proof of Proposition 6.13 of \cite{fefferman2017topologically}. To prove the bounds for $\widehat{L^\delta}(\mu)$, we prove \eqref{eqn:near-op-bd} and \eqref{eqn:lip-L-delta} for $\widehat{L^\delta_1}(\mu)$ as a representative example. We begin by proving the bound \eqref{eqn:near-op-bd}, which follows from the estimates below: for $i=1,2$,
\begin{align*}
    \text{(i)}\quad &\norm{\frac{1}{\delta} \Gamma_3(\delta\xi,k_\parallel,\delta) \widehat{\beta}_\text{near}^B(\xi)}_{L^2_\xi(\mathbb{R})} \lesssim \delta^{-1} \norm{\delta^2|\xi|^2 \widehat{\beta}_\text{near}^B(\xi)}_{L^2_\xi(\mathbb{R})} \\
    &\lesssim \delta^{\tau-1}\norm{\delta|\xi| \widehat{\beta}_\text{near}^B(\xi)}_{L^2_\xi(\mathbb{R})} \lesssim \delta^\tau \norm{\widehat{\bm{\beta}}_\text{near}}_{L^{2,1}_\xi(\mathbb{R})},\\
    \text{(ii)}\quad & \norm{\big(e^{i\sqrt{3}\delta \xi} - 1\big)\chi(|\xi|\leq \delta^{\tau-1}) \widehat{\ff_i \beta_\text{near}^B} (\xi)}_{L^2_\xi(\mathbb{R})} \lesssim \norm{\delta|\xi| \widehat{\ff_i \beta_\text{near}^B} (\xi)}_{L^2_\xi(\mathbb{R})} \\
    &\lesssim \delta \norm{\widehat{\ff_i \beta_\text{near}^B} (\xi)}_{L^{2,1}_\xi(\mathbb{R})} \lesssim \delta \norm{\widehat{\bm{\beta}}_\text{near}}_{L^{2,1}_\xi(\mathbb{R})},\\
    \text{(iii)}\quad & \norm{\sum_{|m|\geq 1} \widehat{\ff_i \beta_\text{near}^B} \left(\frac{\delta \xi + \frac{4\pi}{\sqrt{3}}m}{\delta}\right)}_{L^2_\xi(\mathbb{R})} = \delta^{-\frac{1}{2}} \norm{\sum_{|m|\geq 1} \widehat{\ff_i \beta_\text{near}^B} \left(\frac{k + \frac{4\pi}{\sqrt{3}}m}{\delta}\right)}_{\LtwoBrill} \\
    &\lesssim \delta \|\widehat{\ff_i \beta_\text{near}^B}\|_{L^{2,1}_\xi(\mathbb{R})} \lesssim \delta \|\widehat{\bm{\beta}}_\text{near}\|_{L^{2,1}_\xi(\mathbb{R})},\\
    \text{(iv)}\quad & \norm{\delta \chi(|\xi|\leq \delta^{\tau-1}) \widetilde{F}_1[\delta \xi;\mathcal{A}\bm{\widetilde{\eta}}_\text{near}(\mu,\delta)]}_{L^2_\xi(\mathbb{R})} \lesssim \delta^{\frac{1}{2}} \norm{ \widetilde{F}_1[k;\mathcal{A}\bm{\widetilde{\eta}}_\text{near}(\mu,\delta)]}_{\LtwoBrill}\\
    &\lesssim \delta^{\frac{1}{2}} \delta^{1-\tau} \norm{\bm{\widetilde{\eta}}_\text{near}}_{\LtwoBrill} \lesssim \delta^{\frac{1}{2}} \delta^{1-\tau} \norm{\bm{\widetilde{\eta}}_\text{near}}_{\LtwoBrill} \lesssim \delta^{1-\tau} \norm{\bm{\widehat{\beta}}_\text{near}}_{L^2_\xi(\mathbb{R})},\\
    \text{(v)}\quad &\norm{\delta \chi(|\xi|\leq \delta^{\tau-1}) \sum_{m \in \mathbb{Z}} \left[ \ff_i\left(\frac{\sqrt{3}}{2}\delta m\right) - \ff_i\left(\frac{\sqrt{3}}{2}\delta (m+2)\right) \right] \beta_\text{near}^B \left(\frac{\sqrt{3}}{2}\delta (m+2)\right) e^{-i \delta \xi\frac{\sqrt{3}}{2}m}}_{L^2_\xi(\mathbb{R})}\\
    & \lesssim \delta^2 \norm{\chi(|\xi|\leq \delta^{\tau-1}) \sum_{m \in \mathbb{Z}} \beta_\text{near}^B \left(\frac{\sqrt{3}}{2}\delta (m+2)\right) e^{-i \delta \xi\frac{\sqrt{3}}{2}m}}_{L^2_\xi(\mathbb{R})}\\
    &\lesssim \delta \norm{\sum_{m \in \mathbb{Z}} \widehat{\beta}_\text{near}^B\left(\frac{\delta \xi + \frac{4\pi}{\sqrt{3}}m}{\delta}\right)}_{L^2_\xi(\mathbb{R})} \lesssim \delta \norm{\widehat{\beta}_\text{near}^B(\xi)}_{L^2_\xi(\mathbb{R})} + \delta \norm{\sum_{|m|\geq 1} \widehat{\beta}_\text{near}^B\left(\frac{\delta \xi + \frac{4\pi}{\sqrt{3}}m}{\delta}\right)}_{L^2_\xi(\mathbb{R})} \\
    &\lesssim \delta \norm{\widehat{\beta}_\text{near}^B}_{L^2_\xi(\mathbb{R})} + \delta^2 \|\widehat{\beta}_\text{near}^B\|_{L^{2,1}_\xi(\mathbb{R})} \lesssim \delta \|\widehat{\bm{\beta}}_\text{near}\|_{L^{2,1}_\xi(\mathbb{R})}.
\end{align*}
Inequality (i) uses that $\delta|\xi|\le \delta^\tau$ in the near-momentum regime, while (ii) follows from the uniform boundedness of $\ff_i$ and $\ff_i'$ for $i=1,2$. Inequality (iii) is a direct application of \eqref{eqn:app-sum-M-bd} by taking $s=1$, and (iv) comes from \eqref{eqn:eta-idft-bd-2}. Finally, the last line of inequality (v) is a consequence of (iii). Using bounds (i)-(v), we complete the proof of \eqref{eqn:near-op-bd} for $\widehat{L_1^\delta}(\mu)$. 

The Lipschitz bound for $\widehat{L_1^\delta}(\mu)$ follows from the estimate below, since this term is the only $\mu$-dependent contribution in $\widehat{L_1^\delta}(\mu)$:
\begin{align*}
    &\norm{\widehat{L_1^\delta}(\mu_1) \widehat{\bm{\eta}}_\text{near} - \widehat{L_1^\delta}(\mu_2) \widehat{\bm{\eta}}_\text{near}}_{L^2_\xi(\mathbb{R})} = \norm{\delta \chi(|\xi|\leq \delta^{\tau-1}) \Big(\widetilde{F}_1[\delta \xi;\mathcal{A}\bm{\widetilde{\eta}}_\text{near}(\mu_1,\delta)] -  \widetilde{F}_1[\delta \xi;\mathcal{A}\bm{\widetilde{\eta}}_\text{near}(\mu_2,\delta)]\Big)}\\
    \lesssim \: & \delta^{\frac{1}{2}} \norm{\mathcal{A}\bm{\widetilde{\eta}}_\text{near}(k;\mu_2,\delta) - \mathcal{A}\bm{\widetilde{\eta}}_\text{near}(k;\mu_2,\delta)}_{\LtwoBrill} \lesssim \delta^{\frac{1}{2}} \delta^{1-\tau} \: |\mu_1 - \mu_2| \lesssim \delta^{1-\tau} \: |\mu_1 - \mu_2|,
\end{align*}
where the last line follows from (iv) and \eqref{eqn:A-map-lip-mu}.

Lastly, we prove the bounds \eqref{eqn:MN-bds-main} for the source terms. We begin with \eqref{eqn:M-limit} and observe that
\begin{align*}
    &\left\langle \bm{\widehat{\Psi_0}}(\xi), \widehat{\bm{\mathcal{M}}}(\xi; \delta) \right\rangle_{L^2_\xi(\mathbb{R})} = \left\langle \widehat{\Psi_0^A}(\xi), \widehat{\mathcal{M}}_1(\xi; \delta) \right\rangle_{L^2_\xi(\mathbb{R})} + \left\langle \widehat{\Psi_0^B}(\xi), \widehat{\mathcal{M}}_2(\xi; \delta) \right\rangle_{L^2_\xi(\mathbb{R})}\\
   =\: & \frac{2}{\sqrt{3}} \left\langle \widehat{\Psi_0^A}(\xi), \chi(|\xi|\leq \delta^{\tau-1}) \widehat{\Psi_0^A}(\xi)\right\rangle_{L^2_\xi(\mathbb{R})} + \frac{2}{\sqrt{3}} \left\langle \widehat{\Psi_0^B}(\xi), \chi(|\xi|\leq \delta^{\tau-1}) \widehat{\Psi_0^B} (\xi)\right\rangle_{L^2_\xi(\mathbb{R})}\\
   & +\frac{2}{\sqrt{3}} \left \langle\widehat{\Psi_0^A}(\xi),  \chi(|\xi|\leq \delta^{\tau-1}) \sum_{|m|\geq 1} \widehat{\Psi_0^A}\left(\frac{\delta \xi + \frac{4\pi}{\sqrt{3}}m}{\delta}\right)\right\rangle_{L^2_\xi(\mathbb{R})} \\
   & + \frac{2}{\sqrt{3}} \left \langle\widehat{\Psi_0^B}(\xi),  \chi(|\xi|\leq \delta^{\tau-1}) \sum_{|m|\geq 1} \widehat{\Psi_0^B}\left(\frac{\delta \xi + \frac{4\pi}{\sqrt{3}}m}{\delta}\right)\right\rangle_{L^2_\xi(\mathbb{R})}.
\end{align*}
Using \eqref{eqn:app-sum-M-bd}, we know that the last two lines are arbitrarily small and vanishes as $\delta \rightarrow 0$. Therefore, we obtain the desired limit \eqref{eqn:M-limit} with
\begin{align*}
    &\lim_{\delta \rightarrow 0}\left\langle \bm{\widehat{\Psi_0}}(\xi), \widehat{\bm{\mathcal{M}}}(\xi; \delta) \right\rangle_{L^2_\xi(\mathbb{R})} = \frac{2}{\sqrt{3}} \left( \norm{\Psi_0^A}^2_{L^2(\mathbb{R})} + \norm{\Psi_0^A}^2_{L^2(\mathbb{R})} \right) = 1.
\end{align*}
Similarly, to prove the limit \eqref{eqn:N-limit} for $\widehat{\bm{\mathcal{N}}}(\xi; \mu, \delta)$, we observe that
\begin{align*}
    &\left\langle \bm{\widehat{\Psi_0}}(\xi), \widehat{\bm{\mathcal{N}}}(\xi; \mu, \delta) \right\rangle_{L^2_\xi(\mathbb{R})} = \left\langle \widehat{\Psi_0^A}(\xi), \widehat{\mathcal{N}}_1(\xi; \mu,\delta) \right\rangle_{L^2_\xi(\mathbb{R})} + \left\langle \widehat{\Psi_0^B}(\xi), \widehat{\mathcal{N}}_2(\xi; \mu,\delta) \right\rangle_{L^2_\xi(\mathbb{R})}\\
    =\: & -\left\langle \widehat{\Psi_0^A}(\xi), \chi(|\xi|\leq \delta^{\tau-1}) \widetilde{I}_{1,\text{ind}}[\delta \xi;\bm \Psi_0, \delta]\right\rangle_{L^2_\xi(\mathbb{R})} + \left \langle \widehat{\Psi_0^A}(\xi), \delta \chi(|\xi|\leq \delta^{\tau-1}) \widetilde{F}_1[\delta \xi;\mathcal{B}(\mu,\delta)]\right \rangle_{L^2_\xi(\mathbb{R})}\\
    &-\left\langle \widehat{\Psi_0^B}(\xi), \chi(|\xi|\leq \delta^{\tau-1}) \widetilde{I}_{2,\text{ind}}[\delta \xi; \bm \Psi_0, \delta]\right\rangle_{L^2_\xi(\mathbb{R})} + \left \langle \widehat{\Psi_0^B}(\xi), \delta \chi(|\xi|\leq \delta^{\tau-1}) \widetilde{F}_2[\delta \xi;\mathcal{B}(\mu,\delta)]\right \rangle_{L^2_\xi(\mathbb{R})}.
\end{align*}
Using a similar bound shown in (iv), the inner products related to $\widetilde{F}_i$ are small and vanish as $\delta \rightarrow 0$. Therefore, we have
\begin{align*}
    \lim_{\delta \rightarrow 0} \left\langle \bm{\widehat{\Psi_0}}(\xi), \widehat{\bm{\mathcal{N}}}(\xi; \mu, \delta) \right\rangle_{L^2_\xi(\mathbb{R})} =& -\left\langle \widehat{\Psi_0^A}(\xi), \chi(|\xi|\leq \delta^{\tau-1}) \widetilde{I}_{1,\text{ind}}[\delta \xi;\bm \Psi_0, \delta]\right\rangle_{L^2_\xi(\mathbb{R})}\\
    & -\left\langle \widehat{\Psi_0^B}(\xi), \chi(|\xi|\leq \delta^{\tau-1}) \widetilde{I}_{2,\text{ind}}[\delta \xi;\bm \Psi_0, \delta]\right\rangle_{L^2_\xi(\mathbb{R})}.
\end{align*}
To finish our proof of \eqref{eqn:N-limit}, we shall show that
\begin{subequations}\label{eqn:app-N-lim}
    \begin{align}
    \lim_{\delta \rightarrow 0} \left\langle \widehat{\Psi_0^A}(\xi), \chi(|\xi|\leq \delta^{\tau-1}) \widetilde{I}_{1,\text{ind}}[\delta \xi;\bm \Psi_0, \delta]\right\rangle_{L^2_\xi(\mathbb{R})} &= -\frac{2}{\sqrt{3}} \langle \Psi_0^A, \mathcal{R}_2^A\rangle_{L^2(\mathbb{R})}, \label{eqn:app-N-lim-1}\\
    \lim_{\delta \rightarrow 0} \left\langle \widehat{\Psi_0^B}(\xi), \chi(|\xi|\leq \delta^{\tau-1}) \widetilde{I}_{2,\text{ind}}[\delta \xi;\bm \Psi_0, \delta]\right\rangle_{L^2_\xi(\mathbb{R})} &= -\frac{2}{\sqrt{3}} \langle \Psi_0^B, \mathcal{R}_2^B\rangle_{L^2(\mathbb{R})},\label{eqn:app-N-lim-2}
\end{align}
\end{subequations}
where $\mathcal{R}_2^A, \mathcal{R}_2^B$ are shown in \eqref{eqn:app-rem-delta-2}. The desired limit \eqref{eqn:N-limit} comes directly from \eqref{eqn:app-N-lim} and \eqref{eqn:E2}.

We prove \eqref{eqn:app-N-lim-1} as a representative case; the proof of \eqref{eqn:app-N-lim-2} is analogous. We notice that the Fourier modes with $|m|\geq 1$ in $\widetilde{I}_{1,\text{ind}}[\delta \xi;\bm \Psi_0, \delta]$ are also small and vanish as $\delta \rightarrow 0$. As $\delta \rightarrow 0$, the only two terms in $\widetilde{I}_{1,\mathrm{ind}}[\delta \xi;\bm \Psi_0, \delta]$ in \eqref{eqn:I1-full} that remain nonzero are the following:
\begin{align*}
    \text{(a)}\quad & \frac{2}{\sqrt{3} \delta^2} \Gamma_3[\delta \xi,k_\parallel,\delta] \: \widehat{\Psi_0^B}(\xi) = \frac{2}{\sqrt{3}}\left[\frac{9}{8}k_\parallel^2+\frac{3\sqrt3}{4}k_\parallel\xi+\Big(-\frac{3}{2} e^{i\frac{5}{3}\pi} + \frac{3}{8}\Big)\xi^2\right] + O(\delta),\\
    \text{(b)}\quad & e^{i \frac{\pi}{3}} e^{i \frac{4\pi}{3}} t_1 \sum_{m \in \mathbb{Z}} \Bigg[ \Psi_0^B\left(\frac{\sqrt{3}}{2} \delta (m+2) \right) - \Psi_0^B\left(\frac{\sqrt{3}}{2} \delta m \right)\Bigg] \ff_2\left(\frac{\sqrt{3}}{2} \delta m\right)e^{-ik\frac{\sqrt{3}}{2}m}\\
    =&  e^{i \frac{5\pi}{3}} t_1 \sum_{m \in \mathbb{Z}} \Bigg[ (\sqrt{3} \delta )\partial_{X_1}\Psi_0^B\left(\frac{\sqrt{3}}{2} \delta m\right) + O(\delta^2)\Bigg] \ff_2\left(\frac{\sqrt{3}}{2} \delta m\right)e^{-ik\frac{\sqrt{3}}{2}m}\\
    =& 2 e^{i \frac{5\pi}{3}} t_1 \sum_{m \in \mathbb{Z}} \widehat{\ff_2\partial_{X_1}\Psi_0^B}\left(\frac{\delta \xi + \frac{4\pi}{\sqrt{3}}m}{\delta}\right) + O(\delta).
\end{align*}
Using (a)-(b) and \eqref{eqn:app-rem-delta-2-1}, we obtain the desired limit \eqref{eqn:app-N-lim-1}
\begin{align*}
    &\lim_{\delta \rightarrow 0} \left\langle \widehat{\Psi_0^A}(\xi), \chi(|\xi|\leq \delta^{\tau-1}) \widetilde{I}_{1,\text{ind}}[\delta \xi; \bm \Psi_0, \delta]\right\rangle_{L^2_\xi(\mathbb{R})} \\
    =& \frac{2}{\sqrt{3}} \left\langle \widehat{\Psi_0^A}, \left[\frac{1}{8}k_\parallel^2+\frac{\sqrt3}{4}k_\parallel\xi+\Big(-\frac{3}{2} e^{i\frac{5}{3}\pi} + \frac{3}{8}\Big)\xi^2\right] \widehat{\Psi_0^B}\right\rangle_{L^2_\xi(\mathbb{R})}+ \frac{2}{\sqrt{3}} \left\langle \widehat{\Psi_0^A}, e^{i \frac{5\pi}{3}} t_1  \widehat{\ff_2\partial_{X_1}\Psi_0^B}(\xi)\right\rangle_{L^2_\xi(\mathbb{R})}\\
    =& \frac{2}{\sqrt{3}} \left\langle \Psi_0^A, \left[\frac{1}{8}k_\parallel^2 \Psi_0^B -\frac{\sqrt3i}{4}k_\parallel\partial_{X_1}\Psi_0^B +\Big(\frac{3}{2} e^{i\frac{5}{3}\pi} - \frac{3}{8}\Big) \partial_{X_1}^2 \Psi_0^B\right] \right\rangle_{L^2(\mathbb{R})} \\
    &+ \frac{2}{\sqrt{3}} \left\langle \Psi_0^A, \sqrt{3} e^{i \frac{5\pi}{3}} t_1 \ff_2 \partial_{X_1} \Psi_0^B \right\rangle_{L^2(\mathbb{R})} = -\frac{2}{\sqrt{3}} \langle \Psi_0^A, \mathcal{R}_2^A\rangle_{L^2(\mathbb{R})},
\end{align*}
Thus, we complete our proof of bounds in Lemma \ref{lemma:eta-near-bds}.



\end{document}